\documentclass[12pt]{article}

\usepackage{amssymb,amsfonts,amsthm}
\usepackage{graphicx}
\usepackage{subfigure}
\usepackage{amssymb}
\usepackage{amsthm}
\usepackage{amsmath}
\usepackage{bm}             
\usepackage[mathscr]{eucal}
\usepackage[all,cmtip]{xy}

\usepackage{hyperref}
\hypersetup{
 colorlinks=true,
 linkcolor=blue,
 citecolor=blue,
 filecolor=blue,
 urlcolor=blue
 }

\usepackage{cancel}

\usepackage{psfrag}

\usepackage{wasysym}

\newcommand{\bom}{\bm{\omega}}

\newcommand{\diag}{\mathop{\mathrm{diag}}}

\newcommand{\textfrac}[2]{{\textstyle \frac{#1}{#2}}}

\theoremstyle{plain}
\newtheorem{theorem}{Theorem}[section]

\newtheorem{lemma}[theorem]{Lemma}
\newtheorem{proposition}[theorem]{Proposition}
\newtheorem{conjecture}[theorem]{Conjecture}

\theoremstyle{remark}

\newtheorem{remark}{Remark}[section]

\def\pmb#1{\setbox0=\hbox{$#1$}%
  \kern-.025em\copy0\kern-\wd0
  \kern.05em\copy0\kern-\wd0
  \kern-.025em\raise.0433em\box0}
\def\pmbs#1{\setbox0=\hbox{$\scriptstyle #1$}%
  \kern-.0175em\copy0\kern-\wd0
  \kern.035em\copy0\kern-\wd0
  \kern-.0175em\raise.0303em\box0}

\newcommand{\sfrac}[2]{{\textstyle{\frac{#1}{#2}}}}

\usepackage{tikz}
\usepackage[format=plain,labelfont=bf,font=small]{caption}
\usepackage{xcolor}
\usepackage{float}
\usetikzlibrary{arrows, arrows.meta, calc, intersections, decorations.markings}
\tikzset{
decoration={markings,mark=at position 0.67 with {\arrow[thick,color=gray]{latex}}}
}

\title{\huge{\textbf{Bifurcations and Chaos in Ho\v{r}ava-Lifshitz Cosmology}}}

\begin{document}

\author{
 \\
{~}\\
Juliette Hell*, Phillipo Lappicy** and Claes Uggla***\\
}

\date{}
\maketitle
\thispagestyle{empty}

\vfill

$\ast$\\
Institut f\"ur Mathematik, Freie Universit\"at Berlin\\
Arnimallee 3 - 14195 Berlin, Germany\\
{\tt blanca@math.fu-berlin.de}\\
$\ast \ast$\\
Instituto Superior T\'ecnico, Universidade de Lisboa\\
Av. Rovisco Pais, 1049-001 Lisboa, Portugal\\
$\ast \ast$\\
ICMC, Universidade de S\~ao Paulo\\
Av. trabalhador s\~ao-carlense, S\~ao Carlos, Brazil\\
{\tt lappicy@hotmail.com}\\
$\ast \ast \ast$\\
Department of Physics, Karlstad University\\
S-65188 Karlstad, Sweden\\
{\tt claes.uggla@kau.se}\\

\pagebreak
\setcounter{page}{1}

\begin{abstract}
The nature of generic spacelike singularities in general relativity
is connected with first principles, notably Lorentzian causal structure,
scale invariance and general covariance. To bring a new perspective on how
these principles affect generic spacelike singularities, we consider the
initial singularity in spatially homogeneous Bianchi type VIII and IX vacuum
models in Ho\v{r}ava-Lifshitz gravity, where relativistic first principles
are replaced with anisotropic scalings of Lifshitz type. Within this class
of models, General Relativity is shown to be a bifurcation where chaos
becomes generic. To describe the chaotic features of generic singularities
in Ho\v{r}ava-Lifshitz cosmology, we introduce symbolic dynamics within
Cantor sets and iterated function systems.

\end{abstract}


\section{Introduction}

The last couple of decades have seen considerable progress in our
understanding of generic spacelike singularities in General Relativity (GR).
In particular it has been shown that there are connections between the nature
of such singularities and with three of the foundational first principles of
GR: (i) Lorentzian causal structure, (ii) (conformal) scale invariance,
and (iii) general covariance, i.e., spacetime diffeomorphism invariance,
see e.g. the reviews~\cite{ugg13a,ugg13b} and references therein.

Heuristic arguments by
Belinski\v{\i}, Khalatnikov and Lifshitz (BKL)~\cite{lk63,bkl70,bkl82}
suggest that generic spacelike singularities in GR are \emph{vacuum dominated}
for a broad range of matter sources, i.e., generically such sources
asymptotically become test fields because gravity asymptotically generates
more gravity than matter. For simplicity we will therefore only
consider vacuum models.

The importance of \emph{Lorentzian causal structure} for the nature of generic
spacelike singularities in GR is connected with the locality conjecture
of BKL~\cite{lk63,bkl70,bkl82}. This conjecture states that the asymptotic evolution
toward a generic spacelike singularity in inhomogeneous cosmology is local,
in the sense that each spatial point evolves toward the singularity
independently of its neighbors as a spatially homogeneous model.
By reformulating the Einstein field equations in GR, using the so-called conformally
Hubble-normalized orthonormal frame dynamical systems approach,
the BKL locality conjecture was made more precise
in~\cite{ugg13a,ugg13b,uggetal03,gar04,andetal05,rohugg05,limetal06,heietal09}.
In this formulation, there exists an invariant `local boundary set'
where the partial differential equations (PDEs) of inhomogeneous
cosmology reduce to the ordinary differential equations (ODEs) of
spatially homogeneous cosmology at each spatial point. Moreover, on
the local boundary set, there exists an invariant subset identical
to the attractor of the corresponding ODEs of spatially homogeneous cosmology,
for each spatial point. In this approach, loosely speaking, the BKL locality
conjecture amounts to that the invariant subset on the local boundary
corresponding to the ODE attractor, for each spatial point, form a local
PDE attractor, which describes the detailed nature of generic spacelike
singularities in inhomogeneous cosmology.
Presumably, a necessary condition for such asymptotic local evolution
is asymptotic silence, i.e., that the extreme gravity in the
vicinity of a generic spacelike singularity results in particle
horizons that shrink to zero size toward the
singularity~\cite{ugg13a,ugg13b,uggetal03,rohugg05,limetal06,heietal09}.\footnote{Recent
results indicate that this is not the whole story. There is also a connection
between asymptotic silence, generic spacelike singularities, and infinitely recurring
oscillating inhomogeneous spikes, described by certain inhomogeneous
solutions~\cite{andetal05,lim08,heietal12,heiugg13,lim15,lim21,lim22}.}
If this is the case, the nature of generic spacelike
singularities is connected with asymptotic Lorentzian causal structure induced
by extreme gravity and certain \emph{spatially homogeneous models}.

To describe generic spacelike singularities, it is therefore presumably
essential to understand the properties of spatially homogeneous models,
of which there are two categories: the Bianchi models
and the spherically symmetric Kantowski-Sachs models,
where the latter are too special to be of relevance for
generic singularities. The Bianchi models are divided into class~A and class~B.
In contrast to the general class~B models, the class~A models admit a Hamiltionan
formulation and have a simpler hierarchical structure. We will therefore
henceforth restrict considerations to the class~A Bianchi models.
Because of the BKL locality conjecture, this is further motivated by that
the most general models within this class, the Bianchi type VIII and IX models,
are believed to contain some of the key elements needed to describe generic
spacelike singularities.

The class~A Bianchi models have three-dimensional symmetry
groups, which act simply transitively on the spatially homogeneous slices.
These models thereby admit a symmetry-adapted
spatial (left-invariant) co-frame $\{{\bom}^1,{\bom}^2,{\bom}^3\}$,
such that
\begin{equation}\label{structconst}
d{\bom}^1  =  -{n}_1 \,
{\bom}^2\wedge {\bom}^3\:,\quad d{\bom}^2  =
-{n}_2 \, {\bom}^3\wedge {\bom}^1\:,\quad
d{\bom}^3  = -{n}_3 \, {\bom}^1\wedge
{\bom}^2\:,
\end{equation}
where the structure constants $n_1$, $n_2$, $n_3$ determine the Lie algebras
of the various class~A Bianchi models, defined in Table~\ref{intro:classAmodels},
see also e.g.~\cite{waiell97}:
\begin{table}[H]
\begin{center}
\begin{tabular}{|c|ccc|}
\hline Bianchi type & ${n}_\alpha$ &  ${n}_\beta$ & ${n}_\gamma$ \\ \hline
$\mathrm{IX}$ & $+$& $+$& $+$ \\
$\mathrm{VIII}$ & $-$& $+$& $+$ \\
$\mathrm{VII}_0$ & $0$& $+$ & $+$ \\
$\mathrm{VI}_0$ & $0$ & $-$ & $+$ \\
II & $0$ & $0$ & $+$ \\
I & $0$ & $0$ & $0$ \\\hline
\end{tabular}
\caption{The class~A Bianchi types are characterized by the zeroes and relative
signs of the structure constants $(n_\alpha, n_\beta, n_\gamma)$,
where $(\alpha\beta\gamma)$ is a permutation of
$(123)$. There are equivalent representations associated with
an overall change of sign of the constants,
e.g., another Bianchi type IX representation is $(-,-,-)$. It is also
possible to scale the constants, e.g., in type~IX we can set
${n}_1 = {n}_2 = {n}_3 =1$.} \label{intro:classAmodels}
\end{center}
\end{table}\vspace{-0.5cm}

The class~A Bianchi models form a hierarchical structure, where more special
models are obtained from more general ones by performing Lie
contractions, i.e., by setting structure constants to zero,
which results in the Lie contraction diagram given in Figure~\ref{FIG:hierarchy}.
\begin{figure}[H]\centering
\begin{tikzpicture}
\node (IX) at (-2, 4.5) {\boxed{\text{\footnotesize{Type IX}}}};
\node (VIII) at (2, 4.5) {\boxed{\text{\footnotesize{Type VIII}}}};
\node (VII) at (-2, 3) {\boxed{\text{\footnotesize{Type $\mathrm{VII}_0$}}}};
\node (VI) at (2, 3) {\boxed{\text{\footnotesize{Type $\mathrm{VI}_0$}}}};
\node (II) at (0, 1.5) {\boxed{\text{\footnotesize{Type $\mathrm{II}$}}}};
\node (I) at (0, 0) {\boxed{\text{\footnotesize{Type $\mathrm{I}$}}}};

\draw[->] (IX) -- (VII);\draw[->] (VIII) -- (VI);\draw[->] (VIII) -- (VII);
\draw[->] (VII) -- (II);\draw[->] (VI) -- (II);\draw[->] (II) -- (I);
\end{tikzpicture}
\captionof{figure}{The class~A Bianchi Lie contraction hierarchy.}\label{FIG:hierarchy}
\end{figure}

The field equations of all vacuum GR models, and thus also the class~A Bianchi models,
are conformally \emph{scale-invariant} and thereby admit a scale invariance symmetry.
However, \emph{general covariance} (i.e., diffeomorphism invariance), which also
results in symmetries of the GR vacuum field equations, is broken by the
preferred spatial homogeneous foliations in Bianchi cosmology. The symmetries of
the Einstein vacuum field equations generated by the principle of general covariance
reduce to those generated by the spatial diffeomorphisms that are
compatible with the Bianchi symmetry groups, which are locally characterised by their Lie
algebras defined in Table~\ref{intro:classAmodels}.\footnote{In the present paper,
considerations are spatially local. For an investigation about the role of spatial
topology in a Hamiltonian description of Bianchi models, see~\cite{ashsam91}.}
Furthermore, the symmetry generating spatial diffeomorphisms correspond
to the automorphisms of the Lie algebras, i.e., the linear transformations
of the symmetry adapted spatial frame that leave the associated structure constants
unchanged~\cite{jan79,jan01}.

As discussed in Appendix~\ref{app:dom}, the automorphism groups can be
used to diagonalize the vacuum class~A Bianchi models, which then leaves a diagonal
automorphism group for each model.
As described in Table~\ref{intro:classAmodels} and Figure~\ref{FIG:hierarchy},
the class~A Bianchi types are grouped into a hierarchy defined
by the number of non-zero structure constants: Bianchi types IX and VIII
have three; types $\mathrm{VII}_0$ and $\mathrm{VI}_0$ have two; type II
has one; Bianchi type I has none. Each structure constant that is zero
results in a diagonal automorphism and an associated symmetry, see e.g.,
\cite{jan01,rosetal90a,rosetal90b,heiugg10}, and
references therein. Due to the increasing
number of automorphisms as one goes down in the hierarchy by setting
structure constants to zero (i.e., by performing Lie contractions), a new
symmetry in the Einstein equations appears at each level of the hierarchy.
At the levels below Bianchi type IX and VIII in the class~A Bianchi symmetry
hierarchy, the scale and automorphism groups combine into scale-automorphism
groups, which yields a symmetry hierarchy of the class~A Einstein vacuum field
equations~\cite{rosetal90a,rosetal90b,heiugg10}.

The above hierarchical features are naturally incorporated into the conformally
Hubble-normalized orthonormal frame dynamical systems approach to
Einstein's vacuum field equations. In this approach, each class~A Bianchi
model yields an invariant set of the ODEs, denoted by Bianchi type I, II,
$\mathrm{VI}_0$, $\mathrm{VII}_0$, VIII and IX, respectively.
Moreover, the class~A Bianchi Lie contraction hierarchy results in that
each model in the hierarchy form an invariant boundary set of
the models at the next higher level according to Figure~\ref{FIG:hierarchy}.
Thus the invariant Bianchi type I set, which constitute a circle of fixed points,
the Kasner circle, is the boundary of three physically equivalent invariant Bianchi
type II sets, where each type II set forms a hemisphere filled with heteroclinic orbits
(i.e. solution trajectories) between different points of said circle, see, e.g., \cite{waiell97}.
Apart from these kinematical ramifications, the scale-automorphism symmetry hierarchy also have
dynamical consequences. At the higher levels of the Lie contraction hierarchy the scale and 
scale-automorphism symmetries generate monotone functions, which limit the asymptotic dynamics 
in a hierarchical manner: asymptotically the dynamics toward the initial singularity is pushed in
the state space at the top of the hierarchy (the Bianchi type IX and VIII models)
toward the bottom of the hierarchy, the Bianchi type II and I models, where the two
latter are completely determined by the scale-automorphism symmetries~\cite{heiugg10}.

The scale-automorphism symmetries are complemented by discrete symmetries.
Together these symmetries limit but do not completely determine the
asymptotic dynamics of Bianchi types VIII and IX. Nevertheless, the (past) attractor
in these models is expected to reside on the union of the Bianchi type I and II
boundary sets. Furthermore, the concatenation of heteroclinic type II orbits
yields heteroclinic chains, which are expected to be generically asymptotically
shadowed toward the initial singularity by the type VIII and IX orbits,
see~\cite{bre16,dut19} and references therein. The type II
heteroclinic orbits induce a discrete map that acts on the fixed points of
the Kasner circle, called the Mixmaster map. This map exhibits chaotic properties,
and it is because of this feature GR is said to be chaotic~\cite{khaetal85}.
Note that the above statements are partially supported by several
theorems~\cite{rin00,rin01,heiugg09b,beg10,lieetal10,reitru10,lieetal12,bre16,dut19}.

There thereby exist intricate connections in GR between the nature of generic
spacelike singularities, asymptotic Lorentzian causal structure, spatial homogeneous models,
and hierarchically induced scale and diffeomorphism symmetries. To bring a new
perspective on GR, we therefore ask:
\emph{What happens if the first principles that lead to the
structure of generic spacelike singularities in GR are gradually modified?}

To investigate this question we have to go beyond GR and it is
natural to do so by considering Ho\v{r}ava-Lifshitz (HL) theories.
These theories are based on a preferred foliation of spacetime that
breaks full spacetime diffeomorphism invariance and introduce anisotropic
Lifshitz type scalings between space and time, in analogy with condensed
matter physics~\cite{hor09a,hor09b,muk10}. There are two classes of
HL theories: `projectable' theories for which the lapse only depends
on time, which naturally encompasses spatially homogenous cosmology,
and `nonprojectable' theories with a lapse depending on time
and space, which was shown to result in dynamical inconsistencies
in~\cite{henetal10}.

HL gravity is a gauge theory formulated in terms of a lapse $N$ and a
shift vector $N^i$, which serve as Lagrange multipliers for the constraints
in a Hamiltonian context, and a three-dimensional Riemannian metric $g_{ij}$
on the slices of the preferred foliation. In GR, these objects arise from a 3+1
decomposition of a 4-metric according to,
\begin{equation}\label{genmetric}
\mathbf{g} = -N^2dt\otimes
dt + g_{ij}(dx^i + N^idt)\otimes (dx^j + N^jdt).
\end{equation}
In suitable units and scalings, the dynamics of HL vacuum gravity is governed by the
action
\begin{subequations}
\begin{equation}\label{action}
S = \int N\sqrt{ \det g_{ij}}({\cal T} - {\cal V}) dtd^3x,
\end{equation}
where ${\cal T}$ and ${\cal V}$ are given by
\begin{align}
{\cal T} &= K_{ij}K^{ij} - \lambda (K^k\!_k)^2,\label{kin}\\
{\cal V} &= k_1 R + k_2 R^2 + k_3 R^i\!_jR^j\!_i +
k_4 R^i\!_jC^j\!_i + k_5 C^i\!_jC^j\!_i + k_6 R^3+ \dots\, .\label{calV}
\end{align}
\end{subequations}
Here $K_{ij}$ is the extrinsic curvature, $R$ and $R_{ij}$ are the scalar curvature
and Ricci tensor (of the spatial metric $g_{ij}$), respectively,
$C_{ij}$ is the Cotton-York tensor~\cite{hor09b,henetal10}, 
while the constants $\lambda, k_1,\dots k_6$ are real
parameters. Repeated indices are summed over according to Einstein's
summation convention.\footnote{To study the differences that
arise from imposing spacetime or only spatial diffeomorphism invariance on a
theory, it is illuminating to even go beyond HL theories, as
discussed in~\cite{caretal10}.}

Full spacetime diffeomorphism invariance in GR fixes $\lambda=1$
uniquely and set all parameters of ${\cal V}$ in~\eqref{calV} to
zero, except $k_1=-1$ (i.e., ${\cal V} = -R$), see~\cite{hor09a,hor09b}.
Thus GR is a special case among the HL models.
The introduction of $\lambda$ changes the scaling properties of the field
equations, as does the introduction of additional curvature terms. Since some of
the curvature terms have different scaling properties, sums of such terms in ${\cal V}$
result in that the field equations no longer are scale-invariant. Nevertheless,
as heuristically argued in Appendix~\ref{app:dom}, when there is a sum of curvature terms
in the case of the HL class~A Bianchi models, there is an `asymptotically dominant'
curvature term toward the initial singularity. Since each curvature term exhibits a certain scaling
property, this implies that the corresponding field equations are asymptotically scale-invariant.
Although their scaling properties differ, the HL and GR class~A Bianchi models share the same
Lie contraction hierarchy, see Table~\ref{FIG:hierarchy}, and consequently the same automorphism structure.
Combining the (asymptotic) scale and automorphism symmetry groups for the different levels
of the HL hierarchy \emph{continuously deforms} the corresponding scale-automorphism groups in GR.
This in turn affects the nature of the generic initial Bianchi type VIII and IX singularity.

Although a significant part of the previous literature on the dynamics of cosmological
HL models is about isotropic matter models, see e.g.~\cite{cal09,kirkof09,sot11,leopal19},
the present work is by no means the first dealing with the anisotropic vacuum HL
class~A Bianchi models, see e.g.~\cite{baketal09,myuetal10a,myuetal10b,baketal10,misetal11,giakam17}.
The present paper, however, identifies and ties mathematical structures to physical first
principles and introduces new mathematical tools, which yield rigorous results about
discrete dynamics induced by heteroclinic chains.

Although established as an interesting research field in
its own right, the present primary purpose of HL gravity is that
these models situate GR in a broader context that makes it
possible to study how a change of first principles affect
generic spacelike singularities. As we will see, this results
in a new perspective, which generates new ideas and tools for how
to study generic singularities not only in HL gravity but also
in GR. This, however, only requires retaining the parameter
$\lambda$ in~\eqref{kin} and the vacuum GR potential ${\cal V} = -R$,
which yield the so-called $\lambda$-$R$ models~\cite{giukie94,belres12,lolpir14}.
For simplicity, we therefore restrict considerations in the main part
of the paper to the vacuum $\lambda$-$R$ class~A Bianchi models.
Nevertheless, we perform a heuristic analysis of the HL models
in Appendix~\ref{app:dom}, which indicates that the generic asymptotic
dynamics toward the singularity 
for a large class of vacuum HL class~A Bianchi models formally coincide with
that of the vacuum $\lambda$-$R$ class~A Bianchi models.
The results in the main part of the paper for the
$\lambda$-$R$ models are thereby also relevant for a broad class of HL models.

In Appendix~\ref{app:dom}, the Hamiltonian formulation for the spatially
homogenous vacuum $\lambda$-$R$ class~A Bianchi models is used  to obtain
the following evolution equations,
\begin{subequations}\label{intro_dynsyslambdaR}
\begin{align}
\Sigma_\alpha^\prime &= 4v(1-\Sigma^2)\Sigma_\alpha + {\cal S}_\alpha,\\
N_\alpha^\prime &= -2(2v\Sigma^2 + \Sigma_\alpha)N_\alpha, \label{intro_dynsyslambdaR_N}
\end{align}
%
for $\alpha = 1,2,3$, and the constraints,
\begin{align}
0 &= 1 - \Sigma^2 - \Omega_k, \label{intro_cons1}\\
0 &= \Sigma_1 + \Sigma_2 + \Sigma_3,\label{intro_cons2}
\end{align}
\end{subequations}
where
\begin{subequations}
\begin{align}
\Sigma^2 &:= \frac16\left(\Sigma_1^2 + \Sigma_2^2 + \Sigma_3^2\right),\label{Sigma}\\
\Omega_k &:= N_1^2 + N_2^2 + N_3^2 - 2N_1N_2 - 2N_2N_3 - 2N_3N_1, \label{Omega_k}\\
{\cal S}_\alpha &:= -4[(N_\beta - N_\gamma)^2 - N_\alpha(2N_\alpha - N_\beta - N_\gamma)].
\end{align}
\end{subequations}

Here $(\alpha\beta\gamma)$ is a permutation of $(123)$. A ${}^\prime$ denotes
the derivative with respect to the chosen time variable, $\tau_-$,
defined in Appendix~\ref{app:dom}, which is in the opposite direction of
physical time. Since we are considering expanding models, $\tau_-\rightarrow\infty$
describes the dynamics toward the initial singularity.
Throughout, $\alpha$-limits ($\tau_- \rightarrow - \infty$),
$\omega$-limits ($\tau_-\rightarrow\infty$), and stability issues
refer to $\tau_-$. The parameter $v$ is related to $\lambda$ according to
\begin{equation}\label{vdef}
v := \frac{1}{\sqrt{2(3\lambda - 1)}}.
\end{equation}
The GR class~A Bianchi models have $\lambda = 1$ and hence $v=1/2$.
Since we are primarily interested in continuous deformations of GR with $v=1/2$,
we restrict $v$ to $v\in (0,1)$, although the bifurcation values $v=0$ and
$v=1$ are briefly mentioned in the next section and in Appendix~\ref{app:dom}.

The equations~\eqref{intro_dynsyslambdaR} are
invariant under permutations of the axes, i.e.,
they are invariant under the transformation
\begin{equation}\label{permSYM}
(\Sigma_1,\Sigma_2,\Sigma_3,N_1,N_2,N_3)\mapsto
(\Sigma_\alpha,\Sigma_\beta,\Sigma_\gamma,N_\alpha,N_\beta,N_\gamma),
\end{equation}
where $(\alpha\beta\gamma)$ is a permutation of $(123)$, i.e., $(\alpha\beta\gamma)\in \mathrm{S}_3$.

As defined in Appendix~\ref{app:dom}, the variables $N_\alpha$ are equal to the
structure constants $n_\alpha$ multiplied with positive time dependent functions.
Thus there is a one-to-one correspondence between the zeroes and signs of
$n_\alpha$ and $N_\alpha$, as seen by a comparison of Tables~\ref{intro:classAmodels}
and~\ref{BianchiInvSets}.
\begin{table}[H]
\begin{center}
\begin{tabular}{|c|ccc|c|c|}
\hline Bianchi type & ${N}_\alpha$ &  ${N}_\beta$ & ${N}_\gamma$ & Dim & Scale-automorphism induced dynamics \\ \hline
$\mathrm{IX}$ & $+$& $+$& $+$ & $4$ & One monotone function \\
$\mathrm{VIII}$ & $-$& $+$& $+$ & $4$ & One monotone function \\
$\mathrm{VII}_0$ & $0$& $+$ & $+$ & $3$ & Two monotone functions\\
$\mathrm{VI}_0$ & $0$ & $-$ & $+$ & $3$ & Two monotone functions\\
II & $0$ & $0$ & $+$ & $2$ & Hemispheres of heteroclinic orbits \\
I & $0$ & $0$ & $0$ & $1$ & Kasner circle of fixed points\\
\hline
\end{tabular}
\caption{The invariant class~A Bianchi sets of~\eqref{intro_dynsyslambdaR},
characterized by different signs and zeroes of the variables $(N_\alpha, N_\beta,
N_\gamma)$, where $(\alpha\beta\gamma)$ is a permutation of
$(123)$. 
Dim denotes the dimension of the physical state space satisfying the
constraints~\eqref{intro_cons1} and~\eqref{intro_cons2}.
The scale-automorphism group induces a dynamical structure for each Bianchi type, derived in Appendix~\ref{app:heterosym}.
} \label{BianchiInvSets}
\end{center}
\end{table}\vspace{-0.5cm}
This in turn results in a correspondence between Bianchi
types and invariant sets in~\eqref{intro_dynsyslambdaR}. Thus $N_1=N_2=N_3=0$
leads to the invariant Bianchi type I set, which yields a circle of fixed points,
called the \emph{Kasner circle}\footnote{In~\cite{giukie94,belres12,lolpir14}
it was discussed if $\lambda$-$R$ gravity and GR were equivalent in the asymptotically
spatially flat case with ultra-local dynamics, i.e., locally Bianchi type I. This is
supported in the present work by the common description of the Bianchi type I set as the Kasner
circle $\mathrm{K}^{\ocircle}$. However, as we shall see,
stability of $\mathrm{K}^{\ocircle}$ varies with $v\in (0,1)$.
}, denoted by $\mathrm{K}^{\ocircle}$.
There are three invariant Bianchi type II sets, 
obtained by a single non-zero $N_\alpha$ (and thus a non-zero $n_\alpha$)
while the other two variables $N_\beta$ and $N_\gamma$ are zero
(which corresponds to $n_\beta=n_\gamma=0$), where
$(\alpha\beta\gamma)$ is a permutation of $(123)$. On each Bianchi type II set, the solutions will be shown to be
heteroclinic orbits connecting different fixed points on $\mathrm{K}^{\ocircle}$.
Bianchi type $\mathrm{VI}_0$ has two non-zero variables $N_\alpha$ with opposite signs,
whereas type $\mathrm{VII}_0$ has two non-zero variables $N_\alpha$ with the same sign.
The Bianchi type VIII models have three non-zero variables $N_\alpha$ where two of them have
an opposite signs compared to the third, while the Bianchi type IX models
are described by three non-zero variables $N_\alpha$
with the same sign, see Table~\ref{BianchiInvSets}.

Table~\ref{BianchiInvSets} also indicates the dynamical structures
induced by the scale-automorphism group, derived in Appendix~\ref{app:heterosym},
which is  what remains of the first principles of scale and spatial diffeomorphism
invariance in the $\lambda$-$R$ class~A Bianchi models. As will be
seen, monotone functions push the dynamics as $\tau_-\rightarrow\infty$
from the invariant sets at the top of the class~A Bianchi hierarchy
to those at the bottom, in a similar manner as in GR. Moreover, heuristic reasoning
in Appendix~\ref{app:dom} suggests that the asymptotic generic dynamics, as
$\tau_- \rightarrow \infty$, of Bianchi type VIII and IX, described
by~\eqref{intro_dynsyslambdaR}, reside on the union of the invariant Bianchi
type I and II sets, as in GR. The generic asymptotic dynamics is therefore expected to
be described by heteroclinic chains obtained by concatenation of heteroclinic orbits
of the three different type II sets, where the $\omega$-limit of one heteroclinic
orbit in one type II set is the $\alpha$-limit of a subsequent heteroclinic orbit
in another type II set. Note that recent asymptotic proofs in GR exploits
the Bianchi type II heteroclinic chains. From this perspective, an analysis of the
Bianchi type I and II heteroclinic structure is therefore a natural first step in
the asymptotic analysis of the vacuum $\lambda$-$R$ class~A Bianchi models.

To investigate the $\lambda$-$R$ Bianchi type I and II heteroclinic
structure, note that the Bianchi type II sets give rise to the
\emph{Kasner circle map} $\mathcal{K}:\mathrm{K}^{\ocircle}\to \mathrm{K}^{\ocircle}$,
which maps the $\alpha$-limit to the $\omega$-limit of each heteroclinic orbit of type II.
The properties of $\mathcal{K}$, which depend on $v$, give a discrete description of the
properties of the $\lambda$-$R$ type II heteroclinic chains, and thus the expected
generic asymptotic continuous dynamics.

The parameter $v\in (0,1)$ in equation~\eqref{intro_dynsyslambdaR}
situates GR in a broader context. In particular, it will be shown that the
GR value $v=1/2$ corresponds to a bifurcation. More precisely, the case
$v=1/2$, referred to as the `critical case', corresponds to a transition
from a situation without stable fixed points in $\mathrm{K}^\ocircle$
(the subcritical case, $v\in(0,1/2)$) to one with stable fixed points
(the supercritical case, $v\in(1/2,1)$). The existence of stable fixed points
in the supercritical case might tempt someone to conclude that all points in
$\mathrm{K}^\ocircle$ end at one of them by the discrete dynamics of the
Kasner circle map $\mathcal{K}$, yielding finite Bianchi type II heteroclinic
chains, which would prevent asymptotic chaos. However, this is not the case:
there remains a Cantor set associated with infinite Bianchi type II heteroclinic
chains with chaotic dynamics. The critical GR case therefore
represents a transition from non-generic to generic chaos,
and may also exemplify an `attractor crisis', an issue discussed
in~\cite{grebogi82,grebogi83}. More precisely, we will show:
\begin{theorem}
{\bf General relativity $(v=1/2)$ is a bifurcation point as follows: }
\begin{enumerate}
\item[(i)] $v\in(1/2,1)$: The set of points in $\mathrm{K}^{\ocircle}$
associated with infinite Bianchi type II heteroclinic chains is a Cantor set $C$ of
measure zero. Moreover, the Kasner circle map
$\mathcal{K}$ is chaotic in the invariant set $C$.
\item[(ii)] $v=1/2$: The set of points in $\mathrm{K}^{\ocircle}$
associated with infinite Bianchi type II heteroclinic chains 
has full measure.
Moreover, $\mathcal{K}$ is generically chaotic.
\item[(iii)] $v\in(0,1/2)$: All points in $\mathrm{K}^{\ocircle}$ are associated
with infinite Bianchi type II heteroclinic chains. Moreover, the multivalued map
$\mathcal{K}$ is chaotic.
\end{enumerate}
\end{theorem}
Item $(i)$ is proved in Theorems~\ref{CantorTh} and~\ref{chaossub}, which include
bounds on the Hausdorff dimension of $C$.
For an iterative construction of the set $C$, 
see Figure \ref{fig:intro_Cantor}. Item $(ii)$ 
was previously proved in~\cite{bkl70,khaetal85}, see
also~\cite{ugg13a,ugg13b} and references therein.
Item $(iii)$ is shown in Lemma~\ref{SupercritHets} and in this case,
for which the Kasner circle map $\mathcal{K}$ is multivalued,
we conjecture that $\mathcal{K}$ is chaotic on the whole circle $\mathrm{K}^{\ocircle}$, 
which has been partially confirmed in \cite{LappicyDaniel}.
To obtain our results, we use symbolic dynamics, not previously used in GR,
which results in a new description of chaos for generic spacelike singularities.

%
%

The outline of the paper is as follows. Section~\ref{sec:dynsysanalysis}
describes the building blocks for the heteroclinic structure,
the Bianchi type I and II sets, which yield the Kasner circle $\mathrm{K}^{\ocircle}$
and the Kasner circle map
$\mathcal{K}:\mathrm{K}^{\ocircle}\to \mathrm{K}^{\ocircle}$. We
also identify the three dynamically distinct regimes, supercritical, critical,
and subcritical. In the next three sections we focus on the concatenation
of Bianchi type II orbits into heteroclinic chains through iterates of
the Kasner circle map $\mathcal{K}$, and describe associated chaotic aspects.
Section~\ref{sec:critical} sketches known results in the critical GR case.
Section~\ref{sec:superINF} treats the supercritical case using
symbolic dynamics. 
Section~\ref{sec:sub} explores the subcritical case
using iterated function systems. Then Section~\ref{sec:firstprinciples}, primarily,
contains proofs about the asymptotic dynamics for the $\lambda$-$R$ Bianchi type
$\mathrm{VI}_0$ and $\mathrm{VII}_0$ models. 
The main part of the paper is concluded
with Section~\ref{sec:conjectures} which contains dynamical asymptotic
conjectures for the $\lambda$-$R$ Bianchi type VIII and IX models
(and thereby implicitly also for more general HL models).

Appendix~\ref{app:dom} contains a derivation of equation~\eqref{intro_dynsyslambdaR}
and the associated HL equations. It also provides a heuristic analysis of both
the $\lambda$-$R$ and HL Bianchi models, which suggests that their generic asymptotic
dynamics toward the singularity is associated with the Bianchi type I and II
heteroclinic structure, described in the main part of the paper.
In Appendix~\ref{app:heterosym}, the scale-automorphism groups at
each level of the class~A Bianchi Lie contraction hierarchy of
the $\lambda$-$R$ and HL Bianchi models is used to derive monotone
functions and conserved quantities, thereby tying the nature of generic
singularities in GR, $\lambda$-$R$ and HL gravity to physical first principles.
Finally, Appendix~\ref{app:unifying} contains a unified symbolic treatment
of the chaotic regime in the supercritical and critical cases.



\section{Bianchi types I and II}\label{sec:dynsysanalysis}
In this section we describe the Bianchi type I set, i.e.,
the Kasner circle of fixed points, $\mathrm{K}^{\ocircle}$, its
stability features, and the three Bianchi type II sets, which consist
of heteroclinic orbits between fixed points in the set
$\mathrm{K}^{\ocircle}$, thereby yielding the Kasner circle map
$\mathcal{K}:\mathrm{K}^{\ocircle}\to \mathrm{K}^{\ocircle}$.
The heteroclinic orbits of the different type II sets
can subsequently be concatenated to heteroclinic chains
on the Bianchi type I and II boundary sets of Bianchi type VIII and IX;
for the GR case,
see e.g.~\cite{waiell97,heiugg09a,ugg13a,bre16,dut19}.
To illustrate concatenation, we explicitly
construct heteroclinic cycles/chains with period 3 when $v \in [0,1]$.


\subsection{Bianchi type I}

The Bianchi type I set is determined by $N_1=N_2=N_3=0$, which according to
equation~\eqref{intro_dynsyslambdaR} results in the \emph{Kasner circle} 
of fixed points:
\begin{equation}\label{KasnerCircdef}
\mathrm{K}^{\ocircle} := \left\{ (\Sigma_1,\Sigma_2,\Sigma_3,0,0,0)\in \mathbb{R}^6 \Bigm|
\begin{array}{c}
\qquad\quad\, 1-\Sigma^2 = 0, \\
\Sigma_1 + \Sigma_2 + \Sigma_3 = 0
\end{array}
\right\}.
\end{equation}

There are three exceptional points in the set $\mathrm{K}^{\ocircle}$ called the \emph{Taub points},
since they correspond to the Taub representation of Minkowski spacetime in GR, see~\cite{Taub51}.
They are characterized by $(\Sigma_1, \Sigma_2, \Sigma_3)$ as follows:
\begin{equation} \label{Tpoints}
\mathrm{T}_1 := (2,-1,-1), \qquad \mathrm{T}_2 := (-1,2,-1), \qquad \mathrm{T}_3 := (-1,-1,2),
\end{equation}
where $\mathrm{T}_\alpha$, $\alpha = 1,2,3$, is the point in the set
$\mathrm{K}^\ocircle$ where $\Sigma_\alpha$ attains its maximum value
$2$, see Figure~\ref{FIG:BIF}.

The parameter $v$ plays an important role in the dynamics of the
variables $N_\alpha$, $\alpha=1,2,3$, where a bifurcation occurs
at $v=1/2$. This can be seen from the linearization at $\mathrm{T}_{1}$
in~\eqref{intro_dynsyslambdaR}:
\begin{subequations}\label{linTaub}
\begin{alignat}{7}
& &\qquad &&
& & \boxed{\text{\footnotesize{$v< 1/2$}}} &\quad
& \boxed{\text{\footnotesize{$v= 1/2$}}} &\quad
& \boxed{\text{\footnotesize{$v> 1/2$}}} \nonumber\\
N_1' &= -(2v + 2) N_1;&\qquad &-(2v + 2)& \quad
& & <0 &\quad
& < 0 &\quad
& <0, \\
N_2'&= -(2v - 1)N_2; &\qquad &-(2v - 1)& \quad
& &>0 &\quad
&=0 &\quad &<0,\\
N_3'&= -(2v - 1)N_3;  &\qquad &-(2v - 1)& \quad
& &>0 &\quad
& =0 &\quad
& <0.
\end{alignat}
\end{subequations}
%
The Taub point $\mathrm{T}_1$ thereby has one stable variable
$N_{1}$ while $N_{2}$ and $N_{3}$ are central when $v=1/2$,
whereas for $v\neq 1/2$ the Taub point becomes hyperbolic:
$N_{1}$ is stable and both $N_{2}$ and $N_{3}$ are unstable when $v<1/2$,
while all $N_\alpha$ are stable when $v>1/2$.
Using the permutation symmetry~\eqref{permSYM} leads
to similar statements for $\mathrm{T}_2$ and $\mathrm{T}_3$.

In general, linearization of equation~\eqref{intro_dynsyslambdaR_N} at $\mathrm{K}^{\ocircle}$
results in
\begin{equation}\label{linNalpha}
N'_\alpha = -(2v+\Sigma_\alpha|_{\mathrm{K}^{\ocircle}}){N}_\alpha, \qquad \alpha = 1,2,3.
\end{equation}
For each $\alpha=1,2,3$, the stability behaviour of $N_{\alpha}$ changes
when $\Sigma_\alpha|_{\mathrm{K}^{\ocircle}} = - 2v$. We define the
\emph{unstable Kasner arc}, denoted by $\mathrm{int}(A_{\alpha})$,
to be the points in $\mathrm{K}^{\ocircle}$ that are unstable in the
$N_\alpha$ variable, i.e., when $\Sigma_\alpha|_{\mathrm{K}^{\ocircle}} < -2v$. 
The closure of $\mathrm{int}(A_{\alpha})$ is denoted by
$A_\alpha$ 
and is given by
\begin{equation}\label{A_1}
A_\alpha:= \left\{ (\Sigma_1,\Sigma_2,\Sigma_3,0,0,0)\in \mathrm{K}^{\ocircle} \text{ $|$ }
\Sigma_\alpha \leq -2v\right\}.
\end{equation}
%
Due to the axis permutation symmetry \eqref{permSYM}, the
Kasner arcs $A_\alpha$ are symmetric portions of
$\mathrm{K}^{\ocircle}$ with points $\mathrm{Q}_\alpha = -\mathrm{T}_\alpha$ in the middle,
given by
\begin{equation} \label{Qpoints}
\mathrm{Q}_1 := (-2,1,1), \qquad \mathrm{Q}_2 := (1,-2,1), \qquad \mathrm{Q}_3 := (1,1,-2),
\end{equation}
where $\mathrm{Q}_\alpha$ is the point where $\Sigma_\alpha$
attains its minimum value $-2$ in $\mathrm{K}^\ocircle$.

The boundary set $\partial A_\alpha$ consists of two fixed points, which we refer
to as \emph{tangential points}, for reasons explained below, see Figure~\ref{Kcirclemap}.
These tangential points are the Taub points when $v=1/2$, but $v\neq 1/2$
unfolds each Taub point into two non-hyperbolic tangential points, see Figure~\ref{FIG:BIF}. \textcolor{black}{Such unfolding may provide the route for a local description using bifurcation without parameters in \cite{FiLi02,Li15}.}
The tangential points are determined by $\Sigma_\alpha|_{\mathrm{K}^{\ocircle}} = - 2v$,
which taken together with the constraints in~\eqref{KasnerCircdef} yield
\begin{equation}\label{tangentpoints}
\mathrm{t}_{\beta\gamma} := (\Sigma_\alpha,\Sigma_\beta,\Sigma_\gamma) = -v\mathrm{T}_\alpha +
\left(\mathrm{T}_\beta - \mathrm{T}_\gamma\right)\sqrt{(1-v^2)/3},
\end{equation}
where $(\alpha\beta\gamma)$ is a permutation of $(123)$, while
the Taub points were given in~\eqref{Tpoints}.
%
%
%
%
For example, the tangential points for the arcs $A_{2}$ and $A_{3}$
closest to $\mathrm{T}_1$ are given by
\begin{subequations}\label{tangTalpha}
\begin{align}
\mathrm{t}_{12}&
= (v + \sqrt{3(1-v^2)},v -\sqrt{3(1-v^2)},-2v)\\
\mathrm{t}_{13}&
= (v + \sqrt{3(1-v^2)},-2v,v -\sqrt{3(1-v^2)}).
\end{align}
\end{subequations}
The bifurcation at $v=1/2$ induces the stability change of $N_\alpha$ in
equation~\eqref{linNalpha}, where equation~\eqref{tangTalpha} entails that the tangential points $\mathrm{t}_{12}$
and $\mathrm{t}_{13}$ pass through each other at $\mathrm{T}_{1}$
as $v$ crosses the value $1/2$; 
axis permutations result in similar statements for the
other tangential points near the other Taub points, see Figure~\ref{FIG:BIF}. 

We are primarily interested in continuous deformations of GR,
$v=1/2$, and we therefore focus on the interval $v\in(0,1)$.
These models admit three cases, where $(\alpha\beta\gamma)$
is a permutation of $(123)$:
\begin{itemize}
\item[(i)] \emph{The subcritical case} $v\in (0,1/2)$:  The union of the three
arcs $A_\alpha$ cover $\mathrm{K}^{\ocircle}$, 
where both ${N}_\beta$ and ${N}_\gamma$ are unstable in the region
$\mathrm{int}(A_\beta\cap A_\gamma)$ containing $\mathrm{T}_\alpha$.
\item[(ii)] \emph{The critical case} $v=1/2$: The three arcs $A_\alpha$
cover $\mathrm{K}^{\ocircle}$ and each pair of arcs intersect only at a Taub point.
%
\item[(iii)] \emph{The supercritical case} $v\in (1/2,1)$: The union of the three
arcs $A_\alpha$ do not cover $\mathrm{K}^{\ocircle}$.
There is a \emph{closed} region around the Taub points $\mathrm{T}_\alpha$ which is stable,
defined by $S:=\mathrm{K}^{\ocircle}\setminus{\mathrm{int}(A_1)\cup \mathrm{int}(A_2) \cup \mathrm{int}(A_3)}$.
\end{itemize}
Note that the fixed points in $\mathrm{int}(S)$ have negative eigenvalues associated
with the $N_\alpha$ variables, but, for future purposes, we also include the tangential
boundary points in the definition of $S$, for which one of the negative eigenvalues
is replaced by a zero eigenvalue.
%
%
\begin{figure}[H]
\minipage[b]{0.33\textwidth}\centering
\begin{subfigure}\centering
\begin{tikzpicture}[scale=1.3]
    \draw [line width=0.1pt,domain=0:6.28,variable=\t,smooth] plot ({sin(\t r)},{cos(\t r)});

    \draw [ultra thick,domain=-0.28:0.28,variable=\t,smooth] plot ({sin(\t r)},{cos(\t r)});
    \draw [ultra thick,domain=1.81:2.35,variable=\t,smooth] plot ({sin(\t r)},{cos(\t r)});
    \draw [ultra thick,domain=-2.35:-1.81,variable=\t,smooth] plot ({sin(\t r)},{cos(\t r)});

    \draw (0,0.96) circle (0.001pt) node[anchor= north] {\scriptsize{$A_2\cap A_3$}};
    \draw (0.925,-0.4) circle (0.001pt);\node at (0.46,-0.346) {\scriptsize{$A_1\cap A_3$}};
    \draw (-0.925,-0.4) circle (0.001pt);\node at (-0.47,-0.346) {\scriptsize{$A_1\cap A_2$}};

    \draw[shift={(0.25,0.87)},rotate=-20] (-0.1,0) -- (0,0) -- (0,0.2) -- (-0.1,0.2);\filldraw [black] (0.25,1.1) circle (0.001pt) node[anchor= west] {\scriptsize{$\mathrm{t}_{13}$}};
    \draw[shift={(-0.25,0.87)},rotate=20] (0.1,0) -- (0,0) -- (0,0.2) -- (0.1,0.2);\filldraw [black] (-0.25,1.1) circle (0.001pt) node[anchor= east] {\scriptsize{$\mathrm{t}_{12}$}};

    \draw[shift={(-1.07,-0.25)},rotate=-80] (0.1,0.2) -- (0,0.2) -- (0,0) -- (0.1,0)  node[anchor= south east] {\scriptsize{$\mathrm{t}_{32}$}};
    \draw[shift={(-0.65,-0.64)},rotate=140] (0.1,0) -- (0,0) -- (0,0.2) -- (0.1,0.2) node[anchor= north] {\scriptsize{$\mathrm{t}_{31}$}};

    \draw[shift={(1.07,-0.25)},rotate=80] (-0.1,0.2) -- (0,0.2) -- (0,0) -- (-0.1,0)  node[anchor= south west] {\scriptsize{$\mathrm{t}_{23}$}};
    \draw[shift={(0.65,-0.64)},rotate=-140] (-0.1,0) -- (0,0) -- (0,0.2) -- (-0.1,0.2) node[anchor= north] {\scriptsize{$\mathrm{t}_{21}$}};

    \filldraw (1,0) circle (0.001pt) node[anchor= south east] {\scriptsize{$A_3$}};
    \filldraw (0,-1) circle (0.001pt) node[anchor= south] {\scriptsize{$A_1$}};
    \filldraw (-1,0) circle (0.001pt) node[anchor= south west] {\scriptsize{$A_2$}};

    \filldraw [black] (0,1) circle (1.25pt) node[anchor= south] {\scriptsize{$\mathrm{T}_1$}};
    \filldraw [black] (0.88,-0.49) circle (1.25pt) node[anchor= north west] {\scriptsize{$\mathrm{T}_2$}};
    \filldraw [black] (-0.88,-0.49) circle (1.25pt)node[anchor= north east] {\scriptsize{$\mathrm{T}_3$}};

    \draw (0,-0.95) -- (0,-1.05) node[anchor= north] {\scriptsize{$\mathrm{Q}_1$}};
    \draw[rotate=120] (0,-0.95) -- (0,-1.05) node[anchor= west] {\scriptsize{$\mathrm{Q}_3$}};
    \draw[rotate=-120] (0,-0.95) -- (0,-1.05) node[anchor= east] {\scriptsize{$\mathrm{Q}_2$}};
\end{tikzpicture}
\addtocounter{subfigure}{-1}
\captionof{subfigure}{\footnotesize{Subcritical: $v\in (0,1/2)$.}}\label{FIG:BIF1}
\end{subfigure}\endminipage\hfill
\minipage[b]{0.33\textwidth}\centering

\begin{subfigure}\centering
\begin{tikzpicture}[scale=1.3]
    \draw [line width=0.1pt,domain=0:6.28,variable=\t,smooth] plot ({sin(\t r)},{cos(\t r)});

    \draw[color=gray,->]
    (0,0) -- (0,0.343)  node[anchor= south] {\scriptsize{$\Sigma_{1}$}};
    \draw[color=gray,->] (0,0) -- (0.3,-0.166)  node[anchor= south] {\scriptsize{$\Sigma_{2}$}};
    \draw[color=gray,->] (0,0) -- (-0.3,-0.166)  node[anchor= south] {\scriptsize{$\Sigma_{3}$}};

    \filldraw (1,0) circle (0.001pt) node[anchor= south east] {\scriptsize{$A_3$}};
    \filldraw [black] (0,-1) circle (0.001pt) node[anchor= south] {\scriptsize{$A_1$}};
    \filldraw [black] (-1,0) circle (0.001pt) node[anchor= south west] {\scriptsize{$A_2$}};

    \filldraw (0,1) circle (1.25pt) node[anchor= south] {\scriptsize{$\mathrm{t}_{12}=\mathrm{T}_1=\mathrm{t}_{13}$}};
    \filldraw (0.88,-0.49) circle (1.25pt) node[anchor= north west] {\scriptsize{ $\mathrm{T}_2$}};
    \filldraw (-0.88,-0.49) circle (1.25pt)node[anchor= north east] {\scriptsize{$\mathrm{T}_3$}};

    \draw (0,-0.95) -- (0,-1.05) node[anchor= north] {\scriptsize{$\mathrm{Q}_1$}};
    \draw[rotate=120] (0,-0.95) -- (0,-1.05) node[anchor= west] {\scriptsize{$\mathrm{Q}_3$}};
    \draw[rotate=-120] (0,-0.95) -- (0,-1.05) node[anchor= east] {\scriptsize{$\mathrm{Q}_2$}};
\end{tikzpicture}
\addtocounter{subfigure}{-1}
\captionof{subfigure}{\footnotesize{Critical: $v=1/2$.}}\label{FIG:BIF2}
\end{subfigure}\endminipage\hfill
\minipage[b]{0.33\textwidth}\centering

\begin{subfigure}\centering
\begin{tikzpicture}[scale=1.3]
    \draw [line width=0.1pt,domain=0:6.28,variable=\t,smooth] plot ({sin(\t r)},{cos(\t r)});

    \draw [dotted, white, ultra thick, domain=-0.26:0.26,variable=\t,smooth] plot ({sin(\t r)},{cos(\t r)});
    \draw [dotted, white, ultra thick, domain=1.83:2.35,variable=\t,smooth] plot ({sin(\t r)},{cos(\t r)});
    \draw [dotted, white, ultra thick, domain=-2.35:-1.83,variable=\t,smooth] plot ({sin(\t r)},{cos(\t r)});

    \filldraw [black] (0,1) circle (0.001pt) node[anchor= north] {\scriptsize{$\mathrm{int}(S)$}};
    \filldraw [black] (0.88,-0.49) circle (0.001pt) node[anchor= south east] {\scriptsize{$\mathrm{int}(S)$}};
    \filldraw [black] (-0.88,-0.49) circle (0.001pt)node[anchor= south west] {\scriptsize{$\mathrm{int}(S)$}};

    \draw[shift={(0.25,0.87)},rotate=-20] (0.1,0) -- (0,0) -- (0,0.2) -- (0.1,0.2);\filldraw [black] (0.25,1.1) circle (0.001pt) node[anchor= west] {\scriptsize{$\mathrm{t}_{12}$}};
    \draw[shift={(-0.25,0.87)},rotate=20] (-0.1,0) -- (0,0) -- (0,0.2) -- (-0.1,0.2);\filldraw [black] (-0.25,1.125) circle (0.001pt) node[anchor= east] {\scriptsize{$\mathrm{t}_{13}$}};

    \draw[shift={(-1.07,-0.25)},rotate=-80] (-0.1,0.2) -- (0,0.2) -- (0,0) -- (-0.1,0)  node[anchor= east] {\scriptsize{$\mathrm{t}_{31}$}};
    \draw[shift={(-0.65,-0.64)},rotate=140] (-0.1,0) -- (0,0) -- (0,0.2) -- (-0.1,0.2) node[anchor= north] {\scriptsize{$\mathrm{t}_{32}$}};

    \draw[shift={(1.07,-0.25)},rotate=80] (0.1,0.2) -- (0,0.2) -- (0,0) -- (0.1,0)  node[anchor= west] {\scriptsize{$\mathrm{t}_{21}$}};
    \draw[shift={(0.65,-0.64)},rotate=-140] (0.1,0) -- (0,0) -- (0,0.2) -- (0.1,0.2) node[anchor= north] {\scriptsize{$\mathrm{t}_{23}$}};

    \filldraw [black] (1,0) circle (0.001pt) node[anchor= south east] {\scriptsize{$A_3$}};
    \filldraw [black] (0,-1) circle (0.001pt) node[anchor= south] {\scriptsize{$A_1$}};
    \filldraw [black] (-1,0) circle (0.001pt) node[anchor= south west] {\scriptsize{$A_2$}};

    \filldraw [black] (0,1) circle (1.25pt) node[anchor= south] {\scriptsize{$\mathrm{T}_1$}};
    \filldraw [black] (0.88,-0.49) circle (1.25pt) node[anchor= north west] {\scriptsize{$\mathrm{T}_2$}};
    \filldraw [black] (-0.88,-0.49) circle (1.25pt)node[anchor= north east] {\scriptsize{$\mathrm{T}_3$}};

    \draw (0,-0.95) -- (0,-1.05) node[anchor= north] {\scriptsize{$\mathrm{Q}_1$}};
    \draw[rotate=120] (0,-0.95) -- (0,-1.05) node[anchor= west] {\scriptsize{$\mathrm{Q}_3$}};
    \draw[rotate=-120] (0,-0.95) -- (0,-1.05) node[anchor= east] {\scriptsize{$\mathrm{Q}_2$}};
\end{tikzpicture}
\addtocounter{subfigure}{-1}
\captionof{subfigure}{\footnotesize{Supercritical: $v\in (1/2,1)$. }}\label{FIG:BIF3}
\end{subfigure}\endminipage
\captionof{figure}{As $v\in (0,1)$ increases, the arclength of each closed arc
$A_1,A_2$ and $A_3$ decreases. For $v\in (0,1/2)$ the union of all arcs cover
$\mathrm{K}^{\ocircle}$, where the arclength of their intersections (in bold) decreases
when $v$ increases. At $v=1/2$ the arcs only intersect at the Taub points
$\mathrm{T}_1,\mathrm{T}_2,\mathrm{T}_3$. For $v\in (1/2,1)$ the arcs do not intersect and
their union therefore do not cover $\mathrm{K}^{\ocircle}$, which results in the (dashed) set $S$, given by
$S:=\mathrm{K}^{\ocircle}\setminus{\mathrm{int}(A_1)\cup \mathrm{int}(A_2) \cup \mathrm{int}(A_3)}$.}\label{FIG:BIF}
\end{figure}
%

\subsection{Bianchi type II}\label{sec:II}

There are three physically equivalent type II sets, due to~\eqref{permSYM}, each characterized by a single
non-zero variable $N_\alpha$, $\alpha=1,2,3$, where each set yields
a two-dimensional hemisphere. The three hemispheres intersect
only at their common $\mathrm{K}^{\ocircle}$ boundary.
The Bianchi type II set with $N_1\neq 0$, denoted by $\mathrm{II}_1$, is given by
\begin{equation}\label{typeIIhemisphere}
\mathrm{II}_1 := \left\{ (\Sigma_1,\Sigma_2,\Sigma_3,N_1,0,0) \in \mathbb{R}^6 \Bigm|
\begin{array}{c}
\,\, 1-\Sigma^2 - N_1^2=0, \\
\Sigma_1+\Sigma_2+\Sigma_3=0,
\end{array}
N_1\neq 0
\right\},
\end{equation}
while the other two Bianchi type II sets $\mathrm{II}_2$ and $\mathrm{II}_3$ are obtained by
permutation of the axes according to~\eqref{permSYM}. Without loss of generality,
we therefore explicitly only consider $\mathrm{II}_1$.

As follows from~\eqref{intro_dynsyslambdaR}, the evolution equations for $\mathrm{II}_1$
can be written as
%
\begin{subequations}\label{IIevol}
\begin{align}
(\Sigma_1, \Sigma_2, \Sigma_3)^\prime &=
4v\left[(1-\Sigma^2)(\Sigma_1, \Sigma_2, \Sigma_3) + \frac{\mathrm{T}_1}{v} N_1^2\right],\label{sigevol}\\
N_1^\prime &= -2(2v\Sigma^2 + \Sigma_1)N_1.\label{Nalpha}
\end{align}
\end{subequations}
The constraints are given by
$\Sigma_1 + \Sigma_2 + \Sigma_3=0$ and $N_1^2= 1 - \Sigma^2$,
where $\Sigma^2$ is defined in~\eqref{Sigma}.
%
%
%

Using $N_1^2= 1 - \Sigma^2$ to solve for $N_1^2$ results in
that~\eqref{sigevol} can be written as
%
%
%
\begin{equation}\label{sigIIb}
\left[(\Sigma_1, \Sigma_2, \Sigma_3) + \frac{\mathrm{T}_1}{v}\right]^\prime =
4v(1-\Sigma^2)\left[(\Sigma_1, \Sigma_2, \Sigma_3) + \frac{\mathrm{T}_1}{v}\right],
\end{equation}
where $\Sigma_1$ is monotonically increasing for any initial
condition in the interior of $\mathrm{II}_1$.
The term $4v(1-\Sigma^2)$ is an Euler multiplier when $\Sigma^2\neq 1$.
This term is eliminated by an appropriate time rescaling, $(.)'=4v(1-\Sigma^2)\dot{(.)}$,
which leads to $\dot{w}=w$ where
$w := (\Sigma_1,\Sigma_2,\Sigma_3)+\mathrm{T}_1/v$.

Solutions of~\eqref{sigIIb} are therefore
straight lines in $(\Sigma_1, \Sigma_2, \Sigma_3)$-space,
which we parametrize by introducing a variable $\eta\in\mathbb{R}$,
defined by 
\begin{equation}\label{etaprime}
\eta^\prime = 4v(1-\Sigma^2)\eta.
\end{equation}
We thereby obtain
\begin{equation}\label{BIIsol}
(\Sigma_1,\Sigma_2,\Sigma_3) = (\Sigma^{\mathrm{i}}_{1},\Sigma^{\mathrm{i}}_{2},\Sigma^{\mathrm{i}}_{3})\eta +
\frac{\mathrm{T}_1}{v}(\eta-1).
\end{equation}

The straight lines pass through the auxiliary point $\mathrm{Q}_1/v$
outside the physical state space $\mathrm{II}_1$ when $\eta=0$.
A particular straight line solution then enters the physical state space at a point
$p=(\Sigma^{\mathrm{i}}_{1},\Sigma^{\mathrm{i}}_{2},\Sigma^{\mathrm{i}}_{3})\in A_1$
in the set $\mathrm{K}^{\ocircle}$ when $\eta=1$. This point is the $\alpha$-limit
of an associated heteroclinic orbit in $\mathrm{II}_1$, with $\Sigma^2 < 1$,
for which $\eta>1$ is monotonically increasing until the solution ends at its $\omega$-limit
point $p^{\mathrm{f}}=(\Sigma^{\mathrm{f}}_{1},\Sigma^{\mathrm{f}}_{2},\Sigma^{\mathrm{f}}_{3})$
in $\mathrm{K}^{\ocircle}$. The point $p^{\mathrm{f}}$ is determined by the
constraints~\eqref{intro_cons1} and~\eqref{intro_cons2} when $N_1 = N_2 = N_3 = 0$,
and equation~\eqref{BIIsol}. These conditions lead to
two solutions for $\eta$: $\eta = 1$ (for $p$) and $\eta = g$ (for $p^{\mathrm{f}}$), where
\begin{equation}\label{g1i}
g := \frac{1-v^2}{1+v^2 + \Sigma^{\mathrm{i}}_1 v} \geq 1.
\end{equation}
Using the constraints~\eqref{intro_cons1} and~\eqref{intro_cons2}
for $p$ to replace $\Sigma^{\mathrm{i}}_{2}$ and $\Sigma^{\mathrm{i}}_{3}$ with
$\Sigma^{\mathrm{i}}_{1}$, and the latter with $g$ according to the above equation,
give
\begin{equation}\label{ConstANDgalphaNEW}
N_1^2 = 1-\Sigma^2 = \left(\frac{1-v^2}{v^2}\right)(\eta - 1)(g - \eta)g^{-1}.
\end{equation}
Similar results are obtained by axis permutation for $\mathrm{II}_2$
and $\mathrm{II}_3$. Figure~\ref{Kcirclemap} gives an example of a Bianchi
type $\mathrm{II}_1$ heteroclinic orbit, given
by~\eqref{BIIsol}, \eqref{g1i}, \eqref{ConstANDgalphaNEW}, and
its projected straight line in $(\Sigma_1,\Sigma_2,\Sigma_3)$-space.
\begin{figure}[H]
\centering
    \begin{tikzpicture}[scale=1.7]
        \draw (-1,0) arc (180:540:1cm and 0.25cm);
        \draw[white, ultra thick, loosely dashed] (-1,0) arc (180:0:1cm and 0.25cm);
        \filldraw [black] (1,0) circle (0.1pt) node[anchor=west] {$\mathrm{K}^{\ocircle}$};

        \draw (1,0) arc (0:180:1cm and 1cm);
        \node at (1.1,0.6) {$\mathrm{II}_1$};

        \draw[color=gray,dashed,-] (-1.5,-1) -- (-1,0);\filldraw (-1,0) circle (0.001pt) node[anchor=east] {\scriptsize{$\mathrm{t}_{23}$}};
        \draw[color=gray,dashed,-] (-1.5,-1) -- (0.7,-0.2);\filldraw (0.7,-0.2) circle (0.001pt) node[anchor=north] {\scriptsize{$\mathrm{t}_{32}$}};

        \draw[postaction={decorate}] (-0.7,-0.2) arc (180:109:0.45cm and 1.171cm);
        \draw[dashed] (-0.2,0.3) arc (0:90:0.175cm and 0.61cm);

        \draw[color=gray,dotted] (-1.5,-1) -- (-0.7,-0.2);
        \draw[thick,dotted, postaction={decorate}] (-0.7,-0.2) -- (-0.2,0.25);

        \filldraw (-0.7,-0.2) circle (1.25pt) node[anchor=north] {$p$};
        \filldraw [black] (-0.2,0.25) circle (1.25pt) node[anchor=south west] {$p^{\mathrm{f}}$};

        \filldraw [black] (-1.5,-1) circle (0.1pt) node[anchor=east] {\scriptsize{$\frac{\mathrm{Q}_1}{v}$}};
    \end{tikzpicture}
\caption{An example of a Bianchi type II solution; a heteroclinic orbit in the hemisphere $\mathrm{II}_1$.
Its projection is a (dotted) line parametrized by $\eta$ in $(\Sigma_1,\Sigma_2,\Sigma_3)$-space
given by~\eqref{BIIsol}. There are three special points on this line:
the auxiliary point $\mathrm{Q}_1/v$ outside the physical state space $\mathrm{II}_1$ when $\eta=0$, $p$
when $\eta=1$ and $p^{\mathrm{f}}$ when $\eta=g$. 
Furthermore, the nomenclature `tangential points' is explained: they are the points
where $p=p^\mathrm{f}$ and hence where the aforementioned lines are tangential to
$\mathrm{K}^{\ocircle}$.
}\label{Kcirclemap}
\end{figure}

Using equation~\eqref{BIIsol} to eliminate $\eta$ yields the
unparametrized form of the heteroclinic orbits in $\mathrm{II}_1$,
\begin{equation}\label{BIIstraight}
\left(\Sigma^{\mathrm{i}}_{1} + \frac{2}{v}\right)(\Sigma_2 - \Sigma_3) =
\left(\Sigma^{\mathrm{i}}_{2} - \Sigma^{\mathrm{i}}_{3}\right)\left(\Sigma_1 + \frac{2}{v}\right),
\end{equation}
where a cyclic permutation of $(123)$ yields the orbits in $\mathrm{II}_2$
and $\mathrm{II}_3$. Note that equation~\eqref{BIIstraight} is derived
in Appendix~\ref{app:heterosym} from the scale-automorphism group. In
combination with axis permutations this establishes that the Bianchi
type II heteroclinic chains arise from first principles, in GR,
$\lambda$-$R$ and HL gravity.

The type II heteroclinic orbits induce a map between different Kasner states on the Kasner circle, called
\emph{the Kasner circle map} $\mathcal{K}\!:\mathrm{K}^{\ocircle}\to\mathrm{K}^{\ocircle}$. It
maps the $\alpha$-limits to the $\omega$-limits of
heteroclinic orbits in each of the hemispheres $\mathrm{II}_1$, $\mathrm{II}_2$, $\mathrm{II}_3$, see
Figures~\ref{Kcirclemap} and~\ref{FIG:KASNERMAPS}.
\begin{figure}[H]
\minipage[b]{0.37\textwidth}\centering
\begin{subfigure}\centering
    \begin{tikzpicture}[scale=1.3]

    \draw [line width=0.1pt,domain=0:6.28,variable=\t,smooth] plot ({sin(\t r)},{cos(\t r)});

    \draw [very thick, domain=-0.29:0.29,variable=\t,smooth] plot ({0.975*sin(\t r)},{0.975*cos(\t r)});
    \draw [rotate=120,very thick, domain=-0.29:0.29,variable=\t,smooth] plot ({0.975*sin(\t r)},{0.975*cos(\t r)});
    \draw [rotate=-120,very thick, domain=-0.29:0.29,variable=\t,smooth] plot ({0.975*sin(\t r)},{0.975*cos(\t r)});

    \draw[color=gray,dashed] (0,-2.85) -- (0.98,-0.29);
    \draw[color=gray,dashed] (0,-2.85) -- (-0.98,-0.29);

    \draw[color=gray,rotate=120,dashed] (0,-2.85) -- (0.98,-0.29);
    \draw[color=gray,rotate=120,dashed] (0,-2.85) -- (-0.98,-0.29);

    \draw[color=gray,rotate=240,dashed] (0,-2.85) -- (0.98,-0.29);
    \draw[color=gray,rotate=240,dashed] (0,-2.85) -- (-0.98,-0.29);

    \filldraw[gray] (0,-2.85) circle (0.1pt);

    \draw[rotate=-120,color=gray,dotted] (-0.5,0.85) -- (0,-2.85);
    \draw[rotate=-120] (-0.5,0.85);\node at (1.3,0.1) {\scriptsize{$\mathcal{K}(p)$}};
    \draw[rotate=-120] (-0.3,-0.95);\node at (-0.75,0.8) {\scriptsize{$p$}};

    \draw[rotate=-120,white,ultra thick] (-0.5,0.85) -- (-0.26,-0.95);
    \draw[rotate=-120,dotted, thick, postaction={decorate}] (-0.26,-0.95) -- (-0.5,0.85);




    \draw (0,-2.85) circle (0.1pt) node[anchor=north] {$\frac{\mathrm{Q}_1}{v}$};
    \draw[rotate=240] (0,-2.85) circle (0.1pt) node[anchor=east] {$\frac{\mathrm{Q}_2}{v}$};
    \draw[rotate=120]  (0,-2.85) circle (0.1pt) node[anchor=west] {$\frac{\mathrm{Q}_3}{v}$};

    \draw (0,-0.95) -- (0,-1.05) node[anchor= north] {\scriptsize{$\mathrm{Q}_1$}};
    \draw[rotate=120] (0,-0.95) -- (0,-1.05) node[anchor= west] {\scriptsize{$\mathrm{Q}_3$}};
    \draw[rotate=-120] (0,-0.95) -- (0,-1.05) node[anchor= east] {\scriptsize{$\mathrm{Q}_2$}};

    \node at (0,1.2) {\scriptsize{$A_2\cap A_3$}};
    \node at (-1.4,-0.57) {\scriptsize{$A_1\cap A_2$}};
    \node at (1.4,-0.57) {\scriptsize{$A_1\cap A_3$}};

\end{tikzpicture}
    \addtocounter{subfigure}{-1}\captionof{subfigure}{\footnotesize{Subcritical: $v\in (0,1/2)$. }}\label{FIG:KASNERMAPS1}
\end{subfigure}
\endminipage\hfill
\minipage[b]{0.31\textwidth}\centering

\begin{subfigure}\centering
    \begin{tikzpicture}[scale=1.3]
    \draw [line width=0.1pt,domain=0:6.28,variable=\t,smooth] plot ({sin(\t r)},{cos(\t r)});

    \draw[color=gray,dashed,-] (-1.75,1) -- (1.75,1);
    \draw[color=gray,rotate=120,dashed,-] (-1.75,1) -- (1.75,1);
    \draw[color=gray,rotate=240,dashed,-] (-1.75,1) -- (1.75,1);

    \draw[color=gray,dotted] (0,1) -- (0,-1.99);
    \filldraw (0,-2) circle (0.1pt);

    \draw[white,ultra thick] (0,1) -- (0,-1);
    \draw[dotted, thick, postaction={decorate}] (0,-1) -- (0,1);


    \draw (0,-2) circle (0.1pt) node[anchor=north] {\scriptsize{$2\mathrm{Q}_1$}};
    \draw[rotate=240] (0,-2) circle (0.1pt) node[anchor=east] {\scriptsize{$2\mathrm{Q}_2$}};
    \draw[rotate=120]  (0,-2) circle (0.1pt) node[anchor=west] {\scriptsize{$2\mathrm{Q}_3$}};

    \filldraw [black] (0,1) circle (1.25pt) node[anchor= south] {\scriptsize{$\mathcal{K}(p)=\mathrm{T}_1$}};
    \filldraw [black] (0.88,-0.49) circle (1.25pt) node[anchor= north west] {\scriptsize{$\mathrm{T}_2$}};
    \filldraw [black] (-0.88,-0.49) circle (1.25pt)node[anchor= north east] {\scriptsize{$\mathrm{T}_3$}};

    \draw (0,-0.95) -- (0,-1.05) node[anchor= north] {\scriptsize{$p=\mathrm{Q}_1$}};
    \draw[rotate=120] (0,-0.95) -- (0,-1.05) node[anchor= west] {\scriptsize{$\mathrm{Q}_3$}};
    \draw[rotate=-120] (0,-0.95) -- (0,-1.05) node[anchor= east] {\scriptsize{$\mathrm{Q}_2$}};
\end{tikzpicture}
    \addtocounter{subfigure}{-1}\captionof{subfigure}{\footnotesize{Critical: $v=1/2$. }}\label{FIG:KASNERMAPS2}
\end{subfigure}
\endminipage\hfill
\minipage[b]{0.31\textwidth}\centering
\begin{subfigure}\centering
    \begin{tikzpicture}[scale=1.3]
    \draw [line width=0.1pt,domain=0:6.28,variable=\t,smooth] plot ({sin(\t r)},{cos(\t r)});

    \draw [ultra thick, dotted, white, domain=-0.26:0.26,variable=\t,smooth] plot ({sin(\t r)},{cos(\t r)});
    \draw [ultra thick, dotted, white, domain=1.83:2.35,variable=\t,smooth] plot ({sin(\t r)},{cos(\t r)});
    \draw [ultra thick, dotted, white, domain=3.87:4.39,variable=\t,smooth] plot ({sin(\t r)},{cos(\t r)});

    \draw[color=gray,dashed,-] (0,-1.35) -- (0.685,-0.75);
    \draw[color=gray,dashed,-] (0,-1.35) -- (-0.685,-0.75);

    \draw[color=gray,rotate=120,dashed,-] (0,-1.35) -- (0.685,-0.75);
    \draw[color=gray,rotate=120,dashed,-] (0,-1.35) -- (-0.685,-0.75);

    \draw[color=gray,rotate=-120,dashed,-] (0,-1.35) -- (0.685,-0.75);
    \draw[color=gray,rotate=-120,dashed,-] (0,-1.35) -- (-0.685,-0.75);

    \draw[rotate=120,color=gray,dotted] (0,-1.35) -- (-0.9,0.5);
    \node at (1.1,0.3) {\scriptsize{$p$}};

    \draw[rotate=120,white,ultra thick] (-0.2,-0.95) -- (-0.88,0.45);
    \draw[rotate=120,dotted, thick, postaction={decorate}] (-0.2,-0.95) -- (-0.88,0.45);



    \draw (0,-1.35) circle (0.1pt) node[anchor=north] {\scriptsize{$\frac{\mathrm{Q}_1}{v}$}};
    \draw[rotate=240] (0,-1.35) circle (0.1pt) node[anchor=east] {\scriptsize{$\frac{\mathrm{Q}_2}{v}$}};
    \draw[rotate=120]  (0,-1.35) circle (0.1pt) node[anchor=west] {\scriptsize{$\frac{\mathrm{Q}_3}{v}$}};


    \draw (0,-0.95) -- (0,-1.05); \node at (0,-0.75) {\scriptsize{$\mathcal{K}(p)=\mathrm{Q}_1$}};
    \draw[rotate=120] (0,-0.95) -- (0,-1.05) node[anchor= north east] {\scriptsize{$\mathrm{Q}_3$}};
    \draw[rotate=-120] (0,-0.95) -- (0,-1.05) node[anchor= north west] {\scriptsize{$\mathrm{Q}_2$}};

    \node at (0,1.2) {\scriptsize{$S$}};
    \node at (-1.1,-0.57) {\scriptsize{$S$}};
    \node at (1.1,-0.57) {\scriptsize{$S$}};

\end{tikzpicture}
    \addtocounter{subfigure}{-1}\captionof{subfigure}{\footnotesize{Supercritical: $v\in (1/2,1)$. }}\label{FIG:KASNERMAPS3}
\end{subfigure}
\endminipage
\captionof{figure}{The Kasner circle map $\mathcal{K}$ can be obtained from the straight
lines that emanate from the three auxiliary points $\mathrm{Q}_\alpha /v $, which
intersect with two points in the set $\mathrm{K}^{\ocircle}$: $p$
and $\mathcal{K}(p):=p^\mathrm{f}$.
Each (bold dotted) line represents the projection onto $(\Sigma_1,\Sigma_2,\Sigma_3)$-space
of a heteroclinic orbit from different hemispheres $\mathrm{II}_\alpha$,
originating from the auxiliary point $\mathrm{Q}_\alpha /v$.
Note that the points $\mathrm{Q}_\alpha /v $ approach
$\mathrm{Q}_\alpha$ as $v\to 1$, whereas $\mathrm{Q}_\alpha /v$ goes to infinity as $v\to 0$.
}\label{FIG:KASNERMAPS}
\end{figure}

Each point $p = (\Sigma_\alpha^\mathrm{i},\Sigma_\beta^\mathrm{i},\Sigma_\gamma^\mathrm{i})$
in the set $\mathrm{K}^{\ocircle}$ is thereby mapped to
$p^\mathrm{f} = (\Sigma_\alpha^\mathrm{f},\Sigma_\beta^\mathrm{f},\Sigma_\gamma^\mathrm{f})$
in $\mathrm{K}^{\ocircle}$, where $p^\mathrm{f}$ is obtained from~\eqref{BIIsol}
and permutations thereof by setting $\eta=g$. Thus,
\begin{equation}\label{KasnerCirc}
{\cal K}(p):=
\begin{cases}
g(p)p + (g(p)-1)\frac{\mathrm{T}_\alpha}{v} & \text{ for }\, p  \in A_\alpha\\
p &  \text{ for }  p \notin A_1 \cup A_2 \cup A_3
\end{cases}\,,
\end{equation}
where
\begin{equation}\label{galpha}
g(p):=\frac{1-v^2}{1+v^2 + \Sigma_\alpha^\mathrm{i} v} \geq 1,
\qquad \text{ for } p =(\Sigma_1^\mathrm{i},
\Sigma_2^\mathrm{i}, \Sigma_3^\mathrm{i})  \in A_\alpha ,
\end{equation}
and where the index $\alpha$ in $\Sigma_\alpha^\mathrm{i}$ is the same index as
for $A_\alpha$.

When $v\in [1/2, 1)$ the Kasner circle map $\mathcal{K}$ is well-defined and continuous, since
the unstable arcs $\mathrm{int}(A_1)$, $\mathrm{int}(A_2)$ and $\mathrm{int}(A_3)$ are
disjoint. Note that the set $S$ consists of fixed points of the Kasner circle map $\mathcal{K}$.
In the critical case, $v=1/2$, the Kasner circle map $\mathcal{K}$ is the Mixmaster map,
discussed in Section~\ref{sec:critical}, while the dynamics of $\mathcal{K}$
in the supercritical case $(1/2,1)$ is discussed in Section~\ref{sec:superINF}.

For $v\in [0,1/2)$, however, $\mathcal{K}$ is not a well-defined map, since the unstable
arcs $\mathrm{int}(A_\alpha)$ overlap and points in the overlapping
regions $\mathrm{int}(A_\alpha \cap A_\beta)$ have two possible
Bianchi type II heteroclinic orbits, making $\mathcal{K}$ multivalued.
Moreover, a discontinuity on at least one of the boundary points of $A_\alpha \cap A_\beta$
is inevitable, since one must change the auxiliary vertex $\mathrm{Q}_\alpha/v$
for the map, see Figure~\ref{FIG:KASNERMAPS}. Nevertheless, we can still
define iterates of $\mathcal{K}$ through a family of
piece-wise continuous maps to capture
features of the dynamics, as explored in Section~\ref{sec:sub}.


To describe the expansion properties of the Kasner circle
map~\eqref{KasnerCirc}, it is convenient to first introduce
Misner parametrized variables $(\Sigma_+, \Sigma_-)$ adapted to the arc $A_1$, which,
according to Appendix~\ref{app:dom}, are given by
\begin{subequations}\label{Misner}
\begin{align}
\Sigma_1 &= - 2\Sigma_+,\\
\Sigma_2 &= \Sigma_+ + \sqrt{3} \Sigma_-,\\
\Sigma_3 &= \Sigma_+ - \sqrt{3} \Sigma_-,
\end{align}
\end{subequations}
which leads to $\Sigma^2 = \Sigma_+^2 + \Sigma_-^2$ where
$\Sigma^2=1$ on $\mathrm{K}^{\ocircle}$, due to the constraint~\eqref{intro_cons1},
thereby yielding a circle with unit radius. The variables $(\Sigma_+,\Sigma_-)$ have
the advantage of solving the constraint~\eqref{intro_cons2}, but the
drawback of making the permutation symmetry~\eqref{permSYM} implicit, whereas
it is explicit in  $(\Sigma_1,\Sigma_2,\Sigma_3)$.

The variables $\Sigma_\pm$ lead to the following form
for the Kasner circle map~\eqref{KasnerCirc}:
%
\begin{subequations}\label{K_+-}
\begin{align}
\mathcal{K}_+(\Sigma^\mathrm{i}_+,\Sigma^\mathrm{i}_-) &= g(-2\Sigma^\mathrm{i}_+)\left[\Sigma^\mathrm{i}_+ - \frac{1}{v}\right] + \frac{1}{v},\\
\mathcal{K}_-(\Sigma^\mathrm{i}_+,\Sigma^\mathrm{i}_-) &= g(-2\Sigma^\mathrm{i}_+)\Sigma^\mathrm{i}_-,\label{bouncephi}
\end{align}
\end{subequations}
where $g(-2\Sigma^\mathrm{i}_+)$ is given by~\eqref{galpha} in the Misner parametrization~\eqref{Misner} of $p\in A_1$.
%
%

Next we introduce an angular variable $\varphi$ adapted to $A_1$,
\begin{subequations}\label{polar}
\begin{align}
\Sigma_+ &= \cos(\varphi),\\
\Sigma_- &=\sin(\varphi),
\end{align}
\end{subequations}
which solves the remaining constraint~\eqref{intro_cons1}, since
$\Sigma^2 = \Sigma_+^2 + \Sigma_-^2 = 1$ 
(similar variables $\Sigma_\pm$ with associated angles $\varphi$ can be
introduced for $A_2$ and $A_3$, by permutation of the axes).
The map~\eqref{K_+-} can then be replaced by a map with the arc-length
$\varphi$ of the Kasner unit circle $\mathrm{K}^\ocircle$ as its domain,
\begin{equation}
\mathcal{K}(\varphi^\mathrm{i})=\int_{\mathrm{t}_{23}}^{\varphi^\mathrm{i}} \sqrt{ \left(D\mathcal{K}_+(\varphi)\right)^2 + \left(D\mathcal{K}_-(\varphi)\right)^2}\, d\varphi,
\end{equation}
where $D = d/d\varphi$ on $A_1$, and similarly for $A_2$ and $A_3$.

The derivative of $\mathcal{K}(\varphi^\mathrm{i})$ with respect to
$\varphi^\mathrm{i}$ is the tangent vector, with length
\begin{equation}\label{K'arc}
|D\mathcal{K}(\varphi^\mathrm{i})| = \sqrt{(D\mathcal{K}_+(\varphi^\mathrm{i}))^2 + (D\mathcal{K}_-(\varphi^\mathrm{i}))^2}.
\end{equation}
Applying the chain rule to~\eqref{K_+-} at $\varphi^\mathrm{i}$ 
%
%
yields
\begin{equation}\label{DKg}
|D\mathcal{K}(p)| = g(p)= \frac{1 - v^2}{1 + v^2 - 2\cos(\varphi^\mathrm{i})v},
\end{equation}
where $g$ in~\eqref{galpha} is expressed
in $\varphi^\mathrm{i}$ by means of~\eqref{polar}, which yields
$\Sigma_1^\mathrm{i} = -2\Sigma_+^\mathrm{i} = - 2\cos(\varphi^\mathrm{i})$.
Using the symmetry under axes permutations \eqref{permSYM} proves the following Lemma:
\begin{lemma}\label{KasnerCircMapEXP}
The derivative of the Kasner circle map $\mathcal{K}$ with
respect to the arc-length of the Kasner unit circle $\mathrm{K}^{\ocircle}$,
$\varphi\in A_\alpha$, is given by
\begin{equation}\label{KasnerCircPrime}
|D\mathcal{K}(p)| =
\begin{cases}
g(p) & \text{ for } p  \in A_\alpha\\
1 &  \text{ for }  p \notin A_1 \cup A_2 \cup A_3.
\end{cases}
\end{equation}
\end{lemma}
In other words, the Kasner circle map $\mathcal{K}$ is expanding on the interior of each $A_\alpha$,
but not uniformly\footnote{Hence $\mathcal{K}$ is an example of a non-uniformly
hyperbolic circle map, and it would be interesting to investigate its dynamical
properties with recent mathematical methods developed in~\cite{Katok95,BDV05, dut19} and references therein.}
since $g$ is a varying function that attains 1 at $\partial A_\alpha$. In the arc $A_1$, the
map $\mathcal{K}$ is symmetric with respect to $\mathrm{Q}_1$,
which is due to the permutation of $\Sigma_2$ and $\Sigma_3$ according to~\eqref{permSYM},
and it is monotonically increasing on each side of $\mathrm{Q}_1$ starting from the tangential points $\mathrm{t}_{32}$ and $\mathrm{t}_{23}$, where $g= 1$, until $g$ reaches its maximum $g=(1+v)/(1-v)$ at
$\mathrm{Q}_1$, where $\Sigma_1^\mathrm{i}=-2$, see Figure~\ref{fig:plotg}.
Similar statements hold for $A_2$ and $A_3$ by permuting the axes, as in~\eqref{permSYM}.
\begin{figure}[H]
\centering
    \begin{tikzpicture}[scale=1.2]
        \draw[->] (0,-0.1) -- (0,3.5)node[anchor=south] {$g(p)$};
        \draw[->] (-0.1,0) -- (6.28,0)node[anchor=west] {$p\in \mathrm{K}^{\ocircle}$};

        \draw[color=gray,dashed] (3.14,3) -- (0,3) node[anchor=east] {\color{black} $\frac{1+v}{1-v}$};
        \draw[color=gray,dashed] (3.14,3) -- (3.14,0) node[anchor=north] {\color{black} $\mathrm{Q}_1$};

        \draw[color=gray,dashed] (2.09,1) -- (2.09,0) node[anchor=north] {\color{black} $\mathrm{t}_{32}$};

        \draw[color=gray,dashed] (4.18,1) -- (0,1) node[anchor=east] {\color{black} $1$};
        \draw[color=gray,dashed] (4.18,1) -- (4.18,0) node[anchor=north] {\color{black} $\mathrm{t}_{23}$};

        \draw[-] (2.09,0) -- (4.18,0);

        \draw [domain=2.09:4.18,variable=\t,smooth] plot ({\t},{(1-(1/2)^2)/(1+(1/2)^2+2*(1/2)*cos(\t r))});
    \end{tikzpicture}
\caption{The function $g(p)$ for $p\in A_1$ between the tangential points
$\mathrm{t}_{32}$ and $\mathrm{t}_{23}$.}\label{fig:plotg}
\end{figure}
%

%

%
%
%

In Bianchi types VIII and IX, it is possible to concatenate type II
heteroclinic orbits on the type I and II boundaries
to form heteroclinic chains.
This heteroclinic structure is expected to play a key role
for type VIII and IX when $\tau_-\rightarrow\infty$, and is
the focus of the next three sections.
However, before proceeding, we describe the bifurcations at
$v=0$ and $v=1$. We then construct heteroclinic chains with period 3
when $v\in [0,1]$, as an example of concatenation of heteroclinic Bianchi
type II orbits.

\subsubsection*{The cases $v=0$ and $v=1$}\label{sec:bdrycases}

Even though the cases $v=0$ and $v=1$ are not our main focus,
they are useful in order to obtain results for $v\in(0,1)$, as illustrated
by the construction of the heteroclinic cycles/chains with period 3 below.
In contrast to when $v\in(0,1)$, 
the Kasner circle map $\mathcal{K}$ is not chaotic for $v=0$ and $v=1$, and thus bifurcations occur at these parameter values, see Figure~\ref{FIG:BDRYmaps}.

As $v\rightarrow 0$, the heteroclinic orbits become parallel
lines\footnote{Multiplying~\eqref{BIIstraight} by $v$ and setting $v=0$ results in
$\Sigma_2 - \Sigma^{\mathrm{i}}_{2} = \Sigma_3 - \Sigma^{\mathrm{i}}_{3}$ for
the $v=0$ type $\mathrm{II}_1$ models, in agreement with Figure~\ref{FIG:BDRYmaps}
when $v=0$. In~\cite{giakam17}, the authors considered Bianchi type II models
with $\lambda=1$ and a quadratic curvature term, which is mathematically
equivalent to $v=1/8$ in the $\lambda$-$R$ case, and
a cubic curvature term, which corresponds to $v=0$, as seen in Appendix~\ref{appsubsec:HL}.
In the latter case, it was pointed out that the Kasner parameter
$u = u^{\mathrm{i}} = \sqrt{3} + 1$ results in $u^{\mathrm{i}} = u^{\mathrm{f}}$,
which yields the period 3 heteroclinic cycles in the present formulation. See \cite{LappicyLessard} for a broad discussion on the case $v=0$. 
}
, see Figure~\ref{FIG:BDRYmaps}.
The derivative of the Kasner circle map $\mathcal{K}$, given by~\eqref{KasnerCircPrime},
thereby equals $1$ at any point on $\mathrm{K}^\ocircle$, as is seen from the
limit $v\rightarrow 0$ in equation~\eqref{galpha}. 
Since there is no expansion, the case $v=0$ has a network of
heteroclinic orbits 
that is not associated with chaos, but
see Appendix~\ref{appsubsec:HL} for further discussions on
HL models and their relation to the case $v=0$.
Note the connection with `frame transitions' in, e.g., Bianchi type VI$_{-1/9}$ vacuum models,
and when using an Iwasawa frame in GR~\cite{ugg13a,ugg13b,heietal09,dametal03},
since these also consist of parallel heteroclinic orbits. In contrast to these
situations in GR, however, there are three (instead of two) families of non-expanding
orbits when $v=0$, and no family of expanding type II orbits.


As $v\rightarrow 1$, 
the Kasner circle map $\mathcal{K}$ is not continuous anymore:
it becomes the identity on $\mathrm{K}^\ocircle$, except at each of the three points
$\mathrm{Q}_\alpha$, which are mapped to the entire set $\mathrm{K}^\ocircle$.
In particular, the points $\mathrm{Q}_\alpha$ are mapped
to each other, thereby yielding a network of heteroclinic chains: chains of
period 2 between each two 
points $\mathrm{Q}_\alpha$ and $\mathrm{Q}_\beta$, and chains
with period 3 between the three 
points $\mathrm{Q}_\alpha$. The situation for $v=1$ is somewhat
reminiscent to that of the Bianchi type I Einstein-Vlasov models, where
there is a heteroclinic network associated with the Taub points
$\mathrm{T}_\alpha$, see~\cite{heiugg06}; for a recent paper on
the future dynamics of these Einstein-Vlasov models, see~\cite{leeetal20}.


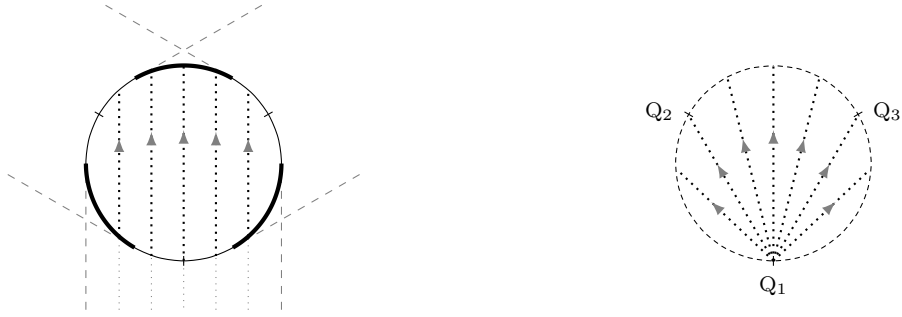
\begin{figure}[H]\centering
\minipage{0.49\textwidth}\centering
\begin{subfigure}\centering
\begin{tikzpicture}[scale=1.3]

    \draw[color=gray,dashed,-] (-1,-1.5) -- (-1,0);
    \draw[color=gray,dashed,-] (1,-1.5) -- (1,0);

    \draw[color=gray,dashed,-,rotate=120] (-1,-1.5) -- (-1,0);
    \draw[color=gray,dashed,-,rotate=240] (-1,-1.5) -- (-1,0);

    \draw[color=gray,dashed,-,rotate=120] (1,-1.5) -- (1,0);
    \draw[color=gray,dashed,-,rotate=240] (1,-1.5) -- (1,0);

    \draw[color=gray,dotted] (0,1) -- (0,-1.5);

    \draw[color=gray,dotted] (-0.33,-1.5) -- (-0.33,0.95);
    \draw[color=gray,dotted] (0.33,-1.5) -- (0.33,0.95);

    \draw[color=gray,dotted] (-0.66,-1.5) -- (-0.66,0.75);
    \draw[color=gray,dotted] (0.66,-1.5) -- (0.66,0.75);

    \draw[white,ultra thick] (0,1) -- (0,-1);
    \draw[dotted,thick, postaction={decorate}] (0,-1) -- (0,1);

    \draw[color=white, ultra thick] (-0.33,-0.95) -- (-0.33,0.95);
    \draw[dotted,thick, postaction={decorate}] (-0.33,-0.95) -- (-0.33,0.95);

    \draw[color=white, ultra thick] (0.33,-0.95) -- (0.33,0.95);
    \draw[dotted,thick, postaction={decorate}] (0.33,-0.95) -- (0.33,0.95);

    \draw[color=white, ultra thick] (-0.66,-0.75) -- (-0.66,0.75);
    \draw[dotted,thick, postaction={decorate}] (-0.66,-0.75) -- (-0.66,0.75);

    \draw[color=white, ultra thick] (0.66,-0.75) -- (0.66,0.75);
    \draw[dotted,thick, postaction={decorate}] (0.66,-0.75) -- (0.66,0.75);

    \draw [ultra thick, domain=-0.52:0.52,variable=\t,smooth] plot ({sin(\t r)},{cos(\t r)});
    \draw [ultra thick, domain=-0.52:0.52,variable=\t,smooth, rotate=120] plot ({sin(\t r)},{cos(\t r)});
    \draw [ultra thick, domain=-0.52:0.52,variable=\t,smooth, rotate=240] plot ({sin(\t r)},{cos(\t r)});

    \draw [domain=0:6.28,variable=\t,smooth] plot ({sin(\t r)},{cos(\t r)});

    \draw (0,-0.95) -- (0,-1.05);
    \draw[rotate=120] (0,-0.95) -- (0,-1.05);
    \draw[rotate=-120] (0,-0.95) -- (0,-1.05);

\end{tikzpicture}
    \addtocounter{subfigure}{-1}\captionof{subfigure}{\footnotesize{For $v\rightarrow 0$,
    the points $\mathrm{Q}_\alpha /v \rightarrow\infty$. Hence the type II heteroclinic orbits are parallel
    lines that emanate from infinity. The overlapping arcs (in bold) have two
    unstable directions. Moreover, $|A_\alpha|=\pi$ and
    $|A_\alpha \cap A_\beta|=\pi/3$. \label{fig:v=0} }}
\end{subfigure}
\endminipage\hfill
\minipage{0.49\textwidth}\centering
\begin{subfigure}\centering
\begin{tikzpicture}[scale=1.3]
    \draw [domain=0:6.3,variable=\t,smooth] plot ({sin(\t r)},{cos(\t r)});

    \draw [ultra thick, dotted, white, domain=0:6.28,variable=\t,smooth] plot ({sin(\t r)},{cos(\t r)});


    \draw (0,-0.95) -- (0,-1.05) node[anchor= north] {\scriptsize{$\mathrm{Q}_1$}};
    \draw[rotate=120] (0,-0.95) -- (0,-1.05) node[anchor= west] {\scriptsize{$\mathrm{Q}_3$}};
    \draw[rotate=-120] (0,-0.95) -- (0,-1.05) node[anchor= east] {\scriptsize{$\mathrm{Q}_2$}};

    \draw[dotted,thick, postaction={decorate}] (0,-1) -- (0,1);

    \draw[dotted,thick, postaction={decorate}] (0,-1) -- (0.47,0.87);
    \draw[dotted,thick, postaction={decorate}] (0,-1) -- (-0.47,0.87);

    \draw[dotted,thick, postaction={decorate}] (0,-1) -- (0.86,0.48);
    \draw[dotted,thick, postaction={decorate}] (0,-1) -- (-0.86,0.48);

    \draw[dotted,thick, postaction={decorate}] (0,-1) -- (0.95,-0.1);
    \draw[dotted,thick, postaction={decorate}] (0,-1) -- (-0.95,-0.1);

    \draw[color=white,dotted] (-1,-1.5) -- (-1,-1.49);
    \draw[color=white,dotted,rotate=120] (-1,-1.5) -- (-1,-1.49);
    \draw[color=white,dotted,rotate=-120] (-1,-1.5) -- (-1,-1.49);

\end{tikzpicture}
    \addtocounter{subfigure}{-1}
    \captionof{subfigure}{\footnotesize{For $v\rightarrow 1$, the points
    $\mathrm{Q}_\alpha /v \rightarrow \mathrm{Q}_\alpha$ on $\mathrm{K}^\ocircle$,
    where $A_\alpha$ consists of a single point $\mathrm{Q}_\alpha$. Hence any point on
    $\mathrm{K}^\ocircle$ can be reached by a type II heteroclinic orbit from $\mathrm{Q}_\alpha$.
    The dashed arcs have three stable directions. \label{fig:v=1}}}
\end{subfigure}
\endminipage
\captionof{figure}{Heteroclinic type II orbits for $v=0$ and $v=1$ projected
onto $(\Sigma_1, \Sigma_2, \Sigma_3)$-space.}\label{FIG:BDRYmaps}
\end{figure}
\subsubsection*{Example of Bianchi type II concatenation: Period 3 chains}

We will now construct heteroclinic chains with period 3 (i.e., period 3 heteroclinic cycles), and describe
how these chains change as the parameter $v\in [0,1]$ varies. For
the GR case $v=1/2$, these chains/cycles have been previously found,
see e.g.~\cite{heiugg09a}. First, note that chains with period 3 consist
of equilateral triangles in the plane of the Kasner circle in the
projected $(\Sigma_1,\Sigma_2,\Sigma_3)$-space, which follows from the
permutation symmetry described in~\eqref{permSYM}, where the corners of
the triangles on $\mathrm{K}^\ocircle$ correspond to physically equivalent
Kasner states, again related by axis permutations. Second, there are two equilateral
triangles for each value of $v \in [0,1)$, which due to~\eqref{permSYM}
are symmetric with respect to reflections with respect to the coordinate
lines $\Sigma_\alpha$, while the two triangles coalesce to a single one
with corners at the points $\mathrm{Q}_\alpha$ when $v=1$.
Third, the triangles depict $(\Sigma_1,\Sigma_2,\Sigma_3)$-space
projections of two different heteroclinic chains on the
Bianchi type II boundary of the Bianchi type VIII and IX state spaces,
with clockwise and anti-clockwise orientation of the projected
heteroclinic chains in $(\Sigma_1,\Sigma_2,\Sigma_3)$-space, see Figure~\ref{FIG:period3}.
\begin{figure}[H]\centering
\minipage{0.29\textwidth}\centering
\begin{subfigure}\centering
\begin{tikzpicture}[scale=1.3]

    \draw[color=white,dashed] (0,-2.85) -- (0,-2.849);
    \draw[color=white,dashed,rotate=120] (0,-2.85) -- (0,-2.849);


    \draw[color=gray,dashed,-] (-1,-1.5) -- (-1,0);
    \draw[color=gray,dashed,-] (1,-1.5) -- (1,0);

    \draw[color=gray,dashed,-,rotate=120] (-1,-1.5) -- (-1,0);
    \draw[color=gray,dashed,-,rotate=240] (-1,-1.5) -- (-1,0);

    \draw[color=gray,dashed,-,rotate=120] (1,-1.5) -- (1,0);
    \draw[color=gray,dashed,-,rotate=240] (1,-1.5) -- (1,0);

    \draw[color=gray, dotted] (0.5,-1.5) -- (0.5,0.85);
    \draw[color=gray, dotted,rotate=120] (0.5,-1.5) -- (0.5,0.85);
    \draw[color=gray, dotted,rotate=240] (0.5,-1.5) -- (0.5,0.85);

    \draw[color=gray,dotted] (-0.5,-1.5) -- (-0.5,0.85);
    \draw[color=gray,dotted,rotate=120] (-0.5,-1.5) -- (-0.5,0.85);
    \draw[color=gray,dotted,rotate=240] (-0.5,-1.5) -- (-0.5,0.85);

    \draw[color=white, ultra thick] (-0.5,-0.85) -- (-0.5,0.85);
    \draw[loosely dotted,thick, postaction={decorate}] (-0.5,-0.85) -- (-0.5,0.85);

    \draw[color=white, ultra thick,rotate=120] (-0.5,-0.85) -- (-0.5,0.85);
    \draw[loosely dotted,thick,rotate=120, postaction={decorate}] (-0.5,-0.85) -- (-0.5,0.85);

    \draw[color=white, ultra thick,rotate=240] (-0.5,-0.85) -- (-0.5,0.85);
    \draw[loosely dotted,thick,rotate=240, postaction={decorate}] (-0.5,-0.85) -- (-0.5,0.85);

    \draw[color=white, ultra thick] (0.5,-0.85) -- (0.5,0.85);
    \draw[ dotted,thick, postaction={decorate}] (0.5,-0.85) -- (0.5,0.85);

    \draw[color=white, ultra thick,rotate=120] (0.5,-0.85) -- (0.5,0.85);
    \draw[ dotted,thick,rotate=120, postaction={decorate}] (0.5,-0.85) -- (0.5,0.85);

    \draw[color=white, ultra thick,rotate=240] (0.5,-0.85) -- (0.5,0.85);
    \draw[ dotted,thick,rotate=240, postaction={decorate}] (0.5,-0.85) -- (0.5,0.85);

    \draw [ultra thick, domain=-0.52:0.52,variable=\t,smooth] plot ({sin(\t r)},{cos(\t r)});
    \draw [ultra thick, domain=-0.52:0.52,variable=\t,smooth, rotate=120] plot ({sin(\t r)},{cos(\t r)});
    \draw [ultra thick, domain=-0.52:0.52,variable=\t,smooth, rotate=240] plot ({sin(\t r)},{cos(\t r)});

    \draw (0,-0.95) -- (0,-1.05);
    \draw[rotate=120] (0,-0.95) -- (0,-1.05);
    \draw[rotate=-120] (0,-0.95) -- (0,-1.05);

    \draw [domain=0:6.28,variable=\t,smooth] plot ({sin(\t r)},{cos(\t r)});

 \end{tikzpicture}
    \addtocounter{subfigure}{-1}\captionof{subfigure}{\footnotesize{ $v=0$. \label{fig:PERIOD3v=0} }}
\end{subfigure}
\endminipage\hfill
\minipage{0.41\textwidth}\centering

\begin{subfigure}\centering
\begin{tikzpicture}[scale=1.3]
    \draw[color=white,dashed,-] (-1,-1.5) -- (-1,-1.49);
    \draw[rotate=120,color=white,dashed,-] (-1,-1.5) -- (-1,-1.49);
    \draw[rotate=-120,color=white,dashed,-] (-1,-1.5) -- (-1,-1.49);

    \draw [domain=0:6.28,variable=\t,smooth] plot ({sin(\t r)},{cos(\t r)});

    \draw [very thick, domain=-0.29:0.29,variable=\t,smooth] plot ({0.975*sin(\t r)},{0.975*cos(\t r)});
    \draw [rotate=120,very thick, domain=-0.29:0.29,variable=\t,smooth] plot ({0.975*sin(\t r)},{0.975*cos(\t r)});
    \draw [rotate=-120,very thick, domain=-0.29:0.29,variable=\t,smooth] plot ({0.975*sin(\t r)},{0.975*cos(\t r)});

    \draw[color=gray,dashed] (0,-2.85) -- (0.98,-0.29);
    \draw[color=gray,dashed] (0,-2.85) -- (-0.98,-0.29);

    \draw[color=gray,rotate=120,dashed] (0,-2.85) -- (0.98,-0.29);
    \draw[color=gray,rotate=120,dashed] (0,-2.85) -- (-0.98,-0.29);

    \draw[color=gray,rotate=240,dashed] (0,-2.85) -- (0.98,-0.29);
    \draw[color=gray,rotate=240,dashed] (0,-2.85) -- (-0.98,-0.29);

    \filldraw[gray] (0,-2.85) circle (0.1pt);

    \draw[color=gray,  dotted] (0,-2.85) -- (0.64,0.75);
    \draw[rotate=120,color=gray, dotted] (0,-2.85) -- (0.64,0.75);
    \draw[rotate=-120,color=gray, dotted] (0,-2.85) -- (0.64,0.75);

    \draw[color=gray,  dotted] (0,-2.85) -- (-0.64,0.75);
    \draw[rotate=120,color=gray, dotted] (0,-2.85) -- (-0.64,0.75);
    \draw[rotate=-120,color=gray, dotted] (0,-2.85) -- (-0.64,0.75);

    \draw[white, ultra thick] (-0.34,-0.93) -- (-0.64,0.75);
    \draw[loosely dotted, thick, postaction={decorate}] (-0.34,-0.93) -- (-0.64,0.75);

    \draw[rotate=120,white, ultra thick] (-0.34,-0.93) -- (-0.64,0.75);
    \draw[rotate=120,loosely dotted, thick, postaction={decorate}] (-0.34,-0.93) -- (-0.64,0.75);

    \draw[rotate=-120,white, ultra thick] (-0.34,-0.93) -- (-0.64,0.75);
    \draw[rotate=-120,loosely dotted, thick, postaction={decorate}] (-0.34,-0.93) -- (-0.64,0.75);

    \draw[white, ultra thick] (0.34,-0.93) -- (0.64,0.75);
    \draw[  dotted, thick, postaction={decorate}] (0.34,-0.93) -- (0.64,0.75);

    \draw[rotate=120,white, ultra thick] (0.34,-0.93) -- (0.64,0.75);
    \draw[rotate=120,  dotted, thick, postaction={decorate}] (0.34,-0.93) -- (0.64,0.75);

    \draw[rotate=-120,white, ultra thick] (0.34,-0.93) -- (0.64,0.75);
    \draw[rotate=-120,  dotted, thick, postaction={decorate}] (0.34,-0.93) -- (0.64,0.75);



    \draw (0,-0.95) -- (0,-1.05);
    \draw[rotate=120] (0,-0.95) -- (0,-1.05);
    \draw[rotate=-120] (0,-0.95) -- (0,-1.05);

 \end{tikzpicture}
    \addtocounter{subfigure}{-1}\captionof{subfigure}{\footnotesize{ $v\in (0,1/2)$. \label{fig:PERIOD3v<1/2} }}
\end{subfigure}
\endminipage\hfill
\minipage{0.29\textwidth}\centering

\begin{subfigure}\centering
    \begin{tikzpicture}[scale=1.3]

    \draw[color=white,dashed] (0,-2.85) -- (0,-2.849);
    \draw[color=white,dashed,rotate=120] (0,-2.85) -- (0,-2.849);

    \draw[color=white,dashed,-] (-1,-1.5) -- (-1,-1.49);
    \draw[rotate=120,color=white,dashed,-] (-1,-1.5) -- (-1,-1.49);
    \draw[rotate=-120,color=white,dashed,-] (-1,-1.5) -- (-1,-1.49);

    \draw [domain=0:6.28,variable=\t,smooth] plot ({sin(\t r)},{cos(\t r)});

    \draw[color=gray,dashed,-] (-1.75,1) -- (1.75,1);
    \draw[color=gray,rotate=120,dashed,-] (-1.75,1) -- (1.75,1);
    \draw[color=gray,rotate=240,dashed,-] (-1.75,1) -- (1.75,1);

    \draw[color=gray,  dotted] (0,-2) -- (0.7,0.7);
    \draw[color=gray,  dotted,rotate=120] (0,-2) -- (0.7,0.7);
    \draw[color=gray,  dotted,rotate=240] (0,-2) -- (0.7,0.7);

    \draw[color=gray,  dotted] (0,-2) -- (-0.7,0.7);
    \draw[color=gray,  dotted,rotate=120] (0,-2) -- (-0.7,0.7);
    \draw[color=gray,  dotted,rotate=240] (0,-2) -- (-0.7,0.7);

    \draw[color=white,ultra thick] (-0.27,-0.96) -- (-0.7,0.7);
    \draw[loosely dotted,thick, postaction={decorate}] (-0.27,-0.96) -- (-0.7,0.7);

    \draw[color=white,ultra thick,rotate=120] (-0.27,-0.96) -- (-0.7,0.7);
    \draw[loosely dotted,thick,rotate=120, postaction={decorate}] (-0.27,-0.96) -- (-0.7,0.7);

    \draw[color=white,ultra thick,rotate=-120] (-0.27,-0.96) -- (-0.7,0.7);
    \draw[loosely dotted,thick,rotate=-120, postaction={decorate}] (-0.27,-0.96) -- (-0.7,0.7);

    \draw[color=white,ultra thick] (0.27,-0.96) -- (0.7,0.7);
    \draw[   dotted,thick, postaction={decorate}] (0.27,-0.96) -- (0.7,0.7);

    \draw[color=white,ultra thick,rotate=120] (0.27,-0.96) -- (0.7,0.7);
    \draw[   dotted,thick,rotate=120, postaction={decorate}] (0.27,-0.96) -- (0.7,0.7);

    \draw[color=white,ultra thick,rotate=-120] (0.27,-0.96) -- (0.7,0.7);
    \draw[   dotted,thick,rotate=-120, postaction={decorate}] (0.27,-0.96) -- (0.7,0.7);

    \draw [domain=0:6.28,variable=\t,smooth] plot ({sin(\t r)},{cos(\t r)});

    \draw (0,-2) circle (0.1pt);
    \draw[rotate=240] (0,-2) circle (0.1pt);
    \draw[rotate=120]  (0,-2) circle (0.1pt);

    \filldraw [black] (0,1) circle (1.25pt);
    \filldraw [black] (0.88,-0.49) circle (1.25pt);
    \filldraw [black] (-0.88,-0.49) circle (1.25pt);

    \draw (0,-0.95) -- (0,-1.05);
    \draw[rotate=120] (0,-0.95) -- (0,-1.05);
    \draw[rotate=-120] (0,-0.95) -- (0,-1.05);
\end{tikzpicture}
    \addtocounter{subfigure}{-1}\captionof{subfigure}{\footnotesize{$v=1/2$. }}\label{FIG:period3v=1/2}
\end{subfigure}
\endminipage
\end{figure}

\begin{figure}[H]\centering
\hspace{2.2cm}
\minipage{0.32\textwidth}\centering
\begin{subfigure}\centering
\begin{tikzpicture}[scale=1.3]

    \draw [domain=0:6.28,variable=\t,smooth] plot ({sin(\t r)},{cos(\t r)});

    \draw [ultra thick, dotted, white, domain=-0.26:0.26,variable=\t,smooth] plot ({sin(\t r)},{cos(\t r)});
    \draw [ultra thick, dotted, white, domain=1.83:2.35,variable=\t,smooth] plot ({sin(\t r)},{cos(\t r)});
    \draw [ultra thick, dotted, white, domain=3.87:4.39,variable=\t,smooth] plot ({sin(\t r)},{cos(\t r)});

    \draw[color=gray,dashed,-] (0,-1.35) -- (0.685,-0.75);
    \draw[color=gray,dashed,-] (0,-1.35) -- (-0.685,-0.75);

    \draw[color=gray,rotate=120,dashed,-] (0,-1.35) -- (0.685,-0.75);
    \draw[color=gray,rotate=120,dashed,-] (0,-1.35) -- (-0.685,-0.75);

    \draw[color=gray,rotate=-120,dashed,-] (0,-1.35) -- (0.685,-0.75);
    \draw[color=gray,rotate=-120,dashed,-] (0,-1.35) -- (-0.685,-0.75);

    \draw[color=gray, dotted] (0,-1.35) -- (0.78,0.6);
    \draw[rotate=120,color=gray, dotted] (0,-1.35) -- (0.78,0.6);
    \draw[rotate=-120,color=gray, dotted] (0,-1.35) -- (0.78,0.6);

    \draw[color=gray, dotted] (0,-1.35) -- (-0.78,0.6);
    \draw[rotate=120,color=gray, dotted] (0,-1.35) -- (-0.78,0.6);
    \draw[rotate=-120,color=gray, dotted] (0,-1.35) -- (-0.78,0.6);

    \draw[white, ultra thick] (-0.15,-0.98) -- (-0.78,0.6);
    \draw[thick,loosely dotted, postaction={decorate}] (-0.15,-0.98) -- (-0.78,0.6);

    \draw[rotate=120,white, ultra thick] (-0.15,-0.98) -- (-0.78,0.6);
    \draw[rotate=120,thick,loosely dotted, postaction={decorate}] (-0.15,-0.98) -- (-0.78,0.6);

    \draw[rotate=-120,white, ultra thick] (-0.15,-0.98) -- (-0.78,0.6);
    \draw[rotate=-120,thick,loosely dotted, postaction={decorate}] (-0.15,-0.98) -- (-0.78,0.6);

    \draw[white, ultra thick] (0.15,-0.98) -- (0.78,0.6);
    \draw[thick,  dotted, postaction={decorate}] (0.15,-0.98) -- (0.78,0.6);

    \draw[rotate=120,white, ultra thick] (0.15,-0.98) -- (0.78,0.6);
    \draw[rotate=120,thick,  dotted, postaction={decorate}] (0.15,-0.98) -- (0.78,0.6);

    \draw[rotate=-120,white, ultra thick] (0.15,-0.98) -- (0.78,0.6);
    \draw[rotate=-120,thick,  dotted, postaction={decorate}] (0.15,-0.98) -- (0.78,0.6);


    \draw[shift={(0,-0.95)}] (0,0) -- (0,-0.1);
    \draw[rotate=120,shift={(0,-0.95)}] (0,0) -- (0,-0.1);
    \draw[rotate=-120,shift={(0,-0.95)}] (0,0) -- (0,-0.1);
\end{tikzpicture}
    \addtocounter{subfigure}{-1}\captionof{subfigure}{\footnotesize{$v\in (1/2,1)$. \label{fig:PERIOD3v>1/2}}}
\end{subfigure}
\endminipage\hfill
\minipage{0.32\textwidth}\centering

\begin{subfigure}\centering
\begin{tikzpicture}[scale=1.3]
    \draw [domain=0:6.3,variable=\t,smooth] plot ({sin(\t r)},{cos(\t r)});

    \draw [ultra thick, dotted, white, domain=0:6.28,variable=\t,smooth] plot ({sin(\t r)},{cos(\t r)});


    \draw (0,-0.95) -- (0,-1.05) node[anchor= north] {\scriptsize{$\mathrm{Q}_1$}};
    \draw[rotate=120] (0,-0.95) -- (0,-1.05) node[anchor= west] {\scriptsize{$\mathrm{Q}_3$}};
    \draw[rotate=-120] (0,-0.95) -- (0,-1.05) node[anchor= east] {\scriptsize{$\mathrm{Q}_2$}};

    \draw[  dotted,thick, postaction={decoration={markings,mark=at position 0.41 with {\arrow[thick,color=gray]{latex reversed}}},decorate}] (0,-1) -- (0.86,0.48);
    \draw[  dotted,thick, postaction={decoration={markings,mark=at position 0.41 with {\arrow[thick,color=gray]{latex reversed}}},decorate}] (0,-1) -- (-0.86,0.48);
    \draw[  dotted,thick, postaction={decoration={markings,mark=at position 0.41 with {\arrow[thick,color=gray]{latex reversed}}},decorate}] (-0.85,0.5) -- (0.85,0.5);

    \draw[white,shift={(0,-0.95)}] (0,0) -- (0,-0.1);
    \draw[white,rotate=120,shift={(0,-0.95)}] (0,0) -- (0,-0.1);
    \draw[white,rotate=-120,shift={(0,-0.95)}] (0,0) -- (0,-0.1);
\end{tikzpicture}
    \addtocounter{subfigure}{-1}\captionof{subfigure}{\footnotesize{$v= 1$. \label{fig:PERIOD3v=1}}}
\end{subfigure}
\endminipage
\hspace{2.2cm}
\captionof{figure}{The two triangles in each figure depict the two periodic heteroclinic chains with period 3
projected onto $(\Sigma_1, \Sigma_2, \Sigma_3)$-space.
As $v\in[0,1]$ increases, these triangles rotate: the densely (sparsely) dotted one rotates clockwise
(counter-clockwise). 
}
\label{FIG:period3}
\end{figure}

The heteroclinic chains with period 3 can be constructed as follows from the $v=0$ case,
for which the period 3 chains are easily obtained due to the simple
heteroclinic structure. 
Without loss of generality, consider the densely dotted triangle
in Figure~\ref{FIG:period3} for $v=0$ and rotate it clockwise by an angle
$\theta\in (0,\pi/6]$. The three prolonged sides of the rotated triangle
intersect each projected $\Sigma_\alpha$ axis (projected onto the plane that contains the Kasner circle
) at the same distance from $\mathrm{K}^\ocircle$ due to the axis permutation symmetry. 
Since the prolonged lines correspond
to Bianchi type II orbits in the physical state space, the points of intersection
are given by $\mathrm{Q}_\alpha/v$ for some $v=v(\theta)$.
Continuity of the rotation and the parametrization
$v(\theta)$ yields the period 3 chains for all $\theta\in [0,\pi/6]$, i.e.,
all $v=v(\theta)\in[0,1]$. The boundary cases $v=0$ and $v=1$ yield
$\lim_{\theta \to 0}\mathrm{Q}_\alpha/v(\theta) \rightarrow \infty$
and $\lim_{\theta \to \pi/6}\mathrm{Q}_\alpha/v(\theta) \rightarrow \mathrm{Q}_\alpha$,
respectively, see Figure~\ref{FIG:period3SUPERPOSITION}.
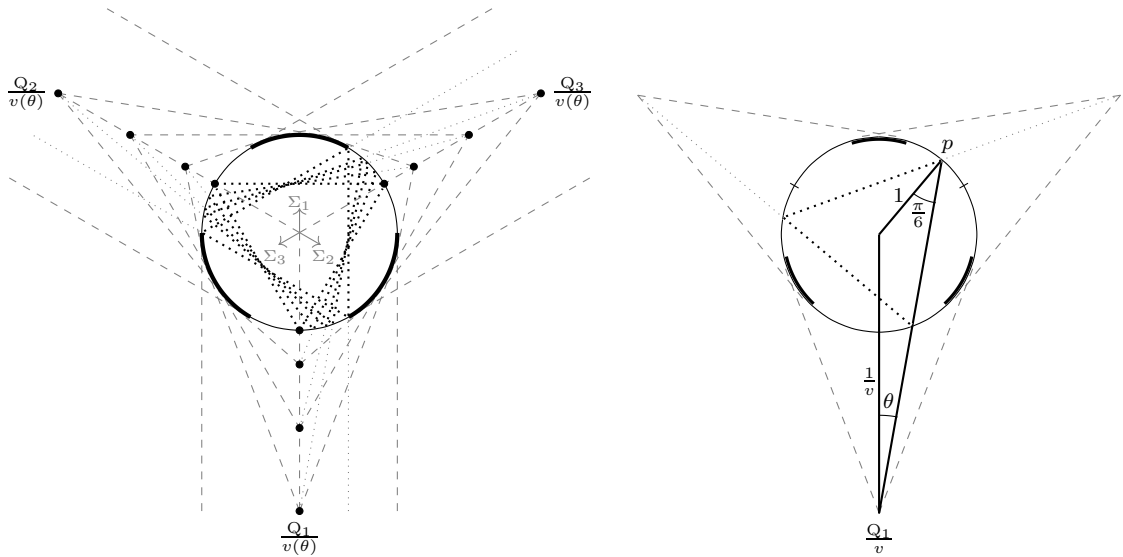
\begin{figure}[H]\centering
\minipage{0.49\textwidth}\centering
\begin{subfigure}\centering
\begin{tikzpicture}[scale=1.3]



    \draw[color=gray,dashed,-] (-1,-2.85) -- (-1,0);
    \draw[color=gray,dashed,-] (1,-2.85) -- (1,0);

    \draw[color=gray,dashed,-,rotate=120] (-1,-2.85) -- (-1,0);
    \draw[color=gray,dashed,-,rotate=240] (-1,-2.85) -- (-1,0);

    \draw[color=gray,dashed,-,rotate=120] (1,-2.85) -- (1,0);
    \draw[color=gray,dashed,-,rotate=240] (1,-2.85) -- (1,0);

    \draw[color=gray,dotted] (0.5,-2.85) -- (0.5,0.85);
    \draw[color=gray,dotted,rotate=120] (0.5,-2.85) -- (0.5,0.85);
    \draw[color=gray,dotted,rotate=240] (0.5,-2.85) -- (0.5,0.85);


    \draw[color=gray,dashed] (0,-2.85) -- (0.98,-0.29);
    \draw[color=gray,dashed] (0,-2.85) -- (-0.98,-0.29);

    \draw[color=gray,rotate=120,dashed] (0,-2.85) -- (0.98,-0.29);
    \draw[color=gray,rotate=120,dashed] (0,-2.85) -- (-0.98,-0.29);

    \draw[color=gray,rotate=240,dashed] (0,-2.85) -- (0.98,-0.29);
    \draw[color=gray,rotate=240,dashed] (0,-2.85) -- (-0.98,-0.29);

    \draw[color=gray, dotted] (0,-2.85) -- (0.64,0.75);
    \draw[rotate=120,color=gray, dotted] (0,-2.85) -- (0.64,0.75);
    \draw[rotate=-120,color=gray, dotted] (0,-2.85) -- (0.64,0.75);


    \draw[color=gray,dashed,-] (-1.75,1) -- (1.75,1);
    \draw[color=gray,rotate=120,dashed,-] (-1.75,1) -- (1.75,1);
    \draw[color=gray,rotate=240,dashed,-] (-1.75,1) -- (1.75,1);

    \draw[color=gray,  dotted] (0,-2) -- (0.7,0.7);
    \draw[color=gray,  dotted,rotate=120] (0,-2) -- (0.7,0.7);
    \draw[color=gray,  dotted,rotate=240] (0,-2) -- (0.7,0.7);


    \draw[color=gray,dashed,-] (0,-1.35) -- (0.685,-0.75);
    \draw[color=gray,dashed,-] (0,-1.35) -- (-0.685,-0.75);

    \draw[color=gray,rotate=120,dashed,-] (0,-1.35) -- (0.685,-0.75);
    \draw[color=gray,rotate=120,dashed,-] (0,-1.35) -- (-0.685,-0.75);

    \draw[color=gray,rotate=-120,dashed,-] (0,-1.35) -- (0.685,-0.75);
    \draw[color=gray,rotate=-120,dashed,-] (0,-1.35) -- (-0.685,-0.75);

    \draw[color=gray, dotted] (0,-1.35) -- (0.78,0.6);
    \draw[rotate=120,color=gray, dotted] (0,-1.35) -- (0.78,0.6);
    \draw[rotate=-120,color=gray, dotted] (0,-1.35) -- (0.78,0.6);


    \filldraw[white] [domain=0:6.28,variable=\t,smooth] plot ({sin(\t r)},{cos(\t r)});


    \draw[  dotted,thick] (0.5,-0.85) -- (0.5,0.85);

    \draw[ dotted,thick,rotate=120] (0.5,-0.85) -- (0.5,0.85);

    \draw[ dotted,thick,rotate=240] (0.5,-0.85) -- (0.5,0.85);


    \draw[ dotted, thick] (0.34,-0.93) -- (0.64,0.75);

    \draw[rotate=120, dotted, thick] (0.34,-0.93) -- (0.64,0.75);

    \draw[rotate=-120, dotted, thick] (0.34,-0.93) -- (0.64,0.75);


    \draw[  dotted,thick] (0.27,-0.96) -- (0.7,0.7);

    \draw[  dotted,thick,rotate=120] (0.27,-0.96) -- (0.7,0.7);

    \draw[  dotted,thick,rotate=-120] (0.27,-0.96) -- (0.7,0.7);



    \draw[thick,  dotted] (0.15,-0.98) -- (0.78,0.6);

    \draw[rotate=120,thick,  dotted] (0.15,-0.98) -- (0.78,0.6);

    \draw[rotate=-120,thick,  dotted] (0.15,-0.98) -- (0.78,0.6);


    \draw[  dotted,thick] (0,-1) -- (0.85,0.45);
    \draw[  dotted,thick] (0,-1) -- (-0.85,0.45);
    \draw[  dotted,thick] (-0.85,0.5) -- (0.85,0.5);


    \draw[color=gray,->](0,0) -- (0,0.25);\node[gray] at (0,0.3) {\tiny{$\Sigma_{1}$}};
    \draw[rotate=120,color=gray,->](0,0) -- (0,0.25); \node[gray] at (-0.25,-0.25) {\tiny{$\Sigma_{3}$}};
    \draw[rotate=-120,color=gray,->](0,0) -- (0,0.25); \node[gray] at (0.25,-0.25) {\tiny{$\Sigma_{2}$}};

    \draw[color=gray, dashed] (0,0) -- (0,-2.85);
    \draw[rotate=120,color=gray, dashed] (0,0) -- (0,-2.85);
    \draw[rotate=-120,color=gray, dashed] (0,0) -- (0,-2.85);

    \draw [ultra thick, domain=-0.52:0.52,variable=\t,smooth] plot ({sin(\t r)},{cos(\t r)});
    \draw [ultra thick, domain=-0.52:0.52,variable=\t,smooth, rotate=120] plot ({sin(\t r)},{cos(\t r)});
    \draw [ultra thick, domain=-0.52:0.52,variable=\t,smooth, rotate=240] plot ({sin(\t r)},{cos(\t r)});

    \draw [domain=0:6.28,variable=\t,smooth] plot ({sin(\t r)},{cos(\t r)});

    \filldraw (0,-1) circle (1pt);
    \filldraw[rotate=120] (0,-1) circle (1pt);
    \filldraw[rotate=-120] (0,-1) circle (1pt); 

    \filldraw (0,-2.85) circle (1pt) node[anchor=north] {\scriptsize{$\frac{\mathrm{Q}_1}{v(\theta)}$}};
    \filldraw[rotate=240] (0,-2.85) circle (1pt) node[anchor=east] {\scriptsize{$\frac{\mathrm{Q}_2}{v(\theta)}$}};
    \filldraw[rotate=120]  (0,-2.85) circle (1pt) node[anchor= west] {\scriptsize{$\frac{\mathrm{Q}_3}{v(\theta)}$}};

    \filldraw (0,-2) circle (1pt);
    \filldraw[rotate=240] (0,-2) circle (1pt);
    \filldraw[rotate=120]  (0,-2) circle (1pt);

    \filldraw (0,-1.35) circle (1pt);
    \filldraw[rotate=240] (0,-1.35) circle (1pt);
    \filldraw[rotate=120]  (0,-1.35) circle (1pt);

 \end{tikzpicture}
    \addtocounter{subfigure}{-1}\captionof{subfigure}{\footnotesize{Superposition of the
    densely dotted heteroclinic chains with period 3 for different $v\in[0,1]$.
    The prolonged sides of each triangle intersect the projected $\Sigma_\alpha$
    axis at some $\mathrm{Q}_\alpha/v(\theta)$ and describe type II orbits.}}
\end{subfigure}
\endminipage\hfill
\minipage{0.49\textwidth}\centering

\begin{subfigure}\centering
\begin{tikzpicture}[scale=1.3]

    \draw [domain=0:6.28,variable=\t,smooth] plot ({sin(\t r)},{cos(\t r)});

    \draw [very thick, domain=-0.29:0.29,variable=\t,smooth] plot ({0.975*sin(\t r)},{0.975*cos(\t r)});
    \draw [rotate=120,very thick, domain=-0.29:0.29,variable=\t,smooth] plot ({0.975*sin(\t r)},{0.975*cos(\t r)});
    \draw [rotate=-120,very thick, domain=-0.29:0.29,variable=\t,smooth] plot ({0.975*sin(\t r)},{0.975*cos(\t r)});

    \draw[color=gray,dashed] (0,-2.85) -- (0.98,-0.29);
    \draw[color=gray,dashed] (0,-2.85) -- (-0.98,-0.29);

    \draw[color=gray,rotate=120,dashed] (0,-2.85) -- (0.98,-0.29);
    \draw[color=gray,rotate=120,dashed] (0,-2.85) -- (-0.98,-0.29);

    \draw[color=gray,rotate=240,dashed] (0,-2.85) -- (0.98,-0.29);
    \draw[color=gray,rotate=240,dashed] (0,-2.85) -- (-0.98,-0.29);

    \draw[thick] (0,-2.85) -- (0.64,0.75);
    \draw[rotate=120,color=gray, dotted] (0,-2.85) -- (0.64,0.75);
    \draw[rotate=-120,color=gray, dotted] (0,-2.85) -- (0.64,0.75);






    \draw[rotate=120,white, ultra thick] (0.34,-0.93) -- (0.64,0.75);
    \draw[rotate=120,  dotted, thick] (0.34,-0.93) -- (0.64,0.75);

    \draw[rotate=-120,white, ultra thick] (0.34,-0.93) -- (0.64,0.75);
    \draw[rotate=-120,  dotted, thick] (0.34,-0.93) -- (0.64,0.75);

    \draw (0,-0.95) -- (0,-1.05);
    \draw[rotate=120] (0,-0.95) -- (0,-1.05);
    \draw[rotate=-120] (0,-0.95) -- (0,-1.05);

    \draw[thick](0,0) -- (0,-2.85);
    \draw[thick](0,0) -- (0.63,0.75);

    \draw [shift={(0,-2.85)},domain=0:0.19,variable=\t,smooth] plot ({sin(\t r)},{cos(\t r)});
    \draw [shift={(0.65,0.82)},domain=-2.5:-3,variable=\t,smooth] plot ({0.5*sin(\t r)},{0.5*cos(\t r)});

    \node at (0.2,0.4) {\scriptsize{$1$}};
    \node at (-0.1,-1.5) {\scriptsize{$\frac{1}{v}$}};
    \node at (0.1,-1.7) {\scriptsize{$\theta$}};
    \node at (0.4,0.2) {\scriptsize{$\frac{\pi}{6}$}};
    \node at (0.7,0.9) {\scriptsize{$p$}};

    \draw (0,-2.85) circle (0.1pt) node[anchor=north] {\scriptsize{$\frac{\mathrm{Q}_1}{v}$}};
    \draw[color=white,dashed,-,rotate=120] (-1,-2.85) -- (-1,-2.84);
    \draw[color=white,dashed,-,rotate=240] (-1,-2.85) -- (-1,-2.84);

 \end{tikzpicture}
    \addtocounter{subfigure}{-1}\captionof{subfigure}{\footnotesize{The angle
    of rotation $\theta=\theta(v)$ is obtained by the law of sines using
    the triangle in the plane of $\mathrm{K}^\ocircle$ between the center
    of $\mathrm{K}^\ocircle$, the point $\mathrm{Q}_1/v$ and the
    vertex $p\in A_3$ of the triangle that describes
    a period 3 chain.}}
\end{subfigure}
\endminipage
\hspace{2.55cm}
\captionof{figure}{As $v$ increases, the triangles rotate clockwise by the
angle $\theta\in [0,\pi/6]$ in~\eqref{p3thetav}.
}\label{FIG:period3SUPERPOSITION}
\end{figure}
Moreover, $v(\theta) = 2 \sin\theta$, or alternatively,
\begin{equation}\label{p3thetav}
\theta(v) = \arcsin\left(\frac{v}{2}\right).
\end{equation}
This equation yields the clockwise (counter-clockwise) rotation of the
densely (sparsely) dotted triangle and can be derived as follows.
The rotation angle $\theta$ for the densely dotted triangle is given by the
angle between the line from $\mathrm{Q}_1/v$ to the center of $\mathrm{K}^\ocircle$
and the Bianchi type II trajectory originating from $\mathrm{Q}_1/v$ and ending
at the vertex of the triangle in the arc $A_3$, which we denote by $p$, see
Figure~\ref{FIG:period3SUPERPOSITION}. Consider the triangle that connects
the center of $\mathrm{K}^\ocircle$, $\mathrm{Q}_1/v$, and the vertex $p$
in the plane of $\mathrm{K}^\ocircle$,
which has unit radius in the $(\Sigma_+,\Sigma_-)$ coordinates.
In these coordinates, the line from the origin $(0,0)$ to
$\mathrm{Q}_1/v$ has length $1/v$, while the unit radius
from $(0,0)$ to the vertex $p$ bisects the
angle of the equilateral triangle, due to the permutation symmetry~\eqref{permSYM},
which yields an angle of $\pi/6$, see Figure~\ref{FIG:period3SUPERPOSITION}.
The law of sines then implies that $\sin\theta = v\sin(\pi/6) = v/2$,
and thereby the above formula. Combining the result in
equation~\eqref{p3thetav} with the geometry
in Figure~\ref{FIG:period3SUPERPOSITION} yields
\begin{equation}
\Sigma_1 = \sqrt{3\left(1 - \left(\frac{v}{2}\right)^2\right)} - \frac{v}{2}
\end{equation}
for the upper vertices (and thus with maximum $\Sigma_1$) of the densely and
sparsely dotted period 3 triangles in Figures~\ref{FIG:period3}
and~\ref{FIG:period3SUPERPOSITION}.

We have thereby proved the following result:
\begin{proposition}\label{lem:Period3}
There are two heteroclinic chains with period 3 for all $v\in [0,1]$.
When projected onto the plane of the Kasner circle $\mathrm{K}^\ocircle$, these
chains/cycles are given by two equilateral triangles within $\mathrm{K}^\ocircle$.
As $v\in[0,1)$ increases, the two triangles rotate in different directions and coalesce
into one when $v=1$.
\end{proposition}
%
%
%

\section{Critical case}\label{sec:critical}

GR belongs to the critical case $v=1/2$ where the
concatenated Bianchi type II orbits describe the heteroclinic chains
that are expected to be asymptotically shadowed by solutions when $\tau_-\rightarrow\infty$
in the Bianchi type VIII and IX models. An example of part of a heteroclinic chain
is given in Figure~\ref{FIG:KASNERMAPGR}.
\begin{figure}[H]\centering
    \begin{tikzpicture}[scale=1.3]
    \draw[rotate=-120,color=gray,dotted] (0,-2) -- (0.98,0.2);
    \draw[color=gray,dotted] (0,-2) -- (-0.8,0.6);
    \draw[rotate=-120,color=gray,dotted] (0,-2) -- (-0.35,0.9);

    \draw[rotate=-120,white,ultra thick] (0.51,-0.86) -- (0.98,0.2);
    \draw[rotate=-120,dotted, thick, postaction={decorate}] (0.51,-0.86) -- (0.98,0.2);

    \draw[white,ultra thick] (-0.32,-0.95) -- (-0.8,0.6);
    \draw[dotted, thick, postaction={decorate}] (-0.32,-0.95) -- (-0.8,0.6);

    \draw[rotate=-120,white,ultra thick] (-0.124,-0.99) -- (-0.35,0.93);
    \draw[rotate=-120,dotted, thick, postaction={decorate}] (-0.124,-0.99) -- (-0.35,0.93);



    \draw[color=gray,dashed,-] (-1.75,1) -- (1.75,1);
    \draw[color=gray,rotate=120,dashed,-] (-1.75,1) -- (1.75,1);
    \draw[color=gray,rotate=240,dashed,-] (-1.75,1) -- (1.75,1);

    \draw (0,-2) circle (0.1pt) node[anchor=north] {\scriptsize{$2\mathrm{Q}_1$}};
    \draw[rotate=240] (0,-2) circle (0.1pt) node[anchor=east] {\scriptsize{$2\mathrm{Q}_2$}};
    \draw[rotate=120]  (0,-2) circle (0.1pt) node[anchor=west] {\scriptsize{$2\mathrm{Q}_3$}};

    \filldraw [black] (0,1) circle (1.25pt) node[anchor= south] {\scriptsize{$\mathrm{T}_1$}};
    \filldraw [black] (0.88,-0.49) circle (1.25pt) node[anchor= north west] {\scriptsize{$\mathrm{T}_2$}};
    \filldraw [black] (-0.88,-0.49) circle (1.25pt)node[anchor= north east] {\scriptsize{$\mathrm{T}_3$}};

    \draw (0,-0.95) -- (0,-1.05);
    \draw[rotate=120] (0,-0.95) -- (0,-1.05);
    \draw[rotate=-120] (0,-0.95) -- (0,-1.05);

    \draw [domain=0:6.28,variable=\t,smooth] plot ({sin(\t r)},{cos(\t r)});

    \node at (-1.11,0.1) {\scriptsize{$p$}};
    \node at (-0.24,-1.14) {\scriptsize{$\mathcal{K}(p)$}};
    \node at (-1.04,0.77) {\scriptsize{$\mathcal{K}^2(p)$}};
    \node at (1.36,-0.23) {\scriptsize{$\mathcal{K}^3(p)$}};
\end{tikzpicture}
\captionof{figure}{The concatenation of three projected heteroclinic orbits
onto $(\Sigma_1,\Sigma_2,\Sigma_3)$-space, which form part of a heteroclinic
chain described by iterates of the Kasner circle
map $\mathcal{K}$.}\label{FIG:KASNERMAPGR}
\end{figure}

In GR it is useful to define the \emph{Kasner parameters}
$(p_1,p_2,p_3)$ on the Kasner circle $\mathrm{K}^{\ocircle}$ according to
\begin{equation}\label{pdef}
{\Sigma}_{\alpha} = 3p_\alpha - 1, \qquad \text{for}\quad \alpha=1,2,3,
\end{equation}
where $p_1 + p_2 + p_3=1 = p_1^2 + p_2^2 + p_3^2$,
due to the constraints~\eqref{KasnerCircdef} on $\mathrm{K}^{\ocircle}$.

The Kasner circle $\mathrm{K}^{\ocircle}$ is described by six sectors
characterized by $p_\alpha < p_\beta < p_\gamma$, where $(\alpha\beta\gamma)$ is a
permutation of $(123)$. All sectors are related by
axis permutations given by~\eqref{permSYM}, see Figure~\ref{FIG:BIF}.
Each sector is half of an arc $\mathrm{int}(A_\alpha)$ when $v=1/2$,
excluding the boundary, which consists of the points $\mathrm{Q}_\alpha$
and $\mathrm{T}_\beta$ or $\mathrm{T}_\gamma$.

The Kasner parameters $(p_1,p_2,p_3)$ can be described by a single parameter $u$ such that
\begin{equation}\label{ueq}
p_\alpha=\frac{-u}{1+u+u^2},\qquad p_\beta=\frac{1+u}{1+u+u^2},\qquad p_\gamma = \frac{u(1+u)}{1+u+u^2},
\end{equation}
where $u\in(1,\infty)$ when $p_\alpha < p_\beta < p_\gamma$, while $u=1$ and $u=\infty$
at the boundary points of the sectors, $\mathrm{Q}_\alpha$ and $\mathrm{T}_\gamma$, respectively.

Invariance of $u$ under axis permutations follows from
$\Sigma_1 \Sigma_2 \Sigma_3 = 2 + 27 p_1 p_2 p_3$ where
\begin{equation}
p_1 p_2 p_3 =  \frac{-u^2(1+u)^2}{(1+u+u^2)^3}\, ,\qquad
\text{where}\quad u\in [1,\infty],
\end{equation}
which is monotone in $u$. In principle $u$ can be replaced by
$\Sigma_1 \Sigma_2 \Sigma_3$ or $p_1p_2p_3$ on $\mathrm{K}^{\ocircle}$.

The following theorem was shown in~\cite{bkl70,khaetal85}, see
also~\cite{ugg13a,ugg13b} and references therein.
\begin{theorem}\label{infchainsGR}
There is only a countable set of points in the set $\mathrm{K}^{\ocircle}$
associated with finite heteroclinic chains ending at a Taub point.
The set of points associated with periodic or infinite
heteroclinic chains is thereby topologically generic and has full measure.
\end{theorem}
There are different points of view regarding the genericity property.
A set is \emph{measure theoretical generic} if it has full measure.
On the other hand, a set is \emph{topologically generic} if it is a countable
intersection of dense open sets. Those definitions are not equivalent.
In physical empirical contexts measure theoretical genericity makes more sense,
since it is a property that is potentially observable.
%

The proof of Theorem~\ref{infchainsGR} follows from describing the Bianchi type II orbits in
the GR case using the \emph{Kasner map} (obtained from the Kasner circle/Mixmaster map
$\mathcal{K}$ in~\eqref{KasnerCirc} when $v=1/2$
by quoting out axis permutations) for the Kasner parameter $u$ in~\eqref{ueq}:
\begin{equation}\label{Gauss}
u\mapsto
\begin{cases}
u-1 & \text{ if } u\geq2\\
\frac{1}{u-1} & \text{ if } u<2
\end{cases},\qquad
u\in (1,+\infty).
\end{equation}

The properties of the Kasner map \eqref{Gauss} are intimately
connected with the properties of continued fraction expansions of $u$,
see~\cite{khaetal85,heiugg09a,rug94,monetal08}.
Using the parameter $u$ and number theory, we obtain
additional facts about heteroclinic chains: 
\begin{itemize}
\item Points in the set $\mathrm{K}^{\ocircle}$ associated with finite
heteroclinic chains correspond to $u\in \mathbb{Q}$, whereas $u\not\in \mathbb{Q}$ yields periodic or
infinite heteroclinic chains.
\item Points in the set $\mathrm{K}^{\ocircle}$ associated with periodic heteroclinic chains
are dense. They correspond to Kasner parameters $u$ with periodic continued fraction expansions.
\item Heteroclinic chains with points that are a finite distance away from the
Taub points are non-generic, whereas chains with points that come arbitrarily
close to the Taub points are generic.
\end{itemize}

The usefulness of the Kasner parameter $u$ in the GR case is due to the simplicity
of the map induced by the Bianchi type II solutions, described in~\eqref{Gauss}.
This simplicity and its relationship to continued fraction expansions
and number theory is lost when $v\neq 1/2$. Nevertheless, for different values of $v$
we will establish some common elements using symbolic dynamics, such as the chaoticity
of the Kasner circle map ${\cal K}$.

Recall that the map $\mathcal{K}$ is \emph{chaotic} if it is
topologically mixing and periodic
orbits are dense. Mathematically the former means that
given any open sets $A,B\subseteq \mathrm{K}^{\ocircle}$, the $n$-th iteration
$\mathcal{K}^n(A)$ intersects $B$ for sufficiently large $n$; 
the latter means that given $p\in \mathrm{K}^{\ocircle}$,
there is a periodic heteroclinic chain $q\in U$ for every
neighborhood $U\subseteq \mathrm{K}^{\ocircle}$ of $p$.
A popular description of chaos includes sensitivity of initial
conditions, but we omit this requirement since it is a consequence
of topological mixing and density of periodic orbits.

In order to prove that the discrete dynamical system generated by
iterates of the Kasner map~\eqref{Gauss} is chaotic, we follow~\cite{mawai92} and introduce
the inverse of the Kasner parameter $x=1/u$, which leads to the \emph{Farey map}
on the unit interval,
%
%
%
\begin{equation}\label{Farey}
x \mapsto \: \left\{\begin{array}{ll}
\frac{x}{1 - x}  & \qquad \text{if}\quad 0 \leq x \leq \sfrac12, \\[1ex]
\frac{1 - x}{x}  &\qquad  \text{if} \quad \sfrac12 \leq x \leq 1.
\end{array}\right.
\end{equation}
Then note that $x = (\sqrt{13} - 1)/6$ is a periodic point with minimal period 3,
see \cite{ma88} and also~~\cite{waiell97,rug94}. Therefore the `period 3 implies
chaos theorem' applies, proved independently by Sharkovsky~\cite{Shark64} and
Li and Yorke~\cite{LiYorke75}. Thus the iterates of the Farey map~\eqref{Farey}
generate a chaotic discrete dynamical system on $[ 0,1 ]$.
This leads to the following theorem:
\begin{theorem}\label{thm:GRchaos}
The Kasner map \eqref{Gauss} is generically chaotic.
\end{theorem}
%

It is also possible to prove chaoticity of the Kasner circle map
$\mathcal{K}$ in the critical case by using symbolic dynamics, although
there is a technicality arising in the encoding of the Taub
points into symbolic sequences. We provide such a new proof in
Appendix~\ref{app:unifying}, which modifies the proof in
Section~\ref{sec:superINF} for the supercritical case $v\in (1/2,1)$,
and shows how chaoticity is carried from the supercritical case to
the critical case $v=1/2$.

\section{Supercritical case}\label{sec:superINF}


In the supercritical case, $v\in (1/2,1)$, the Kasner circle map $\mathcal{K}$ admits
a \emph{closed} set of fixed points called the \emph{stable set} $S$, defined by
$S:=\mathrm{K}^{\ocircle}\setminus{ \mathrm{int}(A_1)\cup \mathrm{int}(A_2) \cup \mathrm{int}(A_3)}$,
where the interior of $S$ contains fixed points of the dynamical system~\eqref{intro_dynsyslambdaR}
with only negative eigenvalues in the eigendirections normal to the Kasner
circle $\mathrm{K}^{\ocircle}$, see Figures~\ref{FIG:BIF} and~\ref{FIG:KASNERMAPS}.
The set $S$ represents the end of heteroclinic chains.
Accordingly, periodic and infinite heteroclinic chains are trajectories under the
map $\mathcal{K}$ never ending at the set $S$.

The set $C$ of initial conditions leading to periodic and
infinite heteroclinic chains is thereby defined by
\begin{equation}\label{defofC}
C:= \{ p\in \mathrm{K}^{\ocircle} \text{ $ | $ } \mathcal{K}^n(p) \notin S  \text{ for all }  n \in \mathbb{N}_0 \}.
\end{equation}
For example, the two chains with period 3 obtained in Lemma~\ref{lem:Period3}, depicted in
Figure~\ref{FIG:period3}, and the three chains with
period 2, see Figure~\ref{FIG:minMAX}, are contained in $C$.

The complement of the set $C$ in $\mathrm{K}^{\ocircle}$ is defined by
\begin{equation}\label{defofF}
F:=\mathrm{K}^{\ocircle} \setminus{C} =  \{ p\in \mathrm{K}^{\ocircle} \text{ $ | $ }
\mathcal{K}^n(p) \in S \text{ for some } n \in \mathbb{N}_0\} .
\end{equation}
Two of our main results describe properties of the set $C$,
and the associated dynamics of the Kasner map $\mathcal{K}$:
\begin{theorem}\label{CantorTh}
The set $C$ is a nonempty Cantor set of Lebesgues measure zero
and a Hausdorff dimension $d_H(C)$ satisfying
\begin{equation}\label{HausCantorBounds}
d_H(C) \in \left[\frac{\log(2)}{\log\left(\frac{2+v^2}{1-v^2}\right)}, \min
\left\{1,\frac{\log(2)}{\log \left(\frac{2(1-v^2)}{1 + 2v^2 - \sqrt{12v^2 - 3}}\right)}\right\}\right].
\end{equation}
%
\end{theorem}
The bounds in \eqref{HausCantorBounds} are positive and well-defined when $v\in (1/2,1)$,
see Figure~\ref{fig:plotdim}. In particular, $12v^2 - 3>0$, where the
$v$-dependent arguments, both larger than 1, 
are the minimum and maximum expansion rates of the Kasner circle map $\mathcal{K}$  restricted
to the Cantor set $C$, as will be shown in Lemma~\ref{MaxMinExpansion}.

As a consequence of Theorem~\ref{CantorTh}, the Cantor set $C$ is non-generic
both in a measure theoretical and a topological sense.
The former is not always true, since there are some Cantor sets with positive measure,
such as the Smith–-Volterra–-Cantor set; the latter follows since Cantor sets by definition are
closed and nowhere dense.
\begin{figure}[H]
\centering
\begin{tikzpicture}[scale=4]
        \fill [lightgray!50, domain=0.5:0.582, variable=\t]
        plot ({\t},{ln(2)/(ln(2+(\t)^2)-ln(1-(\t)^2))})
        -- (0.582,1)
        -- (0.5,1)
        -- cycle;

        \fill [lightgray!50, domain=0.58:0.99, variable=\t]
        (0.58,0.075)
        -- plot ({\t},{ln(2)/(ln(2-2*(\t)^2)-ln(2*(\t)^2+1-(12*(\t)^2-3)^(1/2)))})
        -- (0.99,0.075)
        -- cycle;

        \fill [white, domain=0.5:1, variable=\t]
        (0.5,0.075)
        -- plot ({\t},{ln(2)/(ln(2+(\t)^2)-ln(1-(\t)^2))})
        -- (1,0.075)
        -- cycle;

        \draw[->] (0,0.025) -- (0,1.1)node[anchor=south] {$d_H(C)$};
        \draw[color=gray,dashed] (0.5,1) -- (0,1) node[anchor=east] {\color{black} $1$};
        \draw[->] (-0.05,0.075) -- (1.07,0.075)node[anchor=west] {$v$};
        \draw[color=gray,dashed] (0.5,0.63) -- (0,0.63) node[anchor=east] {\color{black} $\frac{\log(2)}{\log(3)}$};
        \draw[color=gray,dashed] (0.5,1) -- (0.5,0.075) node[anchor=north] {\color{black} \footnotesize{$1/2$}};\draw[dashed] (0.5,1) -- (0.5,0.63);
        \draw[-] (1,0.075) -- (1,0.075) node[anchor=north] {\footnotesize{$1$}};

        \draw [domain=0.5:1,variable=\t,smooth] plot ({\t},{ln(2)/(ln(2+(\t)^2)-ln(1-(\t)^2))});
        \node at (0.71,0.2) {\scriptsize{$\frac{\log(2)}{\log\left(\frac{2+v^2}{1-v^2}\right)}$}};

        \draw[-] (0.5,1) -- (0.58,1);

        \draw [domain=0.577:0.99,variable=\t,smooth] plot ({\t},{ln(2)/(ln(2-2*(\t)^2)-ln(2*(\t)^2+1-(12*(\t)^2-3)^(1/2)))});
        \node at (1.1,0.63) {\scriptsize{$\frac{\log(2)}{\log \left(\frac{2(1-v^2)}{1+2v^2-\sqrt{12 v^2 - 3}}\right)}$}};
\end{tikzpicture}
\caption{The Hausdorff dimension $d_H(C)$ in \eqref{HausCantorBounds} resides in the shaded region.}\label{fig:plotdim}
\end{figure}
\begin{theorem} \label{chaossub}
The Kasner circle map $\mathcal{K}$ restricted to the Cantor set $C$ generates a chaotic discrete dynamical system.
\end{theorem}
Outside the Cantor set $C$ the map $\mathcal{K}$ is not chaotic, since heteroclinic chains end in the stable set $S$ after finitely many iterations.

Let us now compare the GR critical case $v=1/2$ with the supercritical
case $v\in(1/2,1)$, in view of Theorems~\ref{CantorTh} and~\ref{chaossub}.
For GR, the set $S$ is the union of the three Taub points,
while the analog of $C$ is the set of points never reaching the
Taub points under some iteration of $\mathcal{K}$, i.e., 
the set associated with periodic or infinite heteroclinic chains.
This set, however, is not a Cantor set since it is not closed ---
its closure is the whole Kasner circle $\mathrm{K}^{\ocircle}$,
which is different than itself. 
Furthermore, this set is generic in both a measure theoretical and a topological sense,
while its complement 
is countable, see Theorem~\ref{infchainsGR}.
In conclusion, the generic chaos for $v=1/2$ is carried by the non-generic
set $C$ when $v\in(1/2,1)$.

The remaining section is divided into four parts. First, a background
on Cantor sets and their dimensionality.
Second, we describe how $C$ is iteratively constructed, which is the basis for
Theorems~\ref{CantorTh} and~\ref{chaossub}. Third, we characterize
the connected components in each iterative step by means of
symbolic dynamics. Lastly, we prove Theorems~\ref{CantorTh} and~\ref{chaossub}.

\subsection{Background: Cantor sets and Hausdorff dimension}
A non-empty set $C$ in a complete metric space 
is a \emph{Cantor set}
if it is perfect, i.e., it is closed, without isolated points,
and nowhere dense, i.e., its closure has an empty interior.

%
%
%
%
%
%
As an illustration, consider the iteratively constructed ternary Cantor set.
%
Let $T_0$ be the unit
interval. The set $T_{n+1}$ is obtained from $T_n$ by removing the
open middle third of each connected component of $T_n$,
see Figure~\ref{figure4}. Then define $T$ as
\begin{equation}\label{defternary}
T :=\bigcap_{n \in \mathbb{N}_0} T_n.
\end{equation}
In all steps $n\geq 1$ of the construction, we can encode each closed connected
component of $T_n$ by a sequence of symbols $L$ or $R$, which respectively
denotes the left or right connected component from the previous iterations,
see Figure~\ref{figure4}.
From this it follows that $T$ fulfills the abstract definition of a
Cantor set 
and has measure zero. 
A similar procedure will be adapted in order to construct the connected
components of the set $C$ in~\eqref{defofC}, and prove that it is also a Cantor set of measure zero.
\begin{figure}[H]\centering
\begin{tikzpicture}[scale=0.5]
        \draw[line width=2pt] (0,0) -- (27,0) node[anchor=west] {$T_0$};

        \draw[shift={(0,-0.2)},line width=2pt] (0,-1) -- (9,-1);
        \filldraw [shift={(0,-0.2)} ] (4.5,-1) circle (0.001pt) node[anchor=south] {\scriptsize{$L$}};
        \draw[shift={(0,-0.2)},line width=2pt] (18,-1) -- (27,-1) node[anchor=west] {$T_1$};
        \filldraw [shift={(0,-0.2)} ] (22.5,-1) circle (0.001pt) node[anchor=south] {\scriptsize{$R$}};

        \draw[shift={(0,-0.4)},line width=2pt] (0,-2) -- (3,-2);
        \filldraw [shift={(0,-0.4)} ] (1.5,-2) circle (0.001pt) node[anchor=south] {\scriptsize{$LL$}};
        \draw[shift={(0,-0.4)},line width=2pt] (6,-2) -- (9,-2);
        \filldraw [shift={(0,-0.4)} ] (7.5,-2) circle (0.001pt) node[anchor=south] {\scriptsize{$LR$}};
        \draw[shift={(0,-0.4)},line width=2pt] (18,-2) -- (21,-2);
        \filldraw [shift={(0,-0.4)} ] (19.5,-2) circle (0.001pt) node[anchor=south] {\scriptsize{$RL$}};
        \draw[shift={(0,-0.4)},line width=2pt] (24,-2) -- (27,-2) node[anchor=west] {$T_2$};
        \filldraw [shift={(0,-0.4)} ] (25.5,-2) circle (0.001pt) node[anchor=south] {\scriptsize{$RR$}};

        \draw[shift={(0,-0.6)},line width=2pt] (0,-3) -- (1,-3);
        \filldraw [shift={(0,-0.6)} ] (0.5,-3) circle (0.001pt) node[anchor=south] {\scriptsize{$LLL$}};
        \draw[shift={(0,-0.6)},line width=2pt] (2,-3) -- (3,-3);
        \filldraw [shift={(0,-0.6)} ] (2.5,-3) circle (0.001pt) node[anchor=south] {\scriptsize{$LLR$}};
        \draw[shift={(0,-0.6)},line width=2pt] (6,-3) -- (7,-3);
        \filldraw [shift={(0,-0.6)} ] (6.5,-3) circle (0.001pt) node[anchor=south] {\scriptsize{$LRL$}};
        \draw[shift={(0,-0.6)},line width=2pt] (8,-3) -- (9,-3);
        \filldraw [shift={(0,-0.6)} ] (8.5,-3) circle (0.001pt) node[anchor=south] {\scriptsize{$LRR$}};

        \draw[shift={(0,-0.6)},line width=2pt] (18,-3) -- (19,-3);
        \filldraw [shift={(0,-0.6)} ] (18.5,-3) circle (0.001pt) node[anchor=south] {\scriptsize{$RLL$}};
        \draw[shift={(0,-0.6)},line width=2pt] (20,-3) -- (21,-3);
        \filldraw [shift={(0,-0.6)} ] (20.5,-3) circle (0.001pt) node[anchor=south] {\scriptsize{$RLR$}};
        \draw[shift={(0,-0.6)},line width=2pt] (24,-3) -- (25,-3);
        \filldraw [shift={(0,-0.6)} ] (24.5,-3) circle (0.001pt) node[anchor=south] {\scriptsize{$RRL$}};
        \draw[shift={(0,-0.6)},line width=2pt] (26,-3) -- (27,-3) node[anchor=west] {$T_3$};
        \filldraw [shift={(0,-0.6)} ] (26.5,-3) circle (0.001pt) node[anchor=south] {\scriptsize{$RRR$}};

        \filldraw [black] (13.5,-3.75) circle (0.75pt);
        \filldraw [black] (13.5,-4) circle (0.75pt);
        \filldraw [black] (13.5,-4.25) circle (0.75pt);

        \draw[shift={(0,-0.25)},scale=2.5] (0.1,0) -- (0,0) -- (0,0.2) -- (0.1,0.2);
        \draw[shift={(27,-0.25)},scale=2.5] (-0.1,0) -- (0,0) -- (0,0.2) -- (-0.1,0.2);

        \draw[shift={(0,-1.45)},scale=2.5] (0.1,0) -- (0,0) -- (0,0.2) -- (0.1,0.2);
        \draw[shift={(27,-1.45)},scale=2.5] (-0.1,0) -- (0,0) -- (0,0.2) -- (-0.1,0.2);
        \draw[shift={(18,-1.45)},scale=2.5] (0.1,0) -- (0,0) -- (0,0.2) -- (0.1,0.2);
        \draw[shift={(9,-1.45)},scale=2.5] (-0.1,0) -- (0,0) -- (0,0.2) -- (-0.1,0.2);

        \draw[shift={(0,-2.65)},scale=2.5] (0.1,0) -- (0,0) -- (0,0.2) -- (0.1,0.2);
        \draw[shift={(27,-2.65)},scale=2.5] (-0.1,0) -- (0,0) -- (0,0.2) -- (-0.1,0.2);
        \draw[shift={(18,-2.65)},scale=2.5] (0.1,0) -- (0,0) -- (0,0.2) -- (0.1,0.2);
        \draw[shift={(9,-2.65)},scale=2.5] (-0.1,0) -- (0,0) -- (0,0.2) -- (-0.1,0.2);
        \draw[shift={(6,-2.65)},scale=2.5] (0.1,0) -- (0,0) -- (0,0.2) -- (0.1,0.2);
        \draw[shift={(3,-2.65)},scale=2.5] (-0.1,0) -- (0,0) -- (0,0.2) -- (-0.1,0.2);
        \draw[shift={(24,-2.65)},scale=2.5] (0.1,0) -- (0,0) -- (0,0.2) -- (0.1,0.2);
        \draw[shift={(21,-2.65)},scale=2.5] (-0.1,0) -- (0,0) -- (0,0.2) -- (-0.1,0.2);

        \draw[shift={(0,-3.85)},scale=2.5] (0.1,0) -- (0,0) -- (0,0.2) -- (0.1,0.2);
        \draw[shift={(27,-3.85)},scale=2.5] (-0.1,0) -- (0,0) -- (0,0.2) -- (-0.1,0.2);
        \draw[shift={(18,-3.85)},scale=2.5] (0.1,0) -- (0,0) -- (0,0.2) -- (0.1,0.2);
        \draw[shift={(9,-3.85)},scale=2.5] (-0.1,0) -- (0,0) -- (0,0.2) -- (-0.1,0.2);
        \draw[shift={(6,-3.85)},scale=2.5] (0.1,0) -- (0,0) -- (0,0.2) -- (0.1,0.2);
        \draw[shift={(3,-3.85)},scale=2.5] (-0.1,0) -- (0,0) -- (0,0.2) -- (-0.1,0.2);
        \draw[shift={(24,-3.85)},scale=2.5] (0.1,0) -- (0,0) -- (0,0.2) -- (0.1,0.2);
        \draw[shift={(21,-3.85)},scale=2.5] (-0.1,0) -- (0,0) -- (0,0.2) -- (-0.1,0.2);

        \draw[shift={(2,-3.85)},scale=2.5] (0.1,0) -- (0,0) -- (0,0.2) -- (0.1,0.2);
        \draw[shift={(25,-3.85)},scale=2.5] (-0.1,0) -- (0,0) -- (0,0.2) -- (-0.1,0.2);
        \draw[shift={(20,-3.85)},scale=2.5] (0.1,0) -- (0,0) -- (0,0.2) -- (0.1,0.2);
        \draw[shift={(7,-3.85)},scale=2.5] (-0.1,0) -- (0,0) -- (0,0.2) -- (-0.1,0.2);
        \draw[shift={(8,-3.85)},scale=2.5] (0.1,0) -- (0,0) -- (0,0.2) -- (0.1,0.2);
        \draw[shift={(1,-3.85)},scale=2.5] (-0.1,0) -- (0,0) -- (0,0.2) -- (-0.1,0.2);
        \draw[shift={(26,-3.85)},scale=2.5] (0.1,0) -- (0,0) -- (0,0.2) -- (0.1,0.2);
        \draw[shift={(19,-3.85)},scale=2.5] (-0.1,0) -- (0,0) -- (0,0.2) -- (-0.1,0.2);

    \end{tikzpicture}
\caption{ \label{figure4} The iterations $T_n$ for $n=0,1,2,3$ in the construction of the ternary Cantor set.
Note that $T_1$ has two closed connected components, a left and right, denoted by $L$ and $R$.
In the next step, each of those two components $L$ and $R$ has two
further left and right closed connected components in $T_2$, denoted by $LL,LR,RL,RR$. Similarly for $T_3$, and onwards.
}
\end{figure}
%
A natural question regarding Cantor sets is their dimensionality.
There are several notions of dimension, each with its advantages
and disadvantages, see~\cite{Falc13}. By introducing the Hausdorff dimension
a set within the Kasner circle $\mathrm{K}^{\ocircle}$
can have any real dimension
between 0 and 1: more than a discrete set of points, less
than the circle itself. 

Given $d\in [0,\infty)$, for any $\epsilon>0$ the \emph{$d$-Hausdorff measure of $C$} is
\begin{equation}
\mu^d(C):=\lim_{\epsilon\to 0} \inf \left\{ \sum_{i\in\mathbb{N}}
[\mathrm{diam} (U_i)]^d \Bigm| C\subseteq \cup_{i\in\mathbb{N}} U_i \text{ with }
\mathrm{diam}(U_i)\leq \epsilon\right\}.
\end{equation}
That is, consider all coverings $\cup_{i\in\mathbb{N}} U_i$ of $C$ such that each $U_i$ has a
diameter\footnote{Recall that the diameter is defined as $\mathrm{diam} (U_i):=\sup \{\rho(x,y) : x,y\in U_i\}$,
where $\rho(x,y)$ is the metric between $x$ and $y$ in the metric space $(X,\rho)$.}
at most $\epsilon$ minimizing the sum of the $d^{th}$ powers of the diameters.
As $\epsilon$ decreases the number of possible covers is reduced, 
which accounts for the roughness of $C$: the more detailed a shape is, the more impact decreasing
$\epsilon$ has. 
The value of $d$ incorporates the behavior of shapes under rescaling in a $d$-dimensional space:
scaling a set $C$ with a factor $r$ will scale its $d$-Hausdorff measure
with a factor $r^d$.

The \emph{Hausdorff dimension of $C$} is defined as
\begin{equation}
d_H(C):=\inf_{d\geq 0} \{\mu^d(C)=0\}=\sup_{d\geq 0} \{\mu^d(C)=\infty\}.
\end{equation}
The measure $\mu^d(C)$ is therefore $0$ when $d>d_H(C)$ and $\infty$
for $d<d_H(C)$ so that $d_H(C)$ is a value such that the
measure $\mu^d(C)$ jumps from $\infty$ to $0$.
This means that if $\mu^{d}(C)$ is positive and bounded for some $d$, then
this value of $d$ is the Hausdorff dimension $d_H(C)$. Intuitively,
we compare the $d$-dimensional scaling of the surrounding space with the set $C$:
for too large $d$ the set $C$ will be of negligible size
(of measure zero), whereas if $d$ is too small then this leads to an over-sized $C$
(infinite measure).

For example, the ternary Cantor set $T$ in equation~\eqref{defternary} has a Hausdorff
dimension given by $d_H(T)=\log(2)/\log(3) \approx 0.631$,
as can be seen as follows: Divide the Cantor set $T$ into its left
$T_L:=T\cap [0,1/2]$ and right $T_R:=T\cap [1/2,1]$ disjoint parts.
Then $\mu^d(T)=\mu^d(T_L)+\mu^d(T_R)$. 
Moreover, since $T_L$ and $T_R$ have the same measure and are scalings of $T$
by a factor $3^{-1}$, which scales the measure by $3^{-d}$, it follows that
\begin{equation}\label{TernaryCdim}
\mu^d(T)= 2\cdot 3^{-d}\mu^d(T).
\end{equation}
If $\mu^{d}(T)\neq 0,\infty$ for some $d\geq 0$, then it can be divided out
yielding $1=2\cdot 3^{-d}$, and its logarithm provides the desired dimension.

\subsection{Characterization of $C$ through iterations}

Analogously to the ternary Cantor set in equation~\eqref{defternary},
which is obtained by iteratively removing an open middle third from an interval,
the set $C$ in equation~\eqref{defofC} can be iteratively constructed by removing arcs
given by pre-images of $S$ via $\mathcal{K}$ from $\mathrm{K}^{\ocircle}$.

For each $n\in\mathbb{N}_0$, consider the removal process iteratively defined by
\begin{subequations}\label{Fn}
\begin{align}
C_0 &:= \mathrm{K}^{\ocircle},\label{Citerativ0}  \\
F_n &:= \mathrm{int}_{C_n}(\mathcal{K}^{-n}(S)),  \label{CiterativFn}\\
C_{n+1} &:= C_n\setminus{F_n}, \label{CiterativCn+1}
\end{align}
\end{subequations}
where $\mathrm{int}_{C_n}(B):=\mathrm{int}(B\cap C_n)\cup (B\cap \partial C_n)$
denotes the interior of the set $B=\mathcal{K}^{-n}(S)$ relative to $C_n$.
See Figures~\ref{FIG:KASNERMAPSpreimages} and~\ref{FIG:KASNERMAPiterative}
for a visualization of the process defined in
equation~\eqref{Fn}, which we now describe in more detail.

The \emph{removed set} $F_n$ consists of
two different types of points within $C_n$,
$\mathrm{int}(\mathcal{K}^{-n}(S))$ and $\partial C_n$,
since we can write $F_n$ as
$F_n = \mathrm{int}(\mathcal{K}^{-n}(S))\cup \partial C_n$.
This is due to the fact that the points in $F_n$ either have an $n^{th}$ iteration
$\mathcal{K}^{n}(p)$ that falls in the interior of the stable
set, $\mathrm{int}(S)$, or points whose $(n-1)^{th}$
iteration $\mathcal{K}^{n-1}(p)$ ends at one of the tangential
points, which are the boundary points of $C_n$.
Thus the tangential points and their pre-images are all
eventually removed. As a consequence, $C_{n+1}$ is the \emph{closed}
set of points that remains after removing $F_n$ from $C_n$.
The removal procedure, which is analogous to the removal process of the
ternary Cantor set depicted in Figure~\ref{figure4}, is illustrated in
Figures~\ref{FIG:KASNERMAPSpreimages} and~\ref{FIG:KASNERMAPiterative}.
.

The set $C$ is obtained as the intersection of the sets $(C_n)_{n\in \mathbb{N}_0}$ according
to Lemma~\ref{CantorEq} below, which is proved to be a nested sequence of closed sets
in Lemmata~\ref{Lemma2onIN} and~\ref{Lemma1onIN}.
\begin{figure}[H]\centering
\minipage{0.49\textwidth}\centering
\begin{subfigure}\centering
    \begin{tikzpicture}[scale=2]
    \draw [line width=0.1pt,domain=0:6.28,variable=\t,smooth] plot ({sin(\t r)},{cos(\t r)});

    \draw [ultra thick, dotted, white, domain=-0.26:0.26,variable=\t,smooth] plot ({sin(\t r)},{cos(\t r)});
    \draw [rotate=120,ultra thick, dotted, white, domain=-0.26:0.26,variable=\t,smooth] plot ({sin(\t r)},{cos(\t r)});
    \draw [rotate=-120,ultra thick, dotted, white, domain=-0.26:0.26,variable=\t,smooth] plot ({sin(\t r)},{cos(\t r)});

    \node at (0,1.1) {\scriptsize{$\mathrm{int}(S)$}};
    \node at (-1.1,-0.57) {\scriptsize{$\mathrm{int}(S)$}};
    \node at (1.1,-0.57) {\scriptsize{$\mathrm{int}(S)$}};

    \draw[color=gray,dashed,-] (0,-1.35) -- (0.685,-0.75);
    \draw[color=gray,dashed,-] (0,-1.35) -- (-0.685,-0.75);

    \draw[color=gray,rotate=120,dashed,-] (0,-1.35) -- (0.685,-0.75);
    \draw[color=gray,rotate=120,dashed,-] (0,-1.35) -- (-0.685,-0.75);

    \draw[color=gray,rotate=-120,dashed,-] (0,-1.35) -- (0.685,-0.75);
    \draw[color=gray,rotate=-120,dashed,-] (0,-1.35) -- (-0.685,-0.75);

    \draw[color=gray,dotted] (0.28,0.95) -- (0,-1.35);
    \draw[color=gray,dotted] (-0.28,0.95) -- (0,-1.35);

    \draw[color=gray,dotted] (-0.95,-0.26) -- (0,-1.35);
    \draw[color=gray,dotted] (0.95,-0.26) -- (0,-1.35);

    \draw[white, ultra thick] (0.28,0.95) -- (0.044,-1);
    \draw[dotted, thick, postaction={decorate}] (0.044,-1) -- (0.28,0.95);
    \draw[white, ultra thick] (-0.28,0.95) -- (-0.044,-1);
    \draw[dotted, thick, postaction={decorate}] (-0.044,-1) -- (-0.28,0.95);

    \draw[white, ultra thick] (-0.95,-0.26) -- (-0.39,-0.9);
    \draw[dotted, thick, postaction={decorate}] (-0.39,-0.9) -- (-0.95,-0.26);
    \draw[white, ultra thick] (0.95,-0.26) -- (0.39,-0.9);
    \draw[dotted, thick, postaction={decorate}] (0.39,-0.9) -- (0.95,-0.26);


    \draw [line width=2pt, domain=3.19:3.525,variable=\t,smooth] plot ({sin(\t r)},{cos(\t r)});
    \draw [line width=2pt, domain=-3.19:-3.525,variable=\t,smooth] plot ({sin(\t r)},{cos(\t r)});


    \draw[shift={(0.25,0.87)},rotate=-20, scale=0.9] (0.1,0) -- (0,0) -- (0,0.2) -- (0.1,0.2);\filldraw [black] (0.25,1.1) circle (0.001pt) node[anchor= west] {\scriptsize{$\mathrm{t}_{12}$}};
    \draw[shift={(-0.25,0.87)},rotate=20, scale=0.9] (-0.1,0) -- (0,0) -- (0,0.2) -- (-0.1,0.2);\filldraw [black] (-0.25,1.125) circle (0.001pt) node[anchor= east] {\scriptsize{$\mathrm{t}_{13}$}};

    \draw[shift={(-1.07,-0.25)},rotate=-80, scale=0.9] (-0.1,0.2) -- (0,0.2) -- (0,0) -- (-0.1,0)  node[anchor= east] {\scriptsize{$\mathrm{t}_{31}$}};
    \draw[shift={(-0.63,-0.67)},rotate=140, scale=0.9] (-0.1,0) -- (0,0) -- (0,0.2) -- (-0.1,0.2) node[anchor= north east] {\scriptsize{$\mathrm{t}_{32}$}};

    \draw[shift={(1.07,-0.25)},rotate=80, scale=0.9] (0.1,0.2) -- (0,0.2) -- (0,0) -- (0.1,0)  node[anchor= west] {\scriptsize{$\mathrm{t}_{21}$}};
    \draw[shift={(0.63,-0.67)},rotate=-140, scale=0.9] (0.1,0) -- (0,0) -- (0,0.2) -- (0.1,0.2) node[anchor= north west] {\scriptsize{$\mathrm{t}_{23}$}};

    \draw (0,-1) circle (0.001pt) node[anchor= north] {\scriptsize{$F_1$}};
    \draw [rotate=30] (0,-1) circle (0.001pt) node[anchor= north] {\scriptsize{$F_1$}};
    \draw [rotate=-30] (0,-1) circle (0.001pt) node[anchor= north] {\scriptsize{$F_1$}};

    \draw [rotate=11] (0,-1) circle (0.001pt) node[anchor= south] {\scriptsize{$C_2$}};
    \draw [rotate=-11] (0,-1) circle (0.001pt) node[anchor= south] {\scriptsize{$C_2$}};

    \end{tikzpicture}
    \addtocounter{subfigure}{-1}\captionof{subfigure}{\footnotesize{
    The set $F_1$ in $A_1$ is obtained as follows: It is the interior (in the arc $A_1$)
    of the pre-image $\mathcal{K}^{-1}(S)$, and hence contains the tangential points
    $\mathrm{t}_{32}$ and $\mathrm{t}_{23}$, but not the pre-image of the other four
    tangential points. The set $F_1$ has three connected components in $A_1$.
    Removing $F_1$ from $C_1$ yields the (thicker bold) closed set $C_2$ with two connected
    components in $A_1$. Repeating this argument for the other arcs provides the full
    set $C_2$.}}\label{FIG:REMOVAL1}
\end{subfigure}
\endminipage\hfill
\minipage{0.49\textwidth}\centering
\begin{subfigure}\centering
    \begin{tikzpicture}[scale=2]
    \draw [line width=0.1pt,domain=0:6.28,variable=\t,smooth] plot ({sin(\t r)},{cos(\t r)});

    \draw [ultra thick, dotted, white, domain=-0.26:0.26,variable=\t,smooth] plot ({sin(\t r)},{cos(\t r)});
    \draw [rotate=120,ultra thick, dotted, white, domain=-0.26:0.26,variable=\t,smooth] plot ({sin(\t r)},{cos(\t r)});
    \draw [rotate=-120,ultra thick, dotted, white, domain=-0.26:0.26,variable=\t,smooth] plot ({sin(\t r)},{cos(\t r)});

    \node at (0,1.1) {\scriptsize{$\mathrm{int}(S)$}};
    \node at (-1.1,-0.57) {\scriptsize{$\mathrm{int}(S)$}};
    \node at (1.1,-0.57) {\scriptsize{$\mathrm{int}(S)$}};

    \draw[color=gray,dashed,-] (0,-1.35) -- (0.685,-0.75);
    \draw[color=gray,dashed,-] (0,-1.35) -- (-0.685,-0.75);

    \draw[color=gray,rotate=120,dashed,-] (0,-1.35) -- (0.685,-0.75);
    \draw[color=gray,rotate=120,dashed,-] (0,-1.35) -- (-0.685,-0.75);

    \draw[color=gray,rotate=-120,dashed,-] (0,-1.35) -- (0.685,-0.75);
    \draw[color=gray,rotate=-120,dashed,-] (0,-1.35) -- (-0.685,-0.75);

    \draw[color=gray,dotted] (0.28,0.95) -- (0,-1.35);
    \draw[color=gray,dotted] (-0.28,0.95) -- (0,-1.35);

    \draw[color=gray,dotted] (0.6,0.775) -- (0,-1.35);
    \draw[color=gray,dotted] (-0.6,0.775) -- (0,-1.35);

    \draw[color=gray,dotted] (0.83,0.525) -- (0,-1.35);
    \draw[color=gray,dotted] (-0.83,0.525) -- (0,-1.35);

    \draw[color=gray,dotted] (0.875,0.45) -- (0,-1.35);
    \draw[color=gray,dotted] (-0.875,0.45) -- (0,-1.35);

    \draw[color=gray,dotted] (0.975,0.125) -- (0,-1.35);
    \draw[color=gray,dotted] (-0.975,0.125) -- (0,-1.35);

    \draw[color=gray,dotted] (-0.95,-0.26) -- (0,-1.35);
    \draw[color=gray,dotted] (0.95,-0.26) -- (0,-1.35);

    \draw[white,ultra thick] (0.28,0.95) -- (0.044,-1);
    \draw[dotted, thick, postaction={decorate}] (0.044,-1) -- (0.28,0.95);
    \draw[white,ultra thick] (-0.28,0.95) -- (-0.044,-1);
    \draw[dotted, thick, postaction={decorate}] (-0.044,-1) -- (-0.28,0.95);

    \draw[white,ultra thick] (0.6,0.775) -- (0.1,-0.99);
    \draw[dotted, thick, postaction={decorate}] (0.1,-0.99) -- (0.6,0.775);
    \draw[white,ultra thick] (-0.6,0.775) -- (-0.1,-0.99);
    \draw[dotted, thick, postaction={decorate}] (-0.1,-0.99) -- (-0.6,0.775);

    \draw[white,ultra thick] (0.83,0.525) -- (0.165,-0.97);
    \draw[dotted, thick, postaction={decorate}] (0.165,-0.97) -- (0.83,0.525);
    \draw[white,ultra thick] (-0.83,0.525) -- (-0.165,-0.97);
    \draw[dotted, thick, postaction={decorate}] (-0.165,-0.97) -- (-0.83,0.525);

    \draw[white,ultra thick] (0.875,0.45) -- (0.2,-0.95);
    \draw[dotted, thick, postaction={decorate}] (0.2,-0.95) -- (0.875,0.45);
    \draw[white,ultra thick] (-0.875,0.45) -- (-0.2,-0.95);
    \draw[dotted, thick, postaction={decorate}] (-0.2,-0.95) -- (-0.875,0.45);

    \draw[white,ultra thick] (0.975,0.125) -- (0.26,-0.96);
    \draw[dotted, thick, postaction={decorate}] (0.26,-0.96) -- (0.975,0.125);
    \draw[white,ultra thick] (-0.975,0.125) -- (-0.26,-0.96);
    \draw[dotted, thick, postaction={decorate}] (-0.26,-0.96) -- (-0.975,0.125);

    \draw[white,ultra thick] (0.95,-0.26) -- (0.39,-0.9);
    \draw[dotted, thick, postaction={decorate}] (0.39,-0.9) -- (0.95,-0.26);
    \draw[white,ultra thick] (-0.95,-0.26) -- (-0.39,-0.9);
    \draw[dotted, thick, postaction={decorate}] (-0.39,-0.9) -- (-0.95,-0.26);


    \draw [line width=2pt,  domain=3.19:3.525,variable=\t,smooth, rotate=120] plot ({sin(\t r)},{cos(\t r)});
    \draw [line width=2pt,  domain=-3.19:-3.525,variable=\t,smooth, rotate=120] plot ({sin(\t r)},{cos(\t r)});

    \draw [line width=2pt,  domain=3.19:3.525,variable=\t,smooth, rotate=-120] plot ({sin(\t r)},{cos(\t r)});
    \draw [line width=2pt,  domain=-3.19:-3.525,variable=\t,smooth, rotate=-120] plot ({sin(\t r)},{cos(\t r)});

    \draw [line width=2pt,  domain=3.19:3.525,variable=\t,smooth] plot ({sin(\t r)},{cos(\t r)});
    \draw [line width=2pt,  domain=-3.19:-3.525,variable=\t,smooth] plot ({sin(\t r)},{cos(\t r)});

    \draw [line width=5pt,   domain=3.33:3.4,variable=\t,smooth] plot ({sin(\t r)},{cos(\t r)});
    \draw [line width=5pt,   domain=3.24:3.3,variable=\t,smooth] plot ({sin(\t r)},{cos(\t r)});

    \draw [line width=5pt,   domain=-3.33:-3.4,variable=\t,smooth] plot ({sin(\t r)},{cos(\t r)});
    \draw [line width=5pt,   domain=-3.24:-3.3,variable=\t,smooth] plot ({sin(\t r)},{cos(\t r)});
\end{tikzpicture}
    \addtocounter{subfigure}{-1}
    \captionof{subfigure}{\footnotesize{Repeating the argument in \textbf{(a)}
    for the arcs $A_2,A_3$ yields the (bold) closed set $C_2$ with six
    connected components. The set $F_2$ in $A_1$ is obtained by the pre-images
    of the (thin) sets, since those are the points
    reaching $\mathrm{int}(S)$ in two iterations of $\mathcal{K}$. The set $F_2$ has
    six connected components in $A_1$, which when removed yields
    the (thicker bold) closed set $C_3$ in $A_1$.
    Figure \textbf{(c)} reveals more details for the bottom arc $A_1$.}}\label{FIG:REMOVAL2}
\end{subfigure}
\endminipage

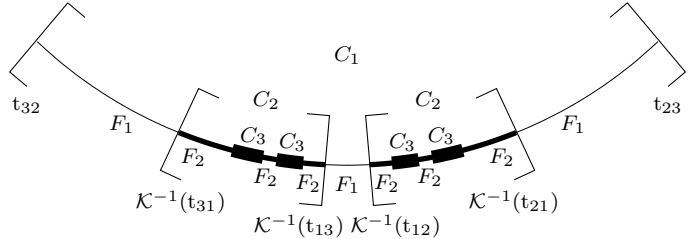
\begin{subfigure}\centering
    \begin{tikzpicture}[scale=6]
    \draw [line width=0.1pt,domain=2.383:3.9,variable=\t,smooth] plot ({sin(\t r)},{cos(\t r)});

    \draw [line width=2pt,  domain=3.19:3.525,variable=\t,smooth] plot ({sin(\t r)},{cos(\t r)});
    \draw [line width=2pt,  domain=-3.19:-3.525,variable=\t,smooth] plot ({sin(\t r)},{cos(\t r)});

    \draw [line width=5pt,   domain=3.33:3.4,variable=\t,smooth] plot ({sin(\t r)},{cos(\t r)});
    \draw [line width=5pt,   domain=3.24:3.3,variable=\t,smooth] plot ({sin(\t r)},{cos(\t r)});

    \draw [line width=5pt,   domain=-3.33:-3.4,variable=\t,smooth] plot ({sin(\t r)},{cos(\t r)});
    \draw [line width=5pt,   domain=-3.24:-3.3,variable=\t,smooth] plot ({sin(\t r)},{cos(\t r)});

    \draw[shift={(-0.62,-0.64)},rotate=140] (-0.05,0) -- (0,0) -- (0,0.2) -- (-0.05,0.2) node[anchor= north] {\scriptsize{$\mathrm{t}_{32}$}};
    \draw[shift={(0.62,-0.64)},rotate=-140] (0.05,0) -- (0,0) -- (0,0.2) -- (0.05,0.2) node[anchor= north] {\scriptsize{$\mathrm{t}_{23}$}};

    \draw[shift={(-0.33,-0.83)},rotate=155] (-0.05,0) -- (0,0) -- (0,0.2) -- (-0.05,0.2) node[anchor= north] {\scriptsize{$\mathcal{K}^{-1}(\mathrm{t}_{31})$}};
    \draw[shift={(0.33,-0.83)},rotate=-155] (0.05,0) -- (0,0) -- (0,0.2) -- (0.05,0.2) node[anchor= north] {\scriptsize{$\mathcal{K}^{-1}(\mathrm{t}_{21})$}};

    \draw[shift={(0.04,-0.89)},rotate=185] (-0.05,0) -- (0,0) -- (0,0.2) -- (-0.05,0.2) node[anchor= north] {\scriptsize{$\mathcal{K}^{-1}(\mathrm{t}_{12})$}};
    \draw[shift={(-0.04,-0.89)},rotate=-185] (0.05,0) -- (0,0) -- (0,0.2) -- (0.05,0.2) node[anchor= north] {\scriptsize{$\mathcal{K}^{-1}(\mathrm{t}_{13})$}};

    \draw  (0,-1) circle (0.001pt) node[anchor= north] {\scriptsize{$F_1$}};
    \draw [ rotate=30] (0,-1) circle (0.001pt) node[anchor= north] {\scriptsize{$F_1$}};
    \draw [ rotate=-30] (0,-1) circle (0.001pt) node[anchor= north] {\scriptsize{$F_1$}};

    \draw [ rotate=20] (0,-1) circle (0.001pt) node[anchor= north] {\scriptsize{$F_2$}};
    \draw [ rotate=10.5] (0,-1) circle (0.001pt) node[anchor= north] {\scriptsize{$F_2$}};
    \draw [ rotate=5] (0,-1) circle (0.001pt) node[anchor= north] {\scriptsize{$F_2$}};

    \draw [ rotate=-20] (0,-1) circle (0.001pt) node[anchor= north] {\scriptsize{$F_2$}};
    \draw [ rotate=-10.5] (0,-1) circle (0.001pt) node[anchor= north] {\scriptsize{$F_2$}};
    \draw [ rotate=-5] (0,-1) circle (0.001pt) node[anchor= north] {\scriptsize{$F_2$}};

    \node[anchor= south] at (0,-0.8) {\scriptsize{$C_1$}};

    \node[anchor= south] at (-0.18,-0.9) {\scriptsize{$C_2$}};
    \node[anchor= south] at (0.18,-0.9) {\scriptsize{$C_2$}};

    \draw [ rotate=12] (0,-1) circle (0.001pt) node[anchor= south] {\scriptsize{$C_3$}};
    \draw [ rotate=-12] (0,-1) circle (0.001pt) node[anchor= south] {\scriptsize{$C_3$}};

    \draw [ rotate=7] (0,-1) circle (0.001pt) node[anchor= south] {\scriptsize{$C_3$}};
    \draw [ rotate=-7] (0,-1) circle (0.001pt) node[anchor= south] {\scriptsize{$C_3$}};

\end{tikzpicture}
    \addtocounter{subfigure}{-1}
    \captionof{subfigure}{\footnotesize{
    In $A_1$ the (thin) set $F_1$ has three connected components, which includes the tangential points $\mathrm{t}_{32}$ and $\mathrm{t}_{23}$, but not the pre-images of the other four tangential points. 
    The (bold) closed set $C_2$ has two connected components, and the (bold) set $F_2$ has six --- both sets contain the pre-images of the four tangential points that are not in $F_1$. The (thicker bold) closed set $C_3$ has four connected components.
    }}\label{FIG:KASNERMAPSpreimagesZOOM}
\end{subfigure}
\captionof{figure}{The removal process of the open sets $F_n$ from the Kasner circle $\mathrm{K}^{\ocircle}$ in equation~\eqref{Fn}. The open set $F_0=\text{int}(S)$ has three connected
components. The closed set $C_1$ consists of the three arcs $A_1$, $A_2$, $A_3$.
Due to the axis permutation~\eqref{permSYM},
we only describe the sets $F_1,C_2$ in \textbf{(a)} and $F_2,C_3$ in \textbf{(b)} in the bottom
arc $A_1$.
}\label{FIG:KASNERMAPSpreimages}
\end{figure}
%


%
\begin{figure}[h!]
\minipage{0.32\textwidth}\centering
\begin{subfigure}\centering
    \begin{tikzpicture}[scale=2]
    \draw [white, shift={(0.04,-0.89)},rotate=185] (0,0.2) circle (0.001pt); 

    \draw [line width=0.1pt,domain=0:6.28,variable=\t,smooth] plot ({sin(\t r)},{cos(\t r)});

    \draw [ultra thick, dotted, white, domain=-0.26:0.26,variable=\t,smooth] plot ({sin(\t r)},{cos(\t r)});
    \draw [rotate=120,ultra thick, dotted, white, domain=-0.26:0.26,variable=\t,smooth] plot ({sin(\t r)},{cos(\t r)});
    \draw [rotate=-120,ultra thick, dotted, white, domain=-0.26:0.26,variable=\t,smooth] plot ({sin(\t r)},{cos(\t r)});

    \draw[shift={(-0.62,-0.64)},rotate=140] (-0.05,0) -- (0,0) -- (0,0.2) -- (-0.05,0.2);
    \draw[shift={(0.62,-0.64)},rotate=-140] (0.05,0) -- (0,0) -- (0,0.2) -- (0.05,0.2);

    \draw[shift={(0.864,-0.216)},rotate=260] (-0.05,0) -- (0,0) -- (0,0.2) -- (-0.05,0.2);
    \draw[shift={(0.244,0.856)},rotate=-20] (0.05,0) -- (0,0) -- (0,0.2) -- (0.05,0.2);

    \draw[shift={(-0.864,-0.216)},rotate=-260] (0.05,0) -- (0,0) -- (0,0.2) -- (0.05,0.2);
    \draw[shift={(-0.244,0.856)},rotate=20] (-0.05,0) -- (0,0) -- (0,0.2) -- (-0.05,0.2);
    \end{tikzpicture}
    \addtocounter{subfigure}{-1}\captionof{subfigure}{Deleting the three (dashed) connected
    components of $F_0$ from the Kasner circle $\mathrm{K}^{\ocircle} = C_0$ yields the closed set $C_1$.}
\end{subfigure}
\endminipage\hfill
\minipage{0.32\textwidth}\centering
\begin{subfigure}\centering
    \begin{tikzpicture}[scale=2]
    \draw [line width=0.1pt,domain=2.383:3.9,variable=\t,smooth] plot ({sin(\t r)},{cos(\t r)});

    \draw [line width=0.1pt,domain=2.383:3.9,variable=\t,smooth, rotate=120] plot ({sin(\t r)},{cos(\t r)});
    \draw [line width=0.1pt,domain=2.383:3.9,variable=\t,smooth, rotate=-120] plot ({sin(\t r)},{cos(\t r)});

    \draw [line width=2pt,  domain=3.19:3.525,variable=\t,smooth] plot ({sin(\t r)},{cos(\t r)});
    \draw [line width=2pt,  domain=-3.19:-3.525,variable=\t,smooth] plot ({sin(\t r)},{cos(\t r)});

    \draw [line width=2pt,  domain=3.19:3.525,variable=\t,smooth, rotate=120] plot ({sin(\t r)},{cos(\t r)});
    \draw [line width=2pt,  domain=-3.19:-3.525,variable=\t,smooth, rotate=120] plot ({sin(\t r)},{cos(\t r)});

    \draw [line width=2pt,  domain=3.19:3.525,variable=\t,smooth, rotate=-120] plot ({sin(\t r)},{cos(\t r)});
    \draw [line width=2pt,  domain=-3.19:-3.525,variable=\t,smooth, rotate=-120] plot ({sin(\t r)},{cos(\t r)});

    \draw[shift={(-0.62,-0.64)},rotate=140] (-0.05,0) -- (0,0) -- (0,0.2) -- (-0.05,0.2);
    \draw[shift={(0.62,-0.64)},rotate=-140] (0.05,0) -- (0,0) -- (0,0.2) -- (0.05,0.2);

    \draw[shift={(-0.33,-0.83)},rotate=155] (-0.05,0) -- (0,0) -- (0,0.2) -- (-0.05,0.2);
    \draw[shift={(0.33,-0.83)},rotate=-155] (0.05,0) -- (0,0) -- (0,0.2) -- (0.05,0.2);

    \draw[shift={(0.04,-0.89)},rotate=185] (-0.05,0) -- (0,0) -- (0,0.2) -- (-0.05,0.2);
    \draw[shift={(-0.04,-0.89)},rotate=-185] (0.05,0) -- (0,0) -- (0,0.2) -- (0.05,0.2);

    \draw[shift={(0.864,-0.216)},rotate=260] (-0.05,0) -- (0,0) -- (0,0.2) -- (-0.05,0.2);
    \draw[shift={(0.244,0.856)},rotate=-20] (0.05,0) -- (0,0) -- (0,0.2) -- (0.05,0.2);

    \draw[shift={(0.883,0.129)},rotate=275] (-0.05,0) -- (0,0) -- (0,0.2) -- (-0.05,0.2);
    \draw[shift={(0.553,0.700)},rotate=-35] (0.05,0) -- (0,0) -- (0,0.2) -- (0.05,0.2);

    \draw[shift={(0.750,0.479)},rotate=305] (-0.05,0) -- (0,0) -- (0,0.2) -- (-0.05,0.2);
    \draw[shift={(0.790,0.410)},rotate=-65] (0.05,0) -- (0,0) -- (0,0.2) -- (0.05,0.2);

    \draw[shift={(-0.864,-0.216)},rotate=-260] (0.05,0) -- (0,0) -- (0,0.2) -- (0.05,0.2);
    \draw[shift={(-0.244,0.856)},rotate=20] (-0.05,0) -- (0,0) -- (0,0.2) -- (-0.05,0.2);

    \draw[shift={(-0.883,0.129)},rotate=-275] (0.05,0) -- (0,0) -- (0,0.2) -- (0.05,0.2);
    \draw[shift={(-0.553,0.700)},rotate=35] (-0.05,0) -- (0,0) -- (0,0.2) -- (-0.05,0.2);

    \draw[shift={(-0.750,0.479)},rotate=-305] (0.05,0) -- (0,0) -- (0,0.2) -- (0.05,0.2);
    \draw[shift={(-0.790,0.410)},rotate=65] (-0.05,0) -- (0,0) -- (0,0.2) -- (-0.05,0.2);

    \end{tikzpicture}
    \addtocounter{subfigure}{-1}\captionof{subfigure}{Removing the nine (thin) connected components
    of $F_1$ from the three arcs of $C_1$ leads to the (bold) closed set $C_2$.}
\end{subfigure}
\endminipage\hfill
\minipage{0.32\textwidth}\centering
\begin{subfigure}\centering
    \begin{tikzpicture}[scale=2]
    \draw [white, shift={(0.244,0.856)},rotate=-20] (0,0.2) circle (0.001pt); 
    \draw [white, shift={(0.04,-0.89)},rotate=185] (0,0.2) circle (0.001pt); 

    \draw [domain=-3.19:-3.525,variable=\t,smooth] plot ({sin(\t r)},{cos(\t r)});
    \draw [domain=3.19:3.525,variable=\t,smooth] plot ({sin(\t r)},{cos(\t r)});

    \draw [line width=0.1pt,domain=-3.19:-3.525,variable=\t,smooth, rotate=120] plot ({sin(\t r)},{cos(\t r)});
    \draw [line width=0.1pt,domain=3.19:3.525,variable=\t,smooth, rotate=120] plot ({sin(\t r)},{cos(\t r)});
    \draw [line width=0.1pt,domain=-3.19:-3.525,variable=\t,smooth, rotate=-120] plot ({sin(\t r)},{cos(\t r)});
    \draw [line width=0.1pt,domain=3.19:3.525,variable=\t,smooth, rotate=-120] plot ({sin(\t r)},{cos(\t r)});

    \draw [line width=2pt,   domain=3.33:3.4,variable=\t,smooth] plot ({sin(\t r)},{cos(\t r)});
    \draw [line width=2pt,   domain=3.24:3.3,variable=\t,smooth] plot ({sin(\t r)},{cos(\t r)});

    \draw [line width=2pt,   domain=-3.33:-3.4,variable=\t,smooth] plot ({sin(\t r)},{cos(\t r)});
    \draw [line width=2pt,   domain=-3.24:-3.3,variable=\t,smooth] plot ({sin(\t r)},{cos(\t r)});

    \draw [line width=2pt,   domain=3.33:3.4,variable=\t,smooth, rotate=120] plot ({sin(\t r)},{cos(\t r)});
    \draw [line width=2pt,   domain=3.24:3.3,variable=\t,smooth, rotate=120] plot ({sin(\t r)},{cos(\t r)});

    \draw [line width=2pt,   domain=-3.33:-3.4,variable=\t,smooth, rotate=120] plot ({sin(\t r)},{cos(\t r)});
    \draw [line width=2pt,   domain=-3.24:-3.3,variable=\t,smooth, rotate=120] plot ({sin(\t r)},{cos(\t r)});

    \draw [line width=2pt,   domain=3.33:3.4,variable=\t,smooth, rotate=-120] plot ({sin(\t r)},{cos(\t r)});
    \draw [line width=2pt,   domain=3.24:3.3,variable=\t,smooth, rotate=-120] plot ({sin(\t r)},{cos(\t r)});

    \draw [line width=2pt,   domain=-3.33:-3.4,variable=\t,smooth, rotate=-120] plot ({sin(\t r)},{cos(\t r)});
    \draw [line width=2pt,   domain=-3.24:-3.3,variable=\t,smooth, rotate=-120] plot ({sin(\t r)},{cos(\t r)});
    \end{tikzpicture}
    \addtocounter{subfigure}{-1}\captionof{subfigure}{Erasing the eighteen (thin) connected components of $F_2$
    from the six components of $C_2$ yields the (bold) closed set $C_3$.}
\end{subfigure}
\endminipage
\captionof{figure}{ \label{fig:intro_Cantor} The iterative construction of the Cantor set $C$:
Start (at the left) with the Kasner circle $\mathrm{K}^{\ocircle}$ and remove
(when going to the right) the (thin) arcs $F_n$ keeping the closed (bold)
arcs $C_{n+1}$. Note that $C_n$ has $3\cdot 2^{n-1}$ connected components for $n\geq 1$, in accordance with Lemma~\ref{Lemma2onIN}.
To avoid clutter, we refrain from drawing the boundaries of the arcs in \textbf{(c)}. \label{FIG:KASNERMAPiterative}
}
\end{figure}
\begin{lemma}\label{CantorEq}
The sets $C$ and $F$ defined respectively in~\eqref{defofC} and~\eqref{defofF} can be written as
\begin{equation}\label{lem:intersec}
C=\bigcap_{n\in \mathbb{N}_0} C_n 
\qquad \text{ and }
\qquad F=\bigcup_{n\in \mathbb{N}_0} F_n,
\end{equation}
where $C_n$ and $F_n$ are defined by~\eqref{Fn}.
\end{lemma}
\begin{proof}
First, we prove the equality for $F$. The inclusion $ \cup_{n\in \mathbb{N}_0} F_n \subseteq F$ follows
from the definition of $F_n$ in~\eqref{CiterativFn} and $F$ in equation~\eqref{defofF}.
The reverse, $F\subseteq  \cup_{n\in \mathbb{N}_0} F_n$, is proved next, where
we show that any point $p\in F$ is also in $F_n$ for some $n\in\mathbb{N}_0$.
For any $p\in F$, there is a minimal $n_0\in \mathbb{N}_0$,
$p\notin \mathcal{K}^{-k}(S)$ for all $k\in\mathbb{N}_0$ such that $k< n_0$, whereas
$p\in  \mathcal{K}^{-k}(S)$ for all $k\geq n_0$.
In other words, $p\in C_k$ and $p\not\in F_k$ for all $k< n_0$.
Consequently, $p$ is not removed in the $n_0-1$ iteration and
$p$ therefore lies in $C_{n_0}=C_{n_0-1}\setminus{F_{n_0-1}}$.
There are two possibilities: Either $p\in F_{n_0}$ or $p\notin F_{n_0}$.
Since the former completes the proof, we consider the latter and
conclude that $p\in F_{n_0+1}$. 
On the one hand, $p\not\in F_{n_0}=\mathrm{int}_{C_{n_0}}(\mathcal{K}^{-n_0}(S))$,
and on the other hand $p\in \mathcal{K}^{-n_0}(S)$. The point $p$ must therefore
lie in the boundary of $F_{n_0}$, which is contained in $F_{n_0+1}$, see Figure~\ref{FIG:KASNERMAPSpreimages}.
%

%
Second, we prove the equality for $C$. The inclusion
$C\subseteq \cap_{n\in \mathbb{N}_0} C_n$ follows from the claim
$C\subseteq C_n$ for all $n\in \mathbb{N}_0$, which is proved by
induction. For the basis, obviously $C\subseteq C_0=\mathrm{K}^{\ocircle}$.
For the induction step, assume $C\subseteq C_n$ for all $n\leq N$
and show that $C\subseteq C_{N+1}$. Note that $C$ in~\eqref{defofC}
and $F_N$ in \eqref{CiterativFn} are disjoint. The induction hypothesis
and~\eqref{CiterativCn+1} yield $C\subseteq C_{N+1}$.
%
%
The reverse inclusion $\cap_{n\in \mathbb{N}_0} C_n  \subseteq C$ follows
from the characterization of $F$ since
points $p\in \cap_{n\in \mathbb{N}_0} C_n$ are never removed in
the iterative construction \eqref{Fn}, that
is, $p\not\in \cup_{n\in \mathbb{N}_0} F_n = F$. Hence,
$p\in \mathrm{K}^{\ocircle} \backslash F = C$.
\end{proof}

The next Lemma describes the maximum and minimum expansion rates of
$\mathcal{K}$ within the set $C$, which are used later to bound the Hausdorff dimension of $C$.
\begin{lemma}\label{MaxMinExpansion}
The extrema of the derivative of the Kasner map within $C$ are
\begin{subequations}
\begin{align}
m &:= \min_{p\in C} |D\mathcal{K}(p)|= \frac{2(1-v^2)}{1+2v^2-\sqrt{12 v^2 - 3}} \label{minC}\\
M &:= \max_{p\in C} |D\mathcal{K}(p)|= \frac{2+v^2}{1-v^2}. \label{maxC}
\end{align}
\end{subequations}
\end{lemma}

The proof is based on two main features. First, the three lines
in $(\Sigma_1,\Sigma_2,\Sigma_3)$-space that connect
each pair of auxiliary points $\mathrm{Q}_1/v$, $\mathrm{Q}_2/v$ and $\mathrm{Q}_3/v$ describe
physically equivalent heteroclinic chains with period 2 in
$\mathrm{K}^\ocircle$,\footnote{Note that these heteroclinic chains with
period 2 are the full unfolding of the Taub points, as seen in Figure~\ref{FIG:minMAX}.
These orbits collapse at the Taub points when $v\rightarrow 1/2$, and disappear when $v\in (0,1/2]$.
The role of the collapse of these objects toward the Taub points when $v\to 1/2$
in Bianchi type VIII and IX is unclear, particularly the (two-dimensional) center
manifold of the tangential points, and the stable manifold of the chain with period 2,
see also Appendix~\ref{app:unifying}.
}
constructed by concatenation of two heteroclinic Bianchi type II orbits
related by axis permutations, see Figure~\ref{FIG:minMAX}.
In particular, the line between $\mathrm{Q}_\alpha/v$ and $\mathrm{Q}_\beta/v$
is characterized by $\Sigma_\gamma=1/v$, which yields a heteroclinic chain of
period 2 under iterates of $\mathcal{K}$ that maps two physically identical
Kasner states (related by interchanging the axes $\Sigma_\alpha$ and
$\Sigma_\beta$) at $A_\alpha$ and $A_\beta$ to each other,
where $(\alpha,\beta,\gamma) = (1,2,3)$ or a permutation thereof.

Second, due to symmetry under axis permutations \eqref{permSYM},
we can without loss of generality restrict attention to the left
half of the arc $A_1$ when considering $|D\mathcal{K}(p)| = g(p)$
(recall Lemma~\ref{KasnerCircMapEXP}), where $g(p)$
monotonically increases between $\mathrm{t}_{32}$ and $\mathrm{Q}_1$, see Figure~\ref{fig:plotg}.
As a consequence the minimum $m$ (maximum $M$) is determined by the left-most (right-most)
point $p_m\in C$ ($p_M\in C$) in this half arc. Moreover, according to~\eqref{galpha},
the coordinate $\Sigma_1$ of $p_m$ and $p_M$ determines $g(p_m)$ and $g(p_M)$, respectively.
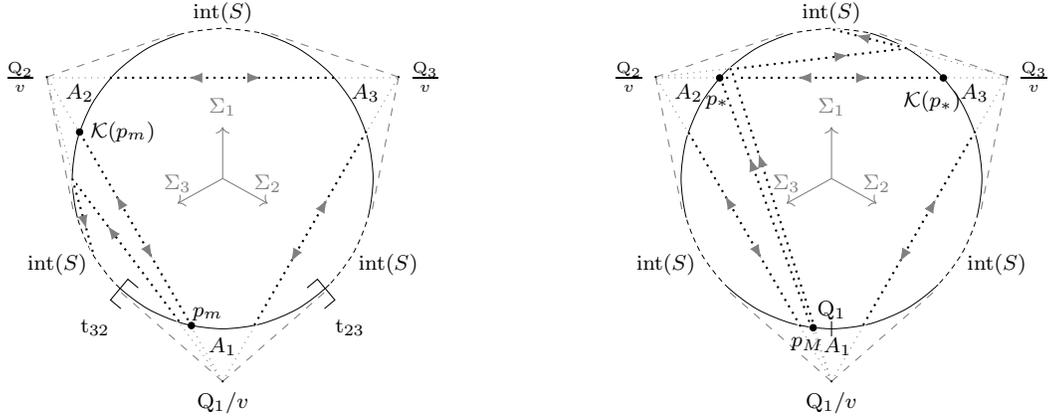
\begin{figure}[H]\centering
\minipage{0.49\textwidth}\centering
\begin{subfigure}\centering
    \begin{tikzpicture}[scale=2]
    \draw [line width=0.1pt,domain=0:6.28,variable=\t,smooth] plot ({sin(\t r)},{cos(\t r)});

    \draw [ultra thick, dotted, white, domain=-0.26:0.26,variable=\t,smooth] plot ({sin(\t r)},{cos(\t r)});
    \draw [rotate=120,ultra thick, dotted, white, domain=-0.26:0.26,variable=\t,smooth] plot ({sin(\t r)},{cos(\t r)});
    \draw [rotate=-120,ultra thick, dotted, white, domain=-0.26:0.26,variable=\t,smooth] plot ({sin(\t r)},{cos(\t r)});

    \draw[color=gray,dashed,-] (0,-1.35) -- (0.685,-0.75);
    \draw[color=gray,dashed,-] (0,-1.35) -- (-0.685,-0.75);

    \draw[color=gray,rotate=120,dashed,-] (0,-1.35) -- (0.685,-0.75);
    \draw[color=gray,rotate=120,dashed,-] (0,-1.35) -- (-0.685,-0.75);

    \draw[color=gray,rotate=-120,dashed,-] (0,-1.35) -- (0.685,-0.75);
    \draw[color=gray,rotate=-120,dashed,-] (0,-1.35) -- (-0.685,-0.75);

    \draw[color=gray,dotted] (0,-1.35) -- (-1.17,0.68);
    \draw[color=gray,dotted] (0,-1.35) -- (1.17,0.68);
    \draw[color=gray,dotted] (-1.17,0.67) -- (1.17,0.67);

    \draw[white, thick] (-0.74,0.67) -- (0.74,0.67);
    \draw[dotted, thick, postaction={decoration={markings,mark=at position 0.41 with {\arrow[thick,color=gray]{latex reversed}}},decorate}] (-0.75,0.67) -- (0.75,0.67);

    \draw[rotate=120,white, ultra thick] (-0.74,0.67) -- (0.74,0.67);
    \draw[rotate=120,dotted, thick, postaction={decoration={markings,mark=at position 0.41 with {\arrow[thick,color=gray]{latex reversed}}},decorate}] (-0.75,0.67) -- (0.75,0.67);
    \draw[rotate=-120,white, ultra thick] (-0.74,0.67) -- (0.74,0.67);
    \draw[rotate=-120,dotted, thick, postaction={decoration={markings,mark=at position 0.41 with {\arrow[thick,color=gray]{latex reversed}}},decorate}] (-0.75,0.67) -- (0.75,0.67);

    \filldraw [rotate=-12 ] (0,-1) circle (0.6pt);     \node at (-0.1,-0.9) {\scriptsize{$p_m$}};
    \filldraw [rotate=-108 ] (0,-1) circle (0.6pt) node[anchor= west] {\scriptsize{$\mathcal{K}(p_m)$}};

    \draw[color=gray,dotted] (0,-1.35) -- (-1,0);
    \draw[color=gray,dotted] (-1.17,0.67) -- (-0.86,-0.5);

    \draw[white,ultra thick] (-0.29,-0.95) -- (-1,0);
    \draw[dotted, thick, postaction={decorate}] (-0.29,-0.95) -- (-1,0);

    \draw[white,ultra thick] (-1,0) -- (-0.86,-0.5);
    \draw[dotted, thick, postaction={decorate}] (-1,0) -- (-0.86,-0.5);

    \node at (0,-1.12) {\scriptsize{$A_1$}};
    \node at (-0.95,0.57) {\scriptsize{$A_2$}};
    \node at (0.95,0.57) {\scriptsize{$A_3$}};

    \node at (0,1.1) {\scriptsize{$\mathrm{int}(S)$}};
    \node at (-1.1,-0.57) {\scriptsize{$\mathrm{int}(S)$}};
    \node at (1.1,-0.57) {\scriptsize{$\mathrm{int}(S)$}};

    \draw[color=gray,->]
    (0,0) -- (0,0.343)  node[anchor= south] {\scriptsize{$\Sigma_{1}$}};
    \draw[color=gray,->] (0,0) -- (0.3,-0.166)  node[anchor= south] {\scriptsize{$\Sigma_{2}$}};
    \draw[color=gray,->] (0,0) -- (-0.3,-0.166)  node[anchor= south] {\scriptsize{$\Sigma_{3}$}};

    \draw (0,-1.35) circle (0.1pt) node[anchor=north] {\scriptsize{$\mathrm{Q}_1/v$}};
    \draw[rotate=240] (0,-1.35) circle (0.1pt) node[anchor=east] {\scriptsize{$\frac{\mathrm{Q}_2}{v}$}};
    \draw[rotate=120]  (0,-1.35) circle (0.1pt) node[anchor=west] {\scriptsize{$\frac{\mathrm{Q}_3}{v}$}};

    \draw[shift={(-0.63,-0.67)},rotate=140, scale=0.9] (-0.1,0) -- (0,0) -- (0,0.2) -- (-0.1,0.2) node[anchor= north east] {\scriptsize{$\mathrm{t}_{32}$}};    \draw[shift={(0.63,-0.67)},rotate=-140, scale=0.9] (0.1,0) -- (0,0) -- (0,0.2) -- (0.1,0.2) node[anchor= north west] {\scriptsize{$\mathrm{t}_{23}$}};

    \end{tikzpicture}
    \addtocounter{subfigure}{-1}\captionof{subfigure}{\footnotesize{The minimum of $g$ on $C$ occurs at $p_m$
    where $\Sigma_3 = 1/v$. 
    Any point in $A_1$ between $\mathrm{t}_{32}$ and $p_m$ eventually ends up in $S$,
    since $\Sigma_3 > 1/v$ monotonically increases on the type $\mathrm{II}_1$
    and $\mathrm{II}_2$ subsets.}}\label{FIG:min}
\end{subfigure}
\endminipage\hfill
\minipage{0.49\textwidth}\centering
\begin{subfigure}\centering
    \begin{tikzpicture}[scale=2]
    \draw [line width=0.1pt,domain=0:6.28,variable=\t,smooth] plot ({sin(\t r)},{cos(\t r)});

    \draw [ultra thick, dotted, white, domain=-0.26:0.26,variable=\t,smooth] plot ({sin(\t r)},{cos(\t r)});
    \draw [rotate=120,ultra thick, dotted, white, domain=-0.26:0.26,variable=\t,smooth] plot ({sin(\t r)},{cos(\t r)});
    \draw [rotate=-120,ultra thick, dotted, white, domain=-0.26:0.26,variable=\t,smooth] plot ({sin(\t r)},{cos(\t r)});

    \draw[color=gray,dashed,-] (0,-1.35) -- (0.685,-0.75);
    \draw[color=gray,dashed,-] (0,-1.35) -- (-0.685,-0.75);

    \draw[color=gray,rotate=120,dashed,-] (0,-1.35) -- (0.685,-0.75);
    \draw[color=gray,rotate=120,dashed,-] (0,-1.35) -- (-0.685,-0.75);

    \draw[color=gray,rotate=-120,dashed,-] (0,-1.35) -- (0.685,-0.75);
    \draw[color=gray,rotate=-120,dashed,-] (0,-1.35) -- (-0.685,-0.75);

    \draw[color=gray,dotted] (0,-1.35) -- (-1.17,0.68);
    \draw[color=gray,dotted] (0,-1.35) -- (1.17,0.68);
    \draw[color=gray,dotted] (-1.17,0.67) -- (1.17,0.67);

    \draw[color=gray,dotted] (0,-1.35) -- (-0.75,0.7);

    \draw[white, ultra thick] (-0.74,0.67) -- (0.74,0.67);
    \draw[dotted, thick, postaction={decoration={markings,mark=at position 0.41 with {\arrow[thick,color=gray]{latex reversed}}},decorate}] (-0.75,0.67) -- (0.75,0.67);

    \draw[rotate=120,white, ultra thick] (-0.74,0.67) -- (0.74,0.67);
    \draw[rotate=120,dotted, thick, postaction={decoration={markings,mark=at position 0.41 with {\arrow[thick,color=gray]{latex reversed}}},decorate}] (-0.75,0.67) -- (0.75,0.67);
    \draw[rotate=-120,white, ultra thick] (-0.74,0.67) -- (0.74,0.67);
    \draw[rotate=-120,dotted, thick, postaction={decoration={markings,mark=at position 0.41 with {\arrow[thick,color=gray]{latex reversed}}},decorate}] (-0.75,0.67) -- (0.75,0.67);

    \draw[white,ultra thick] (-0.13,-1) -- (-0.75,0.7);
    \draw[dotted, thick, postaction={decorate}] (-0.13,-1) -- (-0.75,0.7);

    \draw[color=gray,dotted] (0,-1.35) -- (-0.67,0.74);
    \draw[rotate=-120,color=gray,dotted] (0,-1.35) -- (-1,0);
    \draw[rotate=-120,color=gray,dotted] (-1.17,0.67) -- (-0.86,-0.5);

    \draw[white,ultra thick] (-0.11,-1) -- (-0.67,0.74);
    \draw[dotted, thick, postaction={decorate}] (-0.11,-1) -- (-0.67,0.74);

    \draw[rotate=-120,white,ultra thick] (-0.29,-0.95) -- (-1,0);
    \draw[rotate=-120,dotted, thick, postaction={decorate}] (-0.29,-0.95) -- (-1,0);

    \draw[rotate=-120,white,ultra thick] (-1,0) -- (-0.86,-0.5);
    \draw[rotate=-120,dotted, thick, postaction={decorate}] (-1,0) -- (-0.86,-0.5);

    \filldraw [rotate=-120-12 ] (0,-1) circle (0.6pt); \node at (-0.76,0.51) {\scriptsize{$p_*$}};
    \filldraw [rotate=-120-108 ] (0,-1) circle (0.6pt); \node at (0.67,0.51)  {\scriptsize{$\mathcal{K}(p_*)$}};
    \filldraw [rotate=-7 ] (0,-1) circle (0.6pt);\node at (-0.16,-1.1) {\scriptsize{$p_M$}};

    \node at (0.04,-1.12) {\scriptsize{$A_1$}};
    \node at (-0.95,0.57) {\scriptsize{$A_2$}};
    \node at (0.95,0.57) {\scriptsize{$A_3$}};

    \node at (0,1.1) {\scriptsize{$\mathrm{int}(S)$}};
    \node at (-1.1,-0.57) {\scriptsize{$\mathrm{int}(S)$}};
    \node at (1.1,-0.57) {\scriptsize{$\mathrm{int}(S)$}};

    \draw[shift={(0,-0.95)}] (0,-0.1) -- (0,0);    \node at (0,-0.88) {\scriptsize{$\mathrm{Q}_1$}};

    \draw (0,-1.35) circle (0.1pt) node[anchor=north] {\scriptsize{$\mathrm{Q}_1/v$}};
    \draw[rotate=240] (0,-1.35) circle (0.1pt) node[anchor=east] {\scriptsize{$\frac{\mathrm{Q}_2}{v}$}};
    \draw[rotate=120]  (0,-1.35) circle (0.1pt) node[anchor=west] {\scriptsize{$\frac{\mathrm{Q}_3}{v}$}};

    \draw[color=gray,->](0,0) -- (0,0.343)  node[anchor= south] {\scriptsize{$\Sigma_{1}$}};
    \draw[color=gray,->] (0,0) -- (0.3,-0.166)  node[anchor= south] {\scriptsize{$\Sigma_{2}$}};
    \draw[color=gray,->] (0,0) -- (-0.3,-0.166)  node[anchor= south] {\scriptsize{$\Sigma_{3}$}};

    \end{tikzpicture}
    \addtocounter{subfigure}{-1}
    \captionof{subfigure}{\footnotesize{The maximum of $g$ on $C$ is at $p_M:=\mathcal{K}^{-1}(p_*)$,
    where $p_*$ has $\Sigma_1 = 1/v$.
    Points between $p_M$ and $\mathrm{Q}_1$ eventually end up in $S$, since $\Sigma_1 > 1/v$
    monotonically increases on the type $\mathrm{II}_2$ and $\mathrm{II}_3$ subsets.}}\label{FIG:max}
\end{subfigure}
\endminipage
\captionof{figure}{Depiction of the lines connecting each pair of points $\mathrm{Q}_1/v$, $\mathrm{Q}_2/v$
and $\mathrm{Q}_3/v$, which determine the periodic heteroclinic chains with
period 2; the points $p_m$ and $p_M$ for the extrema of $|D\mathcal{K}| = g$ on $C$;
examples of points with finite heteroclinic chains ending in $S$.}\label{FIG:minMAX}
\end{figure}
\begin{proof}
The point $p_m$ is determined by the periodic heteroclinic chain characterized
by the line between $\mathrm{Q}_1/v$ and $\mathrm{Q}_2/v$ for which $\Sigma_3=1/v$.
This follows since $p_m$ thereby belongs to $C$, and since any point $p$ in
the half arc between $\mathrm{t}_{32}$ and $p_m$ with $g(p)< g(p_m)$ is not in $C$.
This is due to that $p$ eventually ends up in $S$, either directly via a heteroclinic
orbit or by a finite heteroclinic chain, since $\Sigma_3$ monotonically
increases when $\Sigma_3>1/v$ along such heteroclinic orbits and chains, see
Figure~\ref{FIG:minMAX}.

To find the coordinate $\Sigma_1$ of $p_m$ we insert $\Sigma_3=1/v$ into the
constraint~\eqref{intro_cons2}, which yields $\Sigma_2 = -(\Sigma_1 + 1/v)$.
Inserting the values for $\Sigma_2$ and $\Sigma_3$ into the
constraint~\eqref{intro_cons1}, $\Sigma^2 = 1$, results in
\begin{equation}
\Sigma_1^2 + \frac{\Sigma_1}{v} + \frac{1-3v^2}{v^2}=0.
\end{equation}
This equation has two solutions (which coincide with $\mathrm{T}_3$ for $v=1/2$),
where the one with the smaller $\Sigma_1$ yields a point that resides in the
left arc of $A_1$, while the other solution gives the image of this point, which
is in $A_2$, see Figure~\ref{FIG:minMAX}. The relevant solution is therefore the
one with the smaller value
\begin{equation}\label{minSIGMA1}
\Sigma_1 = \frac{-1-\sqrt{12v^2-3}}{2v}.
\end{equation}
Inserting this into $g(p_m)$ in~\eqref{galpha} results
in~\eqref{minC}, $m=|D\mathcal{K}(p_m)| = g(p_m)$, as desired.

Next, we show that that the point $p_M$ is the pre-image $\mathcal{K}^{-1}(p_*)\in A_1$,
where $p_*$ is the point in $A_2$ determined by the line between
$\mathrm{Q}_2 /v$ and $\mathrm{Q}_3 /v$ characterized by $\Sigma_1=1/v$,
see Figure~\ref{FIG:minMAX}. Note that $p_M\in C$, since $\mathcal{K}(p_M)=p_*$
resides on a heteroclinic cycle forming a heteroclinic chain with period 2.
Moreover, the point $p_M$ yields $M$, since any point $p$ in the half arc
between $p_M$ and $\mathrm{Q}_1$ with $g(p) > g(p_M)$ is not in $C$. This follows since $p$
yields an orbit such that $\mathcal{K}(p)$ ends above $p_*$ in Figure~\ref{FIG:max}
with $\Sigma_1>1/v$, either in $S$ or in $A_2$. In the latter case the orbit
is concatenated with other heteroclinic orbits for which $\Sigma_1>1/v$
monotonically increases, which thereby yields a finite heteroclinic chain
that ends at $S$.

The point $\mathcal{K}(p_M)=p_*$ is the $\omega$-limit of the heteroclinic
orbit with $p_M$ as its $\alpha$-limit. Inserting $\eta = g(p_M)$
into~\eqref{BIIsol} yields
\begin{equation}\label{Sigma1pM}
\Sigma_1^{\mathrm{f}} = \Sigma^{\mathrm{i}}_{1}g(p_M) + \frac{2}{v} \left(g(p_M)-1\right),
\end{equation}
where $\Sigma_1^{\mathrm{f}}$ is the value of $\Sigma_1$ at $p_*$ and
$g(p_M)$ is given in~\eqref{galpha}. Moreover, $\Sigma_1^{\mathrm{f}} = 1/v$,
since $p_*$ lies on the line between $\mathrm{Q}_2/v$ and $\mathrm{Q}_3/v$,
which, when inserted into~\eqref{Sigma1pM}, gives
\begin{equation}
\Sigma^{\mathrm{i}}_{1} = -\frac{1 + 5v^2}{v(2 + v^2)}.
\end{equation}
Inserting this result into~\eqref{galpha} yields $M=|D\mathcal{K}(p_M)|=g(p_M)$
and thereby~\eqref{maxC}.
%
%
%
\end{proof}
%

\subsection{Characterization of $C$ by symbolic dynamics}

To know more about the connected components of $C_n$ with $n\geq 1$,
i.e., the bold sets in Figure~\ref{FIG:KASNERMAPiterative},
we introduce symbolic dynamics in a manner similar to that of the ternary Cantor set
in Figure~\ref{figure4}, where each connected component was encoded
by a sequence of two symbols $L$ and $R$. The starting point $n=0$ consists of
the removal of the set $F_0=\mathrm{int}(S)$ from the Kasner circle $C_0=\mathrm{K}^\ocircle$,
which yields $C_1=A_1\cup A_2 \cup A_3$.

Consider any $p\in C_n$ with $n\geq 1$. By the definitions in
equation~\eqref{Fn}, there is a unique symbol $a_n \in \{1,2,3\}$ such
that $\mathcal{K}^n(p)\in \mathrm{int}(A_{a_n})$ for each $n\in\mathbb{N}$.
Since two consecutive iterations are never in the same
unstable arc, $a_n\neq a_{n+1}$ for all $n\in \mathbb{N}$.
Furthermore, any $p\in A_{\alpha}$ has two pre-images of the Kasner circle
map $\mathcal{K}$: one in $A_{\beta}$ and one in $A_\gamma$,
where $(\alpha ,\beta ,\gamma)$ is a permutation of $(123)$, see
Figure~\ref{FIG:KASNERMAPS}. It therefore follows that we can find points
visiting a prescribed sequence of expanding arcs, obtained from
heteroclinic chains. To describe a finite sequence of arcs we introduce
the following notation: $w_n = a_0\ldots a_{n-1}$, where $w_n$ is called a \emph{word}.
Consider the set of all words $w_n$ of length $n \geq 1$, also called
an \emph{alphabet}, and denote this set by
\begin{equation}\label{defword}
    W_n:= \left\{ a_0 \ldots a_{n-1}  \Bigm|
    \begin{array}{c}
        \,\, a_k \in \{1,2,3\} \text{ for } k=0,\ldots, n-1, \\
        \,\,\, a_{k+1}\neq a_k\quad \text{ for } k=0, \ldots, n-2
    \end{array}\right\}.
\end{equation}
It then follows that the alphabet $W_n$ consists of $3\cdot 2^{n-1}$ words
(three possibilities for $a_0$ and two possibilities for each following
$a_k$ due to the restriction $a_{k+1}\neq a_k$).
Points $p\in C$ are encoded by infinite sequences with
$n=\infty$, i.e., $w_\infty\in W_\infty$.

To connect words with the iterative construction of $C$,
we define the set $I(w_n)$ to be the collection of points on
$\mathrm{K}^\ocircle$ that visits the arcs by iterations of $\mathcal{K}$
prescribed by $w_n= a_0 \ldots a_{n-1} \in W_n$.
This is formally expressed as
\begin{equation} \label{defofIN}
I(w_n) := \bigcap_{k=0}^{n-1} \mathcal{K}^{-k}(A_{a_k}),
\end{equation}
where a specific word $w_n$ yields a specific connected closed set $I(w_n)$, as illustrated
in Figure~\ref{FIG:INTERVALSencoding}. Note that for
$p\in I(w_n) = I(a_0\ldots a_{n-1})$ it follows that
$\mathcal{K}^k(p) \in A_{a_k}$, $k=0,\ldots, n-1$, and, in particular, $p \in A_{a_0}$.
%
%

The next lemmata guarantee that the union of $I(w_n)$ for all words $w_n\in W_n$
yield the set $C_n$ in the iterative construction~\eqref{Fn}, and hence that $I(w_n)$
are the connected components of $C_{n}$. Moreover, the family $\{I(w_n)\}_{n\in\mathbb{N}}$
consists of \emph{shrinking nested closed sets}. The veracity of these
claims is illustrated by considering a sequence of one of the connected components
of the bold sets $C_n$ for each $n\geq 1$ in Figure~\ref{FIG:KASNERMAPiterative}
and by a step by step construction of the nested closed sets in
Figure~\ref{FIG:INTERVALSencoding}.
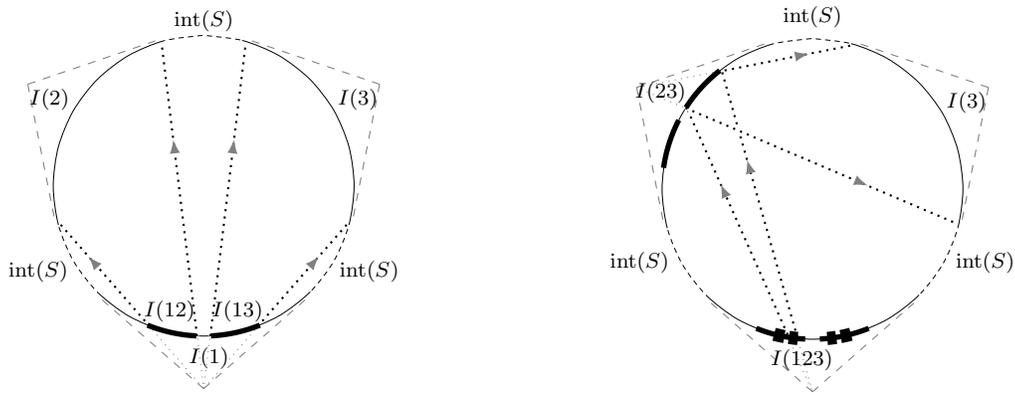
\begin{figure}[H]\centering
\minipage{0.49\textwidth}\centering
\begin{subfigure}\centering
    \begin{tikzpicture}[scale=2]
    \draw [line width=0.1pt,domain=0:6.28,variable=\t,smooth] plot ({sin(\t r)},{cos(\t r)});

    \draw [ultra thick, dotted, white, domain=-0.26:0.26,variable=\t,smooth] plot ({sin(\t r)},{cos(\t r)});
    \draw [rotate=120,ultra thick, dotted, white, domain=-0.26:0.26,variable=\t,smooth] plot ({sin(\t r)},{cos(\t r)});
    \draw [rotate=-120,ultra thick, dotted, white, domain=-0.26:0.26,variable=\t,smooth] plot ({sin(\t r)},{cos(\t r)});

    \node at (0,1.1) {\scriptsize{$\mathrm{int}(S)$}};
    \node at (-1.1,-0.57) {\scriptsize{$\mathrm{int}(S)$}};
    \node at (1.1,-0.57) {\scriptsize{$\mathrm{int}(S)$}};

    \draw[color=gray,dashed,-] (0,-1.35) -- (0.685,-0.75);
    \draw[color=gray,dashed,-] (0,-1.35) -- (-0.685,-0.75);

    \draw[color=gray,rotate=120,dashed,-] (0,-1.35) -- (0.685,-0.75);
    \draw[color=gray,rotate=120,dashed,-] (0,-1.35) -- (-0.685,-0.75);

    \draw[color=gray,rotate=-120,dashed,-] (0,-1.35) -- (0.685,-0.75);
    \draw[color=gray,rotate=-120,dashed,-] (0,-1.35) -- (-0.685,-0.75);

    \draw[color=gray,dotted] (0.28,0.95) -- (0,-1.35);
    \draw[color=gray,dotted] (-0.28,0.95) -- (0,-1.35);

    \draw[color=gray,dotted] (-0.95,-0.26) -- (0,-1.35);
    \draw[color=gray,dotted] (0.95,-0.26) -- (0,-1.35);

    \draw[white,ultra thick] (0.28,0.95) -- (0.044,-1);
    \draw[dotted, thick, postaction={decorate}] (0.044,-1) -- (0.28,0.95);
    \draw[white,ultra thick] (-0.28,0.95) -- (-0.044,-1);
    \draw[dotted, thick, postaction={decorate}] (-0.044,-1) -- (-0.28,0.95);

    \draw[white,ultra thick] (-0.95,-0.26) -- (-0.39,-0.9);
    \draw[dotted, thick, postaction={decorate}] (-0.39,-0.9) -- (-0.95,-0.26);
    \draw[white,ultra thick] (0.95,-0.26) -- (0.39,-0.9);
    \draw[dotted, thick, postaction={decorate}] (0.39,-0.9) -- (0.95,-0.26);


    \draw [line width=2pt,  domain=3.19:3.525,variable=\t,smooth] plot ({sin(\t r)},{cos(\t r)});
    \draw [line width=2pt,  domain=-3.19:-3.525,variable=\t,smooth] plot ({sin(\t r)},{cos(\t r)});


    \draw [rotate=2 ] (0,-1) circle (0.001pt) node[anchor= north] {\scriptsize{$I(1)$}};
    \draw [rotate=125] (0,-1) circle (0.001pt) node[anchor= west] {\scriptsize{$I(3)$}};
    \draw [rotate=-125] (0,-1) circle (0.001pt) node[anchor= east] {\scriptsize{$I(2)$}};

    \draw [ rotate=13] (0,-1) circle (0.001pt) node[anchor= south] {\scriptsize{$I(13)$}};
    \draw [ rotate=-13] (0,-1) circle (0.001pt) node[anchor= south] {\scriptsize{$I(12)$}};

    \end{tikzpicture}
    \addtocounter{subfigure}{-1}\captionof{subfigure}{\footnotesize{The (bold)
    set $C_2$ within $I(1)=A_1$ has two closed connected components given by $I(w_2)$,
    which are encoded by words $w_2=a_0a_1\in W_2$ such that $a_0=1$.
    The set $I(w_2)$ with $w_2=12$ encodes points $p\in A_1=I(1)$
    such that $\mathcal{K}(p)\in A_2$, whereas $w_2=13$ encodes $p\in A_1 = I(1)$
    with $\mathcal{K}(p)\in A_3= I(3)$.}}\label{FIG:encoding1}
\end{subfigure}
\endminipage\hfill
\minipage{0.49\textwidth}\centering
\begin{subfigure}\centering
    \begin{tikzpicture}[scale=2]
    \draw [line width=0.1pt,domain=0:6.28,variable=\t,smooth] plot ({sin(\t r)},{cos(\t r)});

    \draw [ultra thick, dotted, white, domain=-0.26:0.26,variable=\t,smooth] plot ({sin(\t r)},{cos(\t r)});
    \draw [rotate=120,ultra thick, dotted, white, domain=-0.26:0.26,variable=\t,smooth] plot ({sin(\t r)},{cos(\t r)});
    \draw [rotate=-120,ultra thick, dotted, white, domain=-0.26:0.26,variable=\t,smooth] plot ({sin(\t r)},{cos(\t r)});

    \draw [black] (0,1) circle (0.001pt) node[anchor= south] {\scriptsize{$\mathrm{int}(S)$}};
    \draw [black] (0.88,-0.49) circle (0.001pt) node[anchor= west] {\scriptsize{$\mathrm{int}(S)$}};
    \draw [black] (-0.88,-0.49) circle (0.001pt)node[anchor= east] {\scriptsize{$\mathrm{int}(S)$}};

    \draw[color=gray,dashed,-] (0,-1.35) -- (0.685,-0.75);
    \draw[color=gray,dashed,-] (0,-1.35) -- (-0.685,-0.75);

    \draw[color=gray,rotate=120,dashed,-] (0,-1.35) -- (0.685,-0.75);
    \draw[color=gray,rotate=120,dashed,-] (0,-1.35) -- (-0.685,-0.75);

    \draw[color=gray,rotate=-120,dashed,-] (0,-1.35) -- (0.685,-0.75);
    \draw[color=gray,rotate=-120,dashed,-] (0,-1.35) -- (-0.685,-0.75);

    \draw[color=gray,dotted] (-0.6,0.775) -- (0,-1.35);
    \draw[color=gray,dotted] (-0.83,0.525) -- (0,-1.35);

    \draw[color=gray,rotate=240,dotted] (-0.28,0.95) -- (0,-1.35);
    \draw[color=gray,rotate=240,dotted] (-0.95,-0.26) -- (0,-1.35);

    \draw[white,ultra thick] (-0.83,0.525) -- (-0.165,-0.97);
    \draw[dotted, thick, postaction={decorate}] (-0.165,-0.97) -- (-0.83,0.525);
    \draw[white,ultra thick] (-0.6,0.775) -- (-0.1,-0.99);
    \draw[dotted, thick, postaction={decorate}] (-0.1,-0.99) -- (-0.6,0.775);

    \draw[rotate=240,white,ultra thick] (-0.28,0.95) -- (-0.044,-1);
    \draw[rotate=240,dotted, thick, postaction={decorate}] (-0.044,-1) -- (-0.28,0.95);
    \draw[rotate=240,white,ultra thick] (-0.95,-0.26) -- (-0.39,-0.9);
    \draw[rotate=240,dotted, thick, postaction={decorate}] (-0.39,-0.9) -- (-0.95,-0.26);



    \draw [line width=2pt,  domain=3.19:3.525,variable=\t,smooth, rotate=-120] plot ({sin(\t r)},{cos(\t r)});
    \draw [line width=2pt,  domain=-3.19:-3.525,variable=\t,smooth, rotate=-120] plot ({sin(\t r)},{cos(\t r)});

    \draw [line width=2pt,  domain=3.19:3.525,variable=\t,smooth] plot ({sin(\t r)},{cos(\t r)});
    \draw [line width=2pt,  domain=-3.19:-3.525,variable=\t,smooth] plot ({sin(\t r)},{cos(\t r)});

    \draw [line width=5pt,   domain=3.33:3.4,variable=\t,smooth] plot ({sin(\t r)},{cos(\t r)});
    \draw [line width=5pt,   domain=3.24:3.3,variable=\t,smooth] plot ({sin(\t r)},{cos(\t r)});

    \draw [line width=5pt,   domain=-3.33:-3.4,variable=\t,smooth] plot ({sin(\t r)},{cos(\t r)});
    \draw [line width=5pt,   domain=-3.24:-3.3,variable=\t,smooth] plot ({sin(\t r)},{cos(\t r)});

    \draw [rotate=125] (0,-1) circle (0.001pt) node[anchor= west] {\scriptsize{$I(3)$}};
    \draw [rotate=-130] (0,-1) circle (0.001pt) node[anchor= east] {\scriptsize{$I(23)$}};
    \draw [ rotate=-4] (0,-1) circle (0.001pt) node[anchor= north] {\scriptsize{\bf $I(123)$}};

\end{tikzpicture}
    \addtocounter{subfigure}{-1}
    \captionof{subfigure}{\footnotesize{The (thicker bold) set $C_3$ within $I(1)=A_1$ has
    four closed connected components $I(w_3)$, which are encoded by the words $w_3=a_0a_1a_2\in W_3$
    such that $a_0=1$, $\mathcal{K}(p)\in A_{a_1}$, $\mathcal{K}^2(p)\in A_{a_2}$, where
    $p\in A_1$. The sets $I(w_2)$ and $I(w_3)$ for the bottom arc are described in detail
    in Figure \textbf{(c)}.}}\label{FIG:encoding2}
\end{subfigure}
\endminipage

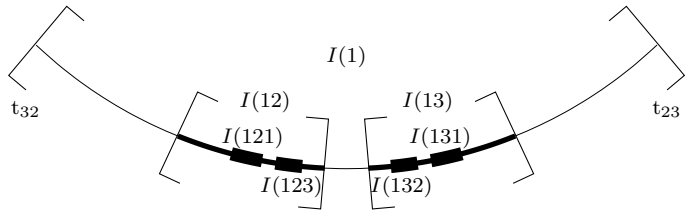
\begin{subfigure}\centering
    \begin{tikzpicture}[scale=6]
    \draw [line width=0.1pt,domain=2.383:3.9,variable=\t,smooth] plot ({sin(\t r)},{cos(\t r)});

    \draw [line width=2pt,  domain=3.19:3.525,variable=\t,smooth] plot ({sin(\t r)},{cos(\t r)});
    \draw [line width=2pt,  domain=-3.19:-3.525,variable=\t,smooth] plot ({sin(\t r)},{cos(\t r)});

    \draw [line width=5pt,   domain=3.33:3.4,variable=\t,smooth] plot ({sin(\t r)},{cos(\t r)});
    \draw [line width=5pt,   domain=3.24:3.3,variable=\t,smooth] plot ({sin(\t r)},{cos(\t r)});

    \draw [line width=5pt,   domain=-3.33:-3.4,variable=\t,smooth] plot ({sin(\t r)},{cos(\t r)});
    \draw [line width=5pt,   domain=-3.24:-3.3,variable=\t,smooth] plot ({sin(\t r)},{cos(\t r)});

    \draw[shift={(-0.62,-0.64)},rotate=140] (-0.05,0) -- (0,0) -- (0,0.2) -- (-0.05,0.2) node[anchor= north] {\scriptsize{$\mathrm{t}_{32}$}};
    \draw[shift={(0.62,-0.64)},rotate=-140] (0.05,0) -- (0,0) -- (0,0.2) -- (0.05,0.2) node[anchor= north] {\scriptsize{$\mathrm{t}_{23}$}};

    \draw[shift={(-0.33,-0.83)},rotate=155] (-0.05,0) -- (0,0) -- (0,0.2) -- (-0.05,0.2);
    \draw[shift={(0.33,-0.83)},rotate=-155] (0.05,0) -- (0,0) -- (0,0.2) -- (0.05,0.2);

    \draw[shift={(0.04,-0.89)},rotate=185] (-0.05,0) -- (0,0) -- (0,0.2) -- (-0.05,0.2);
    \draw[shift={(-0.04,-0.89)},rotate=-185] (0.05,0) -- (0,0) -- (0,0.2) -- (0.05,0.2);

    \node[anchor= south] at (0,-0.8) {\scriptsize{$I(1)$}};

    \node[anchor= south] at (-0.18,-0.9) {\scriptsize{$I(12)$}};
    \node[anchor= south] at (0.18,-0.9) {\scriptsize{$I(13)$}};

    \draw [ rotate=12] (0,-1) circle (0.001pt) node[anchor= south] {\scriptsize{$I(131)$}};
    \draw [ rotate=-12] (0,-1) circle (0.001pt) node[anchor= south] {\scriptsize{$I(121)$}};

    \draw [ rotate=7] (0,-1) circle (0.001pt) node[anchor= north] {\scriptsize{$I(132)$}};
    \draw [ rotate=-7] (0,-1) circle (0.001pt) node[anchor= north] {\scriptsize{$I(123)$}};

\end{tikzpicture}
    \addtocounter{subfigure}{-1}
    \captionof{subfigure}{\footnotesize{The sets $I(w_n)$, for $n=1,2,3$, in the
    arc $A_1$. The (thin) set $C_1$ has one closed connected component in $A_1$
    given by $I(w_1)$ encoded by the word $w_1=1\in W_1$, i.e., $A_1 = I(1)$.
    The (bold) set $C_2$ has two closed connected components in $A_1$ given
    by $I(w_2)$, which are encoded by words $w_2=a_0a_1\in W_2$ such that $a_0=1$.
    The (thicker bold) set $C_3$ has four closed connected components in $A_1$,
    given by $I(w_3)$ encoded by the words $w_3=a_0a_1a_2\in W_3$ such that $a_0=1$.}}\label{FIG:setsIn}
\end{subfigure}
\captionof{figure}{Illustration of the nested sets $I(w_n)$ for $n=1,2,3$. 
The (bold) set $C_1$ has three closed connected component given by arcs $A_1 ,A_2 , A_3$,
described in equation~\eqref{A_1} and Figure~\ref{FIG:BIF}, and they are encoded by symbolic dynamics as
$I(w_1)=A_{w_1}$, with the corresponding word $w_1$ within the alphabet $W_1=\{1,2,3\}$.}\label{FIG:INTERVALSencoding}
\end{figure}
\begin{lemma}\label{Lemma2onIN}
For $n\geq 1$, the set $C_{n}$ is given by the union of the $3\cdot 2^{n-1}$ closed, connected,
disjoint sets $I(w_n)$:
\begin{equation} \label{CN+1comp}
C_{n}= \bigcup_{w_n\in W_n} I(w_n).
\end{equation}
\end{lemma}
\begin{proof}
First, we show that $C_n\subseteq \cup_{w_n\in W_n} I(w_n)$.
By the definitions in~\eqref{Fn}, points $p\in C_{n}$ are not in the $k^{th}$
removed set $F_k$ for all $k=0,\ldots, n-1$. There is thereby
a unique symbol $a_k\in \{1,2,3\}$ such that
$\mathcal{K}^k(p) \in A_{a_k}$ for each $k=0,\ldots, n-1$
and $a_{k+1}\neq a_k $ for all $k=0,\ldots, n-2$. 
Hence $p\in I(w_n)$ for the word $w_n=a_0\ldots a_{n-1}\in W_n$.

Second, $I(w_n)\subseteq C_n$ for all $w_n=a_0\ldots a_{n-1}\in W_n$ since $p\in I(w_n)$
are not in $F_k$ for all $k=0,\ldots, n-1$, as follows from the definition~\eqref{defofIN} and the iterative
construction defined in~\eqref{Fn}. Hence $p\in C_{n}$.

To show disjointedness, consider two different
words in $W_n$ given by $w_n = a_0 \ldots a_{n-1}$ and $\tilde{w}_n = \tilde{a}_0 \ldots
\tilde{a}_{n-1}$ such that $a_k\neq
\tilde{a}_k$ for some $k$. Then any $p\in I(w_n)$ satisfies $\mathcal{K}^k(p)\in A_{a_k}$ whereas
$\mathcal{K}^k(p)\notin A_{\tilde{a}_k}$ since $A_{a_k}$ and $ A_{\tilde{a}_k}$ are disjoint.
Hence $p\notin I(\tilde{w}_n)$.

Finally, the number of closed connected components of $C_n$ is given by the cardinality
of $W_n$, which is $3\cdot 2^{n-1}$.
\end{proof}
\begin{lemma}\label{Lemma1onIN}
For any $n \geq 2$, the set $I(w_{n})$ is a closed nested set of $\mathrm{K}^{\ocircle}$:
\begin{equation}\label{nest}
I(w_{n}) \subseteq  I(w_{n-1}),
\end{equation}
where $w_{n}=a_0\ldots a_{n-1}\in W_n$ and $w_{n-1}=a_0\ldots a_{n-2}\in W_{n-1}$. Moreover,
\begin{equation}\label{INshrinks}
0<\vert I(w_n) \vert < 2\pi \nu^{n-2},
\end{equation}
for some constant $\nu \in (0,1)$, where $\vert \cdot \vert$ denotes the Lebesgue measure.
\end{lemma}
\begin{proof}
The arcs $A_\alpha$ for $\alpha=1,2,3$ are closed, and so are their pre-images under the continuous
map $\mathcal{K}$. Since $I(w_n)$ is an intersection of
closed sets defined in \eqref{defofIN}, it follows that $I(w_n)$ is a closed
set in $\mathrm{K}^{\ocircle}$.

The nested inclusion~\eqref{nest} follows from rewriting equation~\eqref{defofIN} as
\begin{equation}
I(w_{n}) = I(w_{n-1}) \cap \mathcal{K}^{-(n-1)}(A_{a_{n-1}}).
\end{equation}
%
This also implies that $\vert I(w_{n}) \vert>0$,
since $I(w_n)$ is connected and strictly contains two
nonempty closed disjoint subsets given by $ I(w_{n+1})$
for $w_{n+1}=a_0\ldots a_n\in W_{n+1}$, where the word $w_{n+1}$
is an extension of $w_n$ by concatenating a symbol $a_n\neq a_{n-1}$ at the end,
see Figure~\ref{FIG:INTERVALSencoding}.

Next we show equation~\eqref{INshrinks}. 
The Kasner map restricted to the set $I(w_{n})$, given by
$\mathcal{K}: I(a_0\ldots a_{n-1}) \rightarrow I(a_1\ldots a_{n-1})$,
see Figure~\ref{FIG:INTERVALSencoding}, is a diffeomorphism,
which implies that $x=\mathcal{K}(p)$ leads to:
\begin{equation}\label{ineq}
\min_{p\in I(w_{n})} \vert D\mathcal{K}(p) \vert
\cdot \vert I(w_n) \vert \leq \left\vert \int_{I(w_{n})} D\mathcal{K}(p) dp \right\vert
= \left\vert \int_{I(a_1\ldots a_{n-1})}  dx \right\vert
= \vert I(a_1\ldots a_{n-1}) \vert .
\end{equation}
%
%
%
Since the sets $I(w_n)$ are nested as in~\eqref{nest}, and since $I(w_2)$
is a connected component of $C_2$ for some $w_2\in W_2$, as described in equation~\eqref{CN+1comp},
it follows that
\begin{equation}\label{ineq2}
\min_{p\in I(w_n)} \vert D\mathcal{K}(p) \vert
\geq \min_{p\in I(w_2)} \vert D\mathcal{K}(p) \vert \geq  \min_{p\in C_2} \vert D\mathcal{K}(p) \vert =: \nu^{-1} > 1,
\end{equation}
%
%
as illustrated by Figure~\ref{FIG:INTERVALSencoding}. 
The last inequality, which yields $\nu < 1$, follows from that
$C_2$ is bounded away from the tangential points,
see Figure~\ref{FIG:KASNERMAPSpreimages}, and since
the derivative then is strictly bigger than one,
see~\eqref{KasnerCircPrime} and Figure~\ref{fig:plotg}.

We then apply the inequality~\eqref{ineq} recursively together
with~\eqref{ineq2}, which leads to
\begin{equation}\label{shrinkPF}
\vert I(w_n) \vert \leq \nu^{n-2} \cdot \vert I(w_2) \vert
\end{equation}
for some $w_2\in W_2$. Together with $\vert I(w_2) \vert < |\mathrm{K}^{\ocircle}| = 2\pi$
this results in equation~\eqref{INshrinks}.
\end{proof}

Finally, note that the shrinking rate of $|I(w_n)|$ in~\eqref{INshrinks}
is not improved by repeating the recursive procedure in~\eqref{shrinkPF}
once more in order to bound $|I(w_n)|$ by $|I(w_1)|$. Then the right
hand side of~\eqref{ineq2} is replaced with $\nu^{n-2} \cdot \tilde{\nu} \vert I(w_1) \vert$,
where $\tilde{\nu}^{-1}:=\min_{p\in C_1} \vert D\mathcal{K}(p) \vert$.
However, $\tilde{\nu} = 1$, since $C_1$ consists of the three arcs $A_\alpha$, which
contain tangential points for which the minimum 1 is attained, see Figure~\ref{fig:plotg}.

\subsection{Proof of Theorems~\ref{CantorTh} and~\ref{chaossub} }\label{subsec:proofCantor}

Based on the above ingredients we will now prove
Theorems~\ref{CantorTh} and~\ref{chaossub} in six steps.

{\bf First step: Closedness and nonemptiness of the set $C$.}

Since $C$ is obtained in equation~\eqref{lem:intersec} as the intersection
of closed sets $C_n$, defined in~\eqref{Fn}, which in turn is the
union of closed arcs $I(w_n)$ in equation~\eqref{CN+1comp}, $C$ is also closed.

Furthermore, given a word $w_\infty =(a_k)_{k\in
\mathbb{N}_0}\in W_\infty$ and its truncations $w_n=a_0\ldots a_{n-1}$,
the family $\{I(w_n)\}_{n\in\mathbb{N}}$ consists of shrinking nested closed sets such that its intersection 
consists of a single point $p$, due to Lemma~\ref{Lemma1onIN} and Cantor's intersection Theorem in complete metric spaces when $\mathrm{diam}(I(w_n))\to 0$ as $n\to\infty$. Such a point belongs to $I(w_n)$ for all $n\geq 1$,
and also in $C_{n}$ for every $n\geq 0$ due to~\eqref{CN+1comp} and $C_0 = \mathrm{K}^\ocircle$.
Consequently, $p\in C$, which is determined by the
intersection of all $C_{n}$ according to~\eqref{lem:intersec}. In other words,
\begin{equation}\label{intersec}
\bigcap_{n\in \mathbb{N}_0}I(w_n)=  p \in C.
\end{equation}
Note that such a point $p$ is associated with the word
$w_\infty = a_0 a_1 a_2\ldots \in W_\infty$, whereas the next point in
the heteroclinic chain, $\mathcal{K}(p)$, is associated with the sequence $\tilde{w}_\infty = a_1 a_2\ldots$, which
is the word $w_\infty$ without the first symbol $a_0$.
Hence $\mathcal{K}(p)$ lies in the intersection of the
family $\{I(\tilde{w}_n)\}_{n\in\mathbb{N}_0}$. Different points
in the same heteroclinic chain are therefore distinguished by fixing $a_0$.
This notion of deleting the first symbol is also called a shift
to the left, and is used to prove chaoticity in the sixth step.
%

{\bf Second step: No isolated points in $C$.}

Consider a point $p\in C$ and an
$\varepsilon$-neighborhood of $p$ in $\mathrm{K}^{\ocircle}$.
Let $w_\infty =(a_k)_{k\in
\mathbb{N}_0}\in W_\infty$ be a sequence such that $\mathcal{K}^k(p) \in A_{a_k}$
for all $k \in \mathbb{N}_0$.
According to equation~\eqref{INshrinks}, there is an $n\in
\mathbb{N}$ such that $I(w_n) =I(a_0\ldots a_{n-1})$
contains $p$ and has a length smaller than $\varepsilon$.

Next we prove that $I(w_n)$ contains a point $q\in C$ different than $p$ with a
distance smaller than an arbitrary $\varepsilon$, and hence that $p$ is not isolated.
Let $\tilde{a}_{n}\in \{1,2,3\}$ be different than $ a_{n-1}$ and $ a_{n}$.
Hence, the word $\tilde{w}_{n+1}:=a_0\ldots a_{n-1} \tilde{a}_{n}$ is without
repetition and differs from $w_{n+1}=a_0 \ldots a_{n-1}a_{n}$. 
Moreover, the disjoint sets
$I(w_{n+1})$ and $I(\tilde{w}_{n+1})$ are both contained in $I(w_{n})$,
since such arc sequences are nested~\eqref{nest}, see Figure~\ref{FIG:INTERVALSencoding}.
We now show that there is a $q\in I(\tilde{w}_{n+1})$ which is also in $C$,
but different than $p$, since $p\in I(w_{n+1})$.
Consider the family $\{I(\tilde{w}_k)\}_{k\in\mathbb{N}}$ of shrinking
nested closed sets, where $\tilde{w}_k:=a_0\ldots a_{k}$ is the truncation
of the word $\tilde{w}_\infty :=a_0\ldots a_{n-1} \tilde{a}_{n} a_n a_{n+1}\ldots\in W_\infty$.
This guarantees that $a_{n-1}, \tilde{a}_{n}$ and $a_n$ are pair-wise disjoint,
and hence~\eqref{intersec} implies that there is a $q$, which lies in $C$
and in the intersection of $I(\tilde{w}_k)$ for all $k\in\mathbb{N}$, and
consequently in $I(\tilde{w}_{n+1})$.



{\bf Third step: $\overline{C}$ has an empty interior.}

Since $C$ is closed, $\overline{C}=C$, we only have to prove
that $C$ has an empty interior. Consider the same arbitrary point
$p\in C$ and $I(w_n)$ within an $\varepsilon$-neighborhood
for any $\epsilon>0$, as in the second step.
We now show that $I(w_n)$ contains a point
$r\in \mathrm{K}^{\ocircle}\setminus{C}$, 
and hence that $p$ is not an interior point and that there thereby are
no interior points.

Consider $r \in \partial I(w_n)$. 
Since the
restriction $\mathcal{K}:
I(w_n) \rightarrow I(a_1\ldots a_{n-1})$ is a
diffeomorphism, it preserves boundaries, i.e., $\mathcal{K}(r) \in \partial I(a_1\ldots a_{n-1})$.
After $n-1$ iterations, $\mathcal{K}^{n-1}(r) \in \partial I(a_{n-1})$.
Note that $I(a_{n-1})$ is the arc $A_{a_{n-1}}$, and that its
boundary consists of the tangential points, which are in $S$. Hence
$\mathcal{K}^{n-1}(r) \in S$, and thus $r\in \mathrm{K}^{\ocircle}\setminus{C}$,
see Figure~\ref{FIG:KASNERMAPSpreimages}.

{\bf Fourth step: $C$ has measure zero.}

We prove that $F=\mathrm{K}^{\ocircle}\setminus{C}$ has full measure $2\pi$,
and hence that $C$ has measure zero.

Define the relative size of the $n^{th}$ removed set of the iterative construction~\eqref{Fn} as
\begin{equation}\label{qdef}
q_n:= \frac{\vert F_n \vert}{ \vert C_n \vert} \in (0,1),
\end{equation}
which is well-defined, since $C_n$ contains the sets $I_n(w_n)$ of positive length.

Consider the following partial sum of the pairwise disjoint sets $F_k$:
\begin{equation}\label{yetanother}
s_n := \sum_{k=0}^n \vert F_k \vert.
\end{equation}
Then
\begin{equation}
\vert F \vert = s_\infty.
\end{equation}
Applying the definition~\eqref{CiterativCn+1} of $C_{n+1}$ recursively leads to
\begin{equation} \label{partialcomp}
\vert C_{n+1}\vert=2\pi - s_n.
\end{equation}
We then use equations~\eqref{qdef}, \eqref{yetanother} and~\eqref{partialcomp} to obtain
\begin{equation}\label{latenight}
s_{n+1} - s_n = \vert F_{n+1} \vert =  \vert C_{n+1} \vert q_{n+1}= (2\pi - s_n) q_{n+1}.
\end{equation}
%
%
Note that the sequence $(s_n)_{n\in \mathbb{N}_0}$ in~\eqref{yetanother} is increasing and bounded
above by $2\pi$ and thereby converges to the limit $\vert F \vert \in
[0,2\pi]$. On the other hand, the sequence  $(q_n)_{n\in \mathbb{N}_0}$ is
bounded, and thus admits converging subsequences $(q_{n_k})_{k\in \mathbb{N}_0}$ with a limit $q$.
Taking the limit of~\eqref{latenight} results in
%
%
%
\begin{equation} \label{S2pi}
0=\vert F \vert - \vert F \vert= (2\pi - \vert F \vert) q.
\end{equation}
Proving that $F$ has full measure corresponds to excluding $q=0$.
We therefore show that $q_n$ is uniformly bounded away from 0.
First, observe that $\vert  F_n\vert=\vert C_{n}\vert -\vert C_{n+1}\vert$,
as follows from~\eqref{CiterativCn+1}, which enables us to rewrite~\eqref{qdef} as
%
%
%
%
\begin{equation} \label{estiqn2}
q_n= 1-\frac{\vert C_{n+1}\vert}{\vert C_{n}\vert}.
\end{equation}
We now show that the quotient $\vert C_{n+1}\vert / \vert C_{n}\vert$
is uniformly bounded away from 1, which follows from the expansion property
of the Kasner circle map in Lemma~\ref{KasnerCircMapEXP}.
Note that $C_{n+1}\subseteq C_2$ for every $n\geq 1$, and that
$\mathcal{K}(C_{n+1})\subseteq C_n$, from which it follows that
\begin{equation} \label{ineqquot}
\min_{p\in C_{2}} \vert D\mathcal{K}(p) \vert \cdot  \vert
C_{n+1} \vert \leq\min_{p\in C_{n+1}} \vert D\mathcal{K}(p) \vert \cdot  \vert C_{n+1} \vert \leq \vert C_n \vert,
\end{equation}
or, equivalently,
\begin{equation}\label{yetanotherineq}
\frac{\vert C_{n+1}\vert}{\vert C_{n}\vert} \leq \frac{1}{\mathrm{min}_{p\in C_{n+1}}
\vert D\mathcal{K}(p) \vert } \leq \frac{1}{ \mathrm{min}_{p\in C_{2}} \vert D\mathcal{K}(p)  \vert }=\nu <1,
\end{equation}
where the last inequality was shown in connection with equation~\eqref{ineq2}.
Moreover, $F_1$ removes a whole neighborhood of the tangential points,
as follows from~\eqref{Fn}, see Figure~\ref{FIG:KASNERMAPSpreimages}.
Since $\vert C_{n+1}\vert / \vert C_{n}\vert$
is uniformly bounded away from 1, it follows that $q_n$ is uniformly bounded away from 0,
and hence any converging subsequence of $(q_n)_{n\in \mathbb{N}_0}$ has a limit $q>0$.

{\bf Fifth step: Bounds on the Hausdorff dimension of $C$.}

The bounds~\eqref{HausCantorBounds} follow from Proposition 6 in~\cite{Pesin96},
which we simplify and adapt to our situation and notation.
\begin{proposition}\label{Prop:pesin}
Consider the iterative construction of $C$ given by $C_n$ in~\eqref{Fn}, with
connected components $I(w_n)$ in~\eqref{defofIN} for $w_n\in W_n$.
Suppose there are closed sets $I_*(w_n)$ and $I^*(w_n)$ of Lebesgue measure
$|I_*(w_n)|=c/\lambda^{*}$ and $|I^*(w_n)|=c/\lambda_*$ for some
$c\in\mathbb{R}_+$ and $0<1/\lambda^*\leq 1/ \lambda_*<1$ such that
$I_*(w_n)\subseteq I(w_n) \subseteq I^*(w_n)$
where the sets $\mathrm{int} (I_*(w_n))$ and $\mathrm{int} (I_*(\tilde{w}_n))$
are disjoint for different words $w_n\neq \tilde{w}_n$.
Then,
\begin{equation}
\dfrac{\log(2)}{\log(\lambda^*)}\leq \dim_{H} (C) \leq \dfrac{\log(2)}{\log(\lambda_*)}.
\end{equation}
\end{proposition}

Recall that $C_{n}$ in~\eqref{Fn} is obtained by a non-uniform contraction of $C_{n-1}$
with a contraction rate given by the inverse of the expansion rate~\eqref{KasnerCircPrime}.
We exclude the case $n=1$, which only divides $\mathrm{K}^\ocircle$ into the three physically
equivalent arcs $A_\alpha$ and $C$ into three identical parts with the same dimension,
one in each arc, see Figure~\ref{FIG:BIF}.
Then the following sets satisfy the above hypothesis: for $n>1$ let $I_*(w_n)$
and $I^*(w_n)$ be uniform contractions of the set $I(w_{n-1})$ with respective contraction
rates being the inverse of the expansion rates given by
$\lambda^* := M= \max_{p\in C} |D\mathcal{K}(p)|$ and
$\lambda_* := m =\min_{p\in C} |D\mathcal{K}(p)|$ so that $c := |I(w_{n-1})|$.
Disjointness follows from the proof of Lemma~\ref{Lemma2onIN}, which showed that
$I(w_n)$ and $I(w_k)$ are disjoint, while Lemma~\ref{MaxMinExpansion} gave the desired
bounds $M$ and $m<M$. Note also that $C$ lies within $\mathrm{K}^{\ocircle}$ and contains
no interval, and thus that its Hausdorff dimension has to be less than 1.

Although the bounds~\eqref{HausCantorBounds} now have been proven,
it is useful to provide an intuitive non-rigorous reasoning of
this proof: our Cantor set lies between two standard Cantor
sets with removed sets being uniformly scaled by the inverse of the minimum
and maximum expansion of the Kasner map in~\eqref{KasnerCircPrime}.

The Cantor set $C$ can be divided into three
identical parts: the intersection of $C$ with each arc $A_\alpha$
for $\alpha=1,2,3$, which are the three connected
components of the first iterate $C_1$ in the construction~\eqref{Fn} of $C$.
Since those three sets are disjoint,
\begin{equation}
d_H(C)=d_H(C \cap A_\alpha).
\end{equation}
%
After such a first iterate, the construction is similar to the standard
Cantor set in an interval: three parts of each arc $A_\alpha$ are
removed, yielding two remaining
subarcs, which are the two connected components of $C_2\cap A_\alpha$.
The left and right parts of $C$ in the two connected components of
$C_2\cap A_\alpha$ are called $C^L$ and $C^R$, in analogy with the ternary
Cantor set argument in~\eqref{TernaryCdim}, see Figure~\ref{fig:intro_Cantor}. Then,
\begin{equation}\label{HausC}
\mu^d(C \cap A_\alpha)=\mu^d(C^L)+\mu^d(C^R)=2 \delta^{d}\mu^d(C \cap A_\alpha),
\end{equation}
where the first equality holds since the sets $C^L$ and $C^R$ are disjoint; the
second since such sets have the same measure and are obtained by
contracting $C \cap A_\alpha$ with a factor $\delta<1$, which is the
inverse of the expansion of the set $C^L$ according to the Kasner circle map,
scaled with the power of the dimension $d$.

Note that the contraction rate $\delta<1$ is not uniform, since the expansion
of the Kasner circle map is not uniform. Moreover, each iteration has
a different contraction rate given
by $|C_{k+1}|/|C_{k}|$. We therefore obtain the following bounds:
\begin{equation}\label{HausCbounds}
2 M^{-d}\mu^d(C \cap A_\alpha) \leq \mu^d(C \cap A_\alpha)\leq 2 m^{-d}\mu^d(C \cap A_\alpha).
\end{equation}

If $\mu^d(C\cap A_\alpha )\neq 0$ and $\infty$ for some $d\geq 0$,
which we refrain from proving, we obtain
\begin{equation}\label{HausCbounds2}
2 M^{-d}\leq 1\leq 2 m^{-d},
\end{equation}
where the logarithm implies the desired bounds~\eqref{HausCantorBounds}.
There remains to show that $\mu^{d_M}(C \cap A_\alpha)<\infty$
and $\mu^{d^*}(C^i_1\cap C)\geq \epsilon >0$ in order to make the
above proof rigorous. This, however, follows in a similar manner as
for the usual ternary Cantor set, or, alternatively, see~\cite{Pesin96}.

{\bf Sixth step: Chaoticity of $\mathcal{K}$ on  $C$.}

To determine chaoticity of $\mathcal{K}$ on the Cantor
set $C$, we establish topological conjugacy with the \emph{shift map},
$\sigma:W_\infty\to W_\infty$, which shifts a sequence to the right, i.e.,
$\sigma (a_0a_1a_2\dots) := a_1a_2a_3\dots$, since the shift map is
well-known to be chaotic, see Chapter 1.6 in~\cite{Dev03}.
Note that $W_\infty$ is a subspace of all possible infinite sequences, as two adjacent symbols can not coincide, and thus $\sigma$ is a one-sided subshift of finite type.
To accomplish this we construct an \emph{encoding map}, which is a homeomorphism
$h:C\to W_\infty$ such that $h\circ \mathcal{K}=\sigma\circ h$, i.e., we need to
establish the following commutative diagram:
\begin{equation}\label{diagramCONJ}
\xymatrix{
C \ar[r]^{\mathcal{K}} \ar[d]_h  & C \ar[d]^h \\
W_\infty \ar[r]_\sigma & W_\infty }
\end{equation}

If such a map $h$ exists, we say that the discrete dynamical systems $\mathcal{K}$ and
$\sigma$ are \emph{topologically conjugate}. Note that the dynamics of $\mathcal{K}$
and $h$ are equivalent, since $\mathcal{K}=h^{-1}\circ\sigma\circ h$, and hence
fixed points and periodic heteroclinic chains can be translated from one system to the other,
see Chapter 1.7 in~\cite{Dev03}.

We construct the map $h$ so that it encodes each point $p\in C$
into an infinite sequence of three symbols $1,2,3$ without consecutive repetitions,
which accounts for the arcs $A_1,A_2,A_3$ the iterations of $\mathcal{K}^n(p)$ visits, i.e.,
\begin{equation}\label{defofh}
\begin{array}{llll}
h: & C &\rightarrow & W_\infty \\
    & p & \mapsto & h(p) := w_\infty=(a_k)_{k\in \mathbb{N}_0},
\end{array}
\end{equation}
where for each $k$, we define $a_k$ by $\mathcal{K}^k(p)\in A_{a_k}$.
Note that $W_\infty$ is the alphabet of words of infinite length, i.e., the
set~\eqref{defword} when $n=\infty$, and that periodic heteroclinic chains yield
infinite periodic sequences.

Following the heteroclinic orbit that takes $p$ to $\mathcal{K}(p)$ corresponds to a shift
to the right given by $\sigma (a_0a_1a_2\dots):=a_1a_2a_3\dots\,$. In other words,
the diagram in~\eqref{diagramCONJ} commutes. However,
we also have to show that $h$ is bijective, continuous, and that $h^{-1}$ is also continuous.
This follows from the definition of $I(w_n)$ in~\eqref{defofIN} and its properties given
in Lemma~\ref{Lemma1onIN}, as shown next.

The map $h$ is bijective since for any sequence $w_\infty\in W_\infty $
there is a unique point $p\in C$ such that $h(p)=w_\infty$.
This point is $p=\bigcap_{n\in \mathbb{N}_0} I(w_n)$, as in~\eqref{intersec},
where $w_n$ is the $n^{th}$ truncation of the infinite word $w_\infty$. 

The map $h$ is continuous at any point $p\in C$, since the
neighborhood $I(w_n)\cap C$ of $p$, for any $n\in \mathbb{N}_0$, 
only contains points $q\in C$ whose corresponding sequences of symbols $h(q)$
coincide with $h(p)$ for the first $n$ symbols. 

The map $h^{-1}$ is also continuous. For any $\varepsilon>0$, there
is an $n\in \mathbb{N}_0$ such that any two given sequences
$w_\infty,\tilde{w}_\infty\in W_\infty$ for which the first $n$ symbols coincide,
both $h^{-1}(w_\infty)$ and $h^{-1}(\tilde{w}_\infty)$ are in
$I(w_n)=I(\tilde{w}_n)$ with $\vert I(w_n)\vert<\varepsilon$,
due to~\eqref{INshrinks} in Lemma~\ref{Lemma1onIN}. 

Note that the above proof does not carry over to the critical case with $v=1/2$
since the map $h$ in~\eqref{defofh} only encodes the Cantor set $C$, i.e., it
does not encode the set $S$, which includes the tangential points and
the Taub points. To deal with $v \in (1/2,1)$ and $v=1/2$ in a unified manner,
we make an appropriate modification in Appendix~\ref{app:unifying}.

\section{Subcritical case}\label{sec:sub}

In the subcritical case, $v\in (0,1/2)$, each point in the set $\mathrm{K}^\ocircle$
admits at least one unstable direction and hence the following Lemma holds:
\begin{lemma}\label{SupercritHets}
Every point in the set $\mathrm{K}^{\ocircle}$ admits at least one infinite heteroclinic chain.
\end{lemma}
More precisely, all points in $\mathrm{K}^\ocircle \setminus \mathrm{int}(A_\alpha\cap A_\beta)$
have one unstable direction, whereas points in $ \mathrm{int}(A_\alpha\cap A_\beta)$
have two unstable directions, see Figure~\ref{FIG:BIF}. A point
within $\mathrm{int}(A_\alpha\cap A_\beta)$ thereby admits two different
heteroclinic connections on the hemispheres $\mathrm{II}_\alpha$ and
$\mathrm{II}_\beta$ given by \eqref{typeIIhemisphere}, which induces
a multivalued Kasner circle map $\mathcal{K}$,
see Figure~\ref{FIG:KASNERMAPS}. To deal with this situation
we interpret $\mathcal{K}$ as a collection of maps 
on the Kasner circle $\mathrm{K}^\ocircle$, i.e., we will reformulate $\mathcal{K}$ as
a so-called expansive iterated function system.


Recall that the Kasner circle map~\eqref{KasnerCirc} is expanding
due to equation~\eqref{KasnerCircPrime}. 
However, the usual definition of an \emph{iterated function system} (IFS)
is based on a family of contractions
in a metric space $X$, i.e.,
$\mathcal{F}:=\{ f_i: X\to X \text{ $|$ } i=1,...,N, f_i\in C^1 \text{ and } |f_i'|<1\}$.
According to~\cite{Hutch81}, there exists a unique
nonempty compact set $\mathcal{A}\subseteq X$ called the \emph{attractor} of
$\mathcal{F}$, which satisfies $\mathcal{A}=\overline{\cup_{i=1}^N f_i(\mathcal{A})}$.

An example is the ternary Cantor set $T$, iteratively constructed in~\eqref{defternary},
which can be seen as the attractor of the IFS given by $\{ f_L,f_R:[0,1]\to [0,1]\}$,
where the left and right maps are $f_L(x):=x/3$ and $f_R(x):=x/3+2/3$, respectively.
Then the $n^{th}$-step of the construction $T_n$ consists of the union
of all its connected components given by the image
$f_{i_n}\circ ...  \circ f_{i_1}([0,1])$ for some ${i_1},...,{i_n}\in \{L,R\}$,
see Figure \ref{figure4}.

Fewer efforts have been made in understanding families that are not contractions,
although see the construction of Koch curves using expansions in~\cite{PruSan88}
and the more recent work~\cite{Van09}. Both these investigations focus on
generating patterns occurring outside fractal sets and
understanding iterates, which in a non-compact space escape to infinity.
Since we are dealing with expansive iterates of a compact set, the Kasner
circle $\mathrm{K}^\ocircle$, we propose a theory
of expansive iterated function system (eIFS) on compact metric
spaces\footnote{Alternatively, one can consider the usual physical time
direction (i.e., the reverse of the present time direction), for which the
Kasner map becomes a contraction almost everywhere, and seek an attractor and its properties for a non-hyperbolic IFS, see~\cite{LappicyDaniel,ArJuSa17,MatDiaz}
and references therein. 
}.
We define an \emph{expansive iterated function system} (eIFS)
as a family of expansions in a compact metric space $X$,
\begin{equation}\label{eIFS}
\mathcal{F} := \left\{ f_i: X\to X \Bigm| i=1,...,N, \hspace{0.3cm}
\begin{array}{c}
f_i\in C^1 \text{ almost everywhere, } \\
|f_i'|>1 \text{ on dense open sets}
\end{array}
\right\}.
\end{equation}
In the spirit of~\cite{Hutch81}, we define the iterates of $\mathcal{F} $
by the Hutchinson operator:
\begin{equation}
\mathcal{F}^n(p):=\bigcup_{i_1,...,i_n\in \{1,...,N\}} f_{i_n}\circ ...\circ f_{i_1}(p),
\end{equation}
where the $n^{th}$ iterate yields a set consisting of at most $N^n$ points,
since the Hutchinson operator $\mathcal{F}^n(p)$ is defined as
the union over all possible iterates.

We now consider the Kasner circle map~\eqref{KasnerCirc} as an expansive
iterated function system and state a conjecture regarding its dynamics.
The \emph{Kasner circle eIFS} is defined as a collection of eight maps as follows:
\begin{equation}\label{KasnerIFS}
\mathcal{K}:=\{ {\cal K}_{\mu\nu\zeta}(p):\mathrm{K}^{\ocircle}\to
\mathrm{K}^{\ocircle} \text{ $ | $ } \mu=1,2 ; \nu=1,3 ; \zeta=2,3\},
\end{equation}
where each individual map is given by
\begin{equation}\label{KasnerIFSeachmap}
{\cal K}_{\mu\nu\zeta}(p):=
\begin{cases}
f_1(p) & \text{ for } p  \in A_1\backslash \{ (A_1\cap A_2 ) \cup (A_1\cap A_3 )\}\\
f_2(p) & \text{ for } p  \in A_2\backslash \{ (A_2\cap A_1 ) \cup (A_2\cap A_3 )\}\\
f_3(p) & \text{ for } p  \in A_3\backslash \{ (A_3\cap A_1 ) \cup (A_3\cap A_2 )\}\\
f_\mu(p) & \text{ for } p  \in A_1\cap A_2 \\
f_\nu(p) & \text{ for } p  \in A_2\cap A_3 \\
f_\zeta(p) & \text{ for } p  \in A_1\cap A_3
\end{cases}
\end{equation}
where $f_* (p):=g(p)p + \left(g(p)-1\right)\mathrm{T}_*/v$ such that the symbol
$*$ is to be replaced by $1,2,3$ or $\mu\in \{1,2\},\nu\in \{2,3\},\zeta\in \{1,3\}$.

Any point that is not in the overlap regions, e.g.
$p \in A_1\backslash \{ (A_1\cap A_2 ) \cup (A_1\cap A_3 )\}$,
has the same image under all maps ${\cal K}_{\mu\nu\zeta}(p)$,
independently of the indices $\mu,\nu,\zeta$. On the other hand, points
in the overlap regions, e.g. $p \in A_1\cap A_2$, have
two different maps with different images: ${\cal K}_{1\nu\zeta}(p)$
and ${\cal K}_{2\nu\zeta}(p)$, independently of the indices $\nu$ and $\zeta$.
This combinatorial problem of choosing between two maps for each of
the three overlapping arcs yields the eight maps in~\eqref{KasnerIFS}.

Due to~\eqref{KasnerCirc}, each map~\eqref{KasnerIFSeachmap} is expanding
and $C^1$ everywhere in $\mathrm{K}^\ocircle$,
except at certain tangential boundary points
$\partial (A_\alpha \cap A_\beta)$ where the derivative
is one. Nevertheless, if ${\cal K}_{\mu\nu\zeta}(p)$ is
discontinuous at such a tangential point, there is another
${\cal K}_{\mu'\nu'\zeta'}(p)$ that is both $C^1$ and strictly expanding
at this point.

The iterates of the Kasner eIFS are given by its Hutchinson operator:
\begin{equation}\label{KasnerHutchinson}
\mathcal{K}^n(p):=
\bigcup_{\substack{\mu_k=1,2 ; \hspace{0.1cm} \nu_k=2,3 ;
\hspace{0.1cm} \zeta_k=1,3\\ \text{for }  k=1,...,n}}
{\cal K}_{\mu_n\nu_n\zeta_n}\circ ...\circ {\cal K}_{\mu_1\nu_1\zeta_1}(p).
\end{equation}

This allows us to formulate the following conjecture:
\begin{conjecture}\label{conj:eIFSchaos}
The Kasner circle eIFS is chaotic when $v\in (0,1/2)$.
\end{conjecture}
%

We expect that the Kasner circle eIFS is chaotic due to the
expanding properties of each map of the Kasner circle
eIFS~\eqref{KasnerIFS}. 
However, the notion of chaos still has to be
further developed for multivalued maps. 
We suggest two different approaches to tackle this problem.
First, one can attempt to generalize the notion of chaotic discrete dynamical systems to
eIFS using the Hausdorff distance between sets, since the image under
the Hutchinson operator of the Kasner map in~\eqref{KasnerHutchinson} is a set of points.
Second, one can try to incorporate different
symbols $\mu\nu\zeta$ in the definition of chaos, and require that
topological mixing occurs for some, for generic,
or for all symbols $\mu\nu\zeta$.
Roughly speaking, this means that there are chaotic realizations of the eIFS.
Such a realization of chaos for some symbols $\mu\nu\zeta$ has been achieved for the Kasner multivalued map in \cite{LappicyDaniel}.

We also expect that there are two special iterations in the Hutchinson
operator~\eqref{KasnerHutchinson} which dictate the chaotic dynamics.
Whenever some iterate of a point $p$ reaches the overlap $A_\alpha\cap A_\beta$,
there are two choices of maps: one corresponding to orbits originating
from the auxiliary point $\mathrm{Q}_\alpha /v$ and one from $\mathrm{Q}_\beta /v$.
Consider the iteration
${\cal K}_{\mu_n\nu_n\zeta_n}\circ ...\circ {\cal K}_{\mu_1\nu_1\zeta_1}(p)$,
related to a symbolic sequence
$(\mu_k,\nu_k,\zeta_k)_{k\in\mathbb{N}_0}$ that always selects the map with
minimum expansion rate among the two choices, and define it to be $\mathcal{K}^n_m(p)$.
Similarly, denote by $\mathcal{K}^n_M(p)$ the iteration that always selects
the map with maximum expansion among the two choices.
These maps are uniquely determined for each point that is not a Taub point.
We expect the dynamics of $\mathcal{K}^n_m$ and $\mathcal{K}^n_M$ to quantify
how chaotic the full dynamics turns out to be, although there are several
technical problems which need to be resolved, especially in connection with the Taub points.

Note that there is redundancy in the iteration of the
maps~\eqref{KasnerIFSeachmap} in the Hutschinson
operator~\eqref{KasnerHutchinson}, e.g., a point
that is not in any overlap region
$p\in A_\alpha\backslash \{(A_\alpha\cap A_\beta )\cup (A_\alpha\cap A_\gamma )\}$
is mapped by ${\cal K}_{\mu\nu\zeta}(p)$ for all $\mu\nu\zeta$ in the
Hutchinson operator~\eqref{KasnerHutchinson}.
Nevertheless, all these images coincide and consist of a single point.
To avoid redundancy, one can give alternative descriptions
of the multivalued Kasner circle map~\eqref{KasnerCirc},
which affect the number of maps in an eIFS. For instance,
instead of considering the maps~\eqref{KasnerIFS} within the
eIFS framework in~\eqref{eIFS}, one can consider a family of \emph{three}
$C^1$ maps such that the domain of each map corresponds to $A_\alpha$.
In this and similar descriptions, one has to be careful about how images
of some maps should be contained within the domain of a different map in order
to have a well-defined iteration. 
To circumvent this problem, and have the whole Kasner
circle $\mathrm{K}^\ocircle$ as the domain, we choose the maps
in~\eqref{KasnerIFS}. The drawback with this choice is that each
map~\eqref{KasnerIFSeachmap} is discontinuous at certain tangential
points.

Even though the overall dynamical structure is far from being understood,
there are still special features which can be compared with
the supercritical case. Consider the set $\tilde{C}$ of points in $\mathrm{K}^{\ocircle}$
for which all iterates of the Kasner circle map ${\cal K}$ consist of
exactly one positive eigenvalue in the $N_\alpha$ variables\footnote{Note
that the set $C$ in~\eqref{defofC} for $v\in(1/2,1)$ can also be described by this
formulation. Thus the properties of such a set (of points with exactly one
positive eigenvalue for all iterates) depend on $v$.}, see
Figures~\ref{FIG:BIF} and~\ref{FIG:KASNERMAPS}, i.e.,
\begin{equation}\label{defofCtilde}
\tilde{C}:= \{ p\in \mathrm{K}^{\ocircle} \text{ $ | $ } \mathcal{K}^n(p)
\notin \mathrm{int}((A_1\cap A_2)\cup (A_1\cap A_3)\cup (A_2\cap A_3))  \text{ for all }  n \in \mathbb{N}_0 \}.
\end{equation}
This set is given by the points that never reach the overlaps
$\mathrm{int}(A_\alpha \cap A_\beta)$. The map $\mathcal{K}$ is thereby not a
multivalued map on the set $\tilde{C}$, and thus $\mathcal{K}$ is well-defined.
Furthermore, the set $\tilde{C}$ is not empty since, e.g., there are two
(physically equivalent) period 3 cycles, given by Lemma~\ref{lem:Period3}
and depicted in Figure~\ref{FIG:period3}, since the three vertices of each
triangle do not lie in any of the overlap regions $A_\alpha\cap A_\beta$.

The complement of the set $\tilde{C}$ in $\mathrm{K}^{\ocircle}$ is given by
\begin{equation}\label{defofFtilde}
\tilde{F}:= \{ p\in \mathrm{K}^{\ocircle} \text{ $ | $ }
\mathcal{K}^n(p) \in  \mathrm{int}((A_1\cap A_2)\cup (A_1\cap A_3)\cup (A_2\cap A_3)) \text{ for some } n \in \mathbb{N}_0\} .
\end{equation}
Splitting the dynamics in $\mathrm{K}^\ocircle$ into two
disjoint invariant sets, $\tilde{C}$ and $\tilde{F}$, is
a first step to tackle Conjecture~\ref{conj:eIFSchaos}.
In particular, it has been proved in \cite{LappicyDaniel} that the set $\tilde{F}$ is dense (in the circle) and thereby the generic dynamics occurs in such a set, akin to the generic dynamics outside the Cantor set $C$ in \eqref{defofC} within the supercritical case.
Certain properties of the invariant set $\tilde{C}$, and how they depend on $v$, are still not clear:
Is it a Cantor set or not? What is its 
Hausdorff dimension? What is its internal dynamics and the relation with the dynamics within the invariant set $\tilde{F}$?
\section{First principles and the dynamical hierarchy}\label{sec:firstprinciples}

We now investigate the dynamical consequences summarized in Table~\ref{BianchiInvSets} of first
principles, which for the vacuum $\lambda$-$R$ class~A Bianchi
models reduce to the scale-automorphism groups for the Lie
contraction hierarchy in Figure~\ref{FIG:hierarchy}.
%

The scale-automorphism group for each level of the hierarchy yields monotone functions and
conserved quantities derived in Appendix~\ref{app:lrscaleaut}. As we will see, these monotone functions and
conserved quantities restrict and
push the dynamics toward the initial singularity from the highest
level of the class~A Bianchi hierarchy, Bianchi type IX and VIII, to the lowest
levels of the hierarchy, Bianchi type II and I, for which the dynamics are completely
determined by the scale-automorphism group, as shown in Appendix~\ref{app:lrscaleaut}.
The next level in the hierarchy are the Bianchi type $\mathrm{VI}_0$
and $\mathrm{VII}_0$ models, where the scale-automorphism group give rise to several
quantities that limit the asymptotic dynamics. These quantities yield a
complete qualitative description for this level of the hierarchy, which we
focus on in this section. The asymptotic dynamics of type VIII and IX
form a considerable challenge and we only present some limited results.

\subsubsection*{Bianchi type $\mathrm{VI}_0$ and $\mathrm{VII}_0$}

To obtain the equations for the type $\mathrm{VI}_0$ and $\mathrm{VII}_0$
vacuum $\lambda$-$R$ models we set, without loss of generality, 
$N_1=0$, $N_2>0$, $N_3<0$ for type $\mathrm{VI}_0$,
and $N_1=0$, $N_2>0$, $N_3>0$ for type $\mathrm{VII}_0$.
Since $N_1=0$ selects a special direction, it is natural to replace
$(\Sigma_1,\Sigma_2,\Sigma_3)$ with the $\Sigma_\pm$ Misner variables
given in~\eqref{Misner}. Setting $N_1=0$ in~\eqref{dynsyslambdaR}
and~\eqref{LambdaRquantities} 
yields the evolution equations
\begin{subequations}\label{dynsyslambdaRVIVII}
\begin{align}
\Sigma_+^\prime &= 2(1-\Sigma^2)(1 + 2v\Sigma_+),\label{SpVIVII}\\
\Sigma_-^\prime &= 4v(1-\Sigma^2)\Sigma_- + 2\sqrt{3}(N_2^2 - N_3^2),\label{SmVIVII}\\
N_2^\prime &= -2(2v\Sigma^2 + \Sigma_+ + \sqrt{3}\Sigma_-)N_2,\label{N2primeVIVII}\\
N_3^\prime &= -2(2v\Sigma^2 + \Sigma_+ - \sqrt{3}\Sigma_-)N_3,
\end{align}
%
and the constraint
\begin{equation}\label{constrVIVII}
1 - \Sigma^2 - (N_2 - N_3)^2 =0, \qquad \text{ where } \qquad \Sigma^2 := \Sigma_+^2 + \Sigma_-^2.
\end{equation}
\end{subequations}
%
%
%

Due to the constraint~\eqref{constrVIVII}, the state spaces for the
type $\mathrm{VI}_0$ and $\mathrm{VII}_0$ models with $N_1=0$ are
3-dimensional with a 2-dimensional boundary given by the union of the
invariant type $\mathrm{II}_2$, $\mathrm{II}_3$ and $\mathrm{K}^\ocircle$
sets. Type $\mathrm{VI}_0$ has a relatively compact state-space,
whereas type $\mathrm{VII}_0$ has an unbounded one.
Equation~\eqref{constrVIVII} implies that $\Sigma_+^2 + \Sigma_-^2\leq 1$.
For type $\mathrm{VI}_0$, $(N_2-N_3)^2= N_2^2 + N_3^2 + 2|N_2N_3|$, and
hence~\eqref{constrVIVII} yields $N_2^2 \leq 1 -\Sigma^2$ and
$N_3^2 \leq 1 -\Sigma^2$, where the equalities hold individually for the $\mathrm{II}_2$
and $\mathrm{II}_3$ boundary sets, respectively. For type $\mathrm{VII}_0$,
on the other hand, introducing $N_\pm := N_2\pm N_3$ results in that the
constraint~\eqref{constrVIVII} can be written as $\Sigma^2 + N_-^2 = 1$,
and thus that $\Sigma_\pm$ and $N_-$ are bounded, while $N_+$ is unbounded.

The analysis of the Bianchi type $\mathrm{VI}_0$ and $\mathrm{VII}_0$
scale-automorphism group of the vacuum $\lambda$-$R$ models in
Appendix~\ref{app:lrscaleaut} resulted in three quantities that are
essential for the asymptotics of the dynamical system~\eqref{dynsyslambdaRVIVII}:
\begin{equation}\label{MVIVII}
1 + 2v\Sigma_+, \qquad Z_{\mathrm{sup}} := \frac{(2v + \Sigma_+)^2}{|N_2N_3|} \qquad
Z_{\mathrm{sub}} := \frac{(1 + 2v\Sigma_+)^2}{|N_2N_3|},
\end{equation}
where $Z_{\mathrm{sup}}=Z_{\mathrm{sub}}=Z_\mathrm{crit} = (1 + \Sigma_+)^2/|N_2N_3|$
when $v=1/2$. Due to~\eqref{dynsyslambdaRVIVII}, these quantities satisfy
\begin{subequations}
\begin{align}
(1+2v\Sigma_+)^\prime &= 4v(1 - \Sigma^2)(1+2v\Sigma_+),\label{MonotoneVIVII}\\
Z_{\mathrm{sup}}^\prime &= 4\left[\frac{(1 + 2v\Sigma_+)^2 + (4v^2-1)\Sigma_-^2}{2v + \Sigma_+}\right]Z_{\mathrm{sup}},\label{Z1prim}\\
Z_{\mathrm{sub}}^\prime &= 4(2v + \Sigma_+)Z_{\mathrm{sub}},\label{Z2prim}
\end{align}
\end{subequations}
and hence $Z_\mathrm{crit}^\prime = 4(1 + \Sigma_+)Z_\mathrm{crit}$.

These functions behave differently for the subcritical, critical and
supercritical cases, and they have different asymptotic consequences for
the type $\mathrm{VI}_0$ and $\mathrm{VII}_0$ vacuum $\lambda$-$R$ models,
primarily because the state space of the type $\mathrm{VII}_0$ models is
unbounded. Nevertheless, the two Bianchi types share several features.
For example, $\Sigma_+=-1/(2v)$ is a 2-dimensional invariant subset in the
supercritical case, $v\in(1/2,1)$, both for type $\mathrm{VI}_0$
and $\mathrm{VII}_0$. They also have some
common asymptotic features. In particular, they have the same
$\mathrm{II}_2\cup\mathrm{II}_3\cup\mathrm{K}^\ocircle$ boundary.
In the supercritical case the stable set in the
Kasner circle set $\mathrm{K}^\ocircle$ is
given by $S_{\mathrm{VI}_0} = S_{\mathrm{VII}_0} =
S_{\mathrm{VI}_0,\mathrm{VII}_0} := \mathrm{K}^\ocircle\backslash
\mathrm{int}(A_2 \cup A_3)$, which, due to that $N_1=0$, is different than the
set $S$ in the supercritical Bianchi type $\mathrm{VIII}$ and $\mathrm{IX}$ models,
cf. Figures~\ref{FIG:BIF} and~\ref{FIG:SVIVII}, although both
$S_{\mathrm{VI}_0,\mathrm{VII}_0}$ and $S$ are defined as the sets where
type II heteroclinic chains end. In the subcritical case, type
$\mathrm{VI}_0$ and $\mathrm{VII}_0$ also share the region $A_2\cap A_3$ in
$\mathrm{K}^\ocircle$, where both $N_2$ and $N_3$ are unstable in
$\mathrm{int}(A_2\cap A_3)$.
In the critical case, $A_2\cap A_3$ reduces to the Taub point $\mathrm{T}_1$.
These features are illustrated in Figure~\ref{FIG:SVIVII}.
\begin{figure}[H]
\minipage[b]{0.39\textwidth}\centering
\begin{subfigure}\centering
    \begin{tikzpicture}[scale=1.15]



    \draw[color=gray, dotted] (2.46,1.42) -- (-0.55,0.84);
    \draw[color=white, ultra thick] (0.15,0.97) -- (-0.55,0.84);
    \draw[dotted, thick, postaction={decorate}] (0.15,0.97) -- (-0.55,0.84);

    \draw[color=gray, dotted] (-2.46,1.42) -- (0.91,0.41);
    \draw[color=white, ultra thick] (-0.55,0.84) -- (0.91,0.41);
    \draw[dotted, thick, postaction={decorate}] (-0.55,0.84) -- (0.91,0.41);

    \draw[color=gray, dotted] (2.46,1.42) -- (-0.75,-0.69);
    \draw[color=white, ultra thick] (0.91,0.41) -- (-0.75,-0.69);
    \draw[dotted, thick, postaction={decorate}] (0.91,0.41) -- (-0.75,-0.69);

    \draw [domain=0:6.28,variable=\t,smooth] plot ({sin(\t r)},{cos(\t r)});

    \draw [ultra thick, dotted, white, domain=2.3:4,variable=\t,smooth] plot ({sin(\t r)},{cos(\t r)});

    \draw [very thick, domain=-0.29:0.29,variable=\t,smooth] plot ({0.985*sin(\t r)},{0.985*cos(\t r)});

    \draw[color=gray,rotate=120,dashed] (0,-2.85) -- (0.98,-0.29);
    \draw[color=gray,rotate=120,dashed] (0,-2.85) -- (-0.98,-0.29);

    \draw[color=gray,rotate=240,dashed] (0,-2.85) -- (0.98,-0.29);
    \draw[color=gray,rotate=240,dashed] (0,-2.85) -- (-0.98,-0.29);

    \draw[rotate=240] (0,-2.85) circle (0.1pt) node[anchor=east] {\scriptsize{$\mathrm{Q}_2/v$}};
    \draw[rotate=120]  (0,-2.85) circle (0.1pt) node[anchor=west] {\scriptsize{$\mathrm{Q}_3/v$}};

    \node at (0,1.2) {\scriptsize{$A_2\cap A_3$}};

    \node at (0,-1.2) {\scriptsize{$S_{\mathrm{VI}_0,\mathrm{VII}_0}$}};

\end{tikzpicture}
    \addtocounter{subfigure}{-1}\captionof{subfigure}{\footnotesize{ $v\in (0,1/2)$.
    }}\label{FIG:SVIVIIsub}
\end{subfigure}
\endminipage\hfill
\minipage[b]{0.3\textwidth}\centering

\begin{subfigure}\centering
    \begin{tikzpicture}[scale=1.15]

    \draw[color=gray, dotted] (1.75,1) -- (-0.65,0.76);
    \draw[color=white, ultra thick] (0.5,0.87) -- (-0.65,0.76);
    \draw[dotted, thick, postaction={decorate}] (0.5,0.87) -- (-0.65,0.76);

    \draw[color=gray, dotted] (-1.75,1) -- (0.91,0.41);
    \draw[color=white, ultra thick] (-0.65,0.76) -- (0.91,0.41);
    \draw[dotted, thick, postaction={decorate}] (-0.65,0.76) -- (0.91,0.41);

    \draw[color=gray, dotted] (1.75,1) -- (-0.67,-0.75);
    \draw[color=white, ultra thick] (0.91,0.405) -- (-0.67,-0.75);
    \draw[dotted, thick, postaction={decorate}] (0.91,0.405) -- (-0.67,-0.75);

    \draw [domain=0:6.28,variable=\t,smooth] plot ({sin(\t r)},{cos(\t r)});

    \draw [ultra thick, dotted, white, domain=2.1:4.2,variable=\t,smooth] plot ({sin(\t r)},{cos(\t r)});

    \draw[color=gray,dashed,-] (-1.75,1) -- (1.75,1);

    \draw[color=gray,dashed,-] (-1.75,1) -- (-0.88,-0.49);
    \draw[color=gray,dashed,-] (1.75,1) -- (0.88,-0.49);

    \draw[rotate=240] (0,-2) circle (0.1pt) node[anchor=east] {\scriptsize{$2\mathrm{Q}_2$}};
    \draw[rotate=120]  (0,-2) circle (0.1pt) node[anchor=west] {\scriptsize{$2\mathrm{Q}_3$}};

    \filldraw [black] (0,1) circle (1.25pt) node[anchor= south] {\scriptsize{$\mathrm{T}_1$}};
    \filldraw [black] (0.88,-0.49) circle (1.25pt) node[anchor= north west] {\scriptsize{$\mathrm{T}_2$}};
    \filldraw [black] (-0.88,-0.49) circle (1.25pt)node[anchor= north east] {\scriptsize{$\mathrm{T}_3$}};


    \node at (0,-1.2) {\scriptsize{$S_{\mathrm{VI}_0,\mathrm{VII}_0}$}};

    \draw[color=gray,->](0,0) -- (0,-0.53)  node[anchor= north] {\scriptsize{$\Sigma_{+}$}};
    \draw[color=gray,->] (0,0) -- (0.53,0)  node[anchor= west] {\scriptsize{$\Sigma_{-}$}};
\end{tikzpicture}
     \addtocounter{subfigure}{-1}\captionof{subfigure}{\footnotesize{$v=1/2$,
     }}\label{FIG:SVIVIIcrit}
\end{subfigure}
\endminipage\hfill
\minipage[b]{0.3\textwidth}\centering
\begin{subfigure}\centering
    \begin{tikzpicture}[scale=1.15]

    \draw[color=gray,dotted] (-1.17,0.67) -- (1.17,0.67);

    \draw[white, ultra thick] (-0.74,0.67) -- (0.74,0.67);
    \draw[dotted, thick, postaction={decoration={markings,mark=at position 0.41 with {\arrow[thick,color=gray]{latex reversed}}},decorate}] (-0.75,0.67) -- (0.75,0.67);

    \draw[rotate=-120,color=gray,dotted] (0,-1.35) -- (-1,0);
    \draw[rotate=-120,color=gray,dotted] (-1.17,0.67) -- (-0.86,-0.5);

    \draw[rotate=-120,white,ultra thick] (-0.29,-0.95) -- (-1,0);
    \draw[rotate=-120,dotted, thick, postaction={decorate}] (-0.29,-0.95) -- (-1,0);

    \draw[rotate=-120,white,ultra thick] (-1,0) -- (-0.86,-0.5);
    \draw[rotate=-120,dotted, thick, postaction={decorate}] (-1,0) -- (-0.86,-0.5);

    \draw[color=gray, dotted] (-1.17,0.67) -- (0.91,0.41);
    \draw[color=gray, dotted] (1.17,0.67) -- (-0.39,-0.92);

    \draw[color=white, ultra thick] (-0.78,0.62) -- (0.91,0.41);
    \draw[color=white, ultra thick] (0.91,0.405) -- (-0.39,-0.92);

    \draw[dotted, thick, postaction={decorate}] (-0.78,0.62) -- (0.91,0.41);
    \draw[dotted, thick, postaction={decorate}] (0.91,0.405) -- (-0.39,-0.92);

    \draw [line width=0.1pt,domain=0:6.28,variable=\t,smooth] plot ({sin(\t r)},{cos(\t r)});

    \draw [ultra thick, dotted, white, domain=-0.26:0.26,variable=\t,smooth] plot ({sin(\t r)},{cos(\t r)});
    \draw [ultra thick, dotted, white, domain=1.8:4.5,variable=\t,smooth] plot ({sin(\t r)},{cos(\t r)});


    \draw[color=gray,rotate=120,dashed,-] (0,-1.35) -- (0.685,-0.75);
    \draw[color=gray,rotate=120,dashed,-] (0,-1.35) -- (-0.685,-0.75);

    \draw[color=gray,rotate=-120,dashed,-] (0,-1.35) -- (0.685,-0.75);
    \draw[color=gray,rotate=-120,dashed,-] (0,-1.35) -- (-0.685,-0.75);

    \filldraw [rotate=-120-12 ] (0,-1) circle (0.6pt); 
    \filldraw [rotate=-120-108 ] (0,-1) circle (0.6pt); 


    \node at (0,1.2) {\scriptsize{$S_{\mathrm{VI}_0,\mathrm{VII}_0}$}};
    \node at (0,-1.2) {\scriptsize{$S_{\mathrm{VI}_0,\mathrm{VII}_0}$}};


    \draw[rotate=240] (0,-1.35) circle (0.1pt) node[anchor=east] {\scriptsize{$\mathrm{Q}_2/v$}};
    \draw[rotate=120]  (0,-1.35) circle (0.1pt) node[anchor=west] {\scriptsize{$\mathrm{Q}_3/v$}};

\end{tikzpicture}
    \addtocounter{subfigure}{-1}\captionof{subfigure}{\footnotesize{ $v\in (1/2,1)$.
    }}\label{FIG:SVIVIIsup}
\end{subfigure}
\endminipage
\captionof{figure}{The common stable set $S_{\mathrm{VI}_0,\mathrm{VII}_0}$
for Bianchi type $\mathrm{VI}_0$ and $\mathrm{VII}_0$. In addition, projected
onto $(\Sigma_+,\Sigma_-)$-space, there are illustrative heteroclinic chains
located on the $\mathrm{II}_2\cup\mathrm{II}_3\cup\mathrm{K}^\ocircle$ boundary.
In particular, $v\in(1/2,1)$ admits a heteroclinic cycle/chain with period 2,
which resides on the projected line between $\mathrm{Q}_2/v$ and $\mathrm{Q}_3/v$
characterized by $\Sigma_+=-1/(2v)$.
}\label{FIG:SVIVII}
\end{figure}
\begin{proposition}\label{lem:genVIalpha}
In Bianchi type $\mathrm{VI}_0$ the limit sets (in $\tau_-$) are as follows:
\begin{itemize}
\item[(i)] When $v\in (0,1/2]$, the $\alpha$-limit set for all
orbits resides in the set $A_2\cap A_3$ in $\mathrm{K}^\ocircle$,
where $A_2\cap A_3$ reduces to the Taub point $\mathrm{T}_1$ when $v=1/2$.
The $\omega$-limit set for all orbits resides in the set $S_{\mathrm{VI}_0}$.
\item[(ii)] When $v\in (1/2,1)$, the $\alpha$-limit set for all orbits is
the fixed point $p_{\mathrm{VI}_0}$ given by
\begin{equation}\label{p_VI}
p_{\mathrm{VI}_0}:= \left\{ (\Sigma_+,\Sigma_-,N_2,N_3)=\left(-\frac{1}{2v},0, \frac{\sqrt{1 - 1/(4v^2)}}{2},-\frac{\sqrt{1 - 1/(4v^2)}}{2}\right) \right\}.
\end{equation}
%
Apart from $p_{\mathrm{VI}_0}$,
the $\omega$-limit set of all orbits on the invariant subset
$\Sigma_+ = - 1/(2v)$ consists of the heteroclinic chain with period 2,
while the $\omega$-limit set of all orbits with $\Sigma_+\neq -1/(2v)$ resides
in the set $S_{\mathrm{VI}_0}$.\footnote{Therefore, the invariant set
$\Sigma_+ = - 1/(2v)$ is a co-dimension one stable set for the
heteroclinic chain with period 2. In particular, this set is equivalent to a Bowen's eye,
where the fixed point $p_{\mathrm{VI}_0}$ is surrounded by spiraling orbits toward
the heteroclinic chain with period 2, see~\cite{Tak94} and~\cite{dut19}.}
\end{itemize}
\end{proposition}
\begin{proof}
All type $\mathrm{VI}_0$ orbits satisfy $\Sigma^2 < 1$, and thereby $|\Sigma_+|<1$, while
$\Sigma^2 = 1$ corresponds to the type I boundary set $\mathrm{K}^\ocircle$,
since the constraint~\eqref{constrVIVII} yields $(N_2-N_3)^2 =N_2^2+N_3^2 + 2|N_2N_3|=0$,
and thus $N_2=N_3=0$.

For the subcritical and critical type $\mathrm{VI}_0$ models, with $v\in (0,1/2]$,
the function $1 + 2v\Sigma_+$ is bounded according to
$0 \leq 1 - 2v < 1 + 2v\Sigma_+ \leq 1 + 2v$, and, due to~\eqref{MonotoneVIVII}, it is
monotonically increasing. Thus
$\lim_{\tau_-\rightarrow\pm\infty}\Sigma^2 = 1$ in~\eqref{MonotoneVIVII}, and
hence $\lim_{\tau_-\rightarrow\pm\infty}(N_2,N_3) = (0,0)$, due to the constraints.
Therefore both the $\alpha$- and $\omega$-limit sets for all type $\mathrm{VI}_0$
orbits belong to the set $\mathrm{K}^\ocircle$. It then follows from the stability
properties of $\mathrm{K}^\ocircle$ that the $\alpha$-limit set for these orbits
resides in the set $A_2\cap A_3$ in the subcritical case, $v\in(0,1/2)$, while it
consists of the Taub point $\mathrm{T}_1$ with $\Sigma_+=-1$ in the critical case
$v=1/2$. It also follows for both the subcritical and critical cases that the
$\omega$-limit set of all type $\mathrm{VI}_0$ orbits lies in the stable set
$S_{\mathrm{VI_0}}$.

In the supercritical case, $v\in (1/2,1)$, the function $Z_{\mathrm{sup}}>0$
in~\eqref{MVIVII} is strictly monotonically increasing in
the type $\mathrm{VI}_0$ state space, except at the fixed point
$p_{\mathrm{VI}_0}$, given by~\eqref{p_VI}, 
where $Z_{\mathrm{sup}}$ attains its global minimum, $Z_{\mathrm{sup}}(p_{\mathrm{VI}_0}) = 4(4v^2-1)>0$.
Since $Z_{\mathrm{sup}}$ is strictly monotonically increasing for all non-$p_{\mathrm{VI}_0}$ type
$\mathrm{VI}_0$ orbits, it follows that their $\alpha$-limits
reside at the minimum of $Z_{\mathrm{sup}}$ at $p_{\mathrm{VI}_0}$, see the monotonicity
principle in~\cite{waiell97}, which also yields that
$Z_{\mathrm{sup}}\rightarrow\infty$ as $\tau_-\rightarrow\infty$, for all
non-$p_{\mathrm{VI}_0}$ orbits. Since the numerator $(1+2v\Sigma_+)^2$ of $Z_{\mathrm{sup}}$
in~\eqref{MVIVII} is bounded, it follows that $\lim_{\tau_-\rightarrow\infty}N_2N_3=0$,
and thus that the $\omega$-limit set of all non-$p_{\mathrm{VI}_0}$ supercritical type
$\mathrm{VI}_0$ orbits resides in the $\mathrm{II}_2\cup\mathrm{II}_3\cup\mathrm{K}^\ocircle$
boundary. According to~\eqref{MonotoneVIVII}, $1+2v\Sigma_+=0$ describes an invariant
separatrix surface, which divides the remaining state space into two
disjoint sets, $1+2v\Sigma_+<0$ and $1+2v\Sigma_+>0$, on which $1 + 2v\Sigma_+$
is monotone.\footnote{As described in Appendix~\ref{app:lrscaleaut},
the existence of the invariant set $1+2v\Sigma_+=0$ follows from a discrete symmetry,
which also results in that the flow of~\eqref{MonotoneVIVII} is equivariant under
a change of sign of $1 + 2v\Sigma_+$.}
It follows from monotonicity principle~\cite{waiell97} that the $\omega$-limit set
of all non-$p_{\mathrm{VI}_0}$ orbits on the
invariant set $\Sigma_+ = -1/(2v)$ are given by the boundary, i.e., the
heteroclinic cycle/chain with period 2.
With similar reasoning as in the subcritical and critical cases,
equation~\eqref{MonotoneVIVII} yields that the $\omega$-limit set for all orbits
in the subset $1+2v\Sigma_+<0$ ($1+2v\Sigma_+>0$) resides in the connected
component of the set $S_{\mathrm{VI_0}}$ with $1+2v\Sigma_+<0$ ($1+2v\Sigma_+>0$).
\end{proof}

Let us now turn to type $\mathrm{VII}_0$, but before presenting asymptotic results
we first consider the locally rotationally symmetric (LRS) type
$\mathrm{VII}_0$ subset (for additional information about the LRS models,
see Appendix~\ref{app:lrscaleaut}). This invariant set
is given by $N_-=0$ and $\Sigma_-=0$, where the constraint~\eqref{constrVIVII}
divides the LRS subset into two disjoint invariant sets consisting of the two
lines at $\Sigma_+ = 1$ and $\Sigma_+ = -1$, i.e.,
\begin{subequations}\label{LRS}
\begin{align}
\mathrm{LRS}^\pm &:= 
\left\{ (\Sigma_+,0,N_2,N_3) \in \mathbb{R}^4 \Bigm|
\begin{array}{c}
\,\, \Sigma_+=\pm 1, \\
\,\, N_2=N_3\neq 0
\end{array}
\right\},
\end{align}
\end{subequations}
where the superscript of $\mathrm{LRS}^\pm$ is determined by the sign of $\Sigma_+$.
Let $N := N_2 = N_3>0$. Then the flow on the $\mathrm{LRS}^\pm$ subsets is determined by
\begin{equation}\label{Neqn}
N^\prime = -2\Sigma_+(2v\Sigma_+ + 1)N, \qquad \Sigma_+ = \pm 1.
\end{equation}
On $\mathrm{LRS}^+$, where $\Sigma_+=+1$, the variable $N\in (0,\infty)$ monotonically
decreases from  $\lim_{\tau_-\rightarrow -\infty}N =\infty$  to $0$, and hence
the orbit in the invariant line ends at $\mathrm{Q}_1$ in the set $\mathrm{K}^\ocircle$
for all $v\in(0,1)$. On $\mathrm{LRS}^-$, where $\Sigma_+=-1$, there are three $v$-dependent
cases: the critical case, $v=1/2$, which results in a line of fixed points; the
subcritical case, $v \in (0,1/2)$, which yields an orbit that emanates from
$\mathrm{T}_1$, where $N\in (0,\infty)$ subsequently monotonically increases,
which results in $\lim_{\tau_-\rightarrow\infty}N = \infty$; the supercritical case,
$v \in (1/2,1)$, reverses the flow and leads to an orbit for which
$\lim_{\tau_-\rightarrow -\infty}N =\infty$, while it ends at
$\mathrm{T}_1$.
%
%

The next Propositions address the $\alpha$-limit and $\omega$-limit
sets for the type $\mathrm{VII}_0$ models.
\begin{proposition}\label{lem:genVIVIIomega}
The $\omega$-limit set (in $\tau_-$) for all Bianchi type $\mathrm{VII}_0$ orbits
resides in the stable set $S_{\mathrm{VII}_0}$ in the Kasner circle set
$\mathrm{K}^\ocircle$, apart from three exceptions:
\begin{itemize}
\item[(i)] When $v \in (0,1/2)$, the $\mathrm{LRS}^-$ set consists of an
orbit for which $\lim_{\tau_-\rightarrow\infty} N = \infty$.
\item[(ii)] When $v=1/2$, the $\mathrm{LRS}^-$ set is a line of fixed
points $N_2=N_3=\mathrm{constant}$.
\item[(iii)] When $v\in (1/2,1)$, there is an invariant set of co-dimension
one, characterized by $\Sigma_+=-1/(2v)$, for which the heteroclinic cycle with
period 2 on the $\mathrm{II}_2\cup\mathrm{II}_3\cup\mathrm{K}^\ocircle$
boundary is the $\omega$-limit set.
\end{itemize}
\end{proposition}
\begin{proof}
The first two exceptions follow from the previous analysis of the LRS type
$\mathrm{VII}_0$ subset, due to~\eqref{Neqn}. Consider therefore type
$\mathrm{VII}_0$ non-LRS orbits, i.e., orbits for which $\Sigma_-^2 + N_-^2>0$
and thereby $|\Sigma_+|<1$ due to the constraint. Note that in contrast to the type $\mathrm{VII}_0$ unbounded state space, its boundary is given by the compact set $\mathrm{II}_2\cup\mathrm{II}_3\cup\mathrm{K}^\ocircle$.

In the subcritical and critical cases, $Z_{\mathrm{sub}}>0$ for all
non-$\mathrm{LRS}^-$ orbits. In the $\mathrm{LRS}^-$ case the orbit satisfies
$\lim_{\tau_-\rightarrow\infty}N_2N_3 = \infty$ in the subcritical case, while
$\mathrm{LRS}^-$ yields a line of fixed points with constant $N_2=N_3$
in the critical case. Then note that
\begin{subequations}\label{monp}
\begin{align}
(1 + 2v\Sigma_+)^\prime &= 4vN_-^2(1+2v\Sigma_+), \qquad N_-:=N_2-N_3,\\
(1 + 2v\Sigma_+)''|_{N_-=0} &= 0,\\
(1 + 2v\Sigma_+)'''|_{N_-=0} &= 96(N_2+N_3)^2\Sigma_-^2(1 + 2v\Sigma_+).
\end{align}
\end{subequations}
Thus $(1 + 2v\Sigma_+)$ is monotonically increasing for all non-LRS orbits
(i.e., orbits such that $\Sigma_-^2 + N_-^2>0$), except when $N_-=0$ (and thereby
$\Sigma_-\neq 0$), which corresponds to an inflection point in the growth of
the positive quantity $(1 + 2v\Sigma_+)$, due to~\eqref{monp}. Thus all
non-$\mathrm{LRS}^-$ orbits eventually enter the (positively) invariant set
$\Sigma_+>-2v$. Since, due to~\eqref{Z2prim}, $Z_{\mathrm{sub}}>0$ is strictly
monotonically increasing in the invariant set $\Sigma_+>-2v$, it follows that
$\lim_{\tau_-\rightarrow\infty}Z_{\mathrm{sub}} = \infty$ and thereby
$\lim_{\tau_-\rightarrow\infty}N_2N_3 = 0$. Thus the $\omega$-limit set of all
non-$\mathrm{LRS}^-$ orbits in the subcritical and critical cases resides in
the $\mathrm{II}_2\cup\mathrm{II}_3\cup\mathrm{K}^\ocircle$ boundary set. The
same local analysis of this boundary set as in type $\mathrm{VI}_0$ yields the same
result for the non-$\mathrm{LRS}^-$ orbits in type $\mathrm{VII}_0$.

In the supercritical case, $\Sigma_+= -1/(2v)$ forms an invariant separatrix surface,
which divides the $\mathrm{VII}_0$ state space into two disjoint invariant subsets
with $1+ 2\Sigma_+\neq 0$ on which $1 + 2v\Sigma_+$ is monotone, as in
type $\mathrm{VI}_0$. Due to~\eqref{Z1prim}, $Z_{\mathrm{sup}}>0$
in~\eqref{MVIVII} is strictly monotonically increasing everywhere in the
type $\mathrm{VII}_0$ state space, except at two lines on the invariant subset
$\Sigma_+=-1/(2v)$ given by $\Sigma_-=0$ and thereby $N_2=N_3\pm\sqrt{1- (1/2v)^2}$,
due to the constraint~\eqref{constrVIVII}. However, these lines, denoted
by $L^\pm_{\mathrm{VII}_0}$, are not invariant sets, in contrast to the fixed point $p_{\mathrm{VI}_0}$
in type $\mathrm{VI}_0$, since
$\Sigma_-^\prime|_{L^\pm_{\mathrm{VII}_0}} = \pm 2\sqrt{3}(N_2+N_3)\sqrt{1-(1/2v)^2}$.
This fact in combination with that $N_2=N_3\pm\sqrt{1- (1/2v)^2}$ on the lines
$L^\pm_{\mathrm{VII}_0}$
implies that
$\lim_{\tau_-\rightarrow\infty}Z_{\mathrm{sup}}=\infty$.
Since the numerator
$(2v+\Sigma_+)^2$ of $Z_{\mathrm{sup}}$ in~\eqref{MVIVII} is bounded, 
the unbounded growth of $Z_{\mathrm{sup}}$ implies that $\lim_{\tau_-\rightarrow\infty}N_2N_3 =0$.
Thus at least one of $N_2$ or $N_3$ decays to zero, while the other
variable is asymptotically bounded due to the constraint~\eqref{constrVIVII}.
Hence the $\omega$-limit set for all non-$\mathrm{LRS}^-$ type $\mathrm{VII}_0$
orbits resides in the $\mathrm{II}_2\cup\mathrm{II}_3\cup\mathrm{K}^\ocircle$
boundary set. This in turn leads to the same conclusions for the $\omega$-limit
sets as for the non-$p_{\mathrm{VI}_0}$ orbits in type $\mathrm{VI}_0$.
\end{proof}
\begin{proposition}\label{lem:genVIVIIalpha}
The $\alpha$-limit set (in $\tau_-$) for all Bianchi type $\mathrm{VII}_0$ orbits are as follows:
\begin{itemize}
\item[(i)] When $v \in (0,1/2)$, the $\alpha$-limit set of all non-$\mathrm{LRS}^+$ orbits
reside in the set $A_2\cap A_3$.
The $\mathrm{LRS}^+$ set consists of an
orbit such that $\lim_{\tau_-\rightarrow-\infty} N = \infty$, where $N:=N_2=N_3$.
\item[(ii)] When $v=1/2$, the $\alpha$-limit set of all non-$\mathrm{LRS}$ orbits
is the line of fixed points, $\mathrm{LRS}^-$. 
The $\mathrm{LRS}^+$ set consists of an
orbit for which $\lim_{\tau_-\rightarrow-\infty} N = \infty$.
\item[(iii)] When $v\in (1/2,1)$, all non-$\mathrm{LRS}$ orbits asymptotically
satisfy
\begin{equation}
\lim_{\tau_-\rightarrow-\infty}\Sigma_+=-\frac{1}{2v}, \qquad \lim_{\tau_-\rightarrow-\infty}N_+=\infty, \qquad N_+:=N_2+N_3,
\end{equation}
whereas $\Sigma_-$ and $N_-:=N_2-N_3$ are asymptotically oscillatory, since in coordinates $(\Sigma_-,N_-)=(\sqrt{1-\Sigma_+^2}\cos\psi,\sqrt{1-\Sigma_+^2}\sin\psi)$, the angle $\psi$ is strictly monotonic as $\tau_-\rightarrow -\infty$. Each $\mathrm{LRS}^\pm$ set consists of an orbit
such that $\lim_{\tau_-\rightarrow-\infty} N = \infty$.
\end{itemize}
\end{proposition}

\begin{proof}
The reasoning in the proof of the previous proposition about $\omega$-limits
also yield the basic elements when $\tau_-\rightarrow - \infty$, but,
due to the unboundedness of the type $\mathrm{VII}_0$ state space, there
are some new issues, which did not occur when $\tau_-\rightarrow\infty$.

In the subcritical and critical cases, similar arguments as in the previous discussion
about $\omega$-limit sets lead to the following: $(1+2v\Sigma_+)$ is monotonically decreasing
when $\tau_-\rightarrow - \infty$, which shows that the $\alpha$-limit set for
all non-$\mathrm{LRS}^+$ orbits resides in the set $A_2\cap A_3$ for the subcritical
case, and in the line of fixed points $\mathrm{LRS}^-$ for the critical
case.\footnote{For type $\mathrm{VII}_0$ asymptotics in the critical GR case, see section 7
in~\cite{ring01} and section 4 in~\cite{ring03}. A more detailed analysis of the critical GR
case was performed in~\cite{heiugg09b}, which showed that each fixed
point on the $\mathrm{LRS}^-$ subset 
is the $\alpha$-limit set for a one-parameter set of orbits.}

For the non-LRS orbits in the supercritical case, $(1+2v\Sigma_+)^2>0$ is
monotonically decreasing as $\tau_-\to -\infty$, and orbits thereby approach
the invariant set $\Sigma_+= -1/(2v)$. Similar reasoning
as in the proposition for the $\omega$-limit using the monotone function
$Z_{\mathrm{sup}}$ results in $\lim_{\tau_-\rightarrow-\infty} N_2N_3=\infty$ for the
non-LRS orbits. Since $N_+^2 = N_-^2 + 4N_2N_3$, $N_\pm = N_2\pm N_3$,
and since $N_-^2$ is bounded ($N_-^2 = 1- \Sigma^2< 1$) it follows that
$\lim_{\tau_-\rightarrow-\infty} N_+=\infty$. 
To study the asymptotic behaviour of $\Sigma_-$ and $N_-$, we introduce
polar coordinates for $N_-$ and $\Sigma_-$ and solve the
constraint~\eqref{constrVIVII}. This leads to the following set of
new variables:
\begin{equation}
(\Sigma_+,\Sigma_-, N_-, N_+) =
\left(\Sigma_+,\sqrt{1-\Sigma_+^2}\cos\psi, \sqrt{1-\Sigma_+^2}\sin\psi, \frac{1}{\sqrt{3}M}\right),
\end{equation}
which result in the unconstrained dynamical system
\begin{subequations}\label{VIIunconstrsys}
\begin{align}
\Sigma_+^\prime &= (1-\Sigma_+^2)(1 + 2v\Sigma_+)[1 - \cos(2\psi)],\\
M^\prime &= \left(2\Sigma_+(1+2v\Sigma_+) + (1-\Sigma_+^2)[2v(1 + \cos(2\psi)) + 3M\sin(2\psi)]\right)M,\\
\psi^\prime &= -\frac{2}{M} - (2v + \Sigma_+)\sin(2\psi).\label{psiprime}
\end{align}
\end{subequations}

We have already shown that $\lim_{\tau_-\rightarrow-\infty}\Sigma_+=-1/(2v)$
and $\lim_{\tau_-\rightarrow-\infty}N_+=\infty$, from which it follows that
$\lim_{\tau_-\rightarrow-\infty}M=0$. Due to this, and since $2v + \Sigma_+$
is bounded, equation~\eqref{psiprime} implies
that $\psi$ is strictly monotonic as $\tau_-\rightarrow-\infty$.\footnote{For brevity
we will refrain from deriving explicit asymptotic expressions for $\Sigma_+$, $M$ and $\psi$,
and thereby $\Sigma_\pm$ and $N_2$, $N_3$ in the supercritical type $\mathrm{VII}_0$ case.
However, to derive such expressions there are several different methods one can use, e.g.
those in~\cite{waietal99,leenun18}, or in~\cite{ren07}, or the asymptotic averaging method
used in~\cite{alhetal15}.}
\end{proof}

We now compare the ingredients for the proof given here for the
critical case with well-known proofs in GR. In our approach, $Z_{\mathrm{sup}}$
and $Z_{\mathrm{sub}}$ in~\eqref{MVIVII} become identical in the
critical case:
\begin{equation}\label{critmon}
Z_\mathrm{crit} := Z_{\mathrm{sup}}=Z_{\mathrm{sub}}= \frac{(1+\Sigma_+)^2}{|N_2N_3|},
\end{equation}
where
\begin{equation}\label{critmonp}
Z_\mathrm{crit}^\prime = 4(1+ \Sigma_+)Z_\mathrm{crit}.
\end{equation}
In the traditional approach to the GR case, see~\cite{bog85,waihsu89,ren97,rin01},
two monotone functions $Z_+$ and $Z_-$ were respectively used for type $\mathrm{VI}_0$ and $\mathrm{VII}_0$:
\begin{equation}
Z_\pm := \frac{\Sigma_-^2 + N_\pm^2}{|N_2N_3|},
\end{equation}
where
\begin{equation}
Z_\pm^\prime = \frac{4(1+\Sigma_+)\Sigma_-^2}{\Sigma_-^2 + N_\pm^2}Z_\pm.
\end{equation}
Even though the evolution equation for the type $\mathrm{VII}_0$ function $Z_-$ can be simplified to $Z_-^\prime = 4\Sigma_-^2(1-\Sigma_+)^{-1}Z_-$,
%
%
%
%
%
the functions $Z_\pm$ arguably gives a more cumbersome analysis
than the unified function $Z_\mathrm{crit}=Z_{\mathrm{sup}}=Z_{\mathrm{sub}}$ in~\eqref{critmon}, which naturally arises from the scale-automorphism symmetry of these models.



\subsubsection*{Bianchi types VIII and IX}

In Appendix~\ref{app:lrscaleaut}, we derive the following monotone function
from the scale symmetry of the vacuum $\lambda$-$R$ type VIII and IX models:
\begin{equation}\label{Delta}
\Delta := 3|N_1N_2N_3|^{2/3},
\end{equation}
which, due to~\eqref{intro_dynsyslambdaR}, satisfy
\begin{equation}
\Delta^\prime = -24v\Sigma^2\Delta.
\end{equation}
Thus $\Delta$ is monotonically decreasing when $\Sigma^2 >0$,
and has an inflection point when $\Sigma^2 = 0$, since
%
\begin{equation}
\left.\Delta''\right|_{\Sigma^2 =0} = 0, \qquad
\left.\Delta'''\right|_{\Sigma^2 =0} = - 8 v\left({\cal S}_1^2 + {\cal S}_2^2 + {\cal S}_3^2\right)\Delta,
\end{equation}
where ${\cal S}_1^2 + {\cal S}_2^2 + {\cal S}_3^2>0$.
Therefore,
\begin{equation}\label{Delta0}
\lim_{\tau_-\rightarrow\infty}\Delta = 0,
\end{equation}
and consequently
\begin{equation}\label{Sigma2asymp}
\lim_{\tau_-\rightarrow\infty}\Sigma^2 \leq 1,
\end{equation}
since the definition of $\Omega_k$ in~\eqref{Omega_k} implies
$\Omega_k + \Delta \geq 0$,\footnote{In type $\mathrm{VIII}$, $\Omega_k> 0$ since
the curvature scalar $R < 0$, while $R$ can be positive in type $\mathrm{IX}$.
Using that $R = e^{-4v\beta^\lambda}\bar{R}$ and that $\bar{R}$ is a function of
$\beta^\pm$, as given in~\eqref{barR} in Appendix~\ref{app:dom}, shows that $\bar{R}$
has a negative minimum when $\beta^\pm =0$. Adding a constant so that the minimum
becomes zero, corresponds to adding $\Delta$ to $\Omega_k$, where $\Omega_k + \Delta =0$
when $N_1=N_2=N_3$, which corresponds to $\beta^\pm = 0$. Thus $\Omega_k+\Delta \geq 0$.}
and hence $\Sigma^2 \leq 1 + \Delta$, due to the constraint~\eqref{intro_cons1}.
Moreover, because of~\eqref{Delta0} it follows that at least one of the variables $N_\alpha$
must decay to $0$, for all $v \in (0,1)$. Thus the $\omega$-limit set of all Bianchi type
IX orbits resides in the union of the closure of the type $\mathrm{VII}_0$ subsets,
whereas the $\omega$-limit set of all type VIII orbits lies in closure of the union
of the single type $\mathrm{VII}_0$ subset and the two type $\mathrm{VI}_0$
subsets.\footnote{All initially expanding vacuum $\lambda$-$R$ Bianchi type
I--VIII models are forever expanding, due to that the spatial curvature of
these models satisfies $R\leq 0$. However, all initially expanding type
IX solutions reach a point of maximum expansion and then recollapse.
This has been proven in~\cite{heiugg09b} for GR, and a similar proof
can be given for the $\lambda$-$R$ models. We thereby assume that initial
data in type IX correspond to initially expanding solutions, where we are
interested in the initial singularity when $\tau_-\rightarrow \infty$.}

Similarly to types $\mathrm{VI}_0$ and $\mathrm{VII}_0$, the 
dynamical system~\eqref{intro_dynsyslambdaR} for Bianchi types $\mathrm{VIII}$ and $\mathrm{IX}$ 
admits invariant locally rotationally symmetric (LRS) sets: the (physically equivalent)
type $\mathrm{IX}$ LRS sets are given by
\begin{subequations}\label{LRSIX}
\begin{align}
\mathcal{LRS}_\alpha &:=
\left\{(\Sigma_\alpha,\Sigma_\beta,\Sigma_\gamma,N_\alpha,N_\beta,N_\gamma) \in \mathbb{R}^6 \Bigm|
\begin{array}{c}
\,\, \Sigma_\beta=\Sigma_\gamma, \, N_\beta=N_\gamma, \\
\,\, \text{ satisfying \eqref{intro_cons1}-\eqref{intro_cons2}}
\end{array}
\right\},
\end{align}
\end{subequations}
where $(\alpha\beta\gamma)$ is a permutation of $(123)$, while type $\mathrm{VIII}$ only
admits a single LRS set since one of the variables $N_1$, $N_2$, $N_3$ has an opposite sign compared
to the other two. The LRS type $\mathrm{VIII}$ and $\mathrm{IX}$ sets have three distinct 
dynamical regimes, the subcritical, critical and supercritical cases, and 
boundaries given by the one-dimensional sets $\mathrm{LRS}^\pm$ in \eqref{LRS}.

Next we turn to some dynamical conjectures for Bianchi type VIII and IX.
We expect that the above features will be important ingredients
in future proofs of these conjectures, both in the $\lambda$-$R$
case and for more general HL models.

\section{Dynamical conjectures}\label{sec:conjectures}

Apart from the previous section, the main part of the paper has focussed on
the discrete dynamics of the Kasner circle map $\mathcal{K}$, associated
with the heteroclinic chains obtained by concatenation of Bianchi type II
heteroclinic orbits in the $\lambda$-$R$ Bianchi type VIII and IX models
with $v\in (0,1)$. In particular, we have shown that the
critical case to which GR belongs, $v=1/2$, represents a bifurcation,
where non-generic chaos on a Cantor set for the supercritical case, $v \in (1/2,1)$,
is replaced by generic chaos for the critical and subcritical
cases, $v \in (0,1/2]$.

It remains to connect the discrete dynamics of~$\mathcal{K}$ with
asymptotic continuous dynamics in Bianchi type VIII and IX, described by the
dynamical system~\eqref{intro_dynsyslambdaR}. We therefore conclude with some
dynamical conjectures, which reflect an expected hierarchy of difficulty as
regards possible proofs. The conjectures can be divided into two classes:
(i) if and how many type VIII and IX solutions have an infinite heteroclinic
chain on the Bianchi type I and II boundary as their $\omega$-limit set,
(ii) if generic solutions of type VIII and IX asymptotically approach the
type I boundary or the union of the type I and II boundary sets,
and how this depends on the parameter $v$.


%
\begin{conjecture}\label{conjperiod2}
In the Bianchi type VIII and IX supercritical case, $v\in(1/2,1)$, each
heteroclinic cycle has a stable invariant set
(as $\tau_-\rightarrow\infty$) of co-dimension one.
\end{conjecture}

It should be possible to prove this conjecture with, e.g., the methods used
in~\cite{lieetal10,lieetal12,beg10}, but the situation for the period 2 cycle
is arguably more special than the problems in the aforementioned references,
and other types of proofs might therefore be possible. 
Loosely speaking, the heteroclinic chains with period 2
form the `boundary' of the infinite heteroclinic chains associated with the Cantor
set $C$. It thus seems natural to establish if the period 2 chain has an attracting
set of co-dimension one before addressing the next more ambitious Conjecture (which
contains the previous one as a special case since heteroclinic cycles can be viewed 
as special examples of infinite heteroclinic chains).

\begin{conjecture}\label{conj:supercriticalchain}
In the Bianchi type VIII and IX supercritical case, $v\in (1/2,1)$, each
infinite heteroclinic chain associated with the Cantor set $C$ has a stable
invariant set (as $\tau_-\rightarrow\infty$) of co-dimension one.
\end{conjecture}

Incidentally, the special role of $C$ illustrates that it may not be sufficient
to establish the existence of a stable set for the full global understanding of
asymptotics of a given model, such as the models with dimension 11 or higher
in~\cite{dametal02}. Models with a stable set may thereby be more complicated
than one expects, and in quite interesting ways. Similar conjectures can be
formulated for the set $\tilde{C} \subset \mathrm{K}^\ocircle$ in the subcritical case, 
$v\in (0,1/2)$, e.g.,
\begin{conjecture}\label{conj:subcriticalchain}
In the Bianchi type VIII and IX subcritical case, $v\in (0,1/2)$, each infinite heteroclinic chain
associated with the set $\tilde{C} \subset \mathrm{K}^\ocircle$
in~\eqref{defofCtilde} has a stable invariant set (as $\tau_-\rightarrow\infty$)
of co-dimension one.
\end{conjecture}

To address the more general issues in the two following conjectures, presumably
require more general methods than needed to prove the previous conjectures, see~\cite{beg10,bre16,dut19}.
\begin{conjecture}\label{conj:supercriticalS}
In the Bianchi type VIII and IX supercritical case, $v \in (1/2,1)$, the stable set $S$ on
$\mathrm{K}^\ocircle$ is the attractor ${\cal A}_-$
(as $\tau_-\rightarrow \infty$).
\end{conjecture}

There are subtleties in how to define an attractor,
as discussed in~\cite{Milnor85}. In the present context we are interested
in the behaviour of most orbits in the state space, and thus we deal
with an attractor that attracts generic sets of orbits in the state space.
This set is called the \emph{likely limit set} in~\cite{Milnor85}, which is
the unique maximal attractor. The Kasner circle consists of six physically
equivalent subsets, related by axis permutations~\eqref{permSYM},
which thereby are the six elements in the unique equivalence class
of the quotient of the attractor ${\cal A}_-$ under the action of the symmetric group $\mathrm{S}_3$ according to~\eqref{permSYM}. Combining this feature
with the above conjecture suggests that we refer to the quotient space
${\cal A}_- / \mathrm{S}_3$ as the \emph{physical attractor}.

It is clear from the local analysis of $\mathrm{K}^\ocircle$ that $S$
attracts all nearby orbits. To prove Conjecture~\ref{conj:supercriticalS}
requires establishing that all generic sets of solutions only
have points on $S$ as their $\omega$-limit. Even though we expect that
there is a set of solutions that has the heteroclinic
chains associated with the Cantor set $C$ as their $\omega$-limits,
and thereby not $S$, we believe that this set is non-generic, as
suggested by Conjecture~\ref{conj:supercriticalchain}.


%
\begin{conjecture}\label{conj:subcritical}
In the Bianchi type VIII and IX subcritical and critical cases, $v \in (0,1/2]$,
the attractor ${\cal A}_-$ (as $\tau_- \rightarrow \infty$) consists of the set 
$\mathrm{K}^\ocircle\cup\mathrm{II}_1\cup\mathrm{II}_2\cup\mathrm{II}_3$.
\end{conjecture}

Equation~\eqref{Delta0} shows that $\lim_{\tau_-\rightarrow \infty}\Delta=0$ and hence
that the $\omega$-limits of the type IX (VIII) solutions reside on the union of
the type $\mathrm{VII}_0$ ($\mathrm{VII}_0$ and $\mathrm{VI}_0$), type II
and type I boundary sets. This is the foundation for the proof of the
conjecture in the critical GR case for type IX, see~\cite{rin01,heiugg09b}.
The conjecture for the other cases is highly non-trivial, especially from the
perspective of the present dynamical system formulation. This is due to that
Conjecture~\ref{conj:subcritical} relies on the entire history of solutions
when $\tau_-\rightarrow \infty$ and that the variables in the dynamical
system~\eqref{intro_dynsyslambdaR} do not capture this feature well.

In Appendix~\ref{app:dom} the system~\eqref{intro_dynsyslambdaR} is derived from
a Hamiltonian description of the field equations by a change of variables. As a result the evolution equation for one of the variables, not discussed in the main
text, decouples. The variables in~\eqref{intro_dynsyslambdaR} are thereby
particular coordinates describing a projection of the original state space. It
turns out that the original configuration space variables $\beta^\lambda, \beta^\pm$
in the Hamiltonian formulation are better in capturing the above mentioned history.
This is illustrated by the heuristic moving wall analysis in Appendix~\ref{app:dom},
which shows that on average $\beta^+$ and $\beta^-$ oscillate with increasing
amplitudes. This analysis also shows that excursions of generic solutions into the
type $\mathrm{VI}_0$ or $\mathrm{VII}_0$ subsets (where one of the cross terms
$N_1N_2$, $N_2N_3$, $N_3N_1$ is non-negligible) become increasingly unlikely. Even
though these excursions will happen, they become increasingly rare in some
probabilistic sense as $\tau_-=-\beta^\lambda \rightarrow \infty$,
and thus generic solutions are asymptotically described by sequences of Kasner
states and Bianchi type II solutions; for further details, see Appendix~\ref{app:dom}.
The conjecture is thus based on the assumption that for generic solutions the probability
that the cross terms become non-zero tends to zero when $\tau_-\rightarrow \infty$, which requires
some new statistical measure. 

It is worth noticing that this situation is reminiscent of that in Bianchi
type $\mathrm{VI}_{-1/9}$ and when using an Iwasawa frame in GR,
see~\cite{ugg13a,ugg13b,heietal09} and references therein. The dynamical
systems analysis in~\cite{heietal09} of the suppression of `double transitions'
(non-zero cross terms $N_1N_2$, $N_2N_3$, $N_3N_1$ in the present formulation)
is particularly pertinent, especially since it indicates how hard it will be to rigorously
establish such features by using a dynamical system of the type~\eqref{intro_dynsyslambdaR}.

Presumably, the most difficult conjecture to prove among the ones above
is Conjecture~\ref{conj:subcritical} for the subcritical case, especially
when $v$ becomes increasingly small (if true; when $v=0$ the unstable region
for each of the cross terms is half of
$\mathrm{K}^\ocircle$, and thus the considerations in Appendix~\ref{app:dom}
concerning this bifurcation value should come as no surprise). Nevertheless, can the
methods in~\cite{rin01,heiugg09b} for the GR Bianchi type IX case be developed
and adapted to the $\lambda$-$R$ Bianchi type IX subcritical models with $v \in (0,1/2)$?
Note that these methods do not establish the conjecture for type VIII in GR,
although see~\cite{rin00} for some limited type VIII results, and also~\cite{bre16,dut19}.
The difficulties in GR for type VIII presumably also lead to difficulties
for the subcritical type VIII case and possibly also for
type IX.

All the above conjectures rely on that
$\lim_{\tau_-\rightarrow\infty}(N_1N_2,N_2N_3,N_3N_1)= (0,0,0)$,
and the behaviour of the individual terms $N_1$, $N_2$ and $N_3$.
Possible proofs using the dynamical systems approach presumably
involve the growth and decay of these quantities, which we therefore
now take a closer look at.
Without loss of generality, we describe the evolution equations using
the $\Sigma_\pm$ variables in~\eqref{Misner}, which are
adapted to the $\Sigma_1$-direction.
In Appendix~\ref{app:dom} we derive a system of evolution equations,
which can be written as follows,
\begin{subequations}\label{ODEtypeVIIIIX}
\begin{align}
\Sigma_+^\prime &= 2(1-\Sigma^2)(1+2v\Sigma_+) -6N_1(2N_1 - N_2 - N_3)],\label{ODEVIIIIX:Sp}\\
\Sigma_-^\prime &= 4v(1-\Sigma^2)\Sigma_- + 2\sqrt{3}(N_2 - N_3)(N_2 + N_3 - N_1),\\
N_1^\prime &= -4v\left[\left(\Sigma_+ - \frac{1}{2v}\right)^2 + \Sigma_-^2 - \left(\frac{1}{2v}\right)^2\right]N_1.\label{N1prime}\\
N_2^\prime &= -v\left[\left(\Sigma_+ + \sqrt{3}\Sigma_- + \frac{1}{v}\right)^2 + \left(\sqrt{3}\Sigma_+ - \Sigma_-\right)^2 - \frac{1}{v^2}\right]N_2,\\
N_3^\prime &= -v\left[\left(\Sigma_+ - \sqrt{3}\Sigma_- + \frac{1}{v}\right)^2 + \left(\sqrt{3}\Sigma_+ + \Sigma_-\right)^2 - \frac{1}{v^2}\right]N_3,
\end{align}
and the constraint
\begin{equation}\label{constrpm}
1 = \Sigma^2 + \Omega_k,
\end{equation}
where
\begin{align}
\Sigma^2 &:= \Sigma_+^2 + \Sigma_-^2,\\
\Omega_k &:= N_1^2 + N_2^2 + N_3^2 - 2N_1N_2 - 2N_2N_3 - 2N_3N_1.
\end{align}
\end{subequations}
%
Thus $N_1$ monotonically decreases (increases) when
$(\Sigma_+,\Sigma_-)$ is outside (inside) the circle
\begin{equation}\label{disc:circleN1}
\left(\Sigma_+ - \frac{1}{2v}\right)^2 + \Sigma_-^2 = \left(\frac{1}{2v}\right)^2,
\end{equation}
which has its center at $(\Sigma_+, \Sigma_-)=(1/2v,0)$ and a radius $1/2v$.
It also follows from~\eqref{ODEtypeVIIIIX} that the terms $N_2$ and $N_3$ decrease
(increase) outside (inside) similar circles, 
see Figure~\ref{FIG:circleSINGLE}.

\begin{figure}[H]
\minipage[b]{0.4\textwidth}\centering
\begin{subfigure}\centering
\begin{tikzpicture}[scale=1.3]
    \filldraw [lightgray!40,domain=0:6.28,variable=\t,smooth] plot ({(1.25)*sin(\t r)},{-(1.25)+(1.25)*cos(\t r)});
    \filldraw [rotate=120,lightgray!40,domain=0:6.28,variable=\t,smooth] plot ({(1.25)*sin(\t r)},{-(1.25)+(1.25)*cos(\t r)});
    \filldraw [rotate=-120,lightgray!40,domain=0:6.28,variable=\t,smooth] plot ({(1.25)*sin(\t r)},{-(1.25)+(1.25)*cos(\t r)});

    \draw [domain=0:6.28,variable=\t,smooth] plot ({(1.25)*sin(\t r)},{-(1.25)+(1.25)*cos(\t r)});
    \draw [rotate=120,domain=0:6.28,variable=\t,smooth] plot ({(1.25)*sin(\t r)},{-(1.25)+(1.25)*cos(\t r)});
    \draw [rotate=-120,domain=0:6.28,variable=\t,smooth] plot ({(1.25)*sin(\t r)},{-(1.25)+(1.25)*cos(\t r)});

    \draw [domain=0:6.28,variable=\t,smooth] plot ({sin(\t r)},{cos(\t r)});

    \draw [ultra thick,domain=-0.12:0.12,variable=\t,smooth] plot ({sin(\t r)},{cos(\t r)});
    \draw [ultra thick,domain=1.97:2.21,variable=\t,smooth] plot ({sin(\t r)},{cos(\t r)});
    \draw [ultra thick,domain=-2.21:-1.97,variable=\t,smooth] plot ({sin(\t r)},{cos(\t r)});

    \draw (0,0.96) circle (0.001pt) node[anchor= north] {\scriptsize{$A_2\cap A_3$}};

    \draw[shift={(0.1,0.88)},rotate=-10] (-0.09,0) -- (0,0) -- (0,0.2) -- (-0.09,0.2);
    \draw[shift={(-0.1,0.88)},rotate=10] (0.09,0) -- (0,0) -- (0,0.2) -- (0.09,0.2);

    \draw[shift={(-1.02,-0.45)},rotate=-60] (0.09,0.2) -- (0,0.2) -- (0,0) -- (0.09,0);
    \draw[shift={(-0.72,-0.54)},rotate=120] (0.09,0) -- (0,0) -- (0,0.2) -- (0.09,0.2);

    \draw[shift={(1.02,-0.45)},rotate=60] (-0.09,0.2) -- (0,0.2) -- (0,0) -- (-0.09,0);
    \draw[shift={(0.72,-0.54)},rotate=-120] (-0.09,0) -- (0,0) -- (0,0.2) -- (-0.09,0.2);

    \filldraw (1,0) circle (0.001pt) node[anchor= south east] {\scriptsize{$A_3$}};
    \filldraw (-1,0) circle (0.001pt) node[anchor= south west] {\scriptsize{$A_2$}};

    \filldraw [black] (0,1) circle (1.25pt); \node at (0,1.45) {\scriptsize{$\mathrm{T}_1$}};
    \filldraw [black] (0.88,-0.49) circle (1.25pt); \node at (1.35,-0.75) {\scriptsize{$\mathrm{T}_2$}};
    \filldraw [black] (-0.88,-0.49) circle (1.25pt); \node at (-1.35,-0.75) {\scriptsize{$\mathrm{T}_3$}};

    \draw (0,-0.95) -- (0,-1.05);
    \draw[rotate=120] (0,-0.95) -- (0,-1.05);
    \draw[rotate=-120] (0,-0.95) -- (0,-1.05);
\end{tikzpicture}
\addtocounter{subfigure}{-1}
\captionof{subfigure}{\footnotesize{$v\in(0,1/2)$.}}\label{FIG:singleSUB}
\end{subfigure}\endminipage\hfill
\minipage[b]{0.33\textwidth}\centering

\begin{subfigure}\centering
\begin{tikzpicture}[scale=1.3]
    \filldraw[white] (0,-2.5) circle (0.5pt);

    \filldraw [lightgray!40,domain=0:6.28,variable=\t,smooth] plot ({sin(\t r)},{-1+cos(\t r)});
    \filldraw [rotate=120,lightgray!40,domain=0:6.28,variable=\t,smooth] plot ({sin(\t r)},{-1+cos(\t r)});
    \filldraw [rotate=-120,lightgray!40,domain=0:6.28,variable=\t,smooth] plot ({sin(\t r)},{-1+cos(\t r)});

    \draw [domain=0:6.28,variable=\t,smooth] plot ({sin(\t r)},{-1+cos(\t r)});
    \draw [rotate=120,domain=0:6.28,variable=\t,smooth] plot ({sin(\t r)},{-1+cos(\t r)});
    \draw [rotate=-120,domain=0:6.28,variable=\t,smooth] plot ({sin(\t r)},{-1+cos(\t r)});

    \draw [domain=0:6.28,variable=\t,smooth] plot ({sin(\t r)},{cos(\t r)});

    \filldraw (0,1) circle (1.25pt) node[anchor= south] {\scriptsize{$\mathrm{T}_1$}};
    \filldraw (0.88,-0.49) circle (1.25pt) node[anchor= north west] {\scriptsize{ $\mathrm{T}_2$}};
    \filldraw (-0.88,-0.49) circle (1.25pt)node[anchor= north east] {\scriptsize{$\mathrm{T}_3$}};

    \draw (0,-0.95) -- (0,-1.05);
    \draw[rotate=120] (0,-0.95) -- (0,-1.05);
    \draw[rotate=-120] (0,-0.95) -- (0,-1.05);
\end{tikzpicture}
\addtocounter{subfigure}{-1}
\captionof{subfigure}{\footnotesize{$v=1/2$.}}\label{FIG:singleCRIT}
\end{subfigure}\endminipage\hfill
\minipage[b]{0.26\textwidth}\centering

\begin{subfigure}\centering
\begin{tikzpicture}[scale=1.3]
    \filldraw[white] (0,-2.5) circle (0.5pt);

    \filldraw [lightgray!40,domain=0:6.28,variable=\t,smooth] plot ({(2/3)*sin(\t r)},{-(2/3)+(2/3)*cos(\t r)});
    \filldraw [rotate=120,lightgray!40,domain=0:6.28,variable=\t,smooth] plot ({(2/3)*sin(\t r)},{-(2/3)+(2/3)*cos(\t r)});
    \filldraw [rotate=-120,lightgray!40,domain=0:6.28,variable=\t,smooth] plot ({(2/3)*sin(\t r)},{-(2/3)+(2/3)*cos(\t r)});

    \draw [domain=0:6.28,variable=\t,smooth] plot ({(2/3)*sin(\t r)},{-(2/3)+(2/3)*cos(\t r)});
    \draw [rotate=120,domain=0:6.28,variable=\t,smooth] plot ({(2/3)*sin(\t r)},{-(2/3)+(2/3)*cos(\t r)});
    \draw [rotate=-120,domain=0:6.28,variable=\t,smooth] plot ({(2/3)*sin(\t r)},{-(2/3)+(2/3)*cos(\t r)});

    \draw[dotted, thick, postaction={decoration={markings,mark=at position 0.41 with {\arrow[thick,color=gray]{latex reversed}}},decorate}] (-0.75,0.67) -- (0.75,0.67);
    \draw[rotate=120,dotted, thick, postaction={decoration={markings,mark=at position 0.41 with {\arrow[thick,color=gray]{latex reversed}}},decorate}] (-0.75,0.67) -- (0.75,0.67);
    \draw[rotate=-120,dotted, thick, postaction={decoration={markings,mark=at position 0.41 with {\arrow[thick,color=gray]{latex reversed}}},decorate}] (-0.75,0.67) -- (0.75,0.67);

    \draw [domain=0:6.28,variable=\t,smooth] plot ({sin(\t r)},{cos(\t r)});

    \draw [dotted, white, ultra thick, domain=-0.26:0.26,variable=\t,smooth] plot ({sin(\t r)},{cos(\t r)});
    \draw [dotted, white, ultra thick, domain=1.83:2.35,variable=\t,smooth] plot ({sin(\t r)},{cos(\t r)});
    \draw [dotted, white, ultra thick, domain=-2.35:-1.83,variable=\t,smooth] plot ({sin(\t r)},{cos(\t r)});

    \draw[shift={(0.25,0.87)},rotate=-20] (0.1,0) -- (0,0) -- (0,0.2) -- (0.1,0.2);
    \draw[shift={(-0.25,0.87)},rotate=20] (-0.1,0) -- (0,0) -- (0,0.2) -- (-0.1,0.2);

    \draw[shift={(-1.07,-0.25)},rotate=-80] (-0.1,0.2) -- (0,0.2) -- (0,0) -- (-0.1,0);
    \draw[shift={(-0.65,-0.64)},rotate=140] (-0.1,0) -- (0,0) -- (0,0.2) -- (-0.1,0.2);

    \draw[shift={(1.07,-0.25)},rotate=80] (0.1,0.2) -- (0,0.2) -- (0,0) -- (0.1,0);
    \draw[shift={(0.65,-0.64)},rotate=-140] (0.1,0) -- (0,0) -- (0,0.2) -- (0.1,0.2);


    \filldraw [black] (0,1) circle (1.25pt) node[anchor= south] {\scriptsize{$\mathrm{T}_1$}};
    \filldraw [black] (0.88,-0.49) circle (1.25pt) node[anchor= north west] {\scriptsize{$\mathrm{T}_2$}};
    \filldraw [black] (-0.88,-0.49) circle (1.25pt)node[anchor= north east] {\scriptsize{$\mathrm{T}_3$}};

    \draw (0,-0.95) -- (0,-1.05);
    \draw[rotate=120] (0,-0.95) -- (0,-1.05);
    \draw[rotate=-120] (0,-0.95) -- (0,-1.05);

\end{tikzpicture}
\addtocounter{subfigure}{-1}
\captionof{subfigure}{\footnotesize{$v\in (1/2,1)$. }}\label{FIG:singleSUP}
\end{subfigure}\endminipage
\captionof{figure}{The interior of the (gray) disk opposite to the
location of the Taub point $\mathrm{T}_\alpha\in \mathrm{K}^\ocircle$ in
$(\Sigma_1,\Sigma_2,\Sigma_3)$-space indicates growth of each individual $N_\alpha$,
$\alpha=1,2,3$. Outside their growth region
the individual terms decay. As $v\in (0,1)$ increases, the disks
radii decrease and they move toward the middle, where $\Sigma^2=0$. For $v\in (0,1/2)$ the disks
cover the whole region inside $\Sigma^2=1$. 
At $v=1/2$ the disks only intersects with $\Sigma^2=1$ at the $\Sigma_\alpha$ values of the
Taub points. For $v\in (1/2,1)$ the disks intersect at the $\Sigma_\alpha$
values of the heteroclinic chains with period two, and there is a neighborhood of the location
of the Taub points in $(\Sigma_1,\Sigma_2,\Sigma_3)$-space with decay. 
}\label{FIG:circleSINGLE}
\end{figure}

Expressing the evolution of the cross terms in the $\Sigma_\pm$ variables results in
the equations
\begin{subequations}\label{crossN1}
\begin{align}
(N_1N_2)^\prime &= -8v\left[\left(\Sigma_+ - \frac{1}{8v}\right)^2 +
\left(\Sigma_-+\frac{\sqrt{3}}{8v}\right)^2 -\left(\frac{1}{4v}\right)^2\right](N_1N_2),\\
(N_3N_1)^\prime &= -8v\left[\left(\Sigma_+ - \frac{1}{8v}\right)^2 +
\left(\Sigma_--\frac{\sqrt{3}}{8v}\right)^2 -\left(\frac{1}{4v}\right)^2\right](N_3N_1),\\
(N_2 N_3)^\prime &= -8v\left[\left(\Sigma_+ + \frac{1}{4v}\right)^2 + \Sigma_-^2
- \left(\frac{1}{4v}\right)^2\right](N_2N_3),\label{ODEcross}
\end{align}
\end{subequations}
Hence, e.g., $N_2N_3$ is monotonically decreases (increases) when
$(\Sigma_+,\Sigma_-)$ is outside (inside) the following circle
\begin{equation}\label{disc:circle}
\left(\Sigma_+ + \frac{1}{4v}\right)^2 + \Sigma_-^2 = \left(\frac{1}{4v}\right)^2,
\end{equation}
which has its center at $(\Sigma_+, \Sigma_-)=(-1/4v,0)$ and a radius $1/4v$.
It also follows from~\eqref{crossN1} that the other cross terms decrease
(increase) outside (inside) similar circles, obtained by axis permutations.
In particular, $N_1N_2$ and $N_3N_1$ decay when
$\Sigma_+ < -1/(8v)$ while $N_2N_3$ decays when $\Sigma_+ < -1/(2v)$.
These decay and growth regions are depicted in
Figure~\ref{FIG:circleCROSS}.

\begin{figure}[H]
\minipage[b]{0.33\textwidth}\centering
\begin{subfigure}\centering
\begin{tikzpicture}[scale=1.3]
    \filldraw [lightgray!80,domain=0:6.28,variable=\t,smooth] plot ({(1)*sin(\t r)},{(1)+(1)*cos(\t r)});
    \filldraw [rotate=120,lightgray!80,domain=0:6.28,variable=\t,smooth] plot ({(1)*sin(\t r)},{(1)+(1)*cos(\t r)});
    \filldraw [rotate=-120,lightgray!80,domain=0:6.28,variable=\t,smooth] plot ({(1)*sin(\t r)},{(1)+(1)*cos(\t r)});

    \draw [domain=0:6.28,variable=\t,smooth] plot ({(1)*sin(\t r)},{(1)+(1)*cos(\t r)});
    \draw [rotate=120,domain=0:6.28,variable=\t,smooth] plot ({(1)*sin(\t r)},{(1)+(1)*cos(\t r)});
    \draw [rotate=-120,domain=0:6.28,variable=\t,smooth] plot ({(1)*sin(\t r)},{(1)+(1)*cos(\t r)});

    \draw [domain=0:6.28,variable=\t,smooth] plot ({sin(\t r)},{cos(\t r)});

    \draw [ultra thick,domain=-0.28:0.28,variable=\t,smooth] plot ({sin(\t r)},{cos(\t r)});
    \draw [ultra thick,domain=1.81:2.35,variable=\t,smooth] plot ({sin(\t r)},{cos(\t r)});
    \draw [ultra thick,domain=-2.35:-1.81,variable=\t,smooth] plot ({sin(\t r)},{cos(\t r)});

    \draw (0,0.96) circle (0.001pt) node[anchor= north] {\scriptsize{$A_2\cap A_3$}};

    \draw[shift={(0.25,0.87)},rotate=-20] (-0.1,0) -- (0,0) -- (0,0.2) -- (-0.1,0.2);
    \draw[shift={(-0.25,0.87)},rotate=20] (0.1,0) -- (0,0) -- (0,0.2) -- (0.1,0.2);

    \draw[shift={(-1.07,-0.25)},rotate=-80] (0.1,0.2) -- (0,0.2) -- (0,0) -- (0.1,0);
    \draw[shift={(-0.65,-0.64)},rotate=140] (0.1,0) -- (0,0) -- (0,0.2) -- (0.1,0.2);

    \draw[shift={(1.07,-0.25)},rotate=80] (-0.1,0.2) -- (0,0.2) -- (0,0) -- (-0.1,0);
    \draw[shift={(0.65,-0.64)},rotate=-140] (-0.1,0) -- (0,0) -- (0,0.2) -- (-0.1,0.2);

    \filldraw (1,0) circle (0.001pt) node[anchor= south east] {\scriptsize{$A_3$}};
    \filldraw (-1,0) circle (0.001pt) node[anchor= south west] {\scriptsize{$A_2$}};

    \filldraw [black] (0,1) circle (1.25pt) node[anchor= south] {\scriptsize{$\mathrm{T}_1$}};
    \filldraw [black] (0.88,-0.49) circle (1.25pt) node[anchor= north west] {\scriptsize{$\mathrm{T}_2$}};
    \filldraw [black] (-0.88,-0.49) circle (1.25pt)node[anchor= north east] {\scriptsize{$\mathrm{T}_3$}};

    \draw (0,-0.95) -- (0,-1.05);
    \draw[rotate=120] (0,-0.95) -- (0,-1.05);
    \draw[rotate=-120] (0,-0.95) -- (0,-1.05);
\end{tikzpicture}
\addtocounter{subfigure}{-1}
\captionof{subfigure}{\footnotesize{$v\in (0,1/2)$.}}\label{FIG:crossSUB}
\end{subfigure}\endminipage\hfill
\minipage[b]{0.33\textwidth}\centering

\begin{subfigure}\centering
\begin{tikzpicture}[scale=1.3]
    \filldraw[white] (0,2) circle (0.5pt);
    \filldraw[lightgray!80,rotate=120] (0,2) circle (0.5pt);
    \filldraw[white,rotate=-120] (0,2) circle (0.5pt);

    \filldraw[white] (1,1) circle (0.5pt);
    \filldraw[white,rotate=120] (1,1) circle (0.5pt);
    \filldraw[white,rotate=-120] (1,1) circle (0.5pt);

    \filldraw[white] (-1,1) circle (0.5pt);
    \filldraw[white,rotate=120] (-1,1) circle (0.5pt);
    \filldraw[white,rotate=-120] (-1,1) circle (0.5pt);

    \draw [domain=0:6.28,variable=\t,smooth] plot ({sin(\t r)},{cos(\t r)});

    \filldraw [lightgray!80,domain=0:6.28,variable=\t,smooth] plot ({(1/2)*sin(\t r)},{1/2+(1/2)*cos(\t r)});
    \filldraw [rotate=120,lightgray!80,domain=0:6.28,variable=\t,smooth] plot ({(1/2)*sin(\t r)},{1/2+(1/2)*cos(\t r)});
    \filldraw [rotate=-120,lightgray!80,domain=0:6.28,variable=\t,smooth] plot ({(1/2)*sin(\t r)},{1/2+(1/2)*cos(\t r)});

    \draw [domain=0:6.28,variable=\t,smooth] plot ({(1/2)*sin(\t r)},{1/2+(1/2)*cos(\t r)});
    \draw [rotate=120,domain=0:6.28,variable=\t,smooth] plot ({(1/2)*sin(\t r)},{1/2+(1/2)*cos(\t r)});
    \draw [rotate=-120,domain=0:6.28,variable=\t,smooth] plot ({(1/2)*sin(\t r)},{1/2+(1/2)*cos(\t r)});




    \filldraw (0,1) circle (1.25pt) node[anchor= south] {\scriptsize{$\mathrm{T}_1$}};
    \filldraw (0.88,-0.49) circle (1.25pt) node[anchor= north west] {\scriptsize{ $\mathrm{T}_2$}};
    \filldraw (-0.88,-0.49) circle (1.25pt)node[anchor= north east] {\scriptsize{$\mathrm{T}_3$}};

    \draw (0,-0.95) -- (0,-1.05);
    \draw[rotate=120] (0,-0.95) -- (0,-1.05);
    \draw[rotate=-120] (0,-0.95) -- (0,-1.05);
\end{tikzpicture}
\addtocounter{subfigure}{-1}
\captionof{subfigure}{\footnotesize{$v=1/2$.}}\label{FIG:crossCRIT}
\end{subfigure}\endminipage\hfill
\minipage[b]{0.33\textwidth}\centering

\begin{subfigure}\centering
\begin{tikzpicture}[scale=1.3]
    \filldraw[white] (0,2) circle (0.5pt);
    \filldraw[white,rotate=120] (0,2) circle (0.5pt);
    \filldraw[white,rotate=-120] (0,2) circle (0.5pt);

    \filldraw[white] (1,1) circle (0.5pt);
    \filldraw[white,rotate=120] (1,1) circle (0.5pt);
    \filldraw[white,rotate=-120] (1,1) circle (0.5pt);

    \filldraw[white] (-1,1) circle (0.5pt);
    \filldraw[white,rotate=120] (-1,1) circle (0.5pt);
    \filldraw[white,rotate=-120] (-1,1) circle (0.5pt);


    \draw[dotted, thick, postaction={decoration={markings,mark=at position 0.41 with {\arrow[thick,color=gray]{latex reversed}}},decorate}] (-0.75,0.67) -- (0.75,0.67);
    \draw[rotate=120,dotted, thick, postaction={decoration={markings,mark=at position 0.41 with {\arrow[thick,color=gray]{latex reversed}}},decorate}] (-0.75,0.67) -- (0.75,0.67);
    \draw[rotate=-120,dotted, thick, postaction={decoration={markings,mark=at position 0.41 with {\arrow[thick,color=gray]{latex reversed}}},decorate}] (-0.75,0.67) -- (0.75,0.67);

    \filldraw [lightgray!80,domain=0:6.28,variable=\t,smooth] plot ({(1/3)*sin(\t r)},{(1/3)+(1/3)*cos(\t r)});
    \filldraw [rotate=120,lightgray!80,domain=0:6.28,variable=\t,smooth] plot ({(1/3)*sin(\t r)},{(1/3)+(1/3)*cos(\t r)});
    \filldraw [rotate=-120,lightgray!80,domain=0:6.28,variable=\t,smooth] plot ({(1/3)*sin(\t r)},{(1/3)+(1/3)*cos(\t r)});

    \draw [domain=0:6.28,variable=\t,smooth] plot ({(1/3)*sin(\t r)},{(1/3)+(1/3)*cos(\t r)});
    \draw [rotate=120,domain=0:6.28,variable=\t,smooth] plot ({(1/3)*sin(\t r)},{(1/3)+(1/3)*cos(\t r)});
    \draw [rotate=-120,domain=0:6.28,variable=\t,smooth] plot ({(1/3)*sin(\t r)},{(1/3)+(1/3)*cos(\t r)});


    \draw [domain=0:6.28,variable=\t,smooth] plot ({sin(\t r)},{cos(\t r)});

    \draw [dotted, white, ultra thick, domain=-0.26:0.26,variable=\t,smooth] plot ({sin(\t r)},{cos(\t r)});
    \draw [dotted, white, ultra thick, domain=1.83:2.35,variable=\t,smooth] plot ({sin(\t r)},{cos(\t r)});
    \draw [dotted, white, ultra thick, domain=-2.35:-1.83,variable=\t,smooth] plot ({sin(\t r)},{cos(\t r)});


    \draw[shift={(0.25,0.87)},rotate=-20] (0.1,0) -- (0,0) -- (0,0.2) -- (0.1,0.2);
    \draw[shift={(-0.25,0.87)},rotate=20] (-0.1,0) -- (0,0) -- (0,0.2) -- (-0.1,0.2);

    \draw[shift={(-1.07,-0.25)},rotate=-80] (-0.1,0.2) -- (0,0.2) -- (0,0) -- (-0.1,0);
    \draw[shift={(-0.65,-0.64)},rotate=140] (-0.1,0) -- (0,0) -- (0,0.2) -- (-0.1,0.2);

    \draw[shift={(1.07,-0.25)},rotate=80] (0.1,0.2) -- (0,0.2) -- (0,0) -- (0.1,0);
    \draw[shift={(0.65,-0.64)},rotate=-140] (0.1,0) -- (0,0) -- (0,0.2) -- (0.1,0.2);


    \filldraw [black] (0,1) circle (1.25pt) node[anchor= south] {\scriptsize{$\mathrm{T}_1$}};
    \filldraw [black] (0.88,-0.49) circle (1.25pt) node[anchor= north west] {\scriptsize{$\mathrm{T}_2$}};
    \filldraw [black] (-0.88,-0.49) circle (1.25pt)node[anchor= north east] {\scriptsize{$\mathrm{T}_3$}};

    \draw (0,-0.95) -- (0,-1.05);
    \draw[rotate=120] (0,-0.95) -- (0,-1.05);
    \draw[rotate=-120] (0,-0.95) -- (0,-1.05);

\end{tikzpicture}
\addtocounter{subfigure}{-1}
\captionof{subfigure}{\footnotesize{$v\in (1/2,1)$. }}\label{FIG:crossSUP}
\end{subfigure}\endminipage
\captionof{figure}{The interior of the (dark gray) disk closest to the location of the Taub point $\mathrm{T}_\alpha$ in $(\Sigma_1,\Sigma_2,\Sigma_3)$-space indicates growth of the cross term $N_\beta N_\gamma$, $(\alpha,\beta,\gamma)=(123)$ or a permutation thereof.
Outside their growth region the cross terms decay. As $v\in (0,1)$ increases, the (gray) disks
radii decrease and they move toward the middle, where $\Sigma^2=0$. For $v\in (0,1/2)$, the disks
have parts both outside and inside $\Sigma^2=1$, and in particular, for $v=1/4$,
their boundary circles intersect at the $\Sigma_\alpha$ location of the $\mathrm{Q}_\alpha$
points. At $v=1/2$ the disks only intersects with $\Sigma^2=1$ at the $\Sigma_\alpha$ values
of the Taub points $\mathrm{T}_\alpha$. For $v\in (1/2,1)$ the disks lie inside $\Sigma^2=1$.
}\label{FIG:circleCROSS}
\end{figure}

All the above conjectures are about the dynamical
system~\eqref{intro_dynsyslambdaR}, which describes the dynamics of the vacuum
$\lambda$-$R$ class~A models. In Appendix~\ref{appsubsec:HL} we show how to derive dynamical
systems for the more general HL models, and how the discrete statements
about the dynamical system~\eqref{intro_dynsyslambdaR} translate to these systems.
Moreover, similar heuristic arguments as those in the
$\lambda$-$R$ case suggest that dynamical conjectures, analogous to those above,
can be stated for broad classes of HL models. The results in Appendix~\ref{appsubsec:HL}
also indicate that if one is not able to obtain proofs for the $\lambda$-$R$ case,
then one is not likely to be able to prove analogous results for more general HL models.
In other words, the $\lambda$-$R$ models is a necessary step that needs
to be overcome before attempting to tackle more general HL models.

The above
dynamical conjectures implicitly suggest
that one uses a dynamical systems formulation of the type discussed in this work.
However, there are other possible approaches.
In~\cite{reitru10} the authors used metric variables as the starting point
for their analysis. Arguably the most efficient way to do this is to use
a Lagrangian or Hamiltonian approach, as done in Appendix~\ref{app:dom},
and use a time variable defined by setting ${\cal N}=-1$ in this appendix.
Alternatively, one can use the billiard (metric configuration space) formulation
of Chitr\'e and Misner~\cite{chi72,mis69a,mis69b}, see p. 812 in~\cite{grav73},
and also~\cite{dametal03,heietal09,damlec11}, and attempt to estimate the terms that
are heuristically neglected, which contain the asymptotic history of the solutions,
where the latter is essential for Conjecture~\ref{conj:subcritical}.

\begin{appendix}

\section{\texorpdfstring{Ho\v{r}ava-Lifshitz}{} models}\label{app:dom}

In this appendix we derive the evolution equations~\eqref{intro_dynsyslambdaR}
for the vacuum spatially homogeneous $\lambda$-$R$ class~A Bianchi models.
We also obtain a regular constrained dynamical system for the
Ho\v{r}ava-Lifshitz (HL) class~A Bianchi models. In addition, we heuristically
argue that the heteroclinic structure these models exhibit on the union of
the Bianchi type I and II sets describes the relevant asymptotic dynamical
structure toward the singularity for the $\lambda$-$R$ models and a wide
range of more general HL models. This is further supported by the existence
of a `dominant' Bianchi type I and II invariant set in the HL dynamical systems
formulation which can be identified with the Bianchi type I and II invariant
set for the $\lambda$-$R$ models. The main part of the paper is therefore also
relevant for a broad set of HL models.


Recall that the dynamics of HL gravity is based on the
action~\eqref{action}, where the kinetic part is given by~\eqref{kin} and the
potential by~\eqref{calV}. We consider vacuum spatially homogeneous HL class~A
Bianchi models, for which the Bianchi type VIII and IX models are
the most general ones. These models admit a symmetry-adapted
spatial (left-invariant) co-frame $\{{\bom}^1,{\bom}^2,{\bom}^3\}$,
described in equation~\eqref{structconst}, which we repeat for the reader's
convenience:
\begin{equation}\label{app:structconst}
d{\bom}^1  =  -{n}_1 \,
{\bom}^2\wedge {\bom}^3\:,\quad d{\bom}^2  =
-{n}_2 \, {\bom}^3\wedge {\bom}^1\:,\quad
d{\bom}^3  = -{n}_3 \, {\bom}^1\wedge
{\bom}^2\:,
\end{equation}
where the structure constants $n_1$, $n_2$, $n_3$ determine the Lie algebras of
the three-dimensional simply transitive symmetry groups, which describe
the class~A Bianchi models, see e.g.~\cite{waiell97}, and Table~\ref{intro:classAmodels}.

Expressing the components of the spatial metric in the symmetry adapted
spatial co-frame~\eqref{app:structconst} leads to that they become purely
time-dependent. Since the GR and HL class~A Bianchi models share the same
spatial symmetry adapted frame, they also have the same automorphism groups.
In the present context, automorphisms
are linear transformations of the spatial left-invariant frame that leave
the structure constants of the Bianchi symmetry groups unchanged.
Since the automorphisms are what is left of the symmetry generating spatial
diffeomorphisms, it should come as no surprise that there is
a close connection between them and the momentum/Codazzi vacuum constraints,
which are the same for all GR and HL models, see e.g.~\cite{henetal10}.
In particular, the momentum/Codazzi constraints can be set to zero
by means of the class~A off-diagonal automorphisms, which at the same
time can be used to diagonalize the spatially homogeneous spatial metric,
see~\cite{waiell97,jan01,Mac73} and references therein.\footnote{Diagonalization
and the role of the automorphism group also depends on the spatial topology,
an issue which we neglect. For an investigation about the role of spatial topology
in a Hamiltonian description of Bianchi models, see~\cite{ashsam91}.}
We will use the symmetry adapted co-frame with diagonal class~A metrics
throughout, and we also set the shift vector $N_i$ in~\eqref{genmetric}
to zero. The only remaining constraint is the Hamiltonian/Gauss constraint.

The diagonalized vacuum spatially homogeneous class~A metrics are given by
\begin{equation}\label{threemetric}
\mathbf{g} = -N^2(t)dt\otimes
dt + g_{11}(t)\:{\bom}^1\otimes {\bom}^1 +
g_{22}(t)\:{\bom}^2\otimes {\bom}^2 +
g_{33}(t)\:{\bom}^3\otimes {\bom}^3,
\end{equation}
where the lapse $N=N(t)$ is a non-zero function determining
the particular choice of time variable.
Due to the diagonal time-dependent spatial metric~\eqref{threemetric},
the extrinsic curvature is also diagonal, given by
$(K_{11}, K_{22}, K_{33}) = (\dot{g}_{11}, \dot{g}_{22}, \dot{g}_{33})/(2N)$,
where $\dot{}$ denotes a derivative with respect to $t$. 
Alternatively, raising one of the indices, it takes the form
\begin{equation}\label{extrinsicK}
(K^1\!_1, K^2\!_2, K^3\!_3) 
= \frac{1}{2N}\left(\frac{\dot{g}_{11}}{g_{11}}, \frac{\dot{g}_{22}}{g_{22}}, \frac{\dot{g}_{33}}{g_{33}}\right).
\end{equation}

For the HL class~A Bianchi models, the action~\eqref{action}
expressed in terms of the symmetry adapted co-frame~\eqref{structconst}
yields the field equations for the associated metric~\eqref{threemetric}.
In order to simplify this action as much as possible and thereby obtain
simple Hamiltonian equations, we focus on the kinetic
part ${\cal T}$ in equation~\eqref{kin}, which can be written as
\begin{equation}\label{kin2}
{\cal T} = (K^1\!_1)^2 + (K^2\!_2)^2 + (K^3\!_3)^2 - \lambda (K^1\!_1 + K^2\!_2 + K^3\!_3)^2.
\end{equation}

It follows that ${\cal T}$ is a quadratic form in the time derivatives
of the metric. To simplify ${\cal T}$, we make a variable transformation
from the metric components to the variables $\beta^0,\beta^+,\beta^-$,
first introduced by Misner~\cite{mis69a,mis69b,grav73},
\begin{subequations}\label{Misnerbeta}
\begin{align}
g_{11} &= e^{2( \beta^0 - 2\beta^+)},\\
g_{22} &= e^{2(\beta^0 + \beta^+ + \sqrt{3}\beta^-)},\\
g_{33} &= e^{2(\beta^0 + \beta^+ - \sqrt{3}\beta^-)}.
\end{align}
\end{subequations}
%
%
%
This results in that ${\cal T}$ in equation~\eqref{kin2} takes the form
\begin{equation}\label{calT}
{\cal T} = \dfrac{6}{N^2}
\left[-\left(\frac{3\lambda-1}{2}\right)(\dot{\beta}^0)^2 + (\dot{\beta}^+)^2 + (\dot{\beta}^-)^2\right].
\end{equation}
Note that the character of the quadratic form~\eqref{calT} changes
when $\lambda = 1/3$. Since we are interested
in continuously deforming the GR case $\lambda = 1$, we restrict considerations to
$\lambda > 1/3$. To simplify the kinetic part further, we introduce a new variable
$\beta^\lambda$ and a density-normalized lapse function $\mathcal{N}$, defined by
\begin{subequations}\label{betaLAMBDAandNcal}
\begin{align}
\beta^\lambda &:= \sqrt{\frac{3\lambda-1}{2}}\beta^0,\\
{\cal N} &:= \frac{N}{12\sqrt{g}},
\end{align}
\end{subequations}
where $g = g_{11}g_{22}g_{33} = \exp(6\beta^0)$ is the determinant of the spatial metric
in the symmetry adapted co-frame, which leads to,
\begin{equation}
\sqrt{g}N{\cal T} = \dfrac{1}{2 {\cal N}}\left[-(\dot{\beta}^\lambda)^2 + (\dot{\beta}^+)^2 + (\dot{\beta}^-)^2\right].
\end{equation}

It is convenient to define
\begin{equation}
T := \frac{\sqrt{g} N}{\mathcal{N}}\mathcal{T} = 12g{\cal T},
\end{equation}
so that ${\cal N}T$ is the kinetic part of the Lagrangian for
the present spatially homogeneous models, in
analogy with the GR case, see e.g., ch. 10 in~\cite{waiell97}.
The density-normalized lapse ${\cal N}$ is kept in the kinetic term
${\cal N}T$, since it is needed in order to obtain the Hamiltonian
constraint, which is accomplished by varying ${\cal N}$ in the Hamiltonian.

To proceed to a Hamiltonian description, we introduce the canonical momenta
\begin{equation}
p_\lambda := -\frac{\dot{\beta}^\lambda}{{\cal N}}, \qquad p_\pm :=\frac{\dot{{\beta}}^\pm}{{\cal N}}.
\end{equation}
This leads to that $T$ takes the form
\begin{equation}\label{Tkinetics}
T = \frac12\left(- p_\lambda^2 + p_+^2 + p_-^2\right).
\end{equation}

Similarly to the treatment of the kinetic part, we define
\begin{equation}
V := \sqrt{g}N {\cal V}/{\cal N}=12 g {\cal V}.
\end{equation}
Due to~\eqref{calV},
\begin{equation}\label{VHL}
V = {}^1V + {}^2V + {}^3V + {}^4V + {}^5V + {}^6V + \dots ,
\end{equation}
where
\begin{subequations}\label{pots}
\begin{alignat}{3}
{}^1V &:= 12k_1 gR, &\qquad {}^2V &:= 12k_2 gR^2,  &\qquad
{}^3V &:= 12k_3 g R^i\!_jR^j\!_i,\\
{}^4V &:= 12k_4 g R^i\!_jC^j\!_i,
&\qquad {}^5V &:= 12k_5 g C^i\!_jC^j\!_i, &\qquad {}^6V &:= 12k_6 gR^3.
\end{alignat}
\end{subequations}
The superscripts on ${}^AV$ (where $A = 1,\dots,6$) thereby coincide with the
subscripts of the constants $k_A$ in~\eqref{calV}.

Based on~\eqref{action}, this leads to a Hamiltonian $H$ given by
\begin{equation}\label{LambdaRham}
H := \sqrt{g}N({\cal T} + {\cal V}) = {\cal N}(T + V) = 0,
\end{equation}
where $T$ only depends on the canonical momenta $p_\lambda$, $p_\pm$,
given by~\eqref{Tkinetics}, and $V$ only depends on
$\beta^\lambda$, $\beta^\pm$, given by~\eqref{VHL} and~\eqref{pots}.

In order to derive the ordinary differential equations for these models
via the Hamiltonian equations in terms of the variables
$\beta^\lambda$, $\beta^\pm$ and the canonical momenta $p_\lambda$, $p_\pm$,
we need to compute each ${}^AV(\beta^\lambda,\beta^\pm)$.
We proceed with two cases: one which minimally modifies vacuum GR in the present context,
the vacuum $\lambda$-$R$ models in Section~\ref{subsec:lambdaR}; one which
more generally modifies GR, the HL models in Section~\ref{appsubsec:HL}.
Both cases have a Hamiltonian with the same kinetic part, given in~\eqref{Tkinetics},
but they have different potentials in~\eqref{VHL} and~\eqref{pots}.
The vacuum $\lambda$-$R$ models are obtained by setting
$k_1=-1$, $k_2=k_3=k_4=k_5=k_6=0$ in~\eqref{pots} and thus~\eqref{VHL} yields
$V = {}^1V = -12gR$, i.e., the same potential as in GR.
The more general vacuum HL models are determined by the potentials ${}^AV$
with $A = 1,\dots, 6$, and combinations thereof.

\subsection{$\lambda$-$R$ Class A models}\label{subsec:lambdaR}

The vacuum $\lambda$-$R$ models minimally modify
the vacuum GR models~\cite{giukie94,belres12,lolpir14}.
They are obtained from an action that consists of the generalized
kinetic part in~\eqref{kin}, i.e, by keeping $\lambda$ (GR is obtained by
setting $\lambda=1$), and the vacuum GR potential in~\eqref{calV},
i.e., a potential arising from $-R$ only, and hence when $k_1=-1$
and $k_2=k_3=k_4=k_5=k_6=0$ in~\eqref{calV}. These models suffice for
our goal of illustrating the role of first principles and their connection
with the structure of generic spacelike singularities.

\subsubsection*{Derivation of the $\lambda$-$R$ evolution equations}

To obtain succinct expressions for the spatial curvature, and thereby the
potential $V={}^1V = - 12gR$,
we introduce the following auxiliary quantities (see~\cite{heiugg09a} for a
discussion when one, or several, of the constants $n_1$, $n_2$, $n_3$ is zero),
\begin{subequations}\label{malpha}
\begin{align}
m_1 &:= n_1g_{11} = n_1 e^{2(2v\beta^\lambda - 2\beta^+)},\\
m_2 &:= n_2g_{22} = n_2 e^{2(2v\beta^\lambda + \beta^+ + \sqrt{3}\beta^-)},\\
m_3 &:= n_3g_{33} = n_3 e^{2(2v\beta^\lambda + \beta^+ - \sqrt{3}\beta^-)}.
\end{align}
\end{subequations}
Here we have introduced the parameter $v$, which is defined by
the relation
\begin{equation}
v := \frac{1}{\sqrt{2(3\lambda - 1)}},
\end{equation}
and hence $\beta^0 = 2v\beta^\lambda$ due to~\eqref{betaLAMBDAandNcal}.
The parameter $v$ plays a prominent role in this and the next Appendix, and
in the evolution equations~\eqref{intro_dynsyslambdaR}. Since we are
interested in continuous deformations of GR with $\lambda=1$, and thus
$v=1/2$, we restrict attention to $v\in (0,1)$, although $v=0$ and $v=1$,
which result in bifurcations, will sometimes also be considered.
Specializing the general expression for the spatial curvature in~\cite{elsugg97}
to the diagonal class A Bianchi models leads to\footnote{Note
that the expressions for the extrinsic and the spatial curvature with one upper
and one lower index coincide when using either the presently introduced spatial
co-frame or an associated orthonormal frame for the
diagonal class~A models, as in~\cite{elsugg97}.}
\begin{equation}\label{R11}
R^1\!_1 = \frac{1}{2g}(m_1^2 - (m_2-m_3)^2),
\end{equation}
where $R^1\!_1 = g^{11}R_{11}
= g_{11}^{-1}R_{11}$, and similarly by permutations
for $R^2\!_2$ and $R^3\!_3$. It follows that the spatial scalar curvature
$R = R^1\!_1 + R^2\!_2 +R^3\!_3$ is given by
\begin{equation}\label{Rscalar}
R = -\frac{1}{2g}(m_1^2 + m_2^2 + m_3^2 - 2m_1m_2 - 2m_2m_3 - 2m_3m_1).
\end{equation}
This thereby yields the potential in~\eqref{VHL} and~\eqref{pots} with $k_1=-1$:
\begin{equation}\label{kinV}
V = {}^1V = -12 g R = 6(m_1^2 + m_2^2 + m_3^2 - 2m_1m_2 - 2m_2m_3 - 2m_3m_1),
\end{equation}
where $V$ depends on $\beta^\lambda$ and $\beta^\pm$ via $m_1$, $m_2$ and $m_3$,
according to equation~\eqref{malpha}.

The evolution equations for $\beta^\lambda$, $\beta^\pm$, $p_\lambda$, $p_\pm$
are obtained from Hamilton's equations, where $T$ and $V$
in the Hamiltonian~\eqref{LambdaRham} are given by~\eqref{Tkinetics} and~\eqref{kinV},
respectively, which yields
\begin{subequations}\label{HamiltonEQ}
\begin{align}
\dot{\beta}^\lambda &= \frac{\partial H}{\partial p_\lambda}
= -{\cal N} p_\lambda, \qquad && \dot{p}_\lambda
= - \frac{\partial H}{\partial\beta^\lambda}
= - \mathcal{N} \frac{\partial V}{\partial\beta^\lambda}, \label{betadotham}\\
\dot{\beta}^\pm &= \frac{\partial H}{\partial{p}_\pm}
= {\cal N}{p}_\pm, \qquad &&\dot{p}_\pm
= - \frac{\partial H}{\partial{\beta}^\pm}
= - \mathcal{N} \frac{\partial V}{\partial{\beta}^\pm},
\end{align}
\end{subequations}
while the Hamiltonian constraint $T+V=0$ is obtained by varying ${\cal N}$.

Next, we choose a new time variable
$\tau_-:=-\beta^\lambda$, 
which is directed toward the physical past, since we are considering
expanding models. This is accomplished by setting ${\cal N} =  p_\lambda^{-1}$
in the first equation in~\eqref{betadotham}, and thereby $N = 12\sqrt{g}/ p_\lambda$,
which results in the following evolution equations:
\begin{subequations}\label{HamiltonEQ2}
\begin{align}
\frac{d\beta^\lambda}{d\tau_-} &= -1, \qquad
&&\frac{d p_\lambda}{d\tau_-} =
-\frac{1}{ p_\lambda} \frac{\partial V}{\partial\beta^\lambda}, \\
\frac{d\beta^\pm}{d\tau_-} &= \frac{{p}_\pm}{ p_\lambda}, \qquad &&
\frac{dp_\pm}{d\tau_-} = -\frac{1}{ p_\lambda} \frac{\partial V}{\partial{\beta}^\pm}.
\end{align}
\end{subequations}

We then rewrite the system~\eqref{HamiltonEQ2} and the constraint $T+V=0$
using the non-canonical variable transformation,
\begin{equation}\label{SigmaNvariables}
\Sigma_\pm := - \frac{p_\pm}{ p_\lambda}, \qquad \qquad \qquad N_\alpha
:= - 2\sqrt{3}\left(\frac{m_\alpha}{ p_\lambda}\right),
\end{equation}
while keeping $p_\lambda$. Note that $\Sigma_\pm = d\beta^\pm/d\beta^\lambda = -d\beta^\pm/d\tau_-$.

These variables lead to a decoupling\footnote{More precisely, the variables result in a
skew-product dynamical system where the base dynamics acts in $(\Sigma_\pm,N_1,N_2,N_3)$
while the fiber dynamics acts in $p_\lambda$. This notion was introduced in connection
with ergodic theory in~\cite{Anzai51}.} of the evolution equation for the variable $ p_\lambda$,
\begin{equation}\label{p0prime}
p_\lambda^\prime = -4v(1-\Sigma^2) p_\lambda,
\end{equation}
where ${}^\prime$ denotes the derivative $d/d\tau_-$.
This yields the following reduced system of evolution equations
\begin{subequations}\label{dynsyslambdaR}
\begin{align}
\Sigma_\pm^\prime &= 4v(1-\Sigma^2)\Sigma_\pm + {\cal S}_\pm,\\
N_1^\prime &= -2(2v\Sigma^2 - 2\Sigma_+)N_1,\\
N_2^\prime &= -2(2v\Sigma^2 + \Sigma_+ + \sqrt{3}\Sigma_-)N_2,\\
N_3^\prime &= -2(2v\Sigma^2 + \Sigma_+ - \sqrt{3}\Sigma_-)N_3,
\end{align}
while the Hamiltonian constraint $T+V=0$ results in
\begin{equation}\label{constrpmVIIIIX}
1 - \Sigma^2 - \Omega_k =0,
\end{equation}
\end{subequations}
where
\begin{subequations}\label{LambdaRquantities}
\begin{align}
\Sigma^2 &:= \Sigma_+^2 + \Sigma_-^2,\\
\Omega_k &:= N_1^2 + N_2^2 + N_3^2 - 2N_1N_2 - 2N_2N_3 - 2N_3N_1,\\
{\cal S}_+ &:= 2[(N_2 - N_3)^2 - N_1(2N_1 - N_2 - N_3)],\\
{\cal S}_- &:= 2\sqrt{3}(N_2 - N_3)(N_2 + N_3 - N_1).
\end{align}
\end{subequations}

Note that the variables $\Sigma_\pm$, $N_1$, $N_2$ and $N_3$,
defined in~\eqref{SigmaNvariables}, are \emph{dimensionless}. Dimensions
can be introduced in various ways, but terms in a sum
must all have the same dimension. The constraint~\eqref{constrpmVIIIIX}
is such a sum. Since this sum contains 1,
which 
obviously is dimensionless, it follows that $\Sigma_+$, $\Sigma_-$,
$N_1$, $N_2$ and $N_3$ are dimensionless, and so is the time variable
$\tau_-$, as follows from inspection of~\eqref{dynsyslambdaR}.

The introduction of the Misner parametrization and the associated
$\Sigma_\pm$ variables breaks an axis permutation symmetry, which
can be restored by introducing the variables
\begin{equation}\label{app:MisnerPar}
\Sigma_1 := - 2\Sigma_+, \qquad \Sigma_2:= \Sigma_+ + \sqrt{3}\Sigma_-, \qquad
\Sigma_3 := \Sigma_+ - \sqrt{3}\Sigma_-.
\end{equation}
By multiplying the equation for $\Sigma_+$ with $-2$ and setting ${\cal S}_1 = -2{\cal S}_+$,
we obtain the equation for $\Sigma_1$. Replacing $(123)$ with $(\alpha\beta\gamma)$ then
allows us to write the above system of evolution equations~\eqref{dynsyslambdaR} as
the system~\eqref{intro_dynsyslambdaR}, which is invariant under the axis permutations
in~\eqref{permSYM}. Note that we only need the equation for $\Sigma_+$ and
not the one for $\Sigma_-$ to obtain the system~\eqref{intro_dynsyslambdaR},
a strategy we will use for the HL models.
The vacuum equations for GR are obtained by setting $v=1/2$.\footnote{For a similar derivation
of the GR case, see ch. 10 in~\cite{waiell97}, but note that the present variables $N_\alpha$
differ from those in~\cite{waiell97} by a factor $2\sqrt{3}$.}

\subsubsection*{Heuristic $\lambda$-$R$ considerations}

To obtain some motivation for some of the Conjectures in Section~\ref{sec:conjectures},
we use Misner's heuristic approximation scheme, which he
introduced in order to understand the initial Bianchi type IX singularity in GR,
see~\cite{mis69a,mis69b,jan01} and ch. 10 in~\cite{waiell97}, and apply it
to the class~A Bianchi $\lambda$-$R$ models.
In this scheme, a class~A Bianchi solution toward the past initial singularity
(i.e. when $\tau_- = - \beta^\lambda \rightarrow \infty$) is described as
a `particle' moving in a potential well in $(\beta^+,\beta^-)\in\mathbb{R}^2$ space.

Let us begin with the $\lambda$-$R$ Bianchi type I models for which
$n_1=n_2=n_3=0$, according to Table~\ref{intro:classAmodels}, and thus
$m_1=m_2=m_3=0$ in equation~\eqref{malpha}.
Hence the spatial curvature~\eqref{Rscalar} is identically zero,
and so is the potential~\eqref{kinV}, which implies that the kinetic part~\eqref{Tkinetics}
in the Hamiltonian~\eqref{LambdaRham} determines the dynamics.
It therefore follows that any type I solution can be
described as a `\emph{cosmological particle}' that is moving with the \emph{constant}
velocity
\begin{equation}\label{Velocity}
\vec{V} = (V_+,V_-) = \left(\frac{d\beta^+}{d\tau_-},\frac{d\beta^-}{d\tau_-}\right) =  (-\Sigma_+,-\Sigma_-),
\end{equation}
in $\beta^\pm$-space, due to~\eqref{HamiltonEQ2} and~\eqref{SigmaNvariables}.
Note that $\vec{V}$, $V_\pm$ should not be confused with the potential $V$.
Since the Hamiltonian/Gauss constraint in Bianchi type I reduces to $T=0$,
it follows that $\Sigma_+^2 + \Sigma_-^2 = 1$, and hence that the speed $|\vec{V}|$ of
the `cosmological particle' is $|\vec{V}|= 1$. Thus the fixed points in the Kasner circle
set $\mathrm{K}^\ocircle$ are interpreted in this picture as a particle with a constant
velocity $\vec{V}=(-\Sigma_+,-\Sigma_-)$ and speed $|\vec{V}|= 1$ in
$\beta^\pm$-space.

The $\lambda$-$R$ Bianchi type $\mathrm{II}_1$ models are characterized by
$n_1\neq 0,n_2=n_3=0$, see Table~\ref{intro:classAmodels}, and thus
$m_1\neq 0, m_2=m_3=0$, as follows from~\eqref{malpha}. Similar statements
hold for $\mathrm{II}_2$ and $\mathrm{II}_3$.
The evolution of the $\mathrm{II}_1$ models is determined by
\begin{equation}
T + V = \frac12\left(- p_\lambda^2 + p_+^2 + p_-^2\right) + 6m_1^2 = 0,
\end{equation}
where we recall due to~\eqref{malpha} that
\begin{equation}\label{m1LambdaRheur}
m_1 = n_1e^{2(2v\beta^\lambda - 2\beta^+)} = n_1e^{-4(v\tau_- + \beta^+)},
\end{equation}
where the time variable is given by $\tau_- := - \beta^\lambda$.

The steep exponential potential~\eqref{m1LambdaRheur} is
approximated by setting it to be identically zero
when the exponential in $6m_1^2$ is sufficiently small, and replacing it with an
\emph{infinite potential wall} when the smallness condition is violated.
For a chosen sufficiently small
constant $C\ll 1$, the potential attains this small value 
$C=6n_1^2e^{-8(v\tau_- + \beta^+_0)}$ for some $\beta^+_0\in\mathbb{R}$,
which determines its location in $\beta^+\in\mathbb{R}$ as a function of
the constants $n_1$, $v$, $C$ and the time $\tau_-$, given by
$\beta^+_0 := \log(6n_1^2/C)^{1/8} - v\tau_-$.
The steep potential~\eqref{m1LambdaRheur} is thereby approximated
by a potential that is set to zero
when $\beta^+ > \beta^+_0$, since then $6n_1^2e^{-8(v\tau_- + \beta^+)} < C$,
and an infinite potential wall at $\beta^+ = \beta^+_0$.
As $\tau_-$ increases toward the singularity, the location of the wall
at $\beta^+_0$ moves in the negative $\beta^+$-direction according to
$\beta^+_0 := \log(6n_1^2/C)^{1/8} - v\tau_-$, with a velocity
\begin{equation}\label{1wallv}
\vec{v}_1 = (v_+, v_-) = \left(\frac{d\beta^+_0}{d\tau_-},\frac{d\beta^-_0}{d\tau_-}\right) =  (-v,0),
\end{equation}
and thus the wall has a speed $|\vec{v}_1| = v$ in the negative $\beta^+$-direction.
If $V_+<v_+=-v<0$ (and hence $|V_+|> v > 0$), then the cosmological particle eventually reaches and
bounces against the infinite potential wall given by the Bianchi
type $\mathrm{II}_1$ potential $6m_1^2$, see Figure~\ref{FIG:typeIandIIheuristics}.
This occurs if $-\Sigma_+<-v$, 
which corresponds to $\Sigma_1 < -2v$ in the coordinates~\eqref{app:MisnerPar}.
This coincides with the instability criterion on
$\mathrm{K}^\ocircle$ in the $\Sigma_1$-direction, which defines the unstable
arc $\mathrm{int}(A_1)$ in equation~\eqref{A_1}.

Similarly, one can construct the infinite potential walls for the Bianchi type $\mathrm{II}_2$ and
$\mathrm{II}_3$ models, and obtain analogous results by adapting the $\beta^\pm$
variables to those directions. Such walls, with respective potentials given by
$6m_2^2$ and $6m_3^2$, have the following velocities in the present $\beta^\pm$
coordinates:
\begin{equation}\label{23wallv}
\vec{v}_2 = \frac12(1,\sqrt{3})v,\qquad
\vec{v}_3 = \frac12(1,-\sqrt{3})v .
\end{equation}

The general picture is therefore that a cosmological particle moves in a
Bianchi type I potential $V=0$, with velocity
$\vec{V}=(-\Sigma_+,-\Sigma_-)$, until it encounters a Bianchi type II moving wall
and bounces, thereby obtaining new values of $\Sigma_\pm$, determined by
the Kasner circle map, see Figure~\ref{FIG:typeIandIIheuristics}.
As follows from~\eqref{bouncephi}, a bounce against the $\mathrm{II}_1$ wall
corresponds to
\begin{equation}\label{BounceLaw}
\sin\varphi^\mathrm{f} = \frac{(1 - v^2)\sin\varphi^\mathrm{i}}{1 + v^2 - 2v\cos\varphi^\mathrm{i}},
\end{equation}
where $\varphi^\mathrm{i}$ is the angle the straight line motion of the particle makes with
the $\beta^+$-axis, while $\varphi^\mathrm{f}$ describes the outgoing motion after the bounce,
which is given by the subsequent Kasner solution.
The bounce law~\eqref{BounceLaw} can also be obtained in the present
description by making boost in the $\beta^+$-direction in
$(\beta^\lambda,\beta^\pm)$-space so that the potential wall does not move
and using that the ingoing and outgoing bounce angles then are equal. We will perform
such a boost in the next appendix. Finally, note that~\eqref{BounceLaw} reduces to
the GR case when $v=1/2$, e.g. given in ch. 10 in~\cite{waiell97}.

\begin{figure}[H]
\minipage[b]{0.5\textwidth}\centering
\begin{subfigure}\centering
\begin{tikzpicture}[scale=0.25]

    \filldraw[rotate=120,white] (-9,14) circle (6pt);

    \draw[color=gray,->] (0,-14) -- (0,14)  node[anchor= south] {\scriptsize{$\beta^-$}};
    \draw[color=gray,->] (-14,0) -- (14,0)  node[anchor= west] {\scriptsize{$\beta^+$}};

    \draw (-3,-14) -- (-3,14);
    \draw (-6,-14) -- (-6,14);
    \draw (-9,-14) -- (-9,14);

    \draw[color=gray,->] (-10,1) -- (-13,1)  node[anchor= south west] {\scriptsize{$\vec{v}_1$}};

    \filldraw [black] (-9,14) circle (0.1pt)   node[anchor= south] {\scriptsize{$\beta^+_0(\tau_-)$}};

    \draw[color=gray,->] (6,7.5) -- (3,6)  node[anchor= south] {\scriptsize{$\vec{V}$}};
    \draw[color=gray,dashed] (3,6) -- (-3,3);
    \draw[color=gray,dashed] (-3,3) -- (6,-12);

    \filldraw [black] (6,7.5) circle (6pt);

    \draw[color=gray,dashed] (1.5,3) -- (-7.5,3);

    \draw [shift={(-2.5,3)},domain=0:0.7,variable=\t,smooth] plot ({{cos(\t r)},sin(\t r)});\node at (-0.65,3.8) {\scriptsize{$\varphi^\mathrm{i}$}};
    \draw [shift={(-2.5,3)},domain=0:-1.5,variable=\t,smooth] plot ({{cos(\t r)},sin(\t r)});\node at (-0.65,2) {\scriptsize{$\varphi^\mathrm{f}$}};

\end{tikzpicture}
\addtocounter{subfigure}{-1}
\captionof{subfigure}{\footnotesize{Incoming `particle' with velocity $\vec{V}$ and moving potential $\mathrm{II}_1$ walls.}}\label{FIG:TYPEIANDII}
\end{subfigure}\endminipage\hfill
\minipage[b]{0.5\textwidth}\centering

\begin{subfigure}\centering
\begin{tikzpicture}[scale=0.25]

    \draw[color=gray,->] (0,-14) -- (0,14)  node[anchor= south] {\scriptsize{$\beta^-$}};
    \draw[color=gray,->] (-14,0) -- (14,0)  node[anchor= west] {\scriptsize{$\beta^+$}};

    \draw (-3,-14) -- (-3,14);
    \draw (-6,-14) -- (-6,14);
    \draw (-9,-14) -- (-9,14);

    \draw[color=gray,->] (-10,1) -- (-13,1)  node[anchor= south west] {\scriptsize{$\vec{v}_1$}};

    \draw[rotate=120] (-3,-14) -- (-3,14);
    \draw[rotate=120] (-6,-14) -- (-6,14);
    \draw[rotate=120] (-9,-14) -- (-9,14);

    \draw[rotate=120,color=gray,->] (-10,0) -- (-13,0)  node[anchor= south west] {\scriptsize{$\vec{v}_2$}};

    \draw[rotate=-120] (-3,-14) -- (-3,14);
    \draw[rotate=-120] (-6,-14) -- (-6,14);
    \draw[rotate=-120] (-9,-14) -- (-9,14);

    \draw[rotate=-120,color=gray,->] (-10,0) -- (-13,0)  node[anchor= north west] {\scriptsize{$\vec{v}_3$}};

\end{tikzpicture}
\addtocounter{subfigure}{-1}
\captionof{subfigure}{\footnotesize{Moving potential $\mathrm{II}_1$, $\mathrm{II}_2$ and $\mathrm{II}_3$ walls with respective veloticies $\vec{v}_1$, $\vec{v}_2$ and $\vec{v}_3$.}}\label{FIG:ALLTYPEII}
\end{subfigure}\endminipage
\captionof{figure}{The cosmological particle with velocity $\vec{V}$, determined by the Kasner solutions,
and the level sets of the type $\mathrm{II}_1$ potential described by a moving wall at $\beta_0^+(\tau_-)$
with velocity $\vec{v}_1$, and similarly for $\mathrm{II}_2$ and $\mathrm{II}_3$.
The particle reaches the moving wall $\mathrm{II}_1$ and bounces
with the law given by~\eqref{BounceLaw} when $|V_+|> v$. }\label{FIG:typeIandIIheuristics}
\end{figure}

Next, consider the Bianchi type VIII and IX models.
According to~\eqref{kinV}, the potential consists of the three combined
Bianchi type II potentials given by $6m_1^2,6m_2^2,6m_3^2$, which together
form a \emph{triangular potential well} in $\beta^\pm$-space, plus the
three `cross terms' $-12m_1m_2$, $-12m_2m_3$, $-12m_3m_1$, which form
\emph{cross term walls}. The cross terms are all negative in type IX,
while two are positive and one is negative in type VIII.
They are given by exponentials, which can be approximated
by a region where each individual term can be set to zero and a negative or positive
(depending on its sign) infinite potential wall moving in $\beta^\pm$-space
depending on time $\tau_-$. For example, the reasoning that resulted in
equation~\eqref{1wallv} for the type $\mathrm{II}_1$ potential can be replicated for the cross
term $-12m_2m_3 = -12n_2n_3\exp[4(-2v\tau_- + \beta^+)]$. This leads to a wall with
a velocity $\vec{v}_{2,3} = (2v,0)$, and thus a speed $2v$ in the positive $\beta^+$-direction.
By means of permutation symmetry, similar statements hold for the other cross terms.
According to~\eqref{23wallv}, the Bianchi type $\mathrm{II}_2$ and $\mathrm{II}_3$
potentials yield walls moving with a component of the velocity given
by $v/2$ in the positive $\beta^+$-direction. For sufficiently large $\tau_-$,
the cross term walls will therefore be `hidden' behind the type $\mathrm{II}$
walls, and hence should not affect the asymptotic dynamics, since particles
will bounce off the type $\mathrm{II}$ walls before reaching the cross term walls.
It is therefore expected that only Bianchi type $\mathrm{II}$ potentials play a
role for the asymptotic limit $\tau_- \rightarrow \infty$.

The overall picture is thereby that the asymptotic dynamics is described
by free motion of a cosmological particle in $\beta^\pm$-space with speed $|\vec{V}|=1$
in a Bianchi type I zero potential, followed by bounces against Bianchi type
$\mathrm{II}$ walls moving with speed $v$, if the particle catches up
with the moving walls, which depends on the direction of the cosmological
particle and the speed $v$ of the wall. The Bianchi type
$\mathrm{II}_1$, $\mathrm{II}_2$ and $\mathrm{II}_3$ potentials form a triangular potential
well that is increasing in size in $\beta^\pm$-space as $\tau_-\rightarrow\infty$.
The above heuristic picture amounts to the claim that generic solutions of the
evolution equations~\eqref{dynsyslambdaR} asymptotically follow heteroclinic Bianchi
type II chains, or end in the set $S$ in $\mathrm{K}^\ocircle$ in the supercritical
case $v>1/2$.

There is, however, a subtlety at the corners of the triangular potential well.
Take the corner at $\beta^-=0$. In this case, the type II walls
described by $m_2^2 = m_3^2 \propto 
\exp(4(-2v\tau_- + \beta^+))$ have the same velocity $2v$ in the $\beta^+$-direction,
as has the cross term wall given by $-m_2m_3$. 
Consider therefore a possibly `dangerous' region where the cross term $-m_2m_3$
near the corner of the triangular well might affect the asymptotic dynamics.
To understand this better, define
\begin{equation}\label{heur:Vbar}
\bar{V}(\beta^\pm):= e^{-8v\beta^\lambda}V = e^{8v\tau_-}V.
\end{equation}
Note that the asymptotic dynamics for large times is described by $\tau_-\to\infty$, and consequently for large values of $\bar{V}(\beta^\pm)$.

Without loss of generality we set $n_1=n_2=n_3=1$ in type IX;
$n_1=n_2=1$, $n_3=-1$ in type VIII; $n_1=0$, $n_2=n_3=1$ in type $\mathrm{VII}_0$;
$n_1=0$, $n_2=1$, $n_3=-1$ in type $\mathrm{VI}_0$; $n_1=1$, $n_2=n_3=0$ in
type II. Due to~\eqref{kinV} and~\eqref{malpha}, this leads to the following potentials
\begin{subequations}\label{barvlambdaRpot}
\begin{alignat}{2}
\bar{V}_{\mathrm{II}} &=\, 6e^{-8\beta^+}, &\quad &\text{for type } \mathrm{II},\\
\bar{V}_{\mathrm{VI}_0} &=\, 24e^{4\beta^+}\cosh^2(2\sqrt{3}\beta^-), &\quad &\text{for type } \mathrm{VI}_0,\\
\bar{V}_{\mathrm{VII}_0} &=\, 24e^{4\beta^+}\sinh^2(2\sqrt{3}\beta^-), &\quad &\text{for type } \mathrm{VII}_0,\\
\bar{V}_{\mathrm{VIII}} &=\, 6\left[e^{-8\beta^+} - 4e^{-2\beta^+}\sinh(2\sqrt{3}\beta^-) +
4e^{4\beta^+}\cosh^2(2\sqrt{3}\beta^-)\right], &\quad &\text{for type } \mathrm{VIII},\\
\bar{V}_{\mathrm{IX}} &=\, 6\left[e^{-8\beta^+} - 4e^{-2\beta^+}\cosh(2\sqrt{3}\beta^-) +
4e^{4\beta^+}\sinh^2(2\sqrt{3}\beta^-)\right], &\quad &\text{for type } \mathrm{IX},
\end{alignat}
\end{subequations}
with level sets of $\bar{V}$ depicted in Figures~\ref{FIG:typeIandIIheuristics},
\ref{FIG:potlevelsetsVIVII} and~\ref{FIG:potlevelsets}.

Note that a translation in $\beta^+$ simply rescales the potentials
$\bar{V}_{\mathrm{VI}_0}$ and $\bar{V}_{\mathrm{VII}_0}$
in~\eqref{barvlambdaRpot}, and hence all the potential level curves of
$\bar{V}_{\mathrm{VI}_0}$ and $\bar{V}_{\mathrm{VII}_0}$ have the same shape,
see Figure~\ref{FIG:potlevelsetsVIVII}. Furthermore, $\bar{V}_{\mathrm{VI}_0}$
and $\bar{V}_{\mathrm{VII}_0}$ approach the type $\mathrm{II}_2$ and $\mathrm{II}_3$
potentials exponentially fast as $|\beta^-|$ increases.
Thus $\bar{V}_{\mathrm{VI}_0}$ and $\bar{V}_{\mathrm{VII}_0}$ can be approximated
by the type $\mathrm{II}_2$ and $\mathrm{II}_3$ potentials when
$|\beta^-|>\epsilon$, for some $\epsilon>0$ sufficiently small.
The region $|\beta^-|<\epsilon$, denoted by the \emph{$\epsilon$-corner region},
is where the approximating infinite type II wall description breaks down.
In particular, it is at this region where the cross-term $N_2N_3$ has a
significant role in the dynamical system picture.
\begin{figure}[H]
\minipage[b]{0.5\textwidth}\centering
\begin{subfigure}\centering
\begin{tikzpicture}[scale=0.25]

    \draw[color=gray,->] (0,-14) -- (0,14)  node[anchor= south] {\scriptsize{$\beta^-$}};
    \draw[color=gray,->] (-10,0) -- (18,0)  node[anchor= south] {\scriptsize{$\beta^+$}};

    \draw[color=gray,dashed] (15,0.367) -- (0,0.367) node[anchor= north east] {\scriptsize{$-\epsilon$}};
    \draw[color=gray,dashed] (15,-0.367) -- (0,-0.367) node[anchor= south east] {\scriptsize{$\epsilon$}};


    \draw[domain=1:5,smooth,variable=\x] plot ({5+\x},{exp(-\x)});
    \draw[domain=1:5,smooth,variable=\x] plot ({5+\x},{-exp(-\x)});

    \draw[domain=1:5,smooth,variable=\x] plot ({8+\x},{exp(-\x)});
    \draw[domain=1:5,smooth,variable=\x] plot ({8+\x},{-exp(-\x)});

    \draw[domain=1:5,smooth,variable=\x] plot ({11+\x},{exp(-\x)});
    \draw[domain=1:5,smooth,variable=\x] plot ({11+\x},{-exp(-\x)});

    \draw (6,0.367) -- (-2.682,5.379); 
    \draw[rotate=-120] (6,0.367) -- (-2.682,5.379);

    \draw (9,0.367) -- (-4.182,7.977);
    \draw[rotate=-120] (9,0.367) -- (-4.182,7.977);

    \draw (12,0.367) -- (-5.682,10.575);
    \draw[rotate=-120] (12,0.367) -- (-5.682,10.575);

\end{tikzpicture}
\addtocounter{subfigure}{-1}
\captionof{subfigure}{\footnotesize{Type $\mathrm{VII}_0$.}}\label{FIG:TYPEVII}
\end{subfigure}\endminipage\hfill
\minipage[b]{0.5\textwidth}\centering

\begin{subfigure}\centering
\begin{tikzpicture}[scale=0.25]

    \draw[color=gray,->] (0,-14) -- (0,14)  node[anchor= south] {\scriptsize{$\beta^-$}};
    \draw[color=gray,->] (-10,0) -- (18,0)  node[anchor= south] {\scriptsize{$\beta^+$}};

    \draw[color=gray,dashed] (15,0.367) -- (0,0.367) node[anchor= north east] {\scriptsize{$-\epsilon$}};
    \draw[color=gray,dashed] (15,-0.367) -- (0,-0.367) node[anchor= south east] {\scriptsize{$\epsilon$}};


    \draw[domain=-0.375:0.375,smooth,variable=\x] plot ({6.15-1.15*(\x)^2},{\x});
    \draw[domain=-0.375:0.375,smooth,variable=\x] plot ({9.15-1.15*(\x)^2},{\x});
    \draw[domain=-0.375:0.375,smooth,variable=\x] plot ({12.15-1.15*(\x)^2},{\x});

    \draw (6,0.367) -- (-2.682,5.379); 
    \draw[rotate=-120] (6,0.367) -- (-2.682,5.379);

    \draw (9,0.367) -- (-4.182,7.977);
    \draw[rotate=-120] (9,0.367) -- (-4.182,7.977);

    \draw (12,0.367) -- (-5.682,10.575);
    \draw[rotate=-120] (12,0.367) -- (-5.682,10.575);

\end{tikzpicture}
\addtocounter{subfigure}{-1}
\captionof{subfigure}{\footnotesize{Type $\mathrm{VI}_0$.}}\label{FIG:TYPEVI}
\end{subfigure}\endminipage
\captionof{figure}{Level curves of $\bar{V}$ in~\eqref{barvlambdaRpot} for type $\mathrm{VII}_0$ and $\mathrm{VI}_0$. Each level curve has the same shape under
translation in $\beta^+$. Note that the $\epsilon$-corner region, where the
type II approximation fails, has a fixed size $2\epsilon$ independently of the level curves of
$\bar{V}$.}\label{FIG:potlevelsetsVIVII}
\end{figure}

As derived from first principles in Appendix~\ref{app:heterosym},
in the dynamical systems picture the  asymptotics as $\tau_-\rightarrow\infty$
for solutions of type VIII and IX for all $v\in (0,1)$ reside on the union of
the disjoint relevant type $\mathrm{VI}_0$ and $\mathrm{VII}_0$ boundary sets,
and the union of the type II and I boundary sets. This corresponds to that the
level curves of the potentials $\bar{V}_{\mathrm{VIII}}$ and $\bar{V}_{\mathrm{IX}}$
increasingly possess the same shape for large values of the potential, and that
they are locally described by the relevant type $\mathrm{VI}_0$ and $\mathrm{VII}_0$
potentials. The level curves of the potentials $\bar{V}_{\mathrm{VIII}}$ and
$\bar{V}_{\mathrm{IX}}$ are thereby asymptotically shape
invariant, see Figure~\ref{FIG:potlevelsets}.

In a similar way as for the individual type II potential,
we approximate the type VIII and IX potentials
$V = e^{-8v\tau_-}\bar{V}(\beta^\pm)$ by setting them to be identically
zero when $V<C$ for some chosen small constant $C \ll 1$
(i.e., $\tau_-$ and $\bar{V}$ are large), and an \emph{infinite potential wall}
when $V=C$, consisting of the associated three infinite type II potential walls.
This results in an equilateral triangular potential well, which yields an
increasingly good approximation as $\tau_-$ increases, except at the $\epsilon$-corner
regions, which are asymptotically described by type $\mathrm{VI}_0$, $\mathrm{VII}_0$,
see Figures~\ref{FIG:typeIandIIheuristics} and~\ref{FIG:potlevelsets}.
\begin{figure}[H]
\minipage[b]{0.5\textwidth}\centering
\begin{subfigure}\centering
\begin{tikzpicture}[scale=0.25]

    \draw[color=gray,->] (0,-14) -- (0,14)  node[anchor= south] {\scriptsize{$\beta^-$}};
    \draw[color=gray,->] (-10,0) -- (18,0)  node[anchor= south] {\scriptsize{$\beta^+$}};

    \draw[color=gray,dashed] (15,0.367) -- (0,0.367);
    \draw[color=gray,dashed] (15,-0.367) -- (0,-0.367);

    \draw[rotate=120,color=gray,dashed] (15,0.367) -- (0,0.367);
    \draw[rotate=120,color=gray,dashed] (15,-0.367) -- (0,-0.367);
    \draw[rotate=-120,color=gray,dashed] (15,0.367) -- (0,0.367);
    \draw[rotate=-120,color=gray,dashed] (15,-0.367) -- (0,-0.367);


    \draw[  domain=1:5,smooth,variable=\x] plot ({5+\x},{exp(-\x)});
    \draw[  domain=1:5,smooth,variable=\x] plot ({5+\x},{-exp(-\x)});
    \draw[ rotate=120, domain=1:5,smooth,variable=\x] plot ({5+\x},{exp(-\x)});
    \draw[ rotate=120, domain=1:5,smooth,variable=\x] plot ({5+\x},{-exp(-\x)});
    \draw[ rotate=-120, domain=1:5,smooth,variable=\x] plot ({5+\x},{exp(-\x)});
    \draw[ rotate=-120, domain=1:5,smooth,variable=\x] plot ({5+\x},{-exp(-\x)});

    \draw[  domain=1:5,smooth,variable=\x] plot ({8+\x},{exp(-\x)});
    \draw[  domain=1:5,smooth,variable=\x] plot ({8+\x},{-exp(-\x)});
    \draw[ rotate=120, domain=1:5,smooth,variable=\x] plot ({8+\x},{exp(-\x)});
    \draw[ rotate=120, domain=1:5,smooth,variable=\x] plot ({8+\x},{-exp(-\x)});
    \draw[ rotate=-120, domain=1:5,smooth,variable=\x] plot ({8+\x},{exp(-\x)});
    \draw[ rotate=-120, domain=1:5,smooth,variable=\x] plot ({8+\x},{-exp(-\x)});

    \draw[  domain=1:5,smooth,variable=\x] plot ({11+\x},{exp(-\x)});
    \draw[  domain=1:5,smooth,variable=\x] plot ({11+\x},{-exp(-\x)});
    \draw[ rotate=120, domain=1:5,smooth,variable=\x] plot ({11+\x},{exp(-\x)});
    \draw[ rotate=120, domain=1:5,smooth,variable=\x] plot ({11+\x},{-exp(-\x)});
    \draw[ rotate=-120, domain=1:5,smooth,variable=\x] plot ({11+\x},{exp(-\x)});
    \draw[ rotate=-120, domain=1:5,smooth,variable=\x] plot ({11+\x},{-exp(-\x)});

    \draw (6,0.367) -- (-2.682,5.379); 
    \draw[rotate=120] (6,0.367) -- (-2.682,5.379);
    \draw[rotate=-120] (6,0.367) -- (-2.682,5.379);

    \draw (9,0.367) -- (-4.182,7.977);
    \draw[rotate=120] (9,0.367) -- (-4.182,7.977);
    \draw[rotate=-120] (9,0.367) -- (-4.182,7.977);

    \draw (12,0.367) -- (-5.682,10.575);
    \draw[rotate=120] (12,0.367) -- (-5.682,10.575);
    \draw[rotate=-120] (12,0.367) -- (-5.682,10.575);

\end{tikzpicture}
\addtocounter{subfigure}{-1}
\captionof{subfigure}{\footnotesize{Type IX.}}\label{FIG:TYPEIX}
\end{subfigure}\endminipage\hfill
\minipage[b]{0.5\textwidth}\centering

\begin{subfigure}\centering
\begin{tikzpicture}[scale=0.25]

    \draw[color=gray,->] (0,-14) -- (0,14)  node[anchor= south] {\scriptsize{$\beta^-$}};
    \draw[color=gray,->] (-10,0) -- (18,0)  node[anchor= south] {\scriptsize{$\beta^+$}};

    \draw[color=gray,dashed] (15,0.367) -- (0,0.367);
    \draw[color=gray,dashed] (15,-0.367) -- (0,-0.367);

    \draw[rotate=120,color=gray,dashed] (15,0.367) -- (0,0.367);
    \draw[rotate=120,color=gray,dashed] (15,-0.367) -- (0,-0.367);
    \draw[rotate=-120,color=gray,dashed] (15,0.367) -- (0,0.367);
    \draw[rotate=-120,color=gray,dashed] (15,-0.367) -- (0,-0.367);


    \draw[ rotate=120, domain=1:5,smooth,variable=\x] plot ({5+\x},{exp(-\x)});
    \draw[ rotate=120, domain=1:5,smooth,variable=\x] plot ({5+\x},{-exp(-\x)});

    \draw[ rotate=120, domain=1:5,smooth,variable=\x] plot ({8+\x},{exp(-\x)});
    \draw[ rotate=120, domain=1:5,smooth,variable=\x] plot ({8+\x},{-exp(-\x)});

    \draw[ rotate=120, domain=1:5,smooth,variable=\x] plot ({11+\x},{exp(-\x)});
    \draw[ rotate=120, domain=1:5,smooth,variable=\x] plot ({11+\x},{-exp(-\x)});

    \draw[ domain=-0.375:0.375,smooth,variable=\x] plot ({6.15-1.15*(\x)^2},{\x});
    \draw[ domain=-0.375:0.375,smooth,variable=\x] plot ({9.15-1.15*(\x)^2},{\x});
    \draw[ domain=-0.375:0.375,smooth,variable=\x] plot ({12.15-1.15*(\x)^2},{\x});

    \draw[ rotate=-120,domain=-0.375:0.37,smooth,variable=\x] plot ({6.15-1.15*(\x)^2},{\x});
    \draw[ rotate=-120,domain=-0.375:0.37,smooth,variable=\x] plot ({9.15-1.15*(\x)^2},{\x});
    \draw[ rotate=-120,domain=-0.375:0.37,smooth,variable=\x] plot ({12.15-1.15*(\x)^2},{\x});

    \draw (6,0.367) -- (-2.682,5.379); 
    \draw[rotate=120] (6,0.367) -- (-2.682,5.379);
    \draw[rotate=-120] (6,0.367) -- (-2.682,5.379);

    \draw (9,0.367) -- (-4.182,7.977);
    \draw[rotate=120] (9,0.367) -- (-4.182,7.977);
    \draw[rotate=-120] (9,0.367) -- (-4.182,7.977);

    \draw (12,0.367) -- (-5.682,10.575);
    \draw[rotate=120] (12,0.367) -- (-5.682,10.575);
    \draw[rotate=-120] (12,0.367) -- (-5.682,10.575);

 \end{tikzpicture}
\addtocounter{subfigure}{-1}
\captionof{subfigure}{\footnotesize{Type VIII.}}\label{FIG:TYPEVIII}
\end{subfigure}\endminipage
\captionof{figure}{Level curves for large $\bar{V}$ in \eqref{barvlambdaRpot}
for Bianchi type IX and VIII. As $\tau_-$ increases (and thus $\bar{V}$ becomes
larger), each level curve has a larger perimeter, whereas the 
$\epsilon$-corners have fixed size $2\epsilon$, which asymptotically becomes negligible compared to the remaining perimeter length given by the type II triangle of length $3L(\tau_-)$.
}\label{FIG:potlevelsets}
\end{figure}

According to the above heuristic approximation scheme, a type VIII or IX solution
is described as a free particle moving with unit speed $|\vec{V}|= 1$ in a triangular
potential well with the $\epsilon$-corners cut out (i.e., for now, we assume that the
particle does not enter the corner regions). Similarly to each single
type II wall, the three approximating infinite type II potential walls move with a
speed $|\vec{v}_\alpha| = v$, where each
individual $\mathrm{II}_\alpha$ walls have velocities $\vec{v}_1$,
$\vec{v}_2$, $\vec{v}_3$ in~\eqref{1wallv} and~\eqref{23wallv},
respectively. If the particle hits an infinite type II wall, it bounces according
to that type II wall's bounce law, which is determined by the wall's velocity,
given by~\eqref{BounceLaw}, unless it enters a corner region.
In the critical and subcritical cases, $v\in (0,1/2]$, as $v$ decreases the
increasingly slow wall motion implies that the particle on average bounces against
the walls increasingly often. In the subcritical case, $v\in(0,1/2)$, the regions
$A_\alpha\cap A_\beta$ correspond to directions where the particle might
hit the corner region between $\mathrm{II}_\alpha$ and $\mathrm{II}_\beta$.
In the supercritical case, $v\in(1/2,1)$, the walls move faster and the particle
bounces against walls less often. Note, however, that the particle bounces against
the type II walls infinitely if the velocity directions are associated with the
Cantor set $C$. Otherwise, after a finite number of bounces, the particle acquires
a velocity direction so that it does not catch up with any wall, and thereby it
enters a final Kasner state described by the final velocity
direction.\footnote{This might also happen before the wall description has become a good
approximation. This corresponds to that a solution ends at the stable set $S$ before
starting to shadow the type II and I boundaries.}

The above heuristic picture assumes that the particle does not enter an
$\epsilon$-corner region. However, consider a starting point at some $(\tau_-)_0$
where the length of the equilateral type II triangle is $3L_0$, and the
$\epsilon$-regions have fixed length $6\epsilon$. Since the sides of the triangle
are moving apart with speed $v$, its size $3L(\tau_-)$ increases as $\tau_-\to\infty$,
whereas the $\epsilon$-corners have fixed size $6\epsilon$ independent of $\tau_-$,
which means that the size of the corner regions becomes asymptotically negligible,
i.e., $\lim_{\tau_-\rightarrow\infty}(2\epsilon/L(\tau_-))=0$. This suggests that
the probability that a generic particle (corresponding to non-locally rotationally
symmetric solutions) enters an $\epsilon$-corner region tends to zero as
$\tau_-\rightarrow\infty$, which corresponds to that the `cross terms' asymptotically
tend to zero, in some statistically generic sense.

The above suggest that the type II heteroclinic chains describe the asymptotic
dynamics in the subcritical and critical cases, $v\in(0,1/2]$, in some statistically
generic sense, and that solutions associated with the Cantor set $C$ in the
supercritical case, $v\in (1/2,1)$, also asymptotically shadow type II heteroclinic
chains.

In the above discussion, we have assumed that
the potential wall is not affected by the bounce. This is not quite
the case, but we will argue that the effect is asymptotically negligible.
During the Bianchi type I motion, $p_\lambda$ and $p_\pm$
are all constants (for simplicity set ${\cal N}$ to a constant,
and then note that $\beta^\lambda$ and $\beta^\pm$ are all cyclic variables).
Then write the Hamiltonian type $\mathrm{II}_1$ constraint as
\begin{equation}
\Sigma_+^2 + \Sigma_-^2 + \frac{12n_1^2}{p_\lambda^{2}}e^{-8(v\tau_- + \beta^+)} = 1,
\end{equation}
where the Bianchi type I regime is determined by a wall
$6n_1^2e^{-8(v\tau_- + \beta^+_0)}p_\lambda^{-2} = C \ll 1$,
where $p_\lambda$, which we previously neglected, is a constant
during the Bianchi type I motion of the cosmological particle. However,
$p_\lambda$ changes during a bounce against the wall, where the difference
is determined by the relevant Bianchi type $\mathrm{II}_1$ solution,
as seen as follows. Equations~\eqref{etaprime} and~\eqref{p0prime}
result in that $p_\lambda \propto 1/\eta$. Since we are considering
expanding models, $p_\lambda$ is negative, see~\eqref{betadotham},
and it is monotonically decreasing when
$\tau_-$ is increasing, due to~\eqref{p0prime}. Also,
since $p_\lambda \propto 1/\eta$, then $\eta^\mathrm{i}=1$
and $\eta^\mathrm{f} = g$ 
according to Section~\ref{sec:II}. Therefore,
\begin{equation}\label{pfbounce}
\frac{p_\lambda^\mathrm{f} -  p_\lambda^\mathrm{i}}{p_\lambda^\mathrm{i}}
= \frac{1}{g}-1,
\end{equation}
due to a Bianchi type II bounce.
In other words, the wall moves because of the bounce,
apart from its movement during the cosmological particle's Bianchi
type I motion. However, we will now argue that this effect is
asymptotically negligible for generic solutions.

Equation~\eqref{p0prime} yields $p_\lambda \propto \exp(\int[-4v(1-\Sigma^2)]d\tau_-)$.
In the dynamical systems picture, the above heuristic cosmological particle
description corresponds to that solutions to an increasing extent
shadow the heteroclinic Bianchi type II orbits. They thereby stay
increasingly long times $\tau_-$ near the Kasner circle $\mathrm{K}^\ocircle$
where $1-\Sigma^2 = 0$, while the effects on $p_\lambda$ of a given type II bounce
according to~\eqref{pfbounce} are not affected by when a bounce takes place in the
evolution of the solution. However, the increasing size of the triangular potential well
shows that solutions have a `memory' of their evolution, which is not locally seen,
not in~\eqref{pfbounce} nor in the local eigenvalue analysis of $\mathrm{K}^\ocircle$.
Due to increasingly accurate shadowing, the time $\tau_-$ spent by the particle during
bounces becomes asymptotically negligible compared to the time spent in Bianchi type I
motion. This is due to that the time spent during a bounce is determined by the
time it takes the particle to move the extra distance the wall has moved
because of the bounce. Since this distance is increasingly small when
compared to the size of the increasing triangular well, the (average)
time during a bounce when compared to the (average) time between
bounces, which is increasing due to the increasing size of the
triangular well, becomes asymptotically negligible (this has been illustrated
numerically in GR, see, e.g., Figure 11.1 in~\cite{waiell97}).

The above heuristic discussion suggests that the following `dominant' $\lambda$-$R$
Hamiltonian captures the asymptotic dynamics of the original Hamiltonian:
\begin{equation}\label{LambdaRDOMH}
H_\mathrm{Dom} := {\cal N}\left(\frac12\left(- p_\lambda^2 + p_+^2 + p_-^2\right) + V_\mathrm{Dom}\right),\qquad
V_\mathrm{Dom} := 6m_1^2 + 6m_2^2 + 6m_3^2.
\end{equation}
In other words, the potential $V_\mathrm{Dom}$ consists of the three type II
potentials, without the cross terms in the original potential $V$ in~\eqref{kinV}.
Expressing the Hamiltonian equations obtained from $H_\mathrm{Dom}$ in the variables
$(\Sigma_\pm,N_1,N_2,N_3)$ and time $\tau_-$
yields the same evolution equations as for the $\lambda$-$R$ models
in~\eqref{dynsyslambdaR}, except that the cross terms $N_1N_2$, $N_2N_3$,
and $N_3N_1$ are absent. This system has the same heteroclinic Bianchi type I and II
structure as the $\lambda$-$R$ models. The
above heuristic reasoning suggests that both dynamical systems,
associated with~\eqref{LambdaRham} and~\eqref{LambdaRDOMH}, respectively,
are asymptotically described by $\mathrm{K}^\ocircle$ and the Bianchi type
II heteroclinic chains, in a manner that supports the conjectures in
Section~\ref{sec:conjectures}. The same conclusion is obtained from the billiard formulation
of Chitr\'e and Misner~\cite{chi72,mis69a,mis69b}, see p. 812 in~\cite{grav73},
and also~\cite{dametal03,heietal09,damlec11}. Incidentally, it seems plausible that
the symbolic dynamics methods used in the main text could be used in this billiard formulation.
Also, it would be interesting to see if deforming first principles in some restricted sense
in more general models (such as those discussed in~\cite{dametal03}, e.g., by
considering string theory inspired modifications) can lead to bifurcations that
are related to those in the present case.


\subsection{Ho\v{r}ava-Lifshitz Class A models}\label{appsubsec:HL}

The vacuum spatially homogeneous Ho\v{r}ava-Lifshitz (HL) class~A Bianchi
models have a Hamiltonian with the same kinetic part as the vacuum
$\lambda$-$R$ models, but a potential $V$ in~\eqref{VHL} consisting of a sum
of the potentials ${}^AV$ with $A = 1,\dots,6$,
where the superscript ${}^A$ reflects the constants $k_A$ in~\eqref{calV}, which multiply the
spatial curvature expressions in~\eqref{VHL} and~\eqref{pots}, determined by the spatial
Ricci curvature $R^i\!_j$ and the Cotton (Cotton-York) tensor $C^i\!_j$.

For the diagonal class~A models the only non-zero components of $R^i\!_j$
are the diagonal terms $R^1\!_1$, $R^2\!_2$ and $R^3\!_3$, and similarly for $C^i\!_j$.
The reason for writing these tensors with one upper and lower index is that the
orthonormal frame expressions coincide with those obtained when using the left-invariant
frame in~\eqref{threemetric}. This allows us to specialize the general orthonormal frame
expressions for the spatial curvature and the Cotton tensor in~\cite{elsugg97}
to the present diagonal class~A Bianchi models, where $R^1\!_1$ is given in~\eqref{R11},
the curvature scalar $R$ in~\eqref{Rscalar}, and
\begin{equation}\label{C11}
C^1\!_1 = -\frac{1}{2g^{3/2}}\left(2m_1^3 - (m_2+m_3)\left[m_1^2 + (m_2-m_3)^2\right]\right),
\end{equation}
where $m_1, m_2, m_3$ are defined in~\eqref{malpha}.
The remaining components $R^2\!_2$, $R^3\!_3$ and $C^2\!_2$, $C^3\!_3$
are obtained by permutation of indices $1,2,3$ in the respective formulas
for $R^1\!_1$ and $C^1\!_1$.

As stated in the derivation of the dynamical system~\eqref{dynsyslambdaR} for the
$\lambda$-$R$ models, we only need to compute the equation for $\Sigma_+$ and use
permutations to obtain the equations for the $\Sigma_\alpha$ variables,
as argued after~\eqref{app:MisnerPar}.
In addition, we will exploit that each potential term in $V$ can be written as an exponential
in $\beta^\lambda$ times a function of $\beta^+$ and $\beta^-$, since each
potential term ${}^AV$ has a certain dimensional weight under conformal scalings of the spatial
metric. To make this weight explicit, we define
\begin{subequations}\label{barm}
\begin{align}
\bar{m}_1 &:= g^{-1/3}m_1 = n_1 e^{-4\beta^+}, \\
\bar{m}_2 &:= g^{-1/3}m_2 = n_2 e^{2\beta^+ + 2\sqrt{3}\beta^-}, \\
\bar{m}_3 &:= g^{-1/3}m_3 = n_3 e^{2\beta^+ - 2\sqrt{3}\beta^-}.
\end{align}
\end{subequations}
%
Here we have introduced the convention that variables with an overbar are functions of
$\beta^\pm$ only. Denoting $R^1\!_1$ and $C^1\!_1$ with $R_1$ and $C_1$, respectively,
leads to that the equations~\eqref{R11} and~\eqref{C11} can be written as follows,
by means of~\eqref{barm} and $g=\exp(12v\beta^\lambda)$,
\begin{subequations}\label{RCalpha}
\begin{xalignat}{2}
R_1 &= e^{-4v\beta^\lambda}\bar{R}_1, &\qquad  \bar{R}_1 &= \frac12(\bar{m}_1^2 - \bar{m}_-^2),\\
C_1 &= e^{-6v\beta^\lambda}\bar{C}_1, &\qquad
\bar{C}_1 &= -\frac12\left(2\bar{m}_1^3 - \bar{m}_+(\bar{m}_1^2 + \bar{m}_-^2)\right),
\end{xalignat}
\end{subequations}
where we define $\bar{m}_\pm$ as
\begin{equation}\label{mpm}
\bar{m}_\pm := \bar{m}_2 \pm \bar{m}_3.
\end{equation}
The terms $R_2=R^2\!_2$, $R_3=R^3\!_3$, $C_2=C^2\!_2$ and $C_3=C^3\!_3$
are again obtained by permutations. To obtain succinct expressions, we not only
introduce $\bar{m}_\pm$, but also the following (Misner-like)
parametrization of the diagonal components of $\bar{R}^i\!_j$ and $\bar{C}^i\!_j$:
\begin{subequations}\label{Adecomp}
\begin{xalignat}{2}
\bar{A}_1 &= \frac13\!\left(\bar{A} - 2\bar{A}_+\right), & \quad
\bar{A} &:= \bar{A}_1 + \bar{A}_2 + \bar{A}_3,\\
\bar{A}_2 &= \frac13\!\left(\bar{A} + \bar{A}_+ + \sqrt{3}\bar{A}_-\right), \quad \quad \text{ where } & \quad
\bar{A}_+ &:= \frac12\!\left(\bar{A}_2 + \bar{A}_3 - 2\bar{A}_1\right),\\
\bar{A}_3 &= \frac13\!\left(\bar{A} + \bar{A}_+ - \sqrt{3}\bar{A}_-\right), & \quad
\bar{A}_- &:= \frac{\sqrt{3}}{2}\!\left(\bar{A}_2 - \bar{A}_3\right).
\end{xalignat}
\end{subequations}
The Cotton tensor $C^i\!_j$ is trace-free, and hence
$\bar{C} = \bar{C}_1 + \bar{C}_2 + \bar{C}_3 = 0$ when replacing $\bar{A}$
with $\bar{C}$. Using the expressions in~\eqref{RCalpha} for $\bar{R}_1$, $\bar{C}_1$,
and permutations thereof, in~\eqref{Adecomp} gives
\begin{subequations}\label{barRCpm}
\begin{align}
\bar{R} &= -\frac12\left(\bar{m}_1^2 - 2\bar{m}_1\bar{m}_+ + \bar{m}_-^2\right),\label{barR}\\
\bar{R}_+ &= -\frac12\left(2\bar{m}_1^2 - \bar{m}_1\bar{m}_+ - \bar{m}_-^2\right),\\
\bar{R}_- &= -\frac{\sqrt{3}}{2}\bar{m}_-(\bar{m}_1 - \bar{m}_+),\\
\bar{C}_+ &= \frac34\left(2\bar{m}_1^3 - \bar{m}_1^2\bar{m}_+ - \bar{m}_+\bar{m}_-^2\right),\\
\bar{C}_- &= \frac{\sqrt{3}}{4}\bar{m}_-\left(\bar{m}_1^2 + 2\bar{m}_1\bar{m}_+ - 2 \bar{m}_+^2 - \bar{m}_-^2\right).
\end{align}
\end{subequations}
%

%
%
The above parametrization leads to that potentials
${}^AV$ in~\eqref{pots} take the form:
\begin{subequations}\label{mpots}
\begin{xalignat}{2}
{}^1V &= e^{8v\beta^\lambda}({}^1\bar{V}),&\qquad {}^1\bar{V} &= 12k_1\bar{R},\label{bar1V}\\
{}^2V &= e^{4v\beta^\lambda}({}^2\bar{V}),&\qquad {}^2\bar{V} &= 12k_2\bar{R}^2,\\
{}^3V &= e^{4v\beta^\lambda}({}^3\bar{V}),&\qquad
{}^3\bar{V} &= 12k_3\bar{R}^i\!_j\bar{R}^j\!_i = 4k_3(\bar{R}^2 + 2\bar{R}_+^2 + 2\bar{R}_-^2),\\
{}^4V &= e^{2v\beta^\lambda}({}^4\bar{V}),&\qquad
{}^4\bar{V} &= 12k_4\bar{R}^i\!_j\bar{C}^j\!_i = 8k_4(\bar{R}_+\bar{C}_+ + \bar{R}_-\bar{C_-}),\\
{}^5V &= {}^5\bar{V}, &\qquad
{}^5\bar{V} &= 12k_5\bar{C}^i\!_j\bar{C}^j\!_i = 8k_5(\bar{C}_+^2 + \bar{C}_-^2),\\
{}^6V &= {}^6\bar{V}, &\qquad {}^6\bar{V} &= 12k_6\bar{R}^3,
\end{xalignat}
\end{subequations}
which follows from~\eqref{pots} and~\eqref{Adecomp}.
Thus all potentials ${}^A{V}$ depend explicitly on $\beta^\lambda$
as described in~\eqref{mpots}, whereas ${}^A\bar{V}$ are functions of
$\bar{m}_1$, $\bar{m}_2$, $\bar{m}_3$ due to~\eqref{barRCpm} and~\eqref{mpm},
and thereby of $\beta^\pm$ according to~\eqref{barm}.

Assigning a weight under spatial conformal scalings according to the scale $g^{1/6} = 
e^{2v\beta^\lambda}$ results in that the potentials ${}^AV$ in~\eqref{mpots}
have the following weights, denoted by $[{}^AV]$,
\begin{equation}\label{weights}
[{}^1V] = 4,\qquad [{}^2V] = [{}^3V] = 2,\qquad
[{}^4V] = 1,\qquad [{}^5V] = [{}^6V] = 0.
\end{equation}
%
In other words, all potentials ${}^AV$ have an exponential dependence
on $\beta^\lambda$, but with different powers of $g^{1/6} = 
e^{2v\beta^\lambda}$ according to~\eqref{weights}.
The integer relations between the weights for the different potential terms,
$[{}^1V] = [({}^2V)^2] = [({}^3V)^2] = [({}^4V)^4]$
play a role below.

All HL models share some common features.
First, all HL models have the same automorphism group and associated symmetries
at each level of the class~A Bianchi hierarchy, IX, VIII; $\mathrm{VII}_0$,
$\mathrm{VI}_0$; II; I, see Figure~\ref{FIG:hierarchy}.
Second, each individual curvature term yields a potential
with a certain weight~\eqref{weights}, which results in that each individual
potential has a scale-symmetry which can be combined with the automorphism
group to yield a scale-automorphism group.\footnote{In Appendix~\ref{app:heterosym},
this group is central for the dynamical properties at each level in the
class~A Bianchi hierarchy, for both the $\lambda$-$R$ and HL models.}
For this reason, we now describe the potentials at each level of the class~A Bianchi
hierarchy below type VIII and IX (the potentials of type VIII and IX where given
in~\eqref{mpots}).

The type $\mathrm{VI}_0$ and $\mathrm{VII}_0$ models are characterized
by a single vanishing structure constant $n_1$, $n_2$, $n_3$. Without loss of
generality, let $n_1=0$ (and hence $\bar{m}_1=0$) describe these models.
Equations~\eqref{barm} and~\eqref{mpm} then motivate the definitions
\begin{equation}
\tilde{m}_\pm := e^{-2\beta^+}m_\pm = n_2e^{2\sqrt{3}\beta^-} \pm n_3e^{-2\sqrt{3}\beta^-},
\end{equation}
where we introduce the convention that variables with a $\tilde{}$ on top are functions of
$\beta^-$ only. It follows that $d\tilde{m}_\pm/d\beta^- = 2\sqrt{3}\tilde{m}_\mp$.
Equations~\eqref{mpots} and~\eqref{barRCpm} imply that the type $\mathrm{VII}_0$
and $\mathrm{VI}_0$ potentials with $n_1=0$ can be written as
\begin{subequations}\label{mpotsVIVII}
\begin{xalignat}{2}
{}^1V_{\mathrm{VII}_0,\mathrm{VI}_0} &= e^{4(2v\beta^\lambda + \beta^+)}({}^1\tilde{V}), &\qquad
{}^1\tilde{V} &= -6k_1\tilde{m}_-^2,\label{tilde1V}\\
{}^2V_{\mathrm{VII}_0,\mathrm{VI}_0} &= e^{4(v\beta^\lambda + 2\beta^+)}({}^2\tilde{V}),&\qquad
{}^2\tilde{V} &= 3k_2\tilde{m}_-^4,\\
{}^3V_{\mathrm{VII}_0,\mathrm{VI}_0} &= e^{4(v\beta^\lambda + 2\beta^+)}({}^3\tilde{V}), &\qquad
{}^3\tilde{V} &= 3k_3\tilde{m}_-^2(2\tilde{m}_+^2 + \tilde{m}_-^2),\\
{}^4V_{\mathrm{VII}_0,\mathrm{VI}_0} &= e^{2(v\beta^\lambda + 5\beta^+)}({}^4\tilde{V}), &\qquad
{}^4\tilde{V} &= -6k_4\tilde{m}_+\tilde{m}_-^2(\tilde{m}_+^2 + \tilde{m}_-^2),\\
{}^5V_{\mathrm{VII}_0,\mathrm{VI}_0} &= e^{12\beta^+}({}^5\tilde{V}),&\qquad
{}^5\tilde{V} &= \frac32k_5(7\tilde{m}_+^2\tilde{m}_-^4 + 4\tilde{m}_-^2\tilde{m}_+^4 + \tilde{m}_-^6),\\
{}^6V_{\mathrm{VII}_0,\mathrm{VI}_0} &= e^{12\beta^+}({}^6\tilde{V}),&\qquad
{}^6\tilde{V} &= -\frac32k_6\tilde{m}_-^6.
\end{xalignat}
\end{subequations}

To describe the HL Bianchi type II models, we consider, without loss of
generality, the $\mathrm{II}_1$ models, which are characterized by
$n_2 = n_3 = 0$, which yields $\bar{m}_2 = \bar{m}_3 =0$
and $\bar{m}_\pm =0$. The Bianchi type $\mathrm{II}_1$ potentials
in~\eqref{mpots} with curvature terms~\eqref{barRCpm} are thereby given by
\begin{subequations}\label{potparts}
\begin{align}
{}^1V_{\mathrm{II}_1} &= -6k_1e^{8v\beta^\lambda}\bar{m}_1^2 = -6k_1 n_1^2 e^{8(v\beta^\lambda - \beta^+)}
= -6k_1 n_1^2 e^{-8(v\tau_- + \beta^+)},\label{Vbeta}\\
{}^2V_{\mathrm{II}_1} &= 3k_2e^{4v\beta^\lambda}\bar{m}_1^4 = 3k_2 n_1^4 e^{4(v\beta^\lambda - 4\beta^+)}
= 3k_2 n_1^4 e^{-4(v\tau_- + 4\beta^+)}, \label{Vgamma}\\
{}^3V_{\mathrm{II}_1} &= 9k_3e^{4v\beta^\lambda}\bar{m}_1^4 = 9k_3 n_1^4 e^{4(v\beta^\lambda - 4\beta^+)}
= 9k_3 n_1^4e^{-4(v\tau_- + 4\beta^+)}, \label{Vdelta}\\
{}^4V_{\mathrm{II}_1} &= -12k_4e^{2v\beta^\lambda}\bar{m}_1^5  = -12k_4 n_1^5 e^{2(v\beta^\lambda - 10\beta^+)}
= -12k_4 n_1^5 e^{-2(v\tau_- + 10\beta^+)},\label{Vepsilon}\\
{}^5V_{\mathrm{II}_1} &= 18k_5\bar{m}_1^6 = 18k_5 n_1^6 e^{-24\beta^+}, \label{Vzeta}\\
{}^6V_{\mathrm{II}_1} &= -\frac32 k_6\bar{m}_1^6 = -\frac32 k_6 n_1^6 e^{-24\beta^+}.
\end{align}
\end{subequations}

The common dimensional weight in~\eqref{weights} for ${}^2V$ and ${}^3V$, and for
${}^5V$ and ${}^6V$, motivates that these two pairs of potentials are treated
collectively, i.e.,
\begin{subequations}
\begin{align}
{}^{2,3}V :=&\, {}^2V + {}^3V = 
e^{4v\beta^\lambda}({}^{2,3}\bar{V}) = e^{4v\beta^\lambda}\left[4(3k_2 + k_3)\bar{R}^2 + 8k_3(\bar{R}_+^2 + \bar{R}_-^2)\right],\\
{}^{2,3}V_{\mathrm{VII}_0,\mathrm{VI}_0} =&\, e^{4(v\beta^\lambda + 2\beta^+)}({}^{2,3}\tilde{V})= 3 e^{4(v\beta^\lambda + 2\beta^+)} \tilde{m}_-^2\left[2k_3\tilde{m}_+^2 + (k_2+k_3)\tilde{m}_-^2\right],\\
{}^{2,3}V_{\mathrm{II}_1} =&\,3(k_2+3k_3)n_1^4 e^{4(v\beta^\lambda - 4\beta^+)}, &\hfill
\end{align}
\end{subequations}
and
\begin{subequations}
\begin{align}
{}^{5,6}V :=&\, {}^5V + {}^6V = {}^{5,6}\bar{V} = 8k_5(\bar{C}_+^2 + \bar{C}_-^2) + 12k_6\bar{R}^3,\\
{}^{5,6}V_{\mathrm{VII}_0,\mathrm{VI}_0} =&\, e^{12\beta^+}({}^{5,6}\tilde{V})=\frac32 e^{12\beta^+}\left[k_5(7\tilde{m}_+^2\tilde{m}_-^4 + 4\tilde{m}_-^2\tilde{m}_+^4) + (k_5-k_6)\tilde{m}_-^6\right],\\
{}^{5,6}V_{\mathrm{II}_1} =&\,\frac32\left(12k_5 - k_6\right))n_1^6e^{-24\beta^+},
\end{align}
\end{subequations}
which follows from~\eqref{mpots}, \eqref{mpotsVIVII} and~\eqref{potparts}.

To obtain a unified description of the various potentials, we refer to them
with a superscript $A$, i.e., ${}^AV$, where $A=1, \{2,3\}, 4, \{5,6\}$; thus $A=2,3$
and $A=5,6$ corresponds to ${}^{2,3}V = {}^2V + {}^3V$ and ${}^{5,6}V = {}^5V + {}^6V$,
respectively. We also introduce the constants
\begin{subequations}\label{HLconstants}
\begin{alignat}{7}
{}^1v &= v := \frac{1}{\sqrt{2(3\lambda-1)}}, &\quad {}^{2,3}v &= \frac{v}{4},
&\quad {}^4v &= \frac{v}{10}, &\quad {}^{5,6}v &= 0,\label{HLII1speed}\\
{}^1a &= 2, &\quad {}^{2,3}a &= 4, &\quad {}^4a &= 5, &\quad {}^{5,6}a &= 6, \label{apot}\\
{}^1c &= -12k_1, &\quad {}^{2,3}c &= 6(k_2+3k_3), &\quad {}^4c &= -24k_4, &\quad
{}^{5,6}c &= 36k_5 - 3k_6. \label{cpot}
\end{alignat}
\end{subequations}
The models with ${}^1v, {}^{2,3}v, {}^4v \in (0,1)$ thereby correspond to
$\lambda \in (\frac12,\infty), (\frac{11}{32},\infty),(\frac{67}{200},\infty)$, respectively.

The constants~\eqref{HLconstants} allow us to write the HL potentials at each level of the class~A Bianchi
hierarchy as follows:
\begin{subequations}\label{potuni}
\begin{alignat}{2}
{}^AV &= e^{4av\beta^\lambda}({}^A\bar{V}), &\qquad\qquad &\text{for types } \mathrm{IX} \text{ and } \mathrm{VIII},\\
{}^AV_{\mathrm{VII}_0,\mathrm{VI}_0} &=
e^{2a(2v\beta^\lambda + \beta^+)}({}^A\tilde{V}), &\qquad\qquad &\text{for types }
\mathrm{VII}_0 \text{ and } \mathrm{VI}_0, \text{ with } n_1=0,\\
{}^AV_{\mathrm{II}_1} &= \frac12 c\,n_1^{a}\,e^{4a(v\beta^\lambda - \beta^+)}, &\qquad\qquad &\text{for type } \mathrm{II}_1,\label{uniII_1pot}
\end{alignat}
\end{subequations}
where we refrain from writing the superscript $A$
on $^Aa$, $^Av$ and $^Ac$ for notational brevity, e.g.,
${}^AV_{\mathrm{II}_1} = \frac12 c\,n_1^{a}\,e^{4a(v\beta^\lambda - \beta^+)}
= \frac12 ({}^Ac)(n_1)^{({}^Aa)}\,e^{4({}^Aa)(({}^Av)\beta^\lambda - \beta^+)}$.
As can be seen by inspection, inserting~\eqref{HLconstants} into the above
expressions yield~\eqref{mpots}, \eqref{mpotsVIVII}, with $n_1=0$ for the type
$\mathrm{VII}_0$ and $\mathrm{VI}_0$ potentials, and~\eqref{potparts} for type
$\mathrm{II}_1$.

We will now derive a regular constrained dynamical
system for the HL case, and then perform a heuristic analysis of the
HL models. In the latter case, when there are several potential terms~\eqref{mpots},
we heuristically argue that there is a single dominant potential, and that
the scale-automorphism group for this dominant potential is intimately linked to
the asymptotic dynamics toward the singularity, in the same manner as for the
GR and $\lambda$-$R$ models.

\subsubsection*{Derivation of the HL evolution equations}

To obtain a \emph{regular} dynamical system for the HL models,
we first consider the Hamiltonian equations with the Hamiltonian~\eqref{LambdaRham}
for the variables $\beta^\lambda$, $\beta^\pm$ and the canonical momenta
$p_\lambda$, $p_\pm$. The kinetic part $T$ depends on $p_\lambda$, $p_\pm$
and is given by~\eqref{Tkinetics}, while the potential $V$ depends on $\beta^\lambda$,
$\beta^\pm$ according to~\eqref{VHL} and~\eqref{mpots}.
We then use the same $\Sigma_\pm$ and $\Sigma_\alpha$ variables as in the
$\lambda$-$R$ case, defined in~\eqref{SigmaNvariables} and~\eqref{app:MisnerPar},
and the time variable $\tau_-:=-\beta^\lambda$.
The Hamiltonian equations then result in the evolution equations
\begin{subequations}\label{hamprime}
\begin{align}
\frac{d\beta^\lambda}{d\tau_-} &= -1, \qquad
&&\frac{d p_\lambda}{d\tau_-} =-\Omega_\lambda p_\lambda, \\
\frac{d\beta^\pm}{d\tau_-} &= -\Sigma_\pm, \qquad &&
\frac{d\Sigma_\pm}{d\tau_-} =\Omega_\lambda\Sigma_\pm + {\cal S}_\pm,\label{Salphap}
\end{align}
\end{subequations}
subjected to the constraint
\begin{equation}\label{consf}
\Sigma_+^2 + \Sigma_-^2 + \Omega_k = 1,
\end{equation}
where we have introduced the following dimensionless quantities
\begin{equation}\label{OmS}
\Omega_k := \frac{2}{p_\lambda^2}V, \qquad
\Omega_\lambda := \frac{1}{p_\lambda^2}\frac{\partial V}{\partial\beta^\lambda}, \qquad
{\cal S}_\pm := -\frac{2}{p_\lambda^2}\frac{\partial V}{\partial\beta^\pm}.
\end{equation}
The different quantities in~\eqref{OmS} can be decomposed into objects that are
related to the individual potentials ${}^AV$ for $A = 1, \{2,3\},4,\{5,6\}$ as follows
\begin{subequations}\label{Omk}
\begin{align}
\Omega_k &= {}^1\Omega_k + {}^{2,3}\Omega_k + {}^4\Omega_k + {}^{5,6}\Omega_k,\\
\Omega_\lambda &= {}^1\Omega_\lambda + {}^{2,3}\Omega_\lambda +
{}^4\Omega_\lambda + {}^{5,6}\Omega_\lambda,\\
{\cal S}_\pm &= {}^1{\cal S}_\pm + {}^{2,3}{\cal S}_\pm + {}^4{\cal S}_\pm + {}^{5,6}{\cal S}_\pm, \label{sumSp}
\end{align}
\end{subequations}
where ${}^A\Omega_k:=  2p_\lambda^{-2} ({}^AV)$,
${}^A\Omega_\lambda:=  p_\lambda^{-2}\partial_{\partial\beta^\lambda} ({}^AV)$
and ${}^A{\cal S}_\pm:=  -2p_\lambda^{-2}\partial_{\partial\beta^\pm} ({}^AV)$.
Due to~\eqref{mpots}, $\Omega_\lambda$ is given by
\begin{equation}\label{Omlambda}
\Omega_\lambda = 4v({}^1\Omega_k) + 2v({}^{2,3}\Omega_k) +
v({}^4\Omega_k),
\end{equation}
where the coefficients are related to the scaling weights in~\eqref{weights}.
Note that all ${}^A\bar{V}$ are homogeneous polynomials in
$\bar{m}_1$, $\bar{m}_2$ and $\bar{m}_3$, which are invariant under
permutations of indices, as follows from~\eqref{mpots},
which is a consequence of that the potentials have been constructed
from curvature scalars. It follows from the definitions that ${}^A\Omega_k$
for $A = 1, \{2,3\},4,\{5,6\}$ are also homogeneous polynomials in
$\bar{m}_\alpha$ and that ${}^A\Omega_k$, $\Omega_k$ and $\Omega_\lambda$
are invariant under permutations.

To compute the equation for $\Sigma_+$, 
we need
to compute ${}^A{\cal S}_+$ and ${\cal S}_+$. To do so, note that the
equations~\eqref{barm} and~\eqref{mpm} yield
\begin{equation}\label{dermpm}
\frac{\partial\bar{m}_1}{\partial\beta^+} = - 4\bar{m}_1, \qquad
\frac{\partial\bar{m}_\pm}{\partial\beta^+} = 2\bar{m}_\pm,
\end{equation}
which together with the chain rule and~\eqref{barRCpm} gives
\begin{subequations}\label{barRpmCpm}
\begin{align}
\frac{\partial\bar{R}}{\partial\beta^+} &= 2(2\bar{m}_1^2 - \bar{m}_1\bar{m}_+ - \bar{m}_-^2),\label{DbarRp}\\
\frac{\partial\bar{R}_+}{\partial\beta^+} &= 8\bar{m}_1^2 - \bar{m}_1\bar{m}_+ + 2\bar{m}_-^2,\\
\frac{\partial\bar{R}_-}{\partial\beta^+} &= \sqrt{3}(\bar{m}_1 + 2\bar{m}_+)\bar{m}_-,\\
\frac{\partial\bar{C}_+}{\partial\beta^+} &=
-\frac92\left(4\bar{m}_1^3 - \bar{m}_1^2\bar{m}_+ + \bar{m}_+\bar{m}_-^2\right),\\
\frac{\partial\bar{C}_-}{\partial\beta^+} &=
-\frac{3\sqrt{3}}{2}(\bar{m}_1^2 + 2\bar{m}_+^2 + \bar{m}_-^2) \bar{m}_-.
\end{align}
\end{subequations}
These expressions in combination with the
chain rule applied to~\eqref{mpots} yields each ${}^A{\cal S}_+$ as a
homogeneous polynomial in $\bar{m}_1$, $\bar{m}_2$ and $\bar{m}_3$ of
the same degree as in ${}^A\Omega_k$. The polynomials are multiplied
with $p_\lambda^{-2}$ and exponentials in $\beta^\lambda$,
determined by the weights of ${}^{A}V$ in the same way as for
${}^A\Omega_k$. We then change the $\Sigma_\pm$ variables to
$\Sigma_\alpha$, $\alpha= 1,2,3$,
\begin{equation}\label{PMtoABC}
\left(\Sigma_1,{}^A{\cal S}_1,{\cal S}_1\right):=-2(\Sigma_+,{}^A{\cal S}_+,{\cal S}_+).
\end{equation}
This leads to
\begin{equation}
\Sigma_\alpha^\prime = \Omega_\lambda\Sigma_\alpha + {\cal S}_\alpha, \qquad \alpha=1,2,3,
\end{equation}
where the equations for $\alpha = 2$ and $\alpha=3$ are obtained by cyclic permutations of
$(123)$ in the expressions involving $\bar{m}_1$, $\bar{m}_2$ and $\bar{m}_3$ when $\alpha=1$.
There is thereby no need to compute ${\cal S}_-$.

Our next step is to introduce dimensionless variables that replace $\beta^\lambda$, $\beta^\pm$
in~\eqref{hamprime} and $p_\lambda$ in order to obtain a \emph{regular}
dynamical system. In the $\lambda$-$R$ case, the single potential term given by $V = {}^{1}V$
had a specific weight under conformal scaling transformations, which yielded a symmetry
that decoupled the evolution equation for $p_\lambda$. We now have potential
terms in~\eqref{mpots} with different weights~\eqref{weights}, where adding them breaks
this symmetry. We therefore introduce separate dimensionless variables
${}^AN_\alpha$, $A=1,\{2,3\},4,\{5,6\}$, that respect the different weights and
the polynomial nature of the potential in terms of $\bar{m}_\alpha$. This
leads to twelve variables ${}^AN_\alpha$, three for each value of $A$,
replacing the four variables $\beta^\lambda$, $\beta^\pm$ and $p_\lambda$.
Due to the construction, the equation for the dimensional variable $p_\lambda$,
\begin{equation}\label{HLplambda}
p_\lambda^\prime = - \Omega_\lambda p_\lambda
\end{equation}
decouples, but the redundance of variables ${}^AN_\alpha$ results
in constraints between them.

To explicitly define the variables ${}^AN_\alpha$, which will lead to
explicit constraints, we first consider the type $\mathrm{II}_1$ potentials
${}^AV_{\mathrm{II}_1}$ in equation~\eqref{potparts} and define
${}^A\Omega_{\mathrm{II}_1} := 2({}^AV_{\mathrm{II}_1})/(-p_\lambda)^2$.
Permutation of axis allows us to obtain the dimensionless expression
\begin{equation}
{}^A\Omega_{\mathrm{II}_\alpha} = \frac{2({}^AV_{\mathrm{II}_\alpha})}{(-p_\lambda)^2}
= c e^{4a\,v\,\beta^\lambda}\frac{\bar{m}_\alpha^a}{(-p_\lambda)^{2}},
\end{equation}
where we refrain from writing the superscript $A$ on ${}^Ac$,
${}^Aa$, ${}^Av$ for brevity. We recall that $p_\lambda<0$ for
an expanding model.

Throughout we will restrict considerations to HL models for which
$({}^Ac)(n_\alpha)^{({}^Aa)}\geq 0$ when $n_\alpha \neq 0$, i.e., models
with non-negative Bianchi type II potentials and non-negative
${}^A\Omega_{\mathrm{II}_\alpha}$. This amounts to sign conditions
on the constants $k_A$ in~\eqref{potparts}.
For example, $k_1\leq 0$ in ${}^1V_{\mathrm{II}_1}$ when $n_1\neq 0$,
whereas
\begin{equation}\label{k2356}
k_{2,3} := k_2 + 3k_3\geq 0, \qquad k_{5,6} := 12k_5 - k_6\geq 0,
\end{equation}
are associated with ${}^1V_{\mathrm{II}_{2,3}}$ and ${}^1V_{\mathrm{II}_{5,6}}$,
respectively. Note that the term ${}^4V_{\mathrm{II}_1}$ when $n_1\neq 0$ is
special: it is positive if $k_4$ is negative and $n_1$ is chosen to be positive.
But in Bianchi type VIII, one of the potentials ${}^4V_{\mathrm{II}_1}$,
${}^4V_{\mathrm{II}_2}$, ${}^4V_{\mathrm{II}_3}$ by necessity has an opposite
sign compared to the other two, irrespective of the sign of $k_4$. This occurs
since two of the constants $n_1$, $n_2$, $n_3$ have opposite signs compared to
the third, as in Table~\ref{intro:classAmodels}, and these constants appear
with odd powers in~\eqref{Vepsilon}. For this reason, we exclude the type
VIII models with $k_4 \neq 0$, except when $k_{5,6}>0$, since ${}^{5,6}V$
asymptotically `dominates' over ${}^4V$, as will be seen below.

We then define new dimensionless variables that are linear in $\bar{m}_\alpha$,
as in the $\lambda$-$R$ case. This can be done because the potentials ${}^AV$
are homogeneous polynomials in $\bar{m}_\alpha$, $\alpha = 1, 2, 3$, and thus
the evolution equations has variables that respect these polynomial relationships.
We therefore introduce the following variables,
\begin{equation}\label{Ndefintro}
{}^AN_\alpha := \sqrt[a]{{}^A\Omega_{\mathrm{II}_\alpha}} =
\sqrt[a]{c} e^{4v\beta^\lambda}\frac{\bar{m}_\alpha}{\sqrt[a]{(-p_\lambda)^{2}}} ,
\end{equation}
where we again drop the superscript $A$ on ${}^Av$, ${}^Aa$, ${}^Ac$. This
leads to the following:
\begin{subequations}\label{Ndef}
\begin{align}
{}^1N_\alpha &:= \sqrt{-12k_1} e^{4v\beta^\lambda} \left(\frac{\bar{m}_\alpha}{-p_\lambda}\right),\label{1Ndef}\\
{}^{2,3}N_\alpha &:= \sqrt[4]{6k_{2,3}} e^{v\beta^\lambda} \left(\frac{\bar{m}_\alpha}{\sqrt{-p_\lambda}}\right),\\
{}^4N_\alpha &:= \sqrt[5]{-24k_4} e^{2v\beta^\lambda/5} \left(\frac{\bar{m}_\alpha}{\sqrt[5]{p_\lambda^2}}\right),\\
{}^{5,6}N_\alpha &:= \sqrt[6]{3k_{5,6}} \left(\frac{\bar{m}_\alpha}{\sqrt[3]{-p_\lambda}}\right).
\end{align}
\end{subequations}
As in the $\lambda$-$R$ case, considering the Hamiltonian/Codazzi constraint
in these variables shows that they are dimensionless. Note that the
dimension of $k_A$ determines how both $\beta^\lambda$ and $p_\lambda$
enter the definitions, once one has decided to adapt the variables to the
polynomials in $\bar{m}_\alpha$ (or, equivalently, $m_\alpha$). Finally
we choose to fix the remaining dimensionless constants that can multiply
the ${}^AN_\alpha$ variables by requiring that the coefficients for the
Bianchi type II terms in each ${}^A\Omega_k$ are equal to one.

The evolution equations for the variables ${}^AN_\alpha$, expressed in
the variables $\Sigma_\alpha$, are obtained from the above definitions,
\eqref{hamprime}, \eqref{PMtoABC}, and are given by
\begin{equation}\label{ANeqabs}
{}^AN_\alpha^\prime = -2\left[2({}^Av) + \Sigma_\alpha - \frac{\Omega_\lambda}{{}^Aa}\right]({}^AN_\alpha),
\end{equation}
or, explicitly,
\begin{subequations}\label{ANeq}
\begin{align}
{}^1N_\alpha^\prime &= (\Omega_\lambda - 4v - 2\Sigma_\alpha)({}^1N_\alpha),\label{1Neq}\\
{}^{2,3}N_\alpha^\prime &= \frac12\left(\Omega_\lambda - 2v - 4\Sigma_\alpha\right)({}^{2,3}N_\alpha),\\
{}^4N_\alpha^\prime &= \frac25\left(\Omega_\lambda - v - 5\Sigma_\alpha\right)({}^4N_\alpha),\\
{}^{5,6}N_\alpha^\prime &= \frac13\left(\Omega_\lambda - 6\Sigma_\alpha\right)({}^{5,6}N_\alpha),
\end{align}
\end{subequations}
where $\alpha=1,2,3$, and where ${}^Av$ has been replaced with $v$, defined in~\eqref{vdef},
according to~\eqref{HLII1speed}.

As mentioned, there are twelve variables ${}^AN_\alpha$, since
$\alpha=1,2,3$ and $A=1,\{2,3\},4,\{5,6\}$, and thus they are not all
independent since they are functions of four variables,
$\beta^\lambda$, $\beta^\pm$, $p_\lambda$. The above evolution equations~\eqref{ANeq}
are therefore constrained. Using that there are integer weight relations between
the different potentials ${}^AV$, $A=1,\{2,3\},4,\{5,6\}$,
and inserting the definitions~\eqref{Ndef} into these relations yield
the following constraints
\begin{subequations}\label{Nconstraints}
\begin{alignat}{2}
\sqrt{k_{2,3}^{3}}({}^1N_\alpha)({}^4N_\alpha)^5 &= -k_4\sqrt{-2^5k_1}({}^{2,3}N_\alpha)^6,
&\quad  \alpha &= 1, 2, 3,\label{cons1}\\
k_{2,3}({}^1N_\alpha)({}^{5,6}N_\alpha)^3 &= \sqrt{-k_1k_{5,6}}({}^{2,3}N_\alpha)^4, & \quad \alpha &= 1, 2, 3,\label{cons2}\\
({}^AN_\alpha)({}^BN_\beta) &= ({}^AN_\beta)({}^BN_\alpha), &\quad \quad\alpha\beta &= 12, 23, 31;\\
& & \quad A,B &= 1, \{2,3\}, 4, \{5,6\}.\nonumber
\end{alignat}
\end{subequations}
Only nine of the above constraints above turn out to be independent,
e.g., \eqref{cons1}, \eqref{cons2} and
$({}^1N_\alpha)({}^{2,3}N_\beta) = ({}^1N_\beta)({}^{2,3}N_\alpha)$,
since the other ones can be written in terms of these nine equations if
$k_1, k_{2,3}=k_2+3k_3,k_4, k_{5,6}=12k_5-k_6$ are all non-zero. If any of these
coefficients are zero, one can choose a different set among the available constraints,
but most three of the twelve variables ${}^AN_\alpha$ are independent.

The HL dynamical system thereby consists of the following evolution equations, 
\begin{subequations}\label{fullHLdynsys}
\begin{align}
\Sigma_\alpha^\prime &= \Omega_\lambda\Sigma_\alpha + {\cal S}_\alpha,\\
{}^1N_\alpha^\prime &= (\Omega_\lambda - 4v - 2\Sigma_\alpha)({}^1N_\alpha),\label{1Neq2}\\
{}^{2,3}N_\alpha^\prime &= \frac12\left(\Omega_\lambda - 2v - 4\Sigma_\alpha\right)({}^{2,3}N_\alpha),\\
{}^4N_\alpha^\prime &= \frac25\left(\Omega_\lambda - v - 5\Sigma_\alpha\right)({}^4N_\alpha),\\
{}^{5,6}N_\alpha^\prime &= \frac13\left(\Omega_\lambda - 6\Sigma_\alpha\right)({}^{5,6}N_\alpha),
\end{align}
subjected to the constraints
\begin{alignat}{2}
0 &= \Sigma_1 + \Sigma_2 + \Sigma_3, & &\label{HLusualconst1}\\
1 &= \Sigma^2 + \Omega_k, & \qquad\qquad \Sigma^2 :=\,\, & (\Sigma_1^2+\Sigma_2^2 + \Sigma_3^2)/6,\label{HLusualconst2}\\
\sqrt{k_{2,3}^{3}}({}^1N_\alpha)({}^4N_\alpha)^5 &= -k_4\sqrt{-2^5k_1}({}^{2,3}N_\alpha)^6,
& \quad \alpha &= 1, 2, 3,\label{cons1f}\\
k_{2,3}({}^1N_\alpha)({}^5N_\alpha)^3 &= \sqrt{-k_1k_{5,6}}({}^{2,3}N_\alpha)^4, & \quad \alpha &= 1, 2, 3,\label{cons2f}\\
({}^AN_\alpha)({}^BN_\beta) &= ({}^AN_\beta)({}^BN_\alpha), &\quad \quad \alpha\beta &= 12, 23, 31; \label{cons3f}\\
& & \quad A,B &= 1, \{2,3\}, 4, \{5,6\}.\nonumber
\end{alignat}
\end{subequations}
where 
$\Omega_k$ and $\Omega_\lambda$ are obtained from~\eqref{Omk}
and~\eqref{Omlambda}; ${\cal S}_\alpha$ is obtained by permuting the indices
in ${\cal S}_1$, which is computed from ${\cal S}_+$ according to~\eqref{PMtoABC}. 

There are in total fifteen variables
(three $\Sigma_\alpha$ and twelve ${}^AN_\alpha$) subjected to
eleven constraints (two constraints~\eqref{HLusualconst1}
and~\eqref{HLusualconst2}, which also hold for the $\lambda$-$R$ models and GR,
and nine new independent constraints~\eqref{cons1f},\eqref{cons2f}
and~\eqref{cons3f}, relating the variables
${}^AN_\alpha$). This results in a rather formidable
constrained dynamical system. However, the system~\eqref{fullHLdynsys}
contains several less complicated
special invariant sets, which illustrate the above procedure
of how to obtain the explicit equations. The special invariant
sets are of two types:
\begin{itemize}
\item[(i)] sets obtained by setting some potentials ${}^AV$ to zero,
$A=1,\{2,3\},4,\{5,6\}$, which corresponds to setting the corresponding
variables ${}^AN_\alpha$ with $\alpha=1,2,3$ to zero;
\item[(ii)] special Bianchi types obtained by setting one or more
constants $n_\alpha$ to zero, which implies that the corresponding
variables ${}^AN_\alpha$ are zero, for all values of $A$.
\end{itemize}

To illustrate the algorithmic procedure to obtain explicit equations,
which require expressing ${}^A\Omega_k$, ${}^A\Omega_\lambda$ and
${}^A{\cal S}_\alpha$, based on~\eqref{OmS} in the variables
${}^AN_\alpha$, let us first consider the case where only
$k_1\neq 0$ (recall that $k_1=-1$ yields the $\lambda$-$R$ case).
Thus the invariant set defined by ${}^AN_\alpha=0$ for all $A=\{2,3\},4,\{5,6\}$,
$\alpha=1,2,3$ provides an example of an invariant subset of type
(i) above, where $\Omega_k={}^1\Omega_k$,
$\Omega_\lambda={}^1\Omega_\lambda$, ${\cal S}_\alpha={}^1{\cal S}_\alpha$.
Equations~\eqref{barR} and~\eqref{bar1V} yield
\begin{equation}\label{1Vpot}
\begin{split}
{}^1V &= 12k_1e^{8v\beta^\lambda}\bar{R}
= -6k_1e^{8v\beta^\lambda}(\bar{m}_1^2 - 2\bar{m}_1\bar{m}_+ + \bar{m}_-^2)\\
&= -6k_1e^{8v\beta^\lambda}(\bar{m}_1^2 + \bar{m}_2^2 + \bar{m}_3^2
- 2\bar{m}_1\bar{m}_2 - 2\bar{m}_2\bar{m}_3 - 2\bar{m}_3\bar{m}_1),
\end{split}
\end{equation}
where we have used~\eqref{mpm}.

Together with~\eqref{OmS} and~\eqref{1Ndef}, this leads to
\begin{equation}
\Omega_k = N_1^2 + N_2^2 + N_3^2 - 2N_1N_2 - 2N_2N_3 - 2N_3N_1.
\end{equation}
Equation~\eqref{Omlambda} yields $\Omega_\lambda = 4v\Omega_k$.
It remains to determine ${\cal S}_\alpha$.
Equations~\eqref{1Vpot}, \eqref{DbarRp}, \eqref{mpm} and~\eqref{1Ndef} result in
\begin{equation}
{\cal S}_1 = -4\left[(N_2 - N_3)^2 - N_1(2N_1 - N_2 - N_3)\right],
\end{equation}
where cyclic permutations of $(123)$ yield ${\cal S}_2$ and ${\cal S}_3$.

A comparison shows that the dynamical system with $k_1<0$ is identical to the
system~\eqref{intro_dynsyslambdaR}, \eqref{LambdaRquantities}, for the $\lambda$-$R$
case for which $k_1=-1$. The reason for this is that single curvature terms,
associated with the constant $k_A$, with $A= 1, \{2,3\}, 4$, admit a scaling
symmetry, which correspond to a translation in $\beta^\lambda$. This symmetry
makes it possible to scale $k_A$ with an arbitrary positive number,
and hence scale the negative coefficient $k_1$ so that $k_1=-1$.

To illustrate invariant subsets obtained by restricting to a
particular Bianchi type (i.e., invariant sets of type (ii) above),
we consider Bianchi types I and II. As in the $\lambda$-$R$ case,
the Bianchi type I subset is just the Kasner circle $\mathrm{K}^\ocircle$.
In the Bianchi type $\mathrm{II}_1$ case, $n_2=n_3=0$ implies
${}^AN_2 = {}^AN_3 = 0$ for all $A$. This thereby leaves four variables ${}^AN_1$
when all constants $k_1$, $k_{2,3}$, $k_4$, $k_{5,6}$ are non-zero.
The constraints~\eqref{cons3f} are all identically zero since they involve
${}^AN_2$ and ${}^AN_3$. Similarly the constraints~\eqref{cons1f} and~\eqref{cons2f}
are identically zero for $\alpha=2,3$. This leaves two
constraints~\eqref{cons1f} and~\eqref{cons2f} for $\alpha=1$, and hence
there are two independent variables ${}^AN_1$, $A=1,\{2,3\},4,\{5,6\}$, in Bianchi
type $\mathrm{II}_1$. Similar statements hold for Bianchi type $\mathrm{II}_2$
and $\mathrm{II}_3$. The Bianchi type II models break the formal permutation
symmetry, e.g.,the type $\mathrm{II}_1$ models lead to that $\bar{m}_2$ and
$\bar{m}_3$, and thereby $\bar{m}_\pm$, are set to zero. In this case, it is
convenient to use equation~\eqref{potparts} and compute ${\cal S}_-=0$
from its definition~\eqref{OmS}, since $\beta^-$ is a cyclic variable,
and then transform the result to obtain ${\cal S}_\alpha$, which yields
the following quantities 
\begin{subequations}\label{HLBianchiII}
\begin{align}
\Omega_k &= ({}^1N_1)^2 + ({}^{2,3}N_1)^4
+ ({}^4N_1)^5 + ({}^{5,6}N_1)^6, \\
\Omega_\lambda &= 4v({}^1N_1)^2 + 2v({}^{2,3}N_1)^4 + v({}^4N_1)^5,\\
({\cal S}_1,{\cal S}_2,{\cal S}_3) &= 4\left[2({}^1N_1)^2 + 4({}^{2,3}N_1)^4 + 5({}^4N_1)^5 + 6({}^{5,6}N_1)^6\right] \mathrm{T}_1,
\end{align}
\end{subequations}
from which it is straightforward to obtain the type $\mathrm{II}_1$ dynamical system.

Equation~\eqref{HLBianchiII} illustrates that
${}^A\Omega_k$, ${}^A\Omega_\lambda$ and ${}^A{\cal S}_\alpha$ for $\alpha = 1,2,3$ are
homogeneous polynomials of ${}^AN_\alpha$ of degrees $2, 4, 5, 6$
when $A= 1, \{2,3\}, 4, \{5,6\}$, respectively.
In Bianchi type VIII and IX, the number of terms in ${}^A\Omega_k$, ${}^A\Omega_\lambda$ and
${}^A{\cal S}_\alpha$ increases as the degree of the polynomials become higher due
to an increase in the number of cross terms, which leads to a daunting
number of terms in $\Omega_k$, $\Omega_\lambda$ and ${\cal S}_\alpha$, when
$k_A\neq 0$ for all $A$.

\subsubsection*{Heuristic HL considerations}

Here we will heuristically argue that the heteroclinic network obtained by
concatenation of Bianchi type II orbits in the $\lambda$-$R$ models given
by~\eqref{dynsyslambdaR} describes the asymptotics of a broad class of HL
models~\eqref{fullHLdynsys}.

Equation~\eqref{uniII_1pot} shows that the Bianchi type II potential walls
for the ${}^AV$ potential has a speed ${}^Av$ given by~\eqref{HLII1speed} in the negative
$\beta^+$-direction, which is obtained in the same way as
equation~\eqref{1wallv} in the $\lambda$-$R$ model. Moreover, for the same
reason as in the $\lambda$-$R$ case, the cross terms for a potential
${}^AV$ with $A = 1,\{2,3\},4$ have approximating walls with higher speeds
than the Bianchi type II terms, and are not expected to affect generic
asymptotic dynamics toward the singularity, as will be discussed below;
the case $A=5,6$ is special. Furthermore, a similar argument
holds for the type II terms that belongs to different potentials:
${}^1V_{\mathrm{II}}$ yields slower moving walls than all the other
type II potentials, so if any of these are present, the ${}^1V_{\mathrm{II}}$
contribution to the dynamics is expected to not affect the generic past asymptotic
dynamics; similarly ${}^{2,3}V_{\mathrm{II}}$ is negligible
if ${}^4V$ or ${}^{5,6}V$ (or both) are non-zero; while if $k_{5,6}>0$ all the
other terms are expected to be dominated by ${}^{5,6}V$.

This leads to a situation where the generic past dynamics
is expected to be characterised by a `dominant' Hamiltonian of the form
\begin{equation}
H_\mathrm{Dom} = {\cal N}(T + V_\mathrm{Dom}) = 0,\label{totdomham}
\end{equation}
where the kinetic part $T$ depends on the canonical momenta $ p_\lambda$, $p_\pm$,
according to~\eqref{Tkinetics}, while the `dominant' potential $V_\mathrm{Dom}$ is the
sum of the type $\mathrm{II}_1$, $\mathrm{II}_2$ and $\mathrm{II}_3$ terms in the
dominant potential ${}^AV$, $A= \mathrm{Dom}$, i.e., all potentials with larger ${}^Bv$
and all the cross terms have been dropped. The dominant potential depends on
$\beta^\lambda$, $\beta^\pm$ as follows:
\begin{equation}\label{VdomHL}
V_\mathrm{Dom} =
\frac12 c\left(n_1^a\,e^{4a(v\beta^\lambda - \beta^+)} +
n_2^a\,e^{2a(2v\beta^\lambda + \beta^+ + \sqrt{3}\beta^-)} +
n_3^a\,e^{2a(2v\beta^\lambda + \beta^+ - \sqrt{3}\beta^-)}\right),
\end{equation}
where we drop the superscript ${}^A=\mathrm{Dom}$
on the constants ${}^Ac$, ${}^Aa$ and ${}^Av$ for brevity. The
value of $A$ in~\eqref{VdomHL} is determined by the dominant potential,
i.e., if all potentials with larger $A$
are zero. For example, $A=1$ corresponds to that $k_{2,3}=k_4=k_{5,6}=0$,
while $A=4$ requires $k_{5,6}=0$.

The case $k_{5,6}>0$ requires special attention. Although the ${}^{5,6}V$ potential is
expected to suppress all other potentials toward the singularity, note that ${}^{5,6}V$
only depends on $\beta^\pm$. Hence $\beta^\lambda$ is a cyclic variable in this case
(using e.g. ${\cal N} = \mathrm{constant}$ as determining the time variable), which
results in that $p_\lambda$ becomes a conserved quantity. The term $E=p_\lambda^2/2$
can thus be viewed as an energy for the reduced Hamiltonian problem with
potential ${}^{5,6}V(\beta^\pm)$. This potential yields a
generalized Toda problem in two dimensions, see~\cite{bog85}. Once the existence
of ${}^{5,6}V$ has suppressed the effects of the other potentials, the full remaining
Toda problem must be addressed. In the limit, $E\rightarrow \infty$, one expects that
the dynamics is described by the $v=0$ Bianchi type I and II
heteroclinic network in the dynamical systems picture,
but for small $E$, all terms in ${}^{5,6}V$ comes into play and one can
expect a complicated dynamical behaviour, in agreement with the
discussions in~\cite{baketal09} and~\cite{myuetal10a}.

We can adapt dynamical systems variables to the present dominant Hamiltonian
system~\eqref{totdomham}. Based on~\eqref{VdomHL}, which consists of the
three dominant ${}^AV_{\mathrm{II}_\alpha}$, $A=\mathrm{Dom}$, potentials,
we define the following dimensionless variables
\begin{subequations}\label{HLSigmaN}
\begin{xalignat}{2}
\qquad & \qquad & \tilde{N}_1 &:= \sqrt{cn_1^{a}}\left(\frac{e^{2a(v\beta^\lambda - \beta^+)}}{-p_\lambda}\right), \\
\Sigma_\pm &:= -\frac{p_\pm}{ p_\lambda}, \quad  & \quad
\tilde{N}_2 &:= \sqrt{cn_2^{a}}\left(\frac{e^{a(2v\beta^\lambda + \beta^+ + \sqrt{3}\beta^-)}}{-p_\lambda}\right),\\
\qquad & \qquad & \tilde{N}_3 &:= \sqrt{cn_3^{a}}\left(\frac{e^{a(2v\beta^\lambda + \beta^+ - \sqrt{3}\beta^-)}}{-p_\lambda}\right),
\end{xalignat}
\end{subequations}
%
%
where we again drop the superscript ${}^{A=\mathrm{Dom}}$ for brevity.
In comparison with the variables in~\eqref{SigmaNvariables} for the $\lambda$-$R$ models,
we keep the same variables $\Sigma_\pm$, but the variables $\tilde{N}_\alpha$
are slightly modified. In particular,
there exists an overall factor ${}^Aa/2$ in the exponent times an expression that
is formally the same as in the $\lambda$-$R$ case, but with $v$ replaced with ${}^Av$.

The multiplicative factor ${}^Aa/2$ can be eliminated by a change of the time variable
$\tau_-$ according to $\tilde{\tau}_- := 2\tau_-/{}^Aa$.
Letting ${}^\prime$ denote the new derivative $d/d\tilde{\tau}_-$
yields the following system of evolution equations
\begin{subequations}\label{domsys}
\begin{align}
\Sigma_+^\prime &= 2\left(2v(1-\Sigma^2)\Sigma_+ + \tilde{N}_2^2 + \tilde{N}_3^2 - 2\tilde{N}_1^2\right),\\
\Sigma_-^\prime &= 2\left(2v(1-\Sigma^2)\Sigma_- + \sqrt{3}\tilde{N}_2^2 - \sqrt{3}\tilde{N}_3^2\right),\\
\tilde{N}_1^\prime &= -4(v\Sigma^2 - \Sigma_+)\tilde{N}_1,\\
\tilde{N}_2^\prime &= -2(2v\Sigma^2 + \Sigma_+ + \sqrt{3}\Sigma_-)\tilde{N}_2,\\
\tilde{N}_3^\prime &= -2(2v\Sigma^2 + \Sigma_+ - \sqrt{3}\Sigma_-)\tilde{N}_3,
\end{align}
subjected to the constraint
\begin{equation}
1 - \Sigma^2 - \tilde{N}_1^2 - \tilde{N}_2^2 - \tilde{N}_3^2 = 0.
\end{equation}
\end{subequations}
This dynamical system is formally the same as that in~\eqref{dynsyslambdaR},
but with absent cross terms $\tilde{N}_1\tilde{N}_2$, $\tilde{N}_2\tilde{N}_3$,
$\tilde{N}_3\tilde{N}_1$ and with $v$ replaced by ${}^Av$, where the superscript $A$
refers the dominant potential ${}^{A=\mathrm{Dom}}V$.
Thus the two dynamical systems generated by~\eqref{dynsyslambdaR} and~\eqref{domsys}
share the same Bianchi type I and II heteroclinic structure. Since this structure is
expected to describe the generic asymptotic dynamics toward the singularity
(at least when $k_{5,6} = 0$), the analysis of the heteroclinic structure in the main
part of the paper of the $\lambda$-$R$ models is therefore arguably of relevance
for the generic singularity of a large class of HL models.


Let us now deduce the `dominant' dynamical system~\eqref{domsys}
from the general HL dynamical system~\eqref{fullHLdynsys}. Recall that the
dominant dynamical system was obtained by:
\begin{enumerate}
\item[(i)] setting all potentials ${}^AV$ to zero, except for the
potential ${}^AV$ with $A=\mathrm{Dom}$, i.e., the potential with the largest value
of $A \in 1, \{2,3\}, 4, \{5,6\}$ with non-zero coefficient $k_A$.
This corresponds to the invariant subset of~\eqref{fullHLdynsys} for
which all variables ${}^AN_\alpha$ are set to zero, except for
$A = \mathrm{Dom}$, which leads to three non-zero ${}^\mathrm{Dom}N_\alpha$
variables. Note that the constraints~\eqref{Nconstraints} are automatically satisfied
for subsets that only involve one of the potentials ${}^AV$ where $A=1,\{2,3\},4,\{5,6\}$.
We will refer to this invariant subset of~\eqref{fullHLdynsys}
as the \emph{invariant dominant subset};
\item[(ii)] setting all the cross terms in the potential ${}^AV$
with $A = \mathrm{Dom}$ to zero, which thereby yields
$V_\mathrm{Dom} = {}^AV_{\mathrm{II}_1} + {}^AV_{\mathrm{II}_2} + {}^AV_{\mathrm{II}_3}$ for $A = \mathrm{Dom}$.
In the dynamical system this is achieved by setting all cross terms involving
${}^\mathrm{Dom}N_1$, ${}^\mathrm{Dom}N_2$, ${}^\mathrm{Dom}N_3$
in the invariant dominant subset in~\eqref{fullHLdynsys} to zero.
This results in a system where
$\Omega_k$, $\Omega_\lambda$ and ${\cal S}_\alpha$ are linear in
$({}^{2,3}N_\alpha)^4$, $({}^4N_\alpha)^5$, $({}^{5,6}N_\alpha)^6$, if
$A = \mathrm{Dom}=\{2,3\},4,\{5,6\}$, respectively.
This makes it possible to perform a variable transformation
from the dominant variables ${}^\mathrm{Dom}N_\alpha$ to new variables
${}^\mathrm{Dom}\tilde{N}_\alpha$ which yield quadratic polynomials, e.g.,
${}^{2,3}\tilde{N}_\alpha = ({}^{2,3}N_\alpha)^2$ if $A=\mathrm{Dom}=\{2,3\}$.
To finally obtain the system~\eqref{domsys} from~\eqref{fullHLdynsys},
replace $v$ with ${}^Av$, $A = \mathrm{Dom}$, according to~\eqref{HLII1speed},
and replace ${}^\mathrm{Dom}N_\alpha$ with the new variables $\tilde{N}_\alpha$.
Finally, change the time variable to $\tilde{\tau}_- := 2\tau_-/{}^Aa$,
where ${}^Aa$ is defined in~\eqref{apot} for $A = \mathrm{Dom}$.
\end{enumerate}

The heuristic arguments in this appendix thus suggest that the
$\omega$-limits (as $\tau_-\rightarrow \infty$) for generic Bianchi
type IX solutions (and type VIII, if $A=\mathrm{Dom}\neq 4$, as discussed
previously) of the evolution equation~\eqref{fullHLdynsys} reside on the
union of the Bianchi type I and II subsets on the invariant dominant subset.
Replacing $v$ with ${}^Av$, $A = \mathrm{Dom}$, and using the
Hamiltonian/Gauss constraint to solve for the single ${}^\mathrm{Dom}N_\alpha$
variable in each of the type $\mathrm{II}_\alpha$ subsets, leads to the
equations in $(\Sigma_1,\Sigma_2,\Sigma_3)$-space used in the main text to
discuss the heteroclinic network on the union of the type I and II subsets
for the $\lambda$-$R$ models, if one changes the time variable according to
$\tilde{\tau}_- := 2\tau_-/{}^\mathrm{Dom}a$.\footnote{This is the
reason we obtain the results in~\cite{giakam17} as special cases of
our results for the Bianchi type II $\lambda$-$R$ models. Incidentally, we
could have introduced  the Kasner parameter $u$, as in said reference.
The Kasner map describing how $u$ changes follows from~\eqref{KasnerCirc},
\eqref{pdef} and~\eqref{ueq}. However, the range and domain of $u$ differ from
the critical GR case when $v\neq 1/2$. This suggests that
one should use an extended Kasner parameter, see~\cite{ugg13a}. However, since the
parameter $v$ leads to a complicated expression for the Kasner map for $u$,
we do not pursue this possibility.} The results in the main part of the paper
thereby heuristically apply to a broad class of HL models.

\section{First principles and the Bianchi hierarchy}\label{app:heterosym}

In this appendix, we derive monotone functions and conserved quantities at
each level of the class~A Bianchi hierarchy from the associated scale and
automorphism symmetry hierarchy. These structures are inherited from the first
principles of scale and diffeomorphism invariance, as shown for the GR case
in~\cite{heiugg10}. The present models do not change the automorphism group,
but they do have different scale symmetries, which yield different results
for the $\lambda$-$R$ and HL models. 

\subsection{$\lambda$-$R$ models}\label{app:lrscaleaut}

\subsubsection*{Bianchi types VIII and IX}

We here derive a monotone function, called $\Delta$. The decay
of $\Delta$ in time implies that the type VIII and IX solutions converge to
next level in the class~A Bianchi hierarchy, the union of the invariant
type $\mathrm{VI}_0$ and $\mathrm{VII}_0$ boundary sets, as
discussed in Section~\ref{sec:firstprinciples}.

The Hamiltonian for Bianchi type VIII and IX is characterized by
\begin{equation}\label{appB:Hamilt}
T + V = \frac12\left(-p_\lambda^2 + p_+^2 + p_-^2\right) + 6e^{8v\beta^\lambda}\bar{V}(\beta^\pm)=0.
\end{equation}
The kinetic part defines the DeWitt metric
$\eta_{AB} = \mathrm{diag}(-1,1,1)$ for $A,B = \lambda, \pm$, and its inverse
$\eta^{AB} = \mathrm{diag}(-1,1,1)$, since we can write the kinetic part as
$T =\eta^{AB}p_Ap_B/2$. The diagonal type VIII and IX models admit no
(diagonal) automorphisms since all the structure constants,
$n_1$, $n_2$, $n_3$, are non-zero. However, the field equations of all
vacuum $\lambda$-$R$ models admit a scale symmetry, which thereby leads to a
scale symmetry for the potential in~\eqref{appB:Hamilt}, obtained by translations
in $\beta^\lambda$. Moreover, in the potential $V= 6e^{8v\beta^\lambda}\bar{V}(\beta^\pm)$
the exponent $8v\beta^\lambda$ is clearly timelike with respect to $\eta_{AB}$ in
$(\beta^\lambda,\beta^+,\beta^-)$-space when $v\in (0,1)$.
Similarly as in GR, see ch. 10 in~\cite{waiell97} and~\cite{heiugg10},
this leads to a monotone function, given by
$e^{8v\beta^\lambda}/p_\lambda^2 \propto |N_1N_2N_3|^{2/3}$. Choosing to scale
this with $3$ so that $\Omega_k + \Delta \geq 0$ in Section~\ref{sec:firstprinciples}
yields
\begin{subequations}\label{Deltaapp}
\begin{align}
\Delta :=&\, 3|N_1N_2N_3|^{2/3},\\
\Delta^\prime =&\, -8v\Sigma^2\Delta,
\end{align}
\end{subequations}
where we have used the chain rule and~\eqref{intro_dynsyslambdaR_N}.

\subsubsection*{Bianchi types $\mathrm{VI}_0$ and $\mathrm{VII}_0$}

We now show that the scale-automorphism group for type $\mathrm{VI}_0$
and $\mathrm{VII}_0$ yields the functions ($1+2v\Sigma_+$, $Z_{\mathrm{sub}}$,
$Z_{\mathrm{sup}}$ and $Z_\mathrm{{crit}}$). These functions have different
consequences for the subcritical, supercritical and critical cases,
discussed in Section~\ref{sec:firstprinciples}. In particular,
$1+2v\Sigma_+$ is useful in all cases and is derived first; then we
derive $Z_{\mathrm{sub}}$ ($Z_{\mathrm{sup}}$), which is useful for
the subcritical (supercritical) case, where
$Z_{\mathrm{sub}}=Z_{\mathrm{sup}}=Z_\mathrm{{crit}}$ for
the critical case.

Let $n_1=0$, $n_2=n_3=1$ for type $\mathrm{VII}_0$
and $n_1=0$, $n_2=-n_3=1$ for type $\mathrm{VI}_0$, without loss of generality.
Then, according to equation~\eqref{tilde1V}, the Hamiltonian
is described by,
\begin{equation}\label{HVIVIIapp}
T + V = \frac12\left(-p_\lambda^2 + p_+^2 + p_-^2\right) + 6e^{4(2v\beta^\lambda + \beta^+)}\tilde{m}_-^2 = 0,
\end{equation}
where we recall that
$\tilde{m}_- = n_2e^{2\sqrt{3}\beta^-} - n_3e^{-2\sqrt{3}\beta^-}$.

There are two special cases characterized by $\beta^-=p_-=0$,
discussed in the dynamical systems framework in Section~\ref{sec:firstprinciples}.
The first is given by the locally rotationally symmetric
(LRS) type $\mathrm{VII}_0$ models, which have an extra space-time isometry and thereby
a 4-dimensional multiply transitive isometry group. 
Since $\beta^-=p_-=0$ implies $\tilde{m}_-=0$ and thus $V=0$, both $\beta^\lambda$
and $\beta^+$ become cyclic variables and hence $p_\lambda$ and $p_+$ are constants.
Moreover, the Hamiltonian constraint yields $p_+ = \pm p_\lambda$, which corresponds
to the two invariant disjoint lines $\Sigma_+=\pm 1$, $\Sigma_-=0$, $N:=N_2=N_3$.
The second case results in the special type $\mathrm{VI}_0$ models,
which exist due to the discrete symmetry $\beta^- \rightarrow -\beta^-$,
and correspond to a space-time with a discrete isometry, in contrast to the
continuous extra isometry in the LRS type $\mathrm{VII}_0$ case.
Since $\beta^-=p_-=0$ is an invariant set, it follows that so is
$\Sigma_-=0$, $N_2=-N_3$. Moreover, since $\beta^-=p_-=0$
implies $\tilde{m}_-^2 = \mathrm{constant} >0$, the Hamiltonian constraint yields
$|\Sigma_+| < 1$. Below we will treat the special type $\mathrm{VI}_0$ models
together with the general ones.

Excluding the special cases with $\beta^-=p_-=0$, the exponent
$4(2v\beta^\lambda + \beta^+)$ and $\tilde{m}_-$ in the potential
in~\eqref{HVIVIIapp} shows that there are only two independent variables
in the Hamiltonian, $2v\beta^\lambda + \beta^+$ and $\beta^-$. It hence
follows that there is a cyclic variable and an associated conserved
quantity. The underlying reason is that the models with $n_1=0$ admit a
non-unimodular automorphism in addition to the scale symmetry.
Following~\cite{heiugg10} and~\cite{rosetal90a} and combining the non-unimodular
autmorphism and the scale symmetry appropriately yields a variational
symmetry and thereby a conserved quantity, given by
\begin{equation}\label{VIVIIppp}
p_\lambda - 2vp_+ = \mathrm{constant},
\end{equation}
as follows from Hamilton's equations.

Since $p_\lambda$ is monotone, apart from in the LRS type $\mathrm{VII}_0$ case, as follows
from Hamilton's equation~\eqref{p0prime}, dividing~\eqref{VIVIIppp} with
$p_\lambda$ yields a monotone function, except when $p_\lambda - 2vp_+=0$.
This, however, can only happen in the supercritical case $v\in(1/2,1)$, since
the Hamiltonian constraint yields $|p_\lambda| > |p_+|$. Expressing
the quotient $(p_\lambda - 2vp_+)/p_\lambda$ in the dynamical systems variables
results in
\begin{equation}\label{VIVII1app}
\frac{p_\lambda - 2vp_+}{p_\lambda} = 1 + 2v\Sigma_+,
\end{equation}
which evolves according to
\begin{equation}\label{MonotoneVIVIIapp}
(1+2v\Sigma_+)^\prime = 4v(1 - \Sigma^2)(1+2v\Sigma_+).
\end{equation}
%
Further insights come from explicitly introducing cyclic variables
that respect the kinetic part of the Hamiltonian, which is done next.
More specifically, in the subcritical and supercritical cases,
we make a Lorentz transformation in the $(\beta^\lambda,\beta^\pm)$-space
with respect to $\eta_{AB}$, where these transformations
preserve the form of the kinetic part in~\eqref{HVIVIIapp} by definition,
i.e., $T = \eta^{AB}p_Ap_B/2$. However, note that with respect to $\eta_{AB}$, the
exponent $4(2v\beta^\lambda + \beta^+)$ is spacelike for the subcritical case $v\in(0,1/2)$,
null for the critical case $v=1/2$, and timelike for the supercritical case $v\in(1/2,1)$.
The different causal characters again reflect that a bifurcation takes place when $v=1/2$.

In the subcritical case, $v<1/2$, a boost with velocity $-2v$ results in
\begin{subequations}\label{boost1}
\begin{align}
\tilde{\beta}^\lambda &= \Gamma(\beta^\lambda + 2v\beta^+), &\qquad
\beta^\lambda &= \Gamma(\tilde{\beta}^\lambda - 2v\tilde{\beta}^+),\\
\tilde{\beta}^+ &= \Gamma(2v\beta^\lambda + \beta^+), &\qquad
{\beta}^+ &= \Gamma(-2v\tilde{\beta}^\lambda + \tilde{\beta}^+),
\end{align}
\end{subequations}
where $\Gamma= (1 - (2v)^{2})^{-1/2}$. Hence~\eqref{HVIVIIapp}
is transformed to
\begin{equation}\label{HsubcritVIVIIapp}
T + V = \frac12\left(-\tilde{p}_\lambda^2 + \tilde{p}_+^2 + p_-^2\right) +
6e^{4\tilde{\beta^+}/\Gamma}\tilde{m}_-^2 = 0.
\end{equation}
The fact that $\tilde{p}_\lambda$ is conserved leads to a reduced problem
for $\tilde{\beta}^+$ and $\beta^-$ with energy $E = \tilde{p}_\lambda^2/2$.
Note that the Hamiltonian thereby takes the same form as when $v=0$,
as for the HL models with dominant potential ${}^{5,6}V$. The reduced problem
thereby yields a generalized Toda problem in two dimensions, see~\cite{bog85}.
The conserved quantity $\tilde{p}_\lambda = \Gamma(p_\lambda - 2vp_+)$
results in~\eqref{VIVII1app}, and consequently~\eqref{MonotoneVIVIIapp}.
Moreover, $\tilde{p}_\lambda \neq 0$ due to~\eqref{HsubcritVIVIIapp},
and since we are considering expanding models, $\tilde{p}_\lambda < 0$.
Since $\tilde{p}_\lambda$ has the same sign as $p_\lambda$, it follows
that $\tilde{p}_\lambda/p_\lambda > 0$, which implies that $1 + 2v\Sigma_+>0$
in the subcritical case $v<1/2$. As a consequence, $1 + 2v\Sigma_+$ is a
monotone function in the entire state space, as described in~\eqref{MonotoneVIVIIapp},
apart from when $\Sigma^2 = 1$, which only happens for the Bianchi type I
and the LRS type $\mathrm{VII}_0$ invariant sets.

When $p_-\neq 0$, the Hamiltonian equations for the reduced Toda problem
for $\tilde{\beta}^+$ and $\beta^-$ implies that a solution originates at
$\tilde{\beta^+}\rightarrow - \infty$, and reaches a maximal but finite value
of $\tilde{\beta^+}$, and then turn back and ends at $\tilde{\beta^+}\rightarrow - \infty$,
where the asymptotic origin and end correspond to $\tau_-\to \pm\infty$.
To translate these claims into rigorous dynamical results, note that
the exponential $e^{4\tilde{\beta^+}/\Gamma}$ in the potential
plays a key role. Dividing the conserved quantity $\tilde{p}_\lambda^2$ with
this exponential yields a dimensionless quantity, which when expressed
in the dynamical systems variables results in
\begin{subequations}\label{Z2app}
\begin{align}
Z_{\mathrm{sub}} &= \frac{(1 + 2v\Sigma_+)^2}{|N_2N_3|},\\
Z_{\mathrm{sub}}^\prime &= 4(2v + \Sigma_+)Z_{\mathrm{sub}}.
\end{align}
\end{subequations}
%

In the supercritical case, $v>1/2$, consider a boost with velocity $-1/(2v)$,
i.e.,
\begin{subequations}\label{boost2}
\begin{align}
\tilde{\beta}^\lambda &= \Gamma\left(\beta^\lambda + \frac{\beta^+}{2v}\right), &\qquad
\beta^\lambda &= \Gamma\left(\tilde{\beta}^\lambda - \frac{\tilde{\beta}^+}{2v}\right),\\
\tilde{\beta}^+ &= \Gamma\left(\frac{\beta^\lambda}{2v} + \beta^+\right), &\qquad
{\beta}^+ &= \Gamma\left(-\frac{\tilde{\beta}^\lambda}{2v} + \tilde{\beta}^+\right),
\end{align}
\end{subequations}
where $\Gamma= (1 - (2v)^{-2})^{-1/2}$. This results in that~\eqref{HVIVIIapp}
takes the form
\begin{equation}\label{Hsupercritapp}
T + V = \frac12\left(-\tilde{p}_\lambda^2 + \tilde{p}_+^2 + p_-^2\right) +
6e^{8v\tilde{\beta^\lambda}/\Gamma}\tilde{m}_-^2 = 0.
\end{equation}
In this case $\tilde{p}_+$ is conserved, which implies that
$\tilde{p}_+/p_\lambda = \Gamma(1 + 2v\Sigma_+)$ is
monotone when $\tilde{p}_+\neq 0$, since $p_\lambda$ is monotone.
Setting $\tilde{p}_+=0$ yields the invariant set $1+2v\Sigma_+=0$.
Invariance under the transformation
$(\tilde{\beta^+},\tilde{p}_+) \rightarrow - (\tilde{\beta^+},\tilde{p}_+)$
shows that the models exhibit a discrete symmetry. As a consequence
the invariant subset $1+2v\Sigma_+=0$ forms a separatrix surface which
divides the remaining state space into two disjoint sets. Moreover,
the discrete symmetry results in that the flow
of~\eqref{MonotoneVIVIIapp} is equivariant under a change of sign
of the monotone function $1 + 2v\Sigma_+$.
In addition, the intersection of the special type $\mathrm{VI}_0$ subset,
$\beta^-=p_-=0$ (i.e., $\Sigma_-=0$ and $N_2=-N_3$) and the subset
$1 + 2v\Sigma_+ = 0$ (i.e., $\tilde{\beta^+}=\tilde{p}_+=0$) yields the fixed point
$\Sigma_+ = -1/(2v)$, $\Sigma_-=0$, $N_2 = -N_3 = \sqrt{1 - (2v)^{-2}}$.

In the supercritical case, the above structures are not the only ones
that can be extracted from the scale-automorphism group. As in the type VIII and IX models,
we have a potential with an exponential with a timelike variable with respect
to $\eta_{AB}$ that multiplies a function that depends on spacelike variables
(only $\beta^-$ in this case), see~\eqref{Hsupercritapp}, after the
transformation~\eqref{boost1}.
Following ch. 10 in~\cite{waiell97}, there is a monotone function given by
$6e^{8v\tilde{\beta^\lambda}/\Gamma}/\tilde{p}_\lambda^2$,
except when $\tilde{\beta}^+ = \tilde{p}_+ = \beta^- = p_-=0$ in type
$\mathrm{VI}_0$, i.e., at the fixed point in these models. This results
in that $Z_{\mathrm{sup}} \propto \tilde{p}_\lambda^2 e^{-8v\tilde{\beta^\lambda}/\Gamma}$
is monotone, and expressing this function in the state space variables
results in
\begin{subequations}\label{Mapp}
\begin{align}
Z_{\mathrm{sup}} &= \frac{(2v + \Sigma_+)^2}{N_2N_3},\label{Mdef}\\
Z_{\mathrm{sup}}^\prime &= 4\left[\frac{(1 + 2v\Sigma_+)^2 + (4v^2-1)\Sigma_-^2}{2v + \Sigma_+}\right]Z_{\mathrm{sup}},\label{Mprimapp}
\end{align}
\end{subequations}
where the last equation is obtained from~\eqref{dynsyslambdaRVIVII}.
Hence $Z_{\mathrm{sup}}$ is monotonically increasing, except at the
type $\mathrm{VI}_0$ fixed point~\eqref{p_VI}. 
In type $\mathrm{VII}_0$, the variables $\Sigma_+ = -1/(2v)$, $\Sigma_-=0$ do not correspond
to an invariant subset. If an orbit passes through these values, this
implies that this only yields an inflection point for the monotonically
increasing $Z_{\mathrm{sup}}$.


In the critical GR case, $v=1/2$, we introduce the null variables
\begin{equation}\label{uwapp}
u := \beta^\lambda + \beta^+,\qquad w := \beta^\lambda - \beta^+,
\end{equation}
which results in
\begin{equation}\label{Huwapp}
T + V = -2p_up_w + \frac12 p_-^2 + 6e^{2u}\tilde{m}_-^2 = 0.
\end{equation}
Since $w$ is a cyclic variable, $p_w = (p_\lambda - p_+)/2 = \mathrm{constant} \leq 0$,
where the inequality follows from the Hamiltonian constraint and from $p_\lambda < 0$,
which holds for expanding models. The Hamiltonian constraint implies that the equality
only occurs for the LRS type $\mathrm{VII}_0$ models. Apart from this
special case, $1 + \Sigma_+$ is a monotone function in both type
$\mathrm{VI}_0$ and $\mathrm{VII}_0$ according to~\eqref{MonotoneVIVIIapp} with $v=1/2$.
%
%

In the critical GR case, $v=1/2$, the functions $Z_{\mathrm{sub}}=Z_{\mathrm{sup}} = Z_{\mathrm{crit}}$
in~\eqref{Z2app} and~\eqref{Mapp} yields
%
%
%
\begin{subequations}\label{Mprimcrit}
\begin{align}
Z_{\mathrm{crit}} &= \frac{(1+\Sigma_+)^2}{N_2N_3},\\
Z_{\mathrm{crit}}^\prime &= 4(1+\Sigma_+)Z_{\mathrm{crit}}
\end{align}
\end{subequations}
Thus $Z_{\mathrm{sub}}=Z_{\mathrm{sup}} = Z_{\mathrm{crit}}$ is thereby also a monotone function in the
critical case, except at the LRS type $\mathrm{VII}_0$ subset $\Sigma_+=-1$,
$\Sigma_-=0$, $N_2=N_3=N$.

The underlying reason for the existence of the monotone function
$Z_{\mathrm{sup}} = Z_{\mathrm{crit}}$ in~\eqref{Mprimcrit}
is the scaling property of the potential obtained by a translation in $u$,
see~\eqref{Huwapp}, and the conserved momentum $p_w$, which in turn is
a consequence of the scale-automorphism group. Apart from the
LRS type $\mathrm{VII}_0$ subset where $p_w=0$ and thereby $\Sigma_+ = -1$,
these two features taken together yield the monotone function
$Z_{\mathrm{sup}} = Z_{\mathrm{crit}} \propto p_w^2e^{-4u} = (p_we^{-2u})^2$.
Thus $Z_{\mathrm{sup}}$ is a monotone function when $v\in [1/2,1)$,
but not in the subcritical case $v\in (0,1/2)$. The reason for this is the
change in causal character of the exponent $4(2v\beta^\lambda + \beta^+)$
in the potential and the Hamiltonian
constraint~\eqref{HsubcritVIVIIapp}, which prevents
$\tilde{p}_+^{-2}e^{4\tilde{\beta^+}/\Gamma}$ from being a
monotonically changing `energy', as described in the
qualitative picture of the dynamics in~\cite{uggetal91} and
ch. 10 in~\cite{waiell97} when the exponent is timelike.

Incidentally, the Hamiltonian~\eqref{Hsupercritapp} for type $\mathrm{VI}_0$
is mathematically closely related to the GR Bianchi type II models with a perfect
fluid obeying a linear equation of state $p=w\rho$, $w\in [0,1)$,
where $p$ is the pressure and $\rho$ the energy density,
see~\cite{ugg88}, ch. 10 in~\cite{waiell97}. The difference is that due to
steeper walls in the heuristic wall description of the Hamiltonian~\eqref{Hsupercritapp}
the present models give rise to a heteroclinic cycle, which is not the case for
the type II perfect fluid models. Also note that in the subcritical case, $v\in (0,1/2)$,
the special type $\mathrm{VI}_0$ models with $\beta^- = p_-=0$ yield the same
mathematical problem as the type II models discussed next when
restricted to the type II LRS case with $p_-=0$,
after appropriate translations and rescalings of $\tilde{\beta}^\lambda$, $\tilde{\beta^+}$
and $\tau_-$. Moreover, in the supercritical case, the special type
$\mathrm{VI}_0$ models yield the same mathematical problem as the LRS GR Bianchi type
I models with a perfect fluid with $p=w\rho$.

\subsubsection*{Bianchi types II}

Next we derive the key building block for the heteroclinic structure from
scale-automorphism symmetries, i.e., the straight Bianchi type II trajectories
in $\Sigma_\pm$-space, and thereby those in
$(\Sigma_1, \Sigma_2, \Sigma_3)$-space, given by~\eqref{BIIstraight}.

Without loss of generality, we consider the Bianchi type
$\mathrm{II}_1$ case with the Hamiltonian:
\begin{equation}\label{HamII_I}
T + V = \frac12\left(-p_\lambda^2 + p_+^2 + p_-^2\right) + 6e^{8(v\beta^\lambda - \beta^+)} = 0.
\end{equation}
Since $\beta^-$ is a cyclic variable, $p_-$ is constant.
This occurs since the type II models with $n_2=n_3=0$ admit a unimodular automorphism,
which generates a variational symmetry and thereby the conserved
momentum $p_-$. 
As in the type $\mathrm{VI}_0$ and $\mathrm{VII}_0$ cases, these models also admit a
scale-automorphism symmetry, obtained by combining the scale symmetry with the
remaining non-unimodular automorphism, which yields a variational symmetry and
an additional cyclic variable. This is made explicit by performing a boost in
the $\beta^+$-direction in $(\beta^\lambda,\beta^\pm)$-space with
a velocity $v$, i.e.,
\begin{subequations}\label{boost3}
\begin{align}
\tilde{\beta}^\lambda &= \Gamma(\beta^\lambda - v\beta^+), &\qquad
\beta^\lambda &= \Gamma(\tilde{\beta}^\lambda + v\tilde{\beta}^+),\\
\tilde{\beta}^+ &= \Gamma(-v\beta^\lambda + \beta^+), &\qquad
{\beta}^+ &= \Gamma(v\tilde{\beta}^\lambda + \tilde{\beta}^+),
\end{align}
\end{subequations}
where $\Gamma = (1-v^2)^{-1/2}$. This leads to the following expression,
\begin{equation}\label{HamIIBoost2}
T + V = \frac12\left(-\tilde{p}_\lambda^2 + \tilde{p}_+^2 + p_-^2\right) +
6e^{-8\tilde{\beta}^+/\Gamma} = 0,
\end{equation}
which shows that not only $\beta^-$ but also $\tilde{\beta}^\lambda$
is a cyclic variable. Thus both $p_-$ and $\tilde{p}_\lambda$ are
constant.\footnote{Note that in the heuristic moving wall description,
the wall moves in the \emph{positive} $\beta^+$-direction in
$(\beta^\lambda,\beta^\pm)$-space with a speed $v$,
while the wall moves in the negative $\beta^+$-direction with
time $\tau_- = -\beta^\lambda$. The speed of the wall is also the speed
of the above boost, which thereby transforms the moving wall to a
motionless wall. This yields the bounce law~\eqref{BounceLaw} for
the moving particle by means of the conserved quantities.}

Since both $p_-$ and $\tilde{p}_\lambda = \Gamma(p_\lambda + vp_+)$
are constants, it follows that dividing the following relation
between the constants ${p}_- \propto p_\lambda/v + p_+$
with $-p_\lambda$, and using that $\Sigma_\pm = p_\pm/(- p_\lambda)$, leads to
\begin{equation}\label{BIIpm}
\Sigma_- = \mathrm{constant}\left(\Sigma_+ - \frac{1}{v}\right),
\end{equation}
where the $\mathrm{constant}$ parametrizes the various heteroclinic Bianchi type
II orbits. This equation also holds for the initial values $\Sigma_\pm^{\mathrm{i}}$
of $\Sigma_\pm$ on $\mathrm{K}^\ocircle$ and dividing the above equation with
$\Sigma_-^{\mathrm{i}} = \mathrm{constant}\,(\Sigma_+^{\mathrm{i}} - v^{-1})$ yields
\begin{equation}
\left(\Sigma_+^{\mathrm{i}} - \frac{1}{v}\right)\Sigma_-
= \Sigma_-^{\mathrm{i}}\left(\Sigma_+ - \frac{1}{v}\right).
\end{equation}
%
%
%
Equation~\eqref{BIIstraight} then follows from the definitions
$\Sigma_1 = - 2\Sigma_+$ and $\Sigma_{2,3} = \Sigma_+ \pm \sqrt{3}\Sigma_-$.

\subsubsection*{Bianchi type I}

The Kasner circle of fixed points $\mathrm{K}^\ocircle$ follows
straightforwardly from the scale-automorphism symmetry group.
Bianchi type I is obtained by a Lie contraction of Bianchi type II,
which results in that all structure constants become zero, which yield an
Abelian symmetry group. This leads to one more special automorphism,
which, together with the other automorphisms and the (trivial) scale symmetry,
implies that all variables $\beta^\lambda$, $\beta^+$, $\beta^-$ are cyclic,
and hence that all momenta $p_\lambda$, $p_+$, $p_-$ are conserved.
Thus $\Sigma_+$ and $\Sigma_-$ are constants,
and due to the Hamiltonian constraint, $T + V = (-p_\lambda^2 + p_+^2 + p_-^2)/2 = 0$,
they satisfy $\Sigma_+^2 + \Sigma_-^2=1$.

\subsection{HL models}

Equation~\eqref{HLconstants} and~\eqref{potuni} provide a unified picture
of the individual ${}^AV$ potentials for the HL class~A Bianchi hierarchy,
which we here, for the reader's convenience, repeat:
\begin{subequations}\label{potuni2}
\begin{alignat}{2}
{}^AV &= e^{4av\beta^\lambda}({}^A\bar{V}),&\qquad\qquad &\text{for types } \mathrm{IX} \text{ and } \mathrm{VIII},\\
{}^AV_{\mathrm{VII}_0,\mathrm{VI}_0} &=
e^{2a(2v\beta^\lambda + \beta^+)}({}^A\tilde{V}), &\qquad\qquad &\text{for types }
\mathrm{VII}_0 \text{ and } \mathrm{VI}_0, \text{ with } n_1=0,\\
{}^AV_{\mathrm{II}_1} &= \frac{c\,n_1^{a}}{2} \,e^{4a(v\beta^\lambda - \beta^+)},  &\qquad\qquad &\text{for type } \mathrm{II}_1,\label{HLIIpot2}
\end{alignat}
\end{subequations}
where, for notational brevity, we have refrained from writing the superscript $A$
on ${}^Aa$, ${}^Av$ and ${}^Ac$, where
\begin{subequations}
\begin{alignat}{7}
{}^1v &= v := \frac{1}{\sqrt{2(3\lambda-1)}}, &\quad {}^{2,3}v &= \frac{v}{4},
&\quad {}^4v &= \frac{v}{10}, &\quad {}^{5,6}v &= 0,\label{HLII1v2}\\
{}^1a &= 2, &\quad {}^{2,3}a &= 4, &\quad {}^4a &= 5, &\quad {}^{5,6}a &= 6,\\
{}^1c &= -12k_1, &\quad {}^{2,3}c &= 6k_{2,3}, &\quad {}^4c &= -24k_4, &\quad {}^{5,6}c &= 3k_{5,6}.
\end{alignat}
\end{subequations}

The automorphism group is the same for all models, but the scale-property of the individual
potentials is different for different $A$. Nevertheless, as seen from~\eqref{potuni2}
there is a close relationship, one simply replace the constants ${}^1v=v$, ${}^1a=2$ and ${}^1c=12$
in the $\lambda$-$R$ case with ${}^Av$, ${}^Aa$ and ${}^Ac$ to take care of this difference.
There is thereby a close connection between all single potential term HL models. However, note that
for type IX and VIII the exponent $e^{4av\beta^\lambda}$ is timelike when $A=1,\{2,3\},4$
while it is a constant when $A=5,6$ where $A=5,6$ represent a bifurcation since ${}^{5,6}v=0$.
Hence, for the same reason as for the $\lambda$-$R$ models, $e^{4av\beta^\lambda}/p_\lambda^2$
yields a monotone function for each HL model with a specific value of $A \in 1, \{2,3\}, 4$.
When $A=5,6$, $p_\lambda$ is conserved, which results in that
$({}^{5,6}N_1)({}^{5,6}N_2)({}^{5,6}N_3) = \mathrm{constant}$.

In type $\mathrm{VII}_0$, $\mathrm{VI}_0$ and $\mathrm{II}_1$, one just replaces
the boost in the $\lambda$-$R$ case with an analogous boost that follows
from~\eqref{potuni2}, to obtain similar conserved quantities and monotone
functions for each HL model. However, in the dynamical systems description these
quantities sometimes take a different form due to the different relations
with the associated $^{}N_\alpha$ variables, see~\eqref{Ndefintro} and~\eqref{Ndef},
but e.g., $1 + 2v\Sigma_+$ in the $\lambda$-$R$ type $\mathrm{VII}_0$ and
$\mathrm{VI}_0$ models is just replaced with $1 + 2{}^Av\Sigma_+$.
Similarly in type $\mathrm{II}_1$, $\Sigma_-=\mathrm{constant}(\Sigma_+ - v^{-1})$
is replaced with $\Sigma_-=\mathrm{constant}(\Sigma_+ - ({}^Av)^{-1})$,
and thus models with $A=1,\{2,3\},4$ have formally the same heteroclinic
type II structure as the $\lambda$-$R$ models with $v\in(0,1)$, although
recall that ${}^Av$ for those values of $A$ are differently related to $\lambda$
than ${}^1v=v$, see~\eqref{HLII1v2}. There are thus very strong relationships between
the dynamics of the $\lambda$-$R$ models and the HL models with single
curvature terms as potentials.

Then recall the heuristic argument that asymptotically toward the singularity
there exists a dominant single potential (the one with the largest value of $A$),
and an associated invariant subset in the HL dynamical systems formulation.
With the exception that if this is the ${}^{5,6}V$ potential, which corresponds
to a bifurcation since ${}^{5,6}v=0$, the correspondence between conserved
quantities and monotone functions between the $\lambda$-$R$ models and the
remaining HL models suggests that generic dynamics toward the singularity
is going to be described by the heteroclinic Bianchi type II and I structure
on the dominant invariant subset. This is also suggested by the dominant
Hamiltonian and the associated dominant dynamical system. Hence we conjecture
that the discrete analysis of the heteroclinic structure in the $\lambda$-$R$
case in the main part of the paper is also describing the asymptotic
dynamics of HL models for which ${}^{5,6}V=0$. The above also suggests that the
there are similar dynamical conjectures for these HL models as those
in Section~\ref{sec:conjectures} for the $\lambda$-$R$ case.

\section{A unified critical and supercritical treatment}\label{app:unifying}

In this Appendix we modify the proof about chaos within the non-generic
Cantor set of the supercritical case, given in Section~\ref{subsec:proofCantor},
to also accommodate the critical GR case, in which chaos is generic.
The method pursued in order to achieve chaoticity that suits both the critical
and the supercritical cases is the construction of a topological conjugacy to a
shift map, in analogy with the use of the encoding map $h$ in~\eqref{defofh}.
This yields a new proof for chaos in GR and relates the supercritical
symbolic dynamics construction to the limiting case of GR, thereby providing a
unified treatment of the two cases.

Before we proceed, we mention that there are different ways to incorporate the
methods of symbolic dynamics used in the supercritical case to also
describe chaos in the critical GR case. We will give a description that is a
\emph{continuous} transition from the supercritical case $v>1/2$ to the critical
case $v=1/2$. This is accomplished by designing a new encoding map $\tilde{h}$ which
is continuous in $v\in [1/2,1)$. This new map behaves in a similar manner as the
encoding map $h$ in~\eqref{defofh} for infinite heteroclinic sequences when $v>1/2$,
but it also appropriately encodes points that reach the Taub points when $v=1/2$,
and thus it remains a well-defined homeomorphism in the limit $v\to 1/2$.

The challenge of a unified treatment lies in the following continuity issue.
For all $v\in [1/2,1)$, define the set $C_v$ of points that never reach the
set $S$, as in~\eqref{defofC}. When $v>1/2$ decreases, the set $S$ shrinks and collapses
to the Taub points at $v=1/2$ (i.e., $S = \mathrm{T}_1\cup\mathrm{T}_2\cup\mathrm{T}_3$
when $v=1/2$), where $C_{1/2}$ thereby consists of points that never reach the Taub points
via the Kanser circle map ${\cal K}$.
On the other hand, the heteroclinic chains with period 2, see Figure~\ref{FIG:minMAX},
which are in $C_v$ (and behaves like its `boundary') when $v>1/2$, converge to
the Taub points as $v\to 1/2$, which do not belong to $C_{1/2}$. This implies that the set
$C_v$ is not continuous with respect to the parameter $v$ at $v=1/2$. In other words,
the set $\lim_{v\to 1/2} C_v$ is different than $C_{1/2}$.\footnote{Recall that
$C_v$ is a Cantor set (closed, without isolated points and nowhere dense) for $v\in (1/2,1)$.
On the one hand, the only common feature the set $C_{1/2}$ possesses when compared to $C_v$ with
$v>1/2$ is that both sets have no isolated points, whereas $C_{1/2}$ is not closed,
nor nowhere dense, since the (countably many) pre-images of Taub points are removed from the Kasner circle.
On the other hand, the limiting set $\lim_{v\to 1/2} C_v$ is the whole Kasner circle, and
is thereby closed, but it has no isolated points and is dense.}
Thus one should not expect that the encoding map $h$, given by~\eqref{defofh}, is
continuous (in $v$) at $v=1/2$. To deal with this discrepancy, and guarantee an accurate
continuous transition of non-generic to generic chaos, we also have to encode the Taub
points (and their pre-images) in the limit $v=1/2$, which are the limits of the heteroclinic
chains with period 2 when $v>1/2$.

Two problems arise when trying to encode the Taub points when $v=1/2$.
Consider $(\alpha\beta\gamma)$ a permutation of $(123)$.
First, each Taub point lies in two different arcs,
$\mathrm{T}_\alpha\in \mathrm{A}_\beta\cap \mathrm{A}_{\gamma}$,
and could thereby be described by two different symbols, $\beta$ or $\gamma$,
which would make the encoding map ill-defined.
Second, it is not clear how to encode each Taub point in order to obtain infinite
sequences in $W_\infty$, since it is possible to assign different infinite tails
to the finite heteroclinic chains that end at the Taub points.

To resolve these problems, recall that each heteroclinic chain
with period 2 (where the sequence of points in the set $\mathrm{K}^\ocircle$
of the heteroclinic chain is encoded by $\beta\gamma\beta\gamma\dots$
or $\gamma\beta\gamma\beta\dots$, in Figure~\ref{FIG:minMAX}) converges to the Taub
point $\mathrm{T}_\alpha$ when $v\to 1/2$. In order to guarantee a continuous limit,
it is natural that \emph{both} infinite sequences,
$\overline{\beta\gamma} := \beta\gamma\beta\gamma\dots$
and $\overline{\gamma\beta} := \gamma\beta\gamma\beta\dots$,
encode the Taub point $\mathrm{T}_\alpha$ at $v= 1/2$.
To ensure that the encoding map is well-defined, each Taub point should be encoded
by a single infinite sequence of symbols, and thus the two periodic sequences given
by $\overline{\beta\gamma}$ and $\overline{\gamma\beta}$ will be considered to be
in the same equivalence class in the space of infinite words $W_\infty$.
This assures that each chain with period 2 is encoded by a single infinite sequence
for $v>1/2$, and that each Taub point is encoded by the same sequence when $v=1/2$,
which results in continuity in $v$.

More precisely, we define an equivalence relation $\sim $ in $W_\infty$ as follows.
Two sequences $(a_k)_{k\in \mathbb{N}_0}$ and $(b_k)_{k\in \mathbb{N}_0}$ in $W_\infty$ are
equivalent $(a_k)_{k\in \mathbb{N}_0} \sim  (b_k)_{k\in \mathbb{N}_0}$ if, and only if
there is an $n\in \mathbb{N}_0$ such that $a_k=b_k$ for all $k=0,\ldots, n-1$
with $(a_k)_{k\geq n}= \overline{\beta\gamma}$ and $(b_k)_{k\geq n}= \overline{\gamma\beta}$
for some $\beta\neq\gamma\in \{1,2,3\}$. We then consider the quotient space endowed with the
quotient topology
\begin{equation}\label{Wtilde}
\tilde{W}_\infty := {W_\infty}/ \sim
\end{equation}
whose elements are the equivalence classes of sequences $(a_k)_{k\in
\mathbb{N}_0}$, denoted by $\left[ (a_k)_{k\in \mathbb{N}_0}\right]$.
An equivalence class thereby contains two or one element(s), if the
tail of $(a_k)_{k\in \mathbb{N}}$ is with period 2 or
not, respectively.

The above construction solves the issue of encoding the Taub points,
but it introduces a new problem for the heteroclinic chains with period 2:
for $v>1/2$ the two distinct points of the heteroclinic chain
with period 2 given by $\overline{\beta\gamma}$ and $\overline{\gamma\beta}$
have the same encoding in $\tilde{W}_\infty$, since
$\left[\overline{\beta\gamma}\right]=\left[\overline{\gamma\beta}\right]$.
To guarantee injectivity of the encoding map, we must relate
these two points by means of another equivalence relation. When $v>1/2$
we consider the quotient space
\begin{equation}\label{Ctilde}
\tilde{C}_v:=C_v / \sim,
\end{equation}
where two points $p,q\in \tilde{C}_v$ are equivalent
if, and only if, they are contained in the same heteroclinic
chain with period 2 of the Kasner circle
map ${\cal K}$.
To summarize: elements of $\tilde{C}_v$ are the equivalence classes of points $p$,
denoted by $\left[p\right]$, where an equivalence class contains two or one element(s)
if $p$ has period 2 or not, respectively.

The \emph{encoding map} is now defined as
\begin{equation}
\begin{array}{llll}
\tilde{h}: &D(\tilde{h}) &\rightarrow& \tilde{W}_\infty\\
& [p] & \mapsto & \tilde{h}([p]):=\left[(a_k)_{k \in \mathbb{N}}\right],
\end{array}
\end{equation}
where the domain is the set
\begin{equation}
D(\tilde{h}) = \begin{cases}
\tilde{C}_v & \text{ if } v>1/2,\\
\displaystyle\lim_{v\to 1/2} \tilde{C}_v  & \text{ if } v=1/2,
\end{cases}
\end{equation}
and where the sequence $(a_k)_{k \in \mathbb{N}_0}$ is built as follows:
\begin{enumerate}
\item If $\mathcal{K}^k(p)\neq \mathrm{T}_\alpha$ for all $k\in\mathbb{N}_0$,
then the symbol $a_k$ is uniquely defined for all $k\in\mathbb{N}_0$ as the
index of the open arc where $\mathcal{K}^k(p)$ lies, i.e. $\mathcal{K}^k(p)\in \mathrm{int}(A_{a_k})$.
\item If $\mathcal{K}^n(p)= \mathrm{T}_\alpha$ for some $n\in\mathbb{N}_0$,
where $n$ is the minimum of such values, let
\begin{equation}\label{hTilde}
h([p]):=\left[a_0\ldots a_{n-1} \overline{\beta\gamma}\right],
\end{equation}
where $(\alpha\beta\gamma)$ is a permutation of $(123)$.
\end{enumerate}

For $v>1/2$ the domain is the previously constructed Cantor set $C_v=C$,
which does not contain the Taub points, and thus case 2 above never happens.
Moreover, the  encoding $\tilde{h}$ coincides with $h$ in~\eqref{defofh}, except for
the heteroclinic chains with period 2: they consist of two distinct points
in $C_v$ which are identified in $\tilde{C}_v$ by the equivalence relation
in~\eqref{Ctilde}, and their two encodings in $W_\infty$ are identified in
$\tilde{W}_\infty$ by the equivalence relation in~\eqref{Wtilde} (for example
$\left[\overline{\beta\gamma}\right]=\left[\overline{\gamma\beta}\right]$).
Furthermore, it is only for $v=1/2$ that the Taub points can be reached, so that
case 2 above occurs, where a period 2 tail has been
added in order to obtain a continuous limit.

Hence, $\tilde{h}$ is a well-defined homeomorphism in the following commuting diagram,
\\
\centerline{
\xymatrix{
D(\tilde{h})  \ar[r]^{\mathcal{K}} \ar[d]_{\tilde{h}}  & D(\tilde{h})  \ar[d]^{\tilde{h}} \\
\tilde{W}_\infty \ar[r]_\sigma & \tilde{W}_\infty }}
where $\sigma$ is the shift to the right of sequences, which thereby establishes that
the map $\mathcal{K}$ is chaotic.

\end{appendix}

\subsection*{Acknowledgements}
JH was supported by the DFG collaborative research center SFB647 Space, Time,
Matter. PL was supported by FAPESP, 17/07882-0, 18/18703-1,
in addition to the encouraging, enduring and enlightening discussions
with Bernold Fiedler and Hauke Sprink. CU would like to thank the
Institut f\"ur Mathematik at Freie Universit\"at in Berlin for kind
hospitality.

\bibliographystyle{plain}

\end{document}